\documentclass[a4paper]{article}
\usepackage{hyperref}
\usepackage{graphicx}
\usepackage{mathrsfs}
\usepackage{enumerate}
\usepackage{amsxtra,amssymb,latexsym, amscd,amsthm,amsmath}
\usepackage{indentfirst}
\usepackage{color}
\usepackage[utf8]{inputenc}
\usepackage[mathscr]{eucal}
\usepackage{amsfonts}
\usepackage{graphics}
\usepackage{multirow}
\usepackage{array}
\usepackage{subfigure}
\usepackage{cite}
\usepackage{wrapfig}
\usepackage{tikz}
\usepackage{diagbox}
\usepackage{fullpage}

\usepackage{pgfplots}
\pgfplotsset{compat=1.15}
\usepackage{mathrsfs}
\usetikzlibrary{arrows}

\def\R{{\mathbb R}}

\def\N{{\mathbb N}}

\DeclareMathOperator{\diag}{diag}

\DeclareMathOperator{\sign}{sign}

\DeclareMathOperator{\supp}{supp}

\DeclareMathOperator{\trace}{trace}
\DeclareMathOperator{\tr}{tr}

\newtheorem{theorem}{\bf Theorem}%
\newtheorem{lemma}{\bf Lemma}
\newtheorem{algorithm}{\bf Algorithm}%
\newtheorem{example}{\bf Example}%
\newtheorem{proposition}{\bf Proposition}%
\newtheorem{corollary}{\bf Corollary}%
\newtheorem{definition}{\bf Definition}%
\newtheorem{remark}{\bf Remark}%
\newtheorem{assumption}{\bf Assumption}%

\providecommand{\keywords}[1]
{
  \textbf{\textbf{Keywords:}} #1
}
\definecolor{qqzzff}{rgb}{0,0.6,1}
\definecolor{ududff}{rgb}{0.30196078431372547,0.30196078431372547,1}
\definecolor{xdxdff}{rgb}{0.49019607843137253,0.49019607843137253,1}
\definecolor{ffzzqq}{rgb}{1,0.6,0}
\definecolor{qqzzqq}{rgb}{0,0.6,0}
\definecolor{ffqqqq}{rgb}{1,0,0}
\definecolor{uuuuuu}{rgb}{0.26666666666666666,0.26666666666666666,0.26666666666666666}

\newif\ifcomment
\commentfalse
\commenttrue

\title{Quantum speed-ups for solving semidefinite relaxations of polynomial optimization}

\author{
  Daniel Stilck França\textsuperscript{1,2}
  \and
  Ngoc Hoang Anh Mai\textsuperscript{2,3}
}

\date{} %
\begin{document}
\maketitle

\footnotetext[1]{Department of Mathematical Sciences, University of Copenhagen, 2100 Copenhagen, Denmark. Email: \texttt{dsfranca@math.ku.dk}.}
\footnotetext[2]{ENS de Lyon, Inria, 46 Allée d'Italie, 69364 Lyon Cedex 07, France.}
\footnotetext[3]{Institute of Mathematics, Vietnam Academy of Science and Technology, 18 Hoang Quoc Viet, Cau Giay, Hanoi, Vietnam. Email: \texttt{mnhanh@math.ac.vn}.}

\begin{abstract}
We study quantum algorithms for approximating Lasserre's hierarchy values for polynomial optimization. Let $f,g_1,\ldots,g_m$ be real polynomials in $n$ variables and $f^\star$ the infimum of $f$ over the semialgebraic set $S(g)=\{x: g_i(x)\ge 0\}$. Let $\lambda_k$ be the value of the order-$k$ Lasserre relaxation. Assume either (i) $f^\star=\lambda_k$ and the optimum is attained in the $\ell_1$-ball of radius $1/2$, or (ii) $S(g)$ lies in the simplex $\{x\ge 0: \sum_j x_j\le 1/2\}$, and the constraints define this simplex. After an appropriate coefficient rescaling, we give a quantum algorithm based on matrix multiplicative weights that approximates $\lambda_k$ to accuracy $\varepsilon>0$ with runtime, for fixed $k$,
\[
O(n^k\varepsilon^{-4}+n^{k/2}\varepsilon^{-5}),\qquad
O\!\left(s_g\!\left[n^k\varepsilon^{-4}+\!\left(n^{k}+\!\sum_{i=1}^m n^{k-d_i}\right)^{1/2}\!\varepsilon^{-5}\right]\right),
\]
where $s_g$ bounds the sparsity of the coefficient-matching matrices associated with the constraints. Classical matrix multiplicative-weights methods scale as $O(n^{3k}\mathrm{poly}(1/\varepsilon))$ even in the unconstrained case. As an example, we obtain an $O(n\varepsilon^{-4}+\sqrt{n}\varepsilon^{-5})$ quantum algorithm for portfolio optimization, improving over the classical $O(n^{\omega+1}\log(1/\varepsilon))$ bound with $\omega\approx2.373$.

Our approach builds on and sharpens the analysis of Apeldoorn and Gily\'en for the SDPs arising in polynomial optimization. We also show how to implement the required block encodings without QRAM. Under the stated assumptions, our method achieves a super-quadratic speedup in the problem dimension for computing Lasserre relaxations.
\end{abstract}
\keywords{polynomial optimization, moment, sum of squares, Lasserre's hierarchy, semidefinite programming, matrix multiplicative weights, quantum computing, quantum optimization, quantum algorithm, block encoding, portfolio optimization}
\tableofcontents
\section{Introduction}

Semidefinite programming (SDP) plays a central role in modern optimization, appearing in diverse areas such as combinatorial optimization~\cite{helmberg2000semidefinite}, quantum information theory~\cite{brandao2017quantum}, and polynomial optimization~\cite{lasserre2001global,laurent2009sums}. Although SDPs are convex, they often serve as tractable relaxations of challenging non-convex problems, making them indispensable tools for both theoretical analysis and practical computation.

Recent developments in quantum algorithms have shown that SDPs can, in principle, be solved with polynomial speedups over the best-known classical methods~\cite{brandao2017quantum,van2018improvements,augustino2023quantum,van2020convex,yuan2025exponentialspeedupsstructuredgoemanswilliamson}. Many of these quantum algorithms, including those based on the multiplicative weights update method~\cite{arora2012multiplicative,arora2005fast,kale2007efficient}, achieve this advantage under specific conditions on the input instances. 
In particular, for quantum algorithms rooted in the matrix multiplicative weights (MMW) framework, it is crucial that the SDPs are presented in standard form, are properly normalized  and we have an upper bound on the trace of a solution. 
This requires bounding the operator norm of all constraint matrices, as well as establishing a known bound on the trace of an optimal solution. Furthermore, these quantum solvers often exhibit a suboptimal dependence on the accuracy parameter, so demonstrating that low-accuracy solutions are sufficient for a given application is essential to realizing genuine quantum advantages. Finally, most quantum speedups also need that the matrices defining the constraints and the objective function are highly sparse and we can easily obtain block encodings for them. Analysing these aspects of the quantum solvers is crucial, as some works indicate that such details can cause quantum algorithms for SDPs to only offer advantages at prohibitively large scales~\cite{2502.15426}.

In this work, we investigate these normalization, stability and structural requirements for a broad family of SDPs that arise as sums-of-squares (SOS)~\cite{lasserre2001global} relaxations from both unconstrained and inequality constrained polynomial optimization problems (POPs). Our contributions are twofold. First, we show how to transform these SDPs into sparse standard form with proper normalization, deriving bounds on relevant matrix quantities along the way. Second, we provide stability estimates for both constrained and unconstrained variants of these problems.
As such, we identify regimes for which the quantum algorithms offer a super-quadratic speedup compared to classical algorithms, but at the cost of a polynomial scaling in a precision parameter $\epsilon^{-1}$. 
In particular, we show how to obtain stability results for the POPs in terms of $\ell_1$ norm estimates of solutions assuming the hierarchy converged, making optimization problems over the simplex natural candidates for our techniques.
As an application, we show how to obtain super-quadratic speedups to solve portfolio optimization problems, which is an optimization problem over the simplex.

The remainder of the paper is organized as follows. Section~\ref{sec:background.contribution} reviews the necessary background on semidefinite programming, polynomial optimization, and the framework for achieving quantum speedups, along with a summary of our main contributions. In Section~\ref{sec:pre}, we introduce the notation, definitions, and preliminary results used throughout the paper, including discussions on strong duality, the exactness of SDP relaxations for POPs, the correctness of Hamiltonian updates, and the accuracy of approximate values obtained via their combination with binary search. Section~\ref{sec:solve.sos.relax} describes how quantum SDP solvers can be employed to solve SOS relaxations for specific classes of POPs, and demonstrates their application to the portfolio optimization problem. Finally, Section~\ref{sec:conclusion} concludes with a discussion of potential extensions and open research directions.

\section{Background and Summary of results}
\label{sec:background.contribution}
\subsection{Semidefinite programming}
Consider the following standard form of a semidefinite program (SDP):
\begin{equation}\label{eq:sdp.intro}
\begin{array}{rl}
\lambda^\star=\inf\limits_{X,z}&\tr( CX)\\
\text{s.t.}&X=\diag(X_1,\dots,X_r)\in\prod_{j=1}^r\mathbb S^{N_j}_+\,,\,z\in\R^t\,,\\
&\tr( A_iX)+q_i^\top z =b_i\,,\,i=1,\dots,M\,,
\end{array}
\end{equation}
where $C,A_i\in \prod_{j=1}^r\mathbb S^{N_j}$, $q_i\in\R^t$, and $b_i\in\R$ are given data.
Here, $\mathbb S^{K}$ (resp. $\mathbb S^{K}_+$) denotes the space of real symmetric matrices (resp. cone of positive semidefinite matrices) of size $K$, and $N=\max\{N_1,\dots,N_r\}$ is the size of the largest block.

Although problem~\eqref{eq:sdp.intro} is convex, it is crucial for tackling non-convex problems commonly arising in practice in various areas, such as combinatorial optimization~\cite{goemans1995improved,parrilo2000structured}, quantum information theory~\cite{doherty2012convergence,barthel2012solving}, and polynomial optimization~\cite{parrilo2000structured,helton2012semidefinite,nie2013exact,lasserre2009convex}, among others. Consequently, the development of efficient algorithms for solving semidefinite programs (SDPs) is of significant importance. Existing methods for solving SDPs include interior-point methods~\cite{vandenberghe1996semidefinite,jiang2020faster,dahl2012semidefinite}, Newton-type/IPM variants~\cite{jiang2020faster}, bundle methods~\cite{helmberg2000spectral,helmberg2014spectral}, and augmented Lagrangian frameworks~\cite{yurtsever2019conditional}. 
From a complexity and practicality standpoint, many classical SDP solvers suffer from high polynomial dependence on the matrix dimension and accuracy, and often require substantial memory due to storing dense iterates/factorizations; first-order or low-rank variants mitigate memory to some extent but typically weaken guarantees. There are also quantum proposals, based either on multiplicative weights~\cite{arora2012multiplicative,brandao2017quantum,van2018improvements,lilecture2018,GSLBrando2022} or on quantum interior-point techniques~\cite{augustino2023quantum}, which promise polynomial speedups under suitable normalization and input-access assumptions but inherit sensitivity to accuracy/conditioning and the need for careful instance scaling.

This paper focuses on the complexity and accuracy of SDP-solving algorithms based on matrix multiplicative weights solvers (MMW), both classical and quantum, in the context of polynomial optimization problems.

When $r=1$ and $t=0$, the standard SDP \eqref{eq:sdp.intro} simplifies to:
\begin{equation}\label{eq:primal.intro}
\begin{array}{rl}
\lambda^\star=\inf\limits_{X}&\tr( CX)\\
\text{s.t.}&X\in\mathbb S^{N}_+\,,\,\tr(A_i X) =b_i\,,\,i=1,\dots,M\,.
\end{array}
\end{equation}
Its dual is given by:
\begin{equation}\label{eq:dual.intro}
\begin{array}{rl}
\tau^\star=\sup\limits_{\xi}&b^\top \xi\\
\text{s.t.}&\xi\in\R^M\,,\, C-\sum_{i=1}^M\xi_iA_i\succeq 0\,.
\end{array}
\end{equation}

The interior-point method introduced by Helmberg et al. \cite{helmberg1996interior} solves \eqref{eq:primal.intro}-\eqref{eq:dual.intro} using second-order information to approximate primal-dual solutions.
According to Jiang et al. \cite{jiang2020faster}, its time complexity on a classical computer is: 
\begin{equation}
O\left( N^{1/2} \left(M N^2 + M^{\omega} + N^{\omega}\right)\log\left(\varepsilon^{-1}\right) \right)\,,
\end{equation}
where $\omega\in [2,2.38]$.
Augustino et al. \cite{augustino2023quantum} show that the complexity of solving \eqref{eq:primal.intro} using quantum interior-point method is:
\begin{equation}
O\left( N^{5.5} \varepsilon^{-1}\ \text{polylog}\left(\varepsilon^{-1}\right) \right)\,.
\end{equation}
If $M=O(N^2)$, the complexity of quantum interior-point method reduces to:
\begin{equation}
O\left( N^{3.5} \varepsilon^{-1}\ \text{polylog}\left(\varepsilon^{-1}\right) \right)
\end{equation}
As shown in Lemma \ref{lem:accuracy.interior-point},  the interior-point methods above provide approximate values for the optimal values $\lambda^\star$ and $\tau^\star$ of the semidefinite programs \eqref{eq:sdp.intro}-\eqref{eq:dual.intro}, with an accuracy of 
\begin{equation}
\varepsilon N\,.
\end{equation}

Arora and Kale \cite{arora2012multiplicative} proposed the matrix multiplication weight (MMW) update algorithm, which is paired with binary search to solve SDP \eqref{eq:primal.intro} in a running time of
\begin{equation}
O((MsN+N^\omega)\varepsilon^{-2})\,,
\end{equation}
where $\omega\approx 2.373$ is the exponent of matrix multiplication and $s$ represents the maximum number of nonzero entries per row in $A_i$s and $C$ (see Lemma \ref{lem:MMWU-complexity}).
Unlike the interior-point method, which determines a search direction by solving a linear system of equations to update the iterates, variations of MMW like the Hamiltonian updates algorithm~\cite{GSLBrando2022} adjusts the iterates by verifying infeasibility across affine constraints relative to the parameter $\varepsilon$, without relying on solving linear systems. 

Quantum algorithms for SDPs based on MMW were introduced by Brand\~ao and Svore~\cite{brandao2017quantum}, with further optimizations developed in~\cite{van2018improvements,GSLBrando2022}. The key to their super-quadratic speedup lies in two ingredients:  
(i)~\emph{quantum Gibbs sampling} subroutines, which allow efficient preparation of quantum states proportional to $\exp(-H)$ for relevant matrix inputs $H$, thereby reducing the dependence on the matrix dimension from $N$ to $\sqrt{N}$ in the runtime; and  
(ii)~\emph{gentle measurement} techniques, which enable constraint verification while disturbing the quantum state only minimally, leading to a $\sqrt{M}$ scaling in the number of constraints.  
Both ingredients rely on the ability to implement \emph{block encodings} of the relevant matrices $A_i$ and $C$, which can be nontrivial to implement in general~\cite{Camps2024}. Thus, one of our contributions is to construct efficient block encodings for the relevant matrices for our SDPs.

The running time of the quantum Hamiltonian updates algorithm for solving the SDP \eqref{eq:primal.intro} is now:
\begin{equation}\label{eq:best.quantum.complexity.SDP}
O\left(s(\sqrt{M} + \sqrt{N}\varepsilon^{-1})\varepsilon^{-4}\right),
\end{equation}
where $s$ represents the maximum number of nonzero entries per row in $A_i$s and $C$.

As shown in the work of Apeldoorn and Gily\'en \cite{van2018improvements}, their complexities \eqref{eq:best.quantum.complexity.SDP} requires a constant upper bound on the trace of the optimal solution $X^\star$ to the primal SDP \eqref{eq:primal.intro}, which can be replaced by a constant trace condition for $X^\star$ by embedding the input matrices into larger ones).

Although the polynomial dependence on $\varepsilon^{-1}$  makes Hamiltonian updates exponentially slower than classical interior-point methods in terms of precision, the complexity of Hamiltonian updates is of a lower order concerning problem size---specifically, the number of affine constraints $M$, the sparsity $s$ and the size of the variable matrix $N$---compared to classical interior-point solvers, albeit assuming access to block encodings.

In Lemma \ref{lem:convergence.sdp}, we prove that under strong duality for the problems \eqref{eq:primal.intro}-\eqref{eq:dual.intro}, the accuracy of the approximation for the optimal value $\lambda^\star$  of the primal SDP \eqref{eq:primal.intro}, returned by Hamiltonian updates, is
\begin{equation}
\varepsilon(1+\|\xi^\star\|_{\ell_1})\,,
\end{equation}
where $\xi^\star$ is an optimal solution to the dual SDP \eqref{eq:dual.intro}. In Table \ref{tab:summary.SDP} we summarize the scaling of various SDP solvers.

\begin{table}
\footnotesize
\caption{Summary of classical and quantum algorithms based on Interior-point method and Hamiltonian updates for solving semidefinite programs (SDPs) \eqref{eq:primal.intro}-\eqref{eq:dual.intro} with matrix variable of size $N$ and $M$ affine equality constraints. Here $X^\star$ is some optimal solution to SDP \eqref{eq:primal.intro}; $\xi^\star$ is some optimal solution to SDP \eqref{eq:dual.intro}; and $s$ is the largest number of nonzero entries on each row of $C,A_i$.}
\label{tab:summary.SDP}
\begin{center}
\begin{tabular}{ |p{1.5cm}|p{1.1cm}|p{2.2cm}|p{5.2cm}|p{3cm}| } 
\hline
\multicolumn{2}{|c|}{Method}  &Conditions  & Complexity & Accuracy for approximate optimal value  \\
\hline
Interior-point methods 

(Helmberg et al. \cite{helmberg1996interior})& Classical & Slater's conditions  & $O\left( N^{1/2} \left(M N^2 + M^{\omega} + N^{\omega}\right)\log\left(\varepsilon^{-1}\right) \right)$, $\omega\in [2,2.38]$ 

(Jiang et al. \cite{jiang2020faster})& $\varepsilon N$   

(Lemma \ref{lem:accuracy.interior-point})\\\cline{2-2} \cline{4-4}
 & Quantum &   & $O\left( N^{5.5} \varepsilon^{-1}\ \text{polylog}\left(\varepsilon^{-1}\right) \right)$ 
 
 (Augustino et al. \cite{augustino2023quantum})&   \\\cline{3-4} 
 &  & Slater's conditions
 
 $M=O(N^2)$  & $O\left( N^{3.5} \varepsilon^{-1}\ \text{polylog}\left(\varepsilon^{-1}\right) \right)$ 
 
 (Augustino et al. \cite{augustino2023quantum})&   \\
 \hline
Hamiltonian updates 

(Matrix multiplicative weight update)

(Arora and Kale \cite{arora2012multiplicative})
& Classical & $\tr(X^\star)=1$,

$-I_N\preceq C\preceq I_N$,

 $-I_N\preceq A_i \preceq I_N$  & $O((MsN+N^\omega)\varepsilon^{-2})$ 
 
 (Lemma \ref{lem:MMWU-complexity})& $\varepsilon(1+\|\xi^\star\|_{\ell_1})$
 
 (Lemma \ref{lem:convergence.sdp})   \\ \cline{2-2} \cline{4-4}
 & Quantum &  & $O\left(s(\sqrt{M} + \sqrt{N}\varepsilon^{-1})\varepsilon^{-4}\right)$ 
 
 (Apeldoorn and Gily\'en \cite{van2018improvements})&   \\
\hline
\end{tabular}
\end{center}
\end{table}

In this paper, we apply the Hamiltonian update algorithm to estimate the optimal values of the semidefinite relaxations for the polynomial optimization problem presented in the next section. The key idea behind our main results is that we convert the semidefinite relaxations for the polynomial optimization problem to standard form \eqref{eq:primal.intro} that satisfy the scaling condition on the input data as well as the constant trace condition for $X^\star$, then we compute the sparsity parameter $s$ for the complexity and   compute upper bound for the quantity
 $\|\xi^\star\|_{\ell_1}$ for the accuracy, identifying certain regimes that admit quantum speedups.

 \subsection{Polynomial optimization}
Consider the general form of polynomial optimization problem (POPs) involving $n$ variables, $m$ inequality constraints, and $l$ equality constraints:
\begin{equation}\label{eq:pop.general}
\begin{array}{rl}
f^\star=\min\limits_{x\in\R^n}& f(x)\\
\text{s.t.}&g_i(x)\ge 0\,,\,i=1,\dots,m\,,\\
&h_j(x)= 0\,,\,j=1,\dots,l\,,\\
\end{array}
\end{equation}
where $f,g_i$, and $h_j$ are given polynomials with real coefficients and degrees at most $2d,2d_i$, and $2w_j$, respectively. We also set $d_0=0$, which is the degree of an external constraint $g_0=1$ used later.

POPs of the form \eqref{eq:pop.general} have numerous significant practical applications, including optimal power flow~\cite{wang2020cs,magron2023sparse}, optimal control~\cite{henrion2020}, machine learning~\cite{lasserre2009moments}, and data analysis~\cite{lasserre2010moments}. However, solving \eqref{eq:pop.general} is highly challenging: it is typically non-convex and, in general, NP-hard~\cite{lasserre2015introduction}.

The constraint set for the problem in \eqref{eq:pop.general}, denoted by $S(g,h)$, is a semialgebraic set defined as:
\begin{equation}
S(g,h)=\{x\in\R^n\,:\,g_i(x)\ge 0\,,\,i=1,\dots,m\,,\,h_j(x)= 0\,,\,j=1,\dots,l\}\,,
\end{equation}
where $g=(g_i)_{i=1}^m\subset \R[x]$ and $h=(h_j)_{j=1}^l\subset \R[x]$.
Here, $\R[x]_t=\R[x_1,\dots,x_n]_d$ denotes the ring of polynomials in $x$ of degree at most $t$. Note that $x$ is a vector of $n$ scalar variables, i.e.,
\begin{equation}
x=(x_1,\dots,x_n)\,.
\end{equation}

An important tool to solve optimization problems like \eqref{eq:pop.general} are sums-of-squares polynomials:
\begin{definition}[Sum--of--squares (SOS) cone]\label{def:sos-cone}
Fix integers $n\ge 1$ (number of variables) and $d\ge 0$ (half-degree).
The \emph{SOS cone of degree $2d$} in $\R[x]$ is
\[
  \Sigma[x]_d
  \;:=\;
  \Bigl\{
    p\in\R[x]_{2d}
    \;\Bigm|\;
    \exists\;q_1,\dots,q_s\in\R[x]_{d}
    \text{ such that }
    p(x)=\sum_{i=1}^{s} q_i(x)^2.
  \Bigr\},
\]
\end{definition}

It is well-known~\cite{lasserre2001global} that the sum-of-squares relaxation:
\begin{equation}\label{eq:sos.general.general.pop}
\begin{array}{rl}
\lambda_k=\sup\limits_{\lambda,\sigma_i,\psi_j} & \lambda\\
\text{s.t.}& \lambda\in\R\,,\,\sigma_i\in\Sigma[x]_{k-d_i}\,,\,\psi_j\in\R[x]_{2(k-w_j)}\,,\\
&f-\lambda=\sigma_0+\sum_{i=1}^m\sigma_i g_i+\sum_{j=1}^l\psi_jh_j\,.
\end{array}
\end{equation}
produces a sequence of non-decreasing lower bounds $(\lambda_k)_{k=0}^\infty$ for the optimal value $f^\star$.

Before we further explore the structure of the problem \eqref{eq:sos.general.general.pop}, we introduce some further standard definitions for moment/sum-of-squares relaxations:
\begin{definition}[Multi-index set]\label{def:multiindex}
For non-negative integers \(k,n\) set
\[
  \N_{2k}^{\,n}\;:=\;
  \bigl\{\alpha=(\alpha_1,\dots,\alpha_n)\in\N^{n}\;:\;
          |\alpha|:=\alpha_1+\dots+\alpha_n\le 2k\bigr\}.
\]
It indexes all monomials \(x^{\alpha}=x_1^{\alpha_1}\cdots x_n^{\alpha_n}\)
of total degree at most \(2k\).
\end{definition}

\begin{definition}[Truncated moment sequence]\label{def:moment-seq}
A vector \(y=(y_{\alpha})_{\alpha\in\N_{2k}^{\,n}}\subset\R\)
is called a \emph{truncated moment sequence of order \(2k\)} with respect to a positive measure \(\mu\) if
\begin{equation}
y_{\alpha}=\displaystyle\int_{\R^{n}}x^{\alpha}\,d\mu(x)\,.
\end{equation}
\end{definition}

\begin{definition}[Riesz linear functional]\label{def:riesz}
Given a vector \(y=(y_{\alpha})_{\alpha\in\N_{2k}^{\,n}}\subset\R\),
the \emph{Riesz functional}
\[
  L_{y}:\R[x]_{2k}\longrightarrow\R, \qquad
  L_{y}\!\Bigl(\sum_{\alpha}p_{\alpha}x^{\alpha}\Bigr)
  =\sum_{\alpha}p_{\alpha}\,y_{\alpha},
\]
acts by replacing each monomial coefficient \(x^{\alpha}\) with
its corresponding moment \(y_{\alpha}\).
\end{definition}

\begin{definition}[Moment matrix]\label{def:moment-matrix}
For an integer \(k\ge 0\),
the \emph{moment matrix of order \(k\)} associated with \(y\) is
\[
  M_{k}(y):=\bigl(y_{\alpha+\beta}\bigr)_{\alpha,\beta\in\N_{k}^{\,n}}.
\]
Rows and columns are indexed by monomials of degree \(\le t\).
\end{definition}

\begin{definition}[Localizing matrix]\label{def:localising}
Let \(p(x)=\sum_{\gamma}p_{\gamma}x^{\gamma}\) be a polynomial of degree at most
\(2u\) and fix \(k\ge u\).
The \emph{localizing matrix of order \(k-u\)} associated with \(p\) and \(y\) is
\[
  M_{k-u}(p\,y)
  :=\Bigl(\sum_{\gamma}p_{\gamma}\,y_{\alpha+\beta+\gamma}\Bigr)_
      {\alpha,\beta\in\N_{k-u}^{\,n}}.
\]
\end{definition}

The sum-of-squares relaxation \eqref{eq:sos.general.general.pop} can be equivalently formulated as the following semidefinite program:
\begin{equation}\label{eq:sos.sdp.general.general.pop}
\begin{array}{rl}
\lambda_k=\sup\limits_{\lambda,G_i,\eta_j} & \lambda\\
\text{s.t.}& \lambda\in\R\,,\,G_i\succeq 0\,,\,i=0,\dots,m\,,\,\eta_j\subset \R\,,\,j=1,\dots,l\,,\\
&f-\lambda=v_{k}^\top G_0v_k+\sum_{i=1}^m g_iv_{k-d_i}^\top G_i v_{k-d_i}+\sum_{j=1}^l h_j\eta_j^\top v_{2(k-w_j)}  \,,
\end{array}
\end{equation}
where $v_t$ is the vector of monomials up to degree $t$.

The dual of \eqref{eq:sos.sdp.general.general.pop}, known as moment relaxations, are formulated as:
\begin{equation}\label{eq:mom.relax.pop}
\begin{array}{rl}
\tau_k=\min\limits_y& L_y(f)\\
\text{s.t.}& y=(y_\alpha)_{\alpha\in\N^n_{2k}}\subset \R\,,\\
&M_k(y)\succeq 0\,,\,y_0=1\,,\\
&M_{k-d_i}(g_iy)\succeq 0\,,\,i=1,\dots,m\,,\\
&M_{k-w_j}(h_jy)=0\,,\,j=1,\dots,l\,,
\end{array}
\end{equation}
where $M_t(y)$ is the moment matrix of order $k$ associated with $y$, $M_{t-u}(py)$ is the localizing matrix of order $t-u$ associated with $y$ and a polynomial $p$ of degree at most $2u$, and $L_y$ is the Riesz linear functional on $\R[x]_{2t}$.

The formulations \eqref{eq:sos.sdp.general.general.pop}-\eqref{eq:mom.relax.pop} are primal-dual semidefinite programs originally introduced by Lasserre in \cite{lasserre2001global}.
When $g_1$ represents a ball constraint, i.e., $g_1=R-\|x\|_{\ell_2}^2$ for some $R>0$, Lasserre demonstrated that the sequence $(\lambda_k)_{k=0}^\infty$ converges monotonically to $f^\star$ from below\footnote{The convergence rate of the sequence $(\lambda_k)_{k=0}^\infty$ to $f^\star$  is discussed in detail in \cite{baldi2023effective,laurent2023effective}. In this work, however, our focus is on analyzing the computational complexity involved in solving each semidefinite relaxation.
}.
In this case, Josz and Henrion showed in \cite{josz2016strong} that $\lambda_k=\tau_k$ for all $k\ge \max\{d,d_1,\dots,d_m,w_1,\dots,w_l\}$. 
 
Thus, the optimal values of \eqref{eq:sos.sdp.general.general.pop}-\eqref{eq:mom.relax.pop}, $\lambda_k$ and $\tau_k$, serve as  lower approximations for the optimal value $f^\star$ of the polynomial optimization problem \eqref{eq:pop.general}.
These values enable us to determine whether a local optimal solution $\bar x \in S(g,h)$ for problem \eqref{eq:sos.sdp.general.general.pop}, obtained using classical optimization methods, is a global minimum or not. This is achieved by evaluating the gap $f(\bar x)-\lambda_k$.
Since $\lambda_k\le f ^\star\le f(\bar x)$, a zero gap $f(\bar x)-\lambda_k=0$ implies that $f(\bar x)=f^\star$, confirming global optimality. Furthermore, if the moment matrix satisfies certain rank constraints and the hierarchy has converged, it is possible to extract the solutions from it~\cite{henrion2005detecting}.

However, solving \eqref{eq:sos.sdp.general.general.pop}-\eqref{eq:mom.relax.pop} to compute $\lambda_k$ and $\tau_k$ becomes increasingly challenging because their size grows rapidly as the relaxation order $k$ increases.
To illustrate this, we reformulate problems \eqref{eq:sos.sdp.general.general.pop} and \eqref{eq:mom.relax.pop} in the standard form of a semidefinite program \eqref{eq:sdp.intro}.

For the sum-of-squares relaxation \eqref{eq:sdp.intro}, the corresponding semidefinite program \eqref{eq:sdp.intro} has:
$M=\binom{n+2k}n$ affine equality constraints, 
\begin{equation}
t=\sum_{j=1}^l \binom{n+2(k-w_j)}n=O\left(l \binom{n+2k}n\right)
\end{equation}
scalar variables, and $r=m+1$ positive semidefinite matrices, with the largest size $N=\binom{n+k}n$.

For the moment relaxation \eqref{eq:mom.relax.pop}, the corresponding semidefinite program \eqref{eq:sdp.intro} has:
\begin{equation}
\begin{array}{rl}
M=&1+\sum_{j=1}^l \binom{n+2(k-w_j)}n+\frac{1}{2}\binom{n+k}n[\binom{n+k}n+1]-\binom{n+2k}n\\[5pt]
&+\frac{1}2\sum_{i=1}^m \binom{n+k-d_i}n[\binom{n+k-d_i}n+1]\\[5pt]
=& O\left(l\binom{n+2k}n+m\binom{n+k}n^2\right)
\end{array}
\end{equation}
affine equality constraints, $t=0$ scalar variables and $r=m+1$ positive semidefinite matrices, with the largest size $N=\binom{n+k}n$.

These complexities highlight the exponential growth in size with respect to number of variables $n$ of the polynomial optimization \eqref{eq:pop.general} and the relaxation order $k$ of the semidefinite relaxations \eqref{eq:primal.intro}-\eqref{eq:dual.intro}, making these relaxations computationally intensive to solve.

Therefore, we seek a method for solving the semidefinite relaxations \eqref{eq:primal.intro}-\eqref{eq:dual.intro} that maintains low computational complexity relative to the problem size and does not demand high precision. The method employed below fulfills these expectations.
\subsection{Block encodings}

To efficiently apply the Matrix Multiplicative Weights method/Hamiltonian updates (Algorithms \ref{alg:HU} and \ref{alg:Binary.search.HU}) for solving the SDP relaxation of the polynomial optimization problem \eqref{eq:pop.general}, we adopt the \emph{block-encoding} framework~\cite{Gilyn2019}.
This formalism provides an efficient mechanism for accessing large, structured matrices, thereby enabling significant algorithmic speed-ups in large-scale SDP solvers.

A block encoding represents a matrix $A$ as a normalized sub-block of a unitary operator $U$ such that
\[
U = 
\begin{pmatrix}
A/\alpha & * \\ * & *
\end{pmatrix},
\]
where $\alpha > 0$ is a known normalization factor.

In Appendix~\ref{sec:block.encoding}, we describe how to construct block encodings for the matrices $A_i$ obtained when converting the SOS relaxation~\eqref{eq:sos.sdp.general.general.pop} (with $l=0$) into the standard SDP form~\eqref{eq:primal.intro}.
This construction only requires oracle access to the coefficients of the polynomials $f$ and $g_i$. The resulting block encodings can be generated using $O(n \log k)$ quantum gates, ensuring that this step does not constitute a computational bottleneck for the overall method.

\subsection{Our contribution}
In this paper, we develop quantum algorithms for solving the semidefinite relaxations \eqref{eq:sos.sdp.general.general.pop}-\eqref{eq:mom.relax.pop} of the polynomial optimization problem \eqref{eq:pop.general}. Our approach utilizes binary search and Hamiltonian updates~\cite{GSLBrando2022}.

Our main contributions are:
\begin{itemize}
    \item We show that, after appropriate rescaling of the objective and constraints, the SOS relaxations of certain POPs admit a bounded-trace feasible region. This bound is key for analysing the scaling of the SDPs.
    \item We identify settings—both unconstrained and inequality-constrained POPs—where the SDP data is sparse (each row of $C$ and $A_i$ has few nonzero entries), enabling the Hamiltonian updates framework to offer runtime improvements over classical methods.
    \item For each setting, we prove runtime bounds and accuracy guarantees for our quantum algorithms based on various structural assumptions on the POP, such as the solution being in a certain $\ell_1$ ball or the hierarchy converging after a certain number of levels.
\end{itemize}

Leveraging these contributions we obtain the following results:
\begin{itemize}
\item \textbf{Case 1: Unconstrained polynomial optimization} (see Section \ref{sec:uncons.pop}).

Consider the unconstrained case of problem \eqref{eq:pop.general}, where $m=l=0$. 

If the optimal value of the sum-of-squares relaxation \eqref{eq:sos.sdp.general.general.pop} of some order $k$ matches the optimal value of the polynomial optimization problem \eqref{eq:pop.general}, i.e., $\lambda_k=f^\star$, and $f^\star$ is attained at a point within an $\ell_1$-ball of radius $r\in(0,1)$ centered at the origin, then our quantum Algorithm \ref{alg:uncon.pop} for solving the \eqref{eq:sos.sdp.general.general.pop} has a runtime of 
\begin{equation}\label{eq:complex.case1}
O\left(\left[\binom{n+2k}n^{1/2}+\frac{1}\varepsilon\binom{n+k}n^{1/2}\ \right]\frac{1}{\varepsilon^4}\right)
\end{equation}
to compute an approximate value of $\lambda_k$ with an accuracy of $\frac{\varepsilon }{1-r}$ (see Theorem \ref{theo:accuarcy.approximate.val.uncons}).

To obtain the complexity \eqref{eq:complex.case1}, we transform the sum-of-squares relaxation \eqref{eq:sos.sdp.general.general.pop} to standard semidefinite program \eqref{eq:primal.intro} and apply Apeldoorn--Gily\'en's result~\cite{van2018improvements}. 
Note that each row of the input matrices $C$ and $A_i$ in program \eqref{eq:primal.intro} has at most one nonzero entry.
In addition, the exactness assumption $\lambda_k=f^\star$ allows us to bound the $\ell_1$-norm of the dual solution of program \eqref{eq:primal.intro} via the bound $r$ on the $\ell_1$-norm of the optimal solution for the polynomial optimization problem \eqref{eq:pop.general}.

\if{
(ii) On the other hand, if Slater's condition holds for the sum-of-squares relaxation \eqref{eq:sos.sdp.general.general.pop}, meaning there exists a strictly feasible solution $G_0\succ 0$ with conditional number $\kappa $, Algorithm \ref{alg:uncon.pop} provides an approximate value for  $\lambda_k$ with an accuracy of the minimum of
\begin{equation}\label{eq:accuracy.case1}
\varepsilon\kappa  \binom{n+k}n \quad\text{ or }\quad\varepsilon\kappa \binom{n+2k}n^{1/2} 
\end{equation}
(see Theorem \ref{theo:accuarcy.approximate.val.uncons} (ii)-(iii)).
}\fi

\item \textbf{Case 2: Inequality-constrained polynomial optimization problem} (see Sections \ref{sec:cons.pop}, \ref{sec:ineq.cons.pop.simplex} and \ref{sec:app.portfolio}).

Now consider the inequality-constrained case of problem \eqref{eq:pop.general}, where $l=0$. 

(i) If $\lambda_k=f^\star$ and $f^\star$  is attained at a point within an $\ell_1$-ball of radius $r\in(0,1)$ centered at the origin, with a coefficient rescaling of the $g_i$s, our quantum Algorithm \ref{alg:con.pop.exact} for solving sum-of-squares relaxation \eqref{eq:sos.sdp.general.general.pop} requires
\begin{equation}\label{eq:runtime.inequality.pop}
O\left(s_g\left[\binom{n+2k}{n}^{1/2}+\left( \sum_{i=0}^m\binom{n+k-d_i}n\right)^{1/2}\frac{1}{\varepsilon}\right]\frac{1}{\varepsilon^4}\right)\,,
\end{equation}
to provide an approximate value for $\lambda_k$ with an accuracy of $\frac{\varepsilon }{1-r}$, where $s_g$ is the maximum number of nonzero entries per row of the \textit{coefficient-matching matrices} (defined as in \eqref{eq.tilde.B.gamma}) associated with $g_i$s (see Theorem \ref{theo:convergence.ineq.cons}).
Note that when the sum-of-squares relaxation \eqref{eq:sos.sdp.general.general.pop} is transformed to standard semidefinite program \eqref{eq:primal.intro}, each row of $C$ and $A_i$s has at most $s_g$ nonzero entries and $s_g$ is not larger than the maximum number of nonzero coefficients of each $g_i$.

Consequently, if $f^\star$  is attained at a point within an $\ell_1$-ball of radius $r\in(0,1)$ centered at the origin, and the coefficients of $g_i$s are rescaled appropriately, while the second-order sufficient optimality conditions hold for problem \eqref{eq:pop.general} at every optimal solution, then for sufficiently large $k$, Algorithm \ref{alg:con.pop.exact} yields an approximate value for $f^\star$ with an accuracy of $\frac{\varepsilon }{1-r}$ (see Corollary \ref{coro:sosoc}).

Instead of using the exactness assumption $\lambda_k=f^\star$, we can rely on the inequality constraints (ball and simplex) in the polynomial optimization problem \eqref{eq:pop.general} to bound the $\ell_1$-norm of the dual solution of the relaxation program \eqref{eq:primal.intro}.

(ii) Concretely, if the polynomial optimization problem \eqref{eq:pop.general} includes a ball constraint of the form $c(R-\|x\|_{\ell_2}^2)\ge 0$ with $R\in(0,1)$ and $c>0$, then after rescaling the coefficients of the $g_i$ constraints, Algorithm \ref{alg:con.pop.exact} produces an approximate value for $\lambda_k$  with an accuracy of $\frac{\varepsilon}{1-R}\binom{n+k}n$ (see Theorem \ref{theo:convergence.over.ball}).

(iii) If the constraint set $S(g,h)$ of the polynomial optimization problem \eqref{eq:pop.general} is contained within the simplex $\{x\in\R^n_+\,:\,\sum_{j=1}^n x_j\le r\}$ for some $r\in(0,1)$, meaning that the terms $x^\alpha$, for $\alpha\in\{0,1\}^n\backslash\{0\}$, and polynomials $c_j[r^j-(x_1+\dots+x_n)^j]$ with $c_j>0$, for $j=1,2$, are included among $g_1,\dots,g_m$, then with coefficient rescaling of the $g_i$s, our Algorithm \ref{alg:con.pop.exact} provides an approximate value for $\lambda_k$ with an accuracy of $\frac{\varepsilon}{1-r}$ (see Theorem \ref{theo:convergence.over.simplex}).

(iv) To avoid the exponential number of inequality constraints in $S(g,h)$ in the last result, we define a new semidefinite program \eqref{eq:sdp.relaxation.cons.pop.ineq.aff} with optimal value $\hat \lambda_k$ by relaxing equality affine constraints in sum-of-squares relaxation \eqref{eq:sos.sdp.general.general.pop}.
When the constraint set $S(g,h)$ is contained within the nonnegative orthant $\R^n_+$, it holds that $\lambda_k\le \hat \lambda_k\le f^\star$.
If $S(g,h)$ is contained within the simplex $\{x\in\R^n_+\,:\,\sum_{j=1}^n x_j\le r\}$ for some $r\in(0,1)$ and polynomials $c_j[r^j-(x_1+\dots+x_n)^j]$ with $c_j>0$, for $j=1,2$, are included among $g_1,\dots,g_m$, our quantum Algorithm \ref{alg:con.pop.exact.ineq.aff} provides an approximate value for $\hat \lambda_k$ with an accuracy of $\frac{\varepsilon}{1-r}$ and runtimes \eqref{eq:runtime.inequality.pop} (see Theorem \ref{theo:convergence.over.simplex.ineq.aff}).

(v) By using our Algorithm~\ref{alg:con.pop.exact.ineq.aff}, the optimal value of the portfolio optimization problem can be approximated to within~$\varepsilon$ in time~$O(n\varepsilon^{-4} + \sqrt{n}\varepsilon^{-5})$ on a quantum computer (see Theorem \ref{theo:apply.portfolio}).

\end{itemize}

In Table \ref{tab:summary} we summarize the results for the cases considered above.

\begin{table}
\footnotesize

\caption{Summary of quantum algorithms based on Hamiltonian updates for solving order-$k$ sum-of-squares and moment relaxations \eqref{eq:sos.sdp.general.general.pop}-\eqref{eq:mom.relax.pop} (and \eqref{eq:sdp.relaxation.cons.pop.ineq.aff}-\eqref{eq:sdp.relaxation.cons.pop.dual.ineq.aff}) of the polynomial optimization problem (POP) \eqref{eq:pop.general} with $n$ variables, $m$ inequality constraints and $l$ equality constraints. Here $g$ (resp. $h$) is the set of inequality (resp. equality) constraints in \eqref{eq:pop.general}; $x^\star$ is some optimal solution to \eqref{eq:pop.general}; $c,c_j$ are some positive real; $r,R$ are some reals in $(0,1)$; $s_g$ is the maximum number of nonzero entries per row of the \textit{coefficient-matching matrices} (defined as in \eqref{eq.tilde.B.gamma}) associated with $g_i$s; $n_{k}$ stands for $\binom{n+k}n$; and $n_{m,k}$ stands for $\sum_{i=0}^m\binom{n+k-d_i}n$.}
\label{tab:summary}
\begin{center}
\begin{tabular}{ |p{1.7cm}|p{1.6cm}|p{4cm}|p{2.7cm}|p{4cm}| } 
\hline
Type of POP & Quantum 

algorithm &Conditions &  Accuracy for approximate optimal value & Quantum complexity \\
\hline
Unconstrained ($m=l=0$) &Algorithm \ref{alg:uncon.pop}& $\lambda_k=f^\star$, $\|x^\star\|_{\ell_1}\le r $  & $\frac{\varepsilon }{1-r}$ 

(Theorem \ref{theo:accuarcy.approximate.val.uncons})
& \multirow{2}{3cm}{$O\left(\left(n_{2k}^{1/2}+\frac{1}\varepsilon n_k^{1/2}\ \right)\frac{1}{\varepsilon^4}\right)$} \\ \hline
Inequality-constrained ($l=0$)& Algorithm \ref{alg:con.pop.exact} & $\lambda_k=f^\star$, $\|x^\star\|_{\ell_1}\le r $ & $\frac{\varepsilon }{1-r}$ 

(Theorem \ref{theo:convergence.ineq.cons})
& $O\left(s_g\left(n_{2k}^{1/2}+n_{m,k}^{1/2}\frac{1}{\varepsilon}\right)\frac{1}{\varepsilon^4}\right)$ \\ \cline{3-3}
& & second-order sufficient optimality conditions, $\|x^\star\|_{\ell_1}\le r $& (Corollary \ref{coro:sosoc}) & \\ \cline{3-4}
& &  $c(R-\|x\|_{\ell_2}^2)\in g$& $\frac{\varepsilon n_k}{1-R}$

(Theorem \ref{theo:convergence.over.ball})
 & \\ \cline{3-4}
& & 
$\{x^\alpha\}_{\alpha\in\{0,1\}^n\backslash\{0\}}\subset g$, 

$\{c_j[r^j-(\sum_{j=1}^n x_j)^j]\}_{j=1,2}\subset g$
&$\frac{\varepsilon }{1-r}$ 

(Theorem \ref{theo:convergence.over.simplex})&\\\cline{2-3}
& Algorithm \ref{alg:con.pop.exact.ineq.aff} & $S(g,h)\subset \R_+^n$, 

$\{c_j[r^j-(\sum_{j=1}^n x_j)^j]\}_{j=1,2}\subset g$& (Theorem \ref{theo:convergence.over.simplex.ineq.aff}) &  \\ 
\hline
\end{tabular}
\end{center}
\end{table}

In Appendix \ref{sec:order.quantum.complexities}, we discuss additional quantum complexities of SDP relaxations (both SOS and moment relaxations) for polynomial optimization problems (POPs) in case where the Slater's conditions hold for these SDP relaxations. This appendix is divided into two parts:
In Appendix \ref{sec:other.quan.complex.SOS}, we present the computational complexities and accuracies of our quantum algorithms for solving SOS relaxations of POPs in more general settings where the exactness assumption $\lambda_k = f^\star$ and the structural constraints discussed above do not hold. In Appendix \ref{app:moment.relax}, we further apply our approach to moment relaxations of POPs over the unit ball and the unit sphere. However, the results in these more general cases are less conclusive compared to those reported in the main body of the paper.

\subsection{Comparison with other methods}
\label{sec:comparison}

If we employ interior-point methods~\cite{helmberg1996interior,jiang2020faster,augustino2023quantum}  to solve the sum-of-squares relaxation~\eqref{eq:sos.general.general.pop} of problem~\eqref{eq:pop.general} with $m=l=0$ in Case~1, the time complexity on a classical computer is
\begin{equation}
O\!\left( 
\binom{n+k}{n}^{1/2} 
\left[
\binom{n+2k}{n}\binom{n+k}{n}^2 
+ 
\binom{n+2k}{n}^{\omega} 
+ 
\binom{n+k}{n}^{\omega}
\right]
\log\!\left(\frac{1}{\varepsilon}\right)
\right),
\end{equation}
where $\omega \in [2, 2.38]$.  
On a quantum computer, the corresponding complexity becomes
\begin{equation}
O\!\left(
\binom{n+k}{n}^{3.5}
\frac{1}{\varepsilon}\,
\mathrm{polylog}\!\left(\binom{n+k}{n},\frac{1}{\varepsilon}\right)
\right),
\end{equation}
to compute an approximate value for $\lambda_k$ with accuracy
\begin{equation}
\varepsilon\,\binom{n+k}{n}.
\end{equation}

Compared with these interior-point methods, our quantum algorithm (see~\eqref{eq:complex.case1}) exhibits a higher dependence on the precision parameter $\varepsilon$, scaling polynomially as $O((1/\varepsilon)^5)$, whereas classical interior-point methods depend only logarithmically on $1/\varepsilon$ and the quantum interior-point method scales as $O((1/\varepsilon)\log(1/\varepsilon))$.  
However, our approach achieves substantially smaller exponents in terms of $\binom{n+k}{n}$ and $\binom{n+2k}{n}$—more than four times lower than those in the interior-point setting—yielding superior asymptotic performance for large-scale polynomial optimization problems.  
While the dependence on $\varepsilon$ is worse, this trade-off is advantageous in high-dimensional regimes where problem size dominates the numerical precision requirements.

\medskip

When $t=0$, the standard SDP~\eqref{eq:sdp.intro} is equivalent to
\begin{equation}\label{eq:sdp.intro2}
\begin{array}{rl}
\lambda^\star = \inf\limits_{X} & \tr(CX) \\
\text{s.t.} & X \in \mathbb{S}^{\bar N}_+, \quad \tr(A_i X) = b_i, \quad i = 1,\dots,M,
\end{array}
\end{equation}
where $\bar N = \sum_{j=1}^{r} N_j$, $C, A_i \in \prod_{j=1}^r \mathbb{S}^{N_j}_+$, $q_i \in \mathbb{R}^t$, and $b_i \in \mathbb{R}$.  
In this SDP, while the decision matrix $X$ is dense, the matrices $A_i$ and $C$ exhibit block-diagonal structures.  
The sum-of-squares relaxation~\eqref{eq:sos.general.general.pop} of problem~\eqref{eq:pop.general} with $l=0$ in Case~2 can thus be formulated as the semidefinite program~\eqref{eq:sdp.intro2} with $r = m + 1$, $M = \binom{n+2k}{n}$, and
\begin{equation}
N_{i+1}=\binom{n+k-d_i}{n}\quad\text{and}\quad\bar N = \sum_{i=0}^m \binom{n+k-d_i}{n}.
\end{equation}

Using an interior-point method to solve this SDP, the classical computational complexity is
\begin{equation}
\footnotesize
O\!\left(
\left(\sum_{i=0}^m \binom{n+k-d_i}{n}\right)^{1/2}
\!\left[
\!\left(\sum_{i=0}^m \binom{n+k-d_i}{n}\right)^{\!2}\!\binom{n+2k}{n}
+
\binom{n+2k}{n}^{\omega}
+
\left(\sum_{i=0}^m \binom{n+k-d_i}{n}\right)^{\!\omega}
\!\right]
\log\!\left(\frac{1}{\varepsilon}\right)
\!\right),
\end{equation}
and on a quantum computer it is
\begin{equation}
O\!\left(
\left(\sum_{i=0}^m \binom{n+k-d_i}{n}\right)^{3.5}
\frac{1}{\varepsilon}\,
\mathrm{polylog}\!\left(\frac{1}{\varepsilon}\right)
\right),
\end{equation}
to approximate $\lambda_k$ with accuracy
\begin{equation}
\varepsilon \sum_{i=0}^m \binom{n+k-d_i}{n}.
\end{equation}

Compared with these interior-point complexities, our quantum MMW-based algorithm~\eqref{eq:runtime.inequality.pop} achieves a lower exponent in the parameters $\sum_{i=0}^m \binom{n+k-d_i}n$ ($\le (m+1) \binom{n+k}{n}$), and $\binom{n+2k}{n}$ (noting that $\binom{n+2k}{n} = O(\binom{n+k}{n}^2)$), while exhibiting a higher dependence on $1/\varepsilon$.  
Furthermore, if $\lambda_k = f^\star$ and $f^\star$ is attained at a point within an $\ell_1$-ball of radius $\frac{1}{2}$ centered at the origin, or if the feasible set $S(g,h)$ of problem~\eqref{eq:pop.general} lies within the simplex $\{x \in \mathbb{R}^n_+ : \sum_{j=1}^n x_j \le 1/2\}$, then under suitable coefficient rescaling of $f$ and $g_i$, our quantum algorithm attains improved accuracy in $\varepsilon$.

For example, our quantum algorithm can approximate the optimal value of a portfolio optimization problem (with $k = d_1 = \dots = d_m = 1$, $s_g=1$ and $m = O(n)$) to accuracy $\varepsilon$ in computational time $O(n^{4}\operatorname{poly}(\varepsilon^{-1}))$ on a classical computer and $O(n\varepsilon^{-4} + \sqrt{n}\varepsilon^{-5})$ on a quantum computer (see Theorem~\ref{theo:apply.portfolio}).
For comparison, classical and quantum interior-point methods require $O(n^{3.5}\log(1/\varepsilon))$ and $O(n^{4.5}\varepsilon^{-1}\log(\varepsilon^{-1}))$, respectively, to achieve the same level of accuracy. 
Hence, our approach attains a super-quadratic speedup in the dependence on $n$ relative to interior-point methods in this example.

\medskip

\textbf{Memory comparison.}  
In terms of memory usage, classical interior-point methods require storing and factorizing dense matrices of size $\bar N \times \bar N$, resulting in a memory cost of $O(\bar N^2)$, and typically $O((M+\bar N)^3)$ when forming and solving the normal equations at each iteration. 
In contrast, in our MMW-based approach, the matrices $C$ and $A_j$ are assumed to be \emph{sparse}, with at most $s$ nonzero entries per row.  
This sparsity substantially reduces storage requirements: only the positions and values of the nonzero elements need to be stored or accessed through oracles.  
On a classical computer, this yields a memory cost of $O(s\bar N)$ when exploiting the block-diagonal and sparse structures of $A_j$ and $C$.  
On a quantum computer, these sparse matrices are represented by block encodings constructed from sparse-access oracles, which require only $\mathrm{polylog}(\bar N)$ qubits to store and query.  
Consequently, our method is markedly more memory-efficient than interior-point approaches on both classical and quantum platforms, especially for large-scale SDP relaxations where the data matrices are structured and sparse.

\section{Preliminaries}
\label{sec:pre}

\subsection{Notation and definitions}

Let $\|\cdot\|_{\ell_p}$, $p\in[1,\infty]$, denote the $l_p$-norm of a real vector in $\R^n$.
Then, for any $x\in\R^n$, the following inequality holds: $\|x\|_{\ell_2}\le \|x\|_{\ell_1}\le \sqrt{n}\|x\|_{\ell_2}$.

The trace of a real square matrix $A$, denoted by $\tr(A)$, is the sum of its diagonal elements.
A real symmetric matrix $A$ is said to be positive semidefinite (denoted by $A\succeq 0$) if all of its eigenvalues are nonnegative, and positive definite  (denoted by $A\succ 0$) if all of its eigenvalues are positive.
We write $A\succeq B$ (respectively $A\succ B$) to mean that $A-B\succeq 0$ (respectively $A-B\succ 0$).
Let $\|\cdot\|_F$ and $\|\cdot\|_*$ represent the Frobenius and nuclear norms, respectively, of a real matrix $A\in\R^{m\times n}$, i.e., 
$\|A\|_F=\sqrt{\tr(A^\top A)}$ and $\|A\|_*=\tr(\sqrt{A^\top A})$.
It follows that: $\|A\|_F\le \|A\|_*$.

Let $\mathbb S^N$ be the Hilbert space of real symmetric matrices of size $N$, equipped with the trace inner product $\tr(A^\top B)$ and the Frobenius norm $\|A\|_F=\sqrt{\tr(A^2)}$, for given $A,B\in \mathbb S^N$.
We denote by $\mathbb S^N_+$ the cone of positive semidefinite matrix in $\mathbb S^N$.
Let $I_N$ denote the identity matrix of size $N$.
Given symmetric matrices $X_1,\dots,X_r$, let $\diag(X_1,\dots,X_r)$ denote the block-diagonal matrix with blocks $X_1,\dots,X_r$.
For $X=(X_{ij})_{i,j\in[N]}\in\mathbb S^N$, we define $\text{vec}(X)=\left(X_{ij}\sqrt{2-\delta_{ij}}\right)_{i\le j}$, where $\delta_{ij}$ equals $1$ if $i=j$ and is $0$ otherwise.
Then, $\trace(XY)=\text{vec}(X)^\top \text{vec}(Y)$, for $X,Y\in\mathbb S^N$.
It implies $\|X\|_F=\|\text{vec}(X)\|_{\ell_2}$.
Let $\text{mat}(\cdot)$ denote the inverse map of $\text{vec}(\cdot)$.
If $X=\diag(X_1,\dots,X_r)\in \prod_{j=1}^r \mathbb S^{N_j} =\mathbb S^{N_1}\times \dots \times \mathbb S^{N_r}$, define $\text{Vec}(X)=(\text{vec}(X_1),\dots,\text{vec}(X_r))\in \R^L$,  with $L=\frac{1}{2}\sum_{i=1}^rN_i(N_i+1)$.
Let $\text{Mat}(\cdot)$ denote the inverse map of $\text{Vec}(\cdot)$.

Let $A$ be a real matrix of size $m\times n$.
The $\ell_p\to\ell_q$ norm of $A$ is defined as
$\|A\|_{\ell_p\to \ell_q}=\sup\{\|Ax\|_{\ell_q}\,:\,x\in\R^n\,,\,\|x\|_{\ell_p}=1\}$,
where $p,q\in[1,\infty]$.
It follows that $\|Ax\|_{\ell_q}\le \|A\|_{\ell_p\to \ell_q}\|x\|_{\ell_p}$.
In particular, $\|A\|_{\ell_1\to\ell_1}$ represents the maximum $\ell_1$-norm of a column of $A$ and $\|A\|_{\ell_2\to\ell_2}$ is the largest singular value of $A$.
Let $\sigma_{\max}(\cdot)$ and $\sigma_{\min}(\cdot)$ denote the largest and smallest singular value of a real matrix, respectively.
Similarly, let 
$\lambda_{\max}(\cdot)$ and $\lambda_{\min}(\cdot)$ represent the largest and smallest eigenvalues of a real symmetric matrix, respectively.
Denote by $\det(\cdot)$ the determinant of a square matrix.
The condition number of a real symmetric matrix $A$ is defined by $\kappa(A)=\frac{\lambda_{\max}(A)}{\lambda_{\min}(A)}$.

Let $\R[x]$ denote the ring of polynomials in vectors of variables $x=(x_1,\dots,x_n)$.
For $k\in\N$, define $\N^n_k=\{\alpha\in\N^n\,;\,|\alpha|\le k\}$, where $|\alpha|=\alpha_1+\dots+\alpha_n$.
For $\alpha,\beta\in\N^n$, we write $\alpha\le \beta$ (in lexicographical order) if either $|\alpha|< |\beta|$ or $|\alpha|=|\beta|$ and $\alpha_i>\beta_i$ for the largest $i$ for which $\alpha_i\ne\beta_i$.
Let $v_k$ denote the vector of monomials of degree up to $k$, i.e., $v_k=(x^\alpha)_{\alpha\in\N^n_k}$ with $x^\alpha=x_1^{\alpha_1}\dots x_n^{\alpha_n}$.
A polynomial $p\in \R[x]$ can be written as $p=\sum_{\alpha\in\N^n_k}p_\alpha x^\alpha=\bar p^\top v_k$ with $\bar p=(p_\alpha)_{\alpha\in\N^n_k}\subset \R$ for some $k\in\N$.
The degree of polynomial $p$ is defined as $\deg(p)=\max\{|\alpha|\,:\, p_\alpha\ne 0\}$.
Let $\R[x]_k$ denote the linear space of polynomials in $\R[x]$ of degree at most $k$, and let $\Sigma[x]_k$ represent the cone of sums of squares of polynomials in $\R[x]$ of degree at most $2k$. 

For $p\in\R[x]$, let
$\|p\|_{\ell_q}$ represent the $\ell_q$-norm of the vector of coefficients of $p$.
A polynomial $p=\sum_{\alpha\in\N^n}p_\alpha x^\alpha\in\R[x]$ has its support defined as $\supp(p)=\{\alpha\in\N^n\,:\,p_\alpha\ne 0\}$.
For $\alpha=(\alpha_1,\dots,\alpha_n)\in\N^{n}$ with $|\alpha|=k$,  the multinomial coefficient $\binom{k}{\alpha}$ is defined as $\binom{k}{\alpha}=\frac{k!}{\alpha_1!\dots\alpha_n!}$.

For a given $y=(y_\alpha)_{\alpha\in\N^n_{2k}}$, let $M_k(y)$ denote the moment matrix of order $k$ associated with $y$, i.e., $M_k(y)=(y_{\alpha+\beta})_{\alpha,\beta\in\N^n_k}$.
For $p=\sum_{\gamma\in\N^n_{2u}}p_\gamma x^\gamma \in \R[x]_{2u}$, let $M_{k-u}(py)$ denote the localizing matrix of order $k-u$ associated with $p$ and $y$ of order $k$, i.e., $M_{k-u}(py)=(\sum_{\gamma\in\N^n_{2u}}p_\gamma y_{\alpha+\beta+\gamma})_{\alpha,\beta\in\N^n_{k-u}}$.
Let $L_y:\R[x]_{2k}\to \R$ be the Riesz linear functional defined by $L_y(q)=\sum_{\alpha\in\N^n_{2k}}q_\alpha y_\alpha$, for a given $q=\sum_{\alpha\in\N^n_{2k}}q_\alpha x^\alpha\in\R[x]_{2k}$.
Then, we have the following relations for real symmetric matrices $G$ and $Q$: $L_y(v_k^\top Gv_k)=\tr( G M_k(y))$
and $L_y( pv_{k-u}^\top Qv_{k-u})=\tr( Q M_{k-u}(p y))$.

Let $\mu$ be a positive measure on $\R^n$. For $\alpha\in\N^n$, the quantity $y_\alpha=\int x^\alpha d\mu$ is called the moment of order $\alpha$ with respect to $\mu$.
The collection of moments $y=(y_\alpha)_{\alpha\in\N^n}$ is referred to as the moment sequence with respect to $\mu$.
If we restrict the moments to a finite set of multi-indices $\alpha\in \N^n_d$, the sequence $y=(y_\alpha)_{\alpha\in\N^n_d}$ is called the truncated moment sequence of order $d$ with respect to $\mu$.
For a truncated moment sequence $y=(y_\alpha)_{\alpha\in\N^n_{2k}}$, we define $M_k(\mu)=M_k(y)$ as the moment matrix of order $k$ with $\mu$.
Given $p=\sum_{\gamma\in\N^n_{2u}}p_\gamma x^\gamma \in \R[x]_{2u}$, we denote by $M_{k-u}(p\mu)=M_{k-u}(py)$ the localizing matrix with $p$ and $\mu$ of order $k-u$.
Additionally, the Riesz linear functional $L_y$ is defined by $L_y(q)=\int qd\mu$, for $q\in\R[x]_{2k}$.
For real symmetric matrices $P$ and $Q$, the following relations hold: $\int v_k^\top Gv_k d\mu=\tr( G M_k(\mu))$
and $\int pv_{k-u}^\top Qv_{k-u}d\mu=\tr(Q M_{k-u}(p \mu))$.

\subsection{Strong duality and exactness for semidefinite relaxations of polynomial optimization}
Before applying quantum SDP solvers to sums-of-squares relaxations, it is important to understand when these relaxations are well-behaved from an optimization perspective. In particular, we will review conditions under which strong duality holds and when the SDP hierarchy achieves exactness—that is, the relaxation value equals the true optimum after finitely many levels. This subsection recalls standard optimality conditions, duality theorems, and finite-convergence criteria for polynomial optimization problems, which we will later use to justify our accuracy assumptions and runtime guarantees.
Given $x\in\R^n$, $\eta\in\R^m_+$ and
$\gamma\in \R^l$, the Lagrangian for problem \eqref{eq:pop.general} is defined as:
\begin{equation}
\mathcal{L}( x,\eta,\gamma)\,:=\,
f( x)-\sum_{i=1}^m\eta_ig_i( x)
-\sum_{j=1}^l\gamma_jh_j( x)\,.
\end{equation}

For $x\in S(g,h)$, the active set of inequality constraints is given by:
\begin{equation}
J(x):=\{\,i\in\{1,\dots,m\}: g_i(x)=0\,\}\,.
\end{equation}

We say that the second-order sufficient optimality conditions (see \cite[Chapter 2]{nocedal2006numerical})  hold for problem \eqref{eq:pop.general} at $x^\star\in S(g,h)$ if the following criteria are satisfied:
\begin{enumerate}
\item Karush–Kuhn–Tucker (KKT) conditions: There exist $\eta^\star\in \R^m_+$ and $\gamma^\star\in\R^l$ such that
$\nabla_{x}\mathcal{L}(x^\star,\eta^\star,\gamma^\star)=0$ and $\eta^\star_i\,g_i( x^\star)=0$ for all $i=1,\dots,m$.
\item Strict complementarity: $\eta_i^\star +g_i( x^\star)>0$, for all $i=1,\dots,m$.
\item Linear independence constraint qualification: the family 
\begin{equation}
\Lambda=\{\nabla g_i( x^\star)\}_{i\in J(x^\star)}\cup \{\nabla h_j( x^\star)\}_{j=1}^l
\end{equation}
is linearly independent. 
\item Positive definiteness of the Hessian: $u^\top\nabla_{x}^2\mathcal{L}( x^\star,\eta^\star,\gamma^\star) u>0$ for all $u\in \text{span}(\Lambda)^\perp\backslash \{0\}$.
\end{enumerate}

 \begin{lemma} \label{lem:strong.duality.ball}
(Josz-Henrion \cite{josz2016strong})
Let $f^\star$ be as in \eqref{eq:pop.general}.
Assume $g_1=R-\|x\|_{\ell_2}^2$ for some $R>0$.
Then, zero duality gap between the primal  \eqref{eq:sos.sdp.general.general.pop} and dual \eqref{eq:mom.relax.pop} holds, i.e., $\lambda_k=\tau_k$ and $\tau_k\in\R$, for $k\ge k_{\min}$, where
\begin{equation}\label{eq:k.min}
k_{\min}=\max\{d,d_1,\dots,d_m,w_1,\dots,w_l\}\,.
\end{equation}
Moreover, SDP \eqref{eq:mom.relax.pop} has an optimal solution.
\end{lemma}
Josz and Henrion proved that the set of optimal solutions of \eqref{eq:mom.relax.pop} is compact, ensuring existence of an optimal solution.

\begin{lemma}\label{lem:slater.cond}
(Lasserre \cite[Theorem 3.4 (a)]{lasserre2001global})
If $S(g,h)$ has nonempty interior, then Slater's condition holds for the dual \eqref{eq:mom.relax.pop}  for $k\ge k_{\min}$, where $k_{\min}$ is defined as in \eqref{eq:k.min}.
In this case, $\lambda_k=\tau_k$, $\tau_k\in\R$ and the primal problem \eqref{eq:sos.sdp.general.general.pop} has an optimal solution.
\end{lemma}

The quadratic module $Q_k(g,h)$ associated with $g=(g_i)_{i=1}^m\subset \R[x]$ and $h=(h_j)_{j=1}^l\subset \R[x]$ is given by:
\begin{equation}
Q_k(g,h)=\{\sum_{i=0}^m\sigma_ig_i+\sum_{j=1}^l\psi_j h_j\,:\,\sigma_i\in\Sigma[x]_{k-d_i}\,,\,\psi_j\in\R[x]_{2(k-w_j)}\}
\end{equation}
where $g_0=1$ and $d_0=0$.
We say that $S(g,h)$ is Archimedean if $R-\|x\|_{\ell_2}^2 \in Q_k(g,h)$ for some $R>0$.

The ideal $I(h)$ associated with  $h=(h_j)_{j=1}^l\subset \R[x]$ is given by $I(h)=\{\sum_{j=1}^l\psi_j h_j\,:\,\psi_j\in\R[x]\}$.
We say that $I(h)$ is real radical if for all $f\in I( h)$, there exist $m\in\N$ and a sum-of-squares polynomial $\sigma$ such that $-f^{2m}\in \sigma+ I(h)$.

The following proposition provides a sufficient condition to ensure finite convergence of the sequence $(\lambda_k)_{k\in\N}$ to $f^\star$.
\begin{lemma}\label{lem:finite.conver.ball}
(Lasserre \cite[Theorem 7.5]{lasserre2015introduction}) Assume that $S(g,h)$ is Archimedean, the ideal $I(h)$ is real radical, and the second-order sufficient optimality conditions hold for problem \eqref{eq:pop.general} at every optimal solution. 
Then, $\lambda_k=\tau_k=f^\star$ for some $k\in\N$ and both primal and dual problems \eqref{eq:sos.sdp.general.general.pop}-\eqref{eq:mom.relax.pop} admit optimal solutions.
\end{lemma}
It is worth noting that the real radical property is not generic, and verifying $I(h)$ is real radical must be done on a case-by-case basis.
However, Nie \cite{nie2014optimality} showed that generically, $\tau_k=f^\star$ for sufficiently large $k$, meaning this equality holds on a Zariski open set in the space of coefficients of $g_i,h_j$ for fixed degrees.

The following lemma states that finite convergence occurs at the first order relaxation for convex quadratic programs.
\begin{lemma}[Exactness for convex quadratic programs]
\label{lem:exact.quadratic}
    Let $f=x^\top Ax+q^\top x$ with $A\succeq 0$, $g_{i}=a_i^\top x+b_i$, $i=1,\dots,m$, and $h_j=c_j^\top x+ w_j$, $j=1,\dots,l$.
    Assume that problem \eqref{eq:pop.general} has an optimal solution $x^\star$.
    Then $\lambda_1=f^\star$.
\end{lemma}
\begin{proof}
    It is easy to see that $\lambda_1\le f^\star$.
    To prove the reverse inequality, we apply the KKT conditions for problem $f^\star=\min_{x\in S(g,h)} f(x)$ at $x^\star$.
     Then there exist $\eta^\star\in \R^{m}_+$ and $\gamma^\star\in\R^{l}$ such that
 $\eta_i^\star\,g_i( x^\star)=0$ for all $i=1,\dots,m$, and $\nabla f( x^\star)=\sum_{i=1}^{m}\eta_i^\star\nabla g_i(x^\star)
+\sum_{j=1}^{l}\gamma_j^\star\nabla h_j( x^\star)$.
This follows that
\begin{equation}
\begin{array}{rl}
& 2Ax^\star+q=\sum_{i=1}^{m}\eta_i^\star a_i
+\sum_{j=1}^{l}\gamma_j^\star c_j\\
\Rightarrow& 2(x-x^\star)^\top Ax^\star + q^\top (x-x^\star)=\sum_{i=1}^{m}\eta_i^\star a_i^\top (x-x^\star)
+\sum_{j=1}^{l}\gamma_j^\star c_j^\top (x-x^\star)\\
\Rightarrow&  x^\top Ax+q^\top x -f^\star = (x-x^\star)^\top A(x-x^\star)+\sum_{i=1}^{m}\eta_i^\star(a_i^\top x+b_i)
+\sum_{j=1}^{l}\gamma_j^\star(c_j^\top x+w_j).
\end{array}
\end{equation}
The last equality is due to $f^\star=x^{\star \top} Ax^\star+q^\top x^\star$, $0=\eta_i^\star\,g_i( x^\star)=\eta^\star_i(a_i^\top x^\star+b_i)$ and $0=h_j(x^\star)=c_j^\top x^\star+w_j$.
It implies that $f-f^\star=\sigma_0+\sum_{i=1}^{m}\eta_i^\star g_i
+\sum_{j=1}^{l}\gamma_j^\star h_j$, where $\sigma_0=(x-x^\star)^\top A(x-x^\star)$ is a sum-of-squares polynomial of degree $2$ (since $A\succeq 0$). 
By definition of $\lambda_k$ in \eqref{eq:sos.general.general.pop}, $f^\star\le \lambda_1$.
Hence the equality holds.
\end{proof}

\subsection{Hamiltonian updates}
\label{sec:semidefinite.feasibility}
We are now going to briefly review the Hamiltonian updates framework to solve semidefinite programs of~\cite{GSLBrando2022}.
Consider the semidefinite feasibility problem with trace-one constraint:
\begin{equation}\label{eq:semidefinite.feasibility}
\begin{array}{rl}
\text{Find}&X\in \mathbb S^N\\
\text{s.t.}& X\succeq 0\,,\,\tr(X)=1\,,\,\tr(A_iX)\le  b_i\,,\,i=1,\dots,M\,,
\end{array}
\end{equation}
where $A_i\in  \mathbb S^N$ and $b_i\in\R$ for $i=1,\dots,M$.

We propose the following algorithm to solve problem \eqref{eq:semidefinite.feasibility} approximately.
\begin{algorithm}\label{alg:HU}
Hamiltonian updates
\begin{itemize}
\item Input: $\varepsilon>0$, $A_i\in \mathbb S^N$, $b_i\in\R$, $i=1,\dots,M$.
\item Output: $\eta\in\{0,1\}$ and $X_\varepsilon\in\mathbb S^N$.
\end{itemize}
\begin{enumerate}
\item Initialize $W^{(1)}=I_N$ and set $T=\left\lceil \frac{16\log(N)}{\varepsilon^2}\right\rceil$.
\item For $t=1,\dots,T$:
\begin{enumerate}
\item Compute $P^{(t)}=W^{(t)}/\tr(W^{(t)})$.
\item Check feasibility:
\begin{enumerate}
\item If there exists $j^{(t)}\in\{1,\dots,M\}$ such that 
\begin{equation}
\tr(A_{j^{(t)}}P^{(t)})>b_{j^{(t)}}+\varepsilon\,,
\end{equation}
set $M^{(t)}=\frac{1}{2}(I_N-A_{j^{(t)}})$.
\item Otherwise, set  $\eta=1$, $X_\varepsilon=P^{(t)}$ and terminate.
\end{enumerate}
\item Update $W^{(t+1)}=\exp(-\frac{\varepsilon}{4}\sum_{\tau=1}^t M^{(\tau)})$.
\end{enumerate}
\item If the loop completes, set $\eta=0$ and terminate.
\end{enumerate}
\end{algorithm}

Set the set $S_\varepsilon$  as follows:
\begin{equation}
S_\varepsilon=\left\{ X\in\mathbb S^N\left|\begin{array}{rl}
&X\succeq 0\,,\, \tr(X)=1\,,\\
&\tr(A_iX)\le b_i+\varepsilon\,,\, i=1,\dots,M
\end{array}\right. \right\}\,.
\end{equation}

The following lemma clarifies the interpretation of the output from the Hamiltonian updates algorithm:
\begin{lemma}[Accuracy]
\label{lem:Hamilton.updates}
Let $\varepsilon>0$, $A_i\in \mathbb S^N$, and $b_i\in\R$, for $i=1,\dots,M$.
Assume that $I_N\succeq A_i\succeq -I_N$, for $i=1,\dots,M$.
Let $\eta\in\{0,1\}$ and $X_\varepsilon\in\mathbb S^N$ be the output of Algorithm \ref{alg:HU}. Then:
\begin{enumerate}
\item If $\eta=0$, $S_0=\emptyset$.
\item If $\eta=1$, then $X_\varepsilon\in S_\varepsilon$.
\end{enumerate}
\end{lemma}

\begin{lemma}[Complexity of computing a dense matrix exponential to error $\varepsilon$]
\label{lem:expA-complexity}
Let $A\in\mathbb{C}^{N\times N}$ and let $\|\cdot\|$ denote a submultiplicative operator norm. For any target accuracy $0<\varepsilon<1$, one can compute an approximation $\widetilde{E}$ with
\[
\|\widetilde{E}-e^{A}\|\le\varepsilon
\]
using 
\[
O\!\big(N^\omega [\log\|A\| + \log(1/\varepsilon)]\big)
\]
operations on a classical computer, where $\omega\approx 2.373$ is the exponent of matrix multiplication.
\end{lemma}

\begin{proof}
The scaling-and-squaring method proceeds as follows:

\begin{enumerate}
\item Scale $A$ by $2^{-s}$, with $s = \lceil \log_2(\|A\|/\theta)\rceil$, so that $B = 2^{-s}A$ satisfies $\|B\|\le \theta$ for some small fixed $\theta>0$.
\item Compute a Pad\'e (or truncated Taylor) approximant $\widetilde{R}_m(B)$ of $e^B$ of degree $m$ such that $\|\widetilde{R}_m(B)-e^B\|\le\varepsilon_0 = O(\varepsilon)$. Analytic approximation theory guarantees that $m = O(\log(1/\varepsilon))$ suffices.
\item Recover $e^A$ via repeated squaring: $\widetilde{E} = (\widetilde{R}_m(B))^{2^s}$.
\end{enumerate}

The dominant costs are $m$ matrix multiplications for the approximant and $s$ squarings, each costing $O(N^\omega)$ operations. Substituting $m = O(\log(1/\varepsilon))$ and $s = O(\log\|A\|)$ yields the stated complexity.
\end{proof}

\begin{lemma}[Classical complexity of Hamitonian updates for semidefinite feasibility]
\label{lem:MMWU-complexity}
The complexity of Algorithm \ref{alg:HU}, up to polylogarithmic factors, is as follows:
\begin{itemize}
\item {On a classical computer, the algorithm takes
\[
\widetilde O\Big((N^\omega + M s N)\varepsilon^{-2} \Big),
\]
where $\omega\approx 2.373$ and $s$ is the maximum number of nonzero entries per row of $A_i$s.}
\item On a quantum computer, the algorithm requires
\begin{equation}
O(s(\sqrt{M}+\sqrt{N}\varepsilon^{-1})\varepsilon^{-4})
\end{equation}
operations.
\end{itemize}
\end{lemma}

\begin{proof}
The quantum complexity is derived from the results of Apeldoorn and Gily\'en~\cite{van2018improvements}.
Let us prove the classical complexity based on Lemma \ref{lem:expA-complexity}.
On a classical computer, each iteration $t=1,\dots,T$ requires:

\begin{enumerate}
\item Computing $P^{(t)} = W^{(t)}/\tr(W^{(t)})$, which costs $O(N^2)$.
\item Checking feasibility $\tr(A_i P^{(t)})$ for $i=1,\dots,M$, costing $O(M s N)$ due to the sparsity of $A_i$.
\item Updating $W^{(t+1)} = \exp(\frac{\varepsilon}{4}\sum_{\tau=1}^t M^{(\tau)})$, which is a dense matrix exponential. By Lemma~\ref{lem:expA-complexity}, this costs $O(N^\omega \log(1/\varepsilon))$.
\end{enumerate}

Multiplying the per-iteration cost by $T = O(\log N / \varepsilon^2)$ gives the total complexity
\[
O\Big( [N^\omega \log(1/\varepsilon) +  M s N] \frac{\log N}{\varepsilon^2}\Big).
\]
Hence the result follows.
\end{proof}

\subsection{Binary search using Hamiltonian updates for semidefinite programs with affine equality constraints}

In the following lemma, we provide a lower bound for the evaluation of the objective function of a semidefinite program at a positive semidefinite matrix, in terms of the violation of the affine constraints at that matrix:
\begin{lemma}\label{lem:obj.at.X}
Let $b\in\R^M$, $C\in \mathbb S^N$, and let $\mathcal A:\mathbb S^N\to \R^M$ be a linear operator defined by 
\begin{equation}
\mathcal A(Y)=(\tr(A_1Y),\dots,\tr(A_MY))\,,
\end{equation}
where $A_i\in\mathbb S^N$, for $i=1,\dots,M$.
Consider primal-dual semidefinite programs:
\begin{equation}\label{eq:primal.no.tr.cons}
\begin{array}{rl}
\lambda^\star=\inf\limits_{Y}&\tr(CY)\\
\text{s.t.}&Y\in\mathbb S^{N}_+\,,\,\mathcal A Y=b\,,
\end{array}
\end{equation}
\begin{equation}\label{eq:dual.no.tr.cons}
\begin{array}{rl}
\tau^\star=\sup\limits_{\xi}&b^\top \xi\\
\text{s.t.}&\xi\in\R^M\,,\,C-\mathcal A^\top \xi \succeq 0\,,
\end{array}
\end{equation}
where $\mathcal A^\top:\R^M\to \mathbb S^N$ is the adjoint operator of $\mathcal A$, given by
\begin{equation}
\mathcal A^\top \xi=\sum_{i=1}^M \xi_i A_i\,.
\end{equation}
Assume strong duality holds for problems \eqref{eq:primal.no.tr.cons}-\eqref{eq:dual.no.tr.cons}: $\lambda^\star=\tau^\star$, and both problems admit solutions $Y^\star$  and $\xi^\star$, respectively.
Then, for all $X\in \mathbb S^N_+$, the following holds:
\begin{equation}
\tr(C X) \ge\lambda^\star -\|\xi^\star\|_{\ell_1} \|\mathcal A X-b\|_{\ell_\infty}\,.
\end{equation}
\end{lemma}
\begin{proof}
Set $Z^\star=C-\mathcal A ^\top \xi^\star$.
Then $Z^\star\succeq 0$.
The Lagrangian has the form
\begin{equation}
L(Y,Z,\xi)=\tr(CY) -\tr(Z Y)-\xi^\top (\mathcal A X-b)\,.
\end{equation}
By the Karush–Kuhn–Tucker (KKT) conditions, we have
\begin{equation}\label{eq:KKT1}
0=\frac{\partial L}{\partial Y}(Y^\star,Z^\star,\xi^\star)=C-Z^\star-\mathcal A^\top \xi^\star\,.
\end{equation}
and 
\begin{equation}\label{eq:KKT2}
\tr(Z^\star Y^\star)=0\,.
\end{equation}
Let $X\in \mathbb S^N_+$.
By \eqref{eq:KKT1}, $C=Z^\star+\mathcal{A}^\top \xi^\star$.
Then
\begin{equation}\label{eq:bound4}
\begin{array}{rl}
\tr( C X)-\lambda^\star=&\tr( C( X-Y^\star) )\\
=&\tr(( Z^\star+\mathcal{A}^\top \xi^\star)( X-Y^\star)) \\
=& \tr( Z^\star X) + \tr( \mathcal{A}^\top \xi^\star (X-Y^\star ))\\
=& \tr( Z^\star X) +   \xi^{\star\top}\mathcal{A}(X-Y^\star )\\
\ge &   \xi^{\star\top}(\mathcal{A}X-b )\,. 
\end{array}
\end{equation}
The first equality follows from $\tr( C Y^\star) =\lambda^\star$, the third equality is due to \eqref{eq:KKT2}, and the last inequality is based on the positive semidefiniteness of $Z^\star$ and $X$.
From this, we deduce:
\begin{equation}\label{eq:bound2}
\tr(C X)\ge \lambda^\star -\|\xi^{\star}\|_{\ell_1}\|\mathcal{A}X-b \|_{\ell_\infty}\,.
\end{equation}

\end{proof}

Consider the following semidefinite program with trace one:
\begin{equation}\label{eq:sdp}
\begin{array}{rl}
\lambda^\star=\inf\limits_{X}&\tr( CX)\\
\text{s.t.}&X\in\mathbb S^{N}_+\,,\,\tr(X)= 1\,,\\
&\tr( A_iX) =b_i\,,\,i=1,\dots,M\,,
\end{array}
\end{equation}
where $C,A_i\in\mathbb S^N$, and $b_i\in\R$, for $i=1,\dots,M$, are given.
Define the set $S^{(\lambda)}$ as
\begin{equation}\label{eq:S.rho}
S^{(\lambda)}=\left\{X\in \mathbb S^N_+\left|\begin{array}{rl}
& \tr( A_iX) = b_i\,,\,i=1,\dots,M\,,\\
&\tr( CX) \le \lambda\,,\,\tr(X)=1
\end{array} \right. \right\}\,.
\end{equation}
We propose the following algorithm to numerically solve the SDP \eqref{eq:sdp}:
\begin{algorithm}\label{alg:Binary.search.HU}
Binary search using Hamiltonian updates
\begin{itemize}
\item Input: $\lambda_{\min}\in\R$, $\lambda_{\max}>\lambda_{\min}$, $\varepsilon>0$, $A_i\in \mathbb S^N$, $b_i\in\R$, $i=1,\dots,M$.
\item Output: $\underline \lambda_T\in\R$ and $X_\varepsilon\in\mathbb S^N$.
\end{itemize}
\begin{enumerate}
\item Set $\underline\lambda_0=\lambda_{\min}$, $\overline\lambda_0=\lambda_{\max}$ and $T=\left\lceil \log_2\left(\frac{\lambda_{\max}-\lambda_{\min}}{\varepsilon}\right)\right\rceil$.
\item For $t=0,\dots,T-1$, do
\begin{enumerate}
\item Set $\lambda_t=\frac{\underline\lambda_t+\overline \lambda_t}{2}$.
\item Run Algorithm \ref{alg:HU} (Hamiltonian updates) to test the feasibility of $S^{(\lambda_t)}$ (defined in \eqref{eq:S.rho}) and obtain $\eta_t\in\{0,1\}$ and $X_\varepsilon^{(t)}\in \mathbb S^N$.
\item Based on $\eta_t$:
\begin{enumerate}
\item If $\eta_t=0$, set $\underline \lambda_{t+1}=\lambda_t$ and $\overline \lambda_{t+1}=\overline \lambda_t$.
\item If $\eta_t=1$, set $\underline \lambda_{t+1}=\underline\lambda_t$, $\overline \lambda_{t+1}= \lambda_t$ and $X_\varepsilon=X_\varepsilon^{(t)}$. 
\end{enumerate}
\end{enumerate}
\end{enumerate}
\end{algorithm}
Define the set:
\begin{equation}
S^{(\lambda)}_\varepsilon=\left\{X\in \mathbb S^N_+\left|\begin{array}{rl}
& \tr( A_iX) \le b_i+\varepsilon\,,\,i=1,\dots,M\,,\\
& \tr( -A_iX) \le -b_i+\varepsilon\,,\,i=1,\dots,M\,,\\
&\tr( CX) \le \lambda+\varepsilon\,,\,\tr(X)=1
\end{array} \right. \right\}\,.
\end{equation}
Note that $S^{(\lambda)}_0=S^{(\lambda)}$.

The following lemma provides a necessary and sufficient condition for $S^{(\lambda)}$ to be nonempty:

\begin{lemma}\label{lem:feas.ineq}
Let $C,A_i\in \mathbb S^N$, and $b_i\in\R$, for $i=1,\dots,M$.
Let $\lambda^\star$ be the optimal value of the semidefinite program in \eqref{eq:sdp}.
Assume that the problem \eqref{eq:sdp} has an optimal solution.
Define $S^{(\lambda)}$ as in \eqref{eq:S.rho}.
Then,
\begin{equation}
S^{(\lambda)}\ne\emptyset\quad\Leftrightarrow\quad \lambda\ge \lambda^\star\,.
\end{equation}
\end{lemma}
\begin{proof}
By assumption, there exists an optimal solution $X^\star$ to the semidefinite program \eqref{eq:sdp}.
If $\lambda\ge \lambda^\star$, then $\lambda\ge \lambda^\star=\tr( CX^\star)$, and $X^\star$ is a feasible solution to the problem \eqref{eq:sdp}.
Thus $X^\star\in S^{(\lambda)}$, which shows that $S^{(\lambda)}$ is nonempty. 
Conversely, if $S^{(\lambda)}\ne\emptyset$, there exists $Y\in S^{(\lambda)}$. 
This implies that $Y$ is a feasible solution to semidefinite program \eqref{eq:sdp}, and thus $\lambda\ge \tr( CY)\ge \lambda^\star$.
Therefore, the result follows.
\end{proof}

We present the basic properties of the values returned by binary search in the following two lemmas:

\begin{lemma}\label{lem:induc.ineq}
Let $(\underline \lambda_{t})_{t=0}^T$ and $(\overline \lambda_{t})_{t=0}^T$ be the sequences generated by Algorithm \ref{alg:Binary.search.HU}.
Then, for $t=0,\dots,T-1$, the following holds:
\begin{equation}\label{eq:induc.ineq}
\overline\lambda_{t+1}-\underline\lambda_{t+1}=\frac{1}{2}(\overline\lambda_{t}-\underline\lambda_{t})\,.
\end{equation}
\end{lemma}
\begin{proof}
Let $t\in\{0,\dots,T-1\}$ be fixed and consider Step 2 (c) of Algorithm \ref{alg:Binary.search.HU}.
\begin{itemize}
\item If $\eta_t=0$, then $\underline\lambda_{t+1}=\lambda_t$ and $\overline\lambda_{t+1}=\overline\lambda_t$, which gives 
\begin{equation}
\overline\lambda_{t+1}-\underline\lambda_{t+1}= \overline\lambda_t-\lambda_t=\overline\lambda_t-\frac{\underline\lambda_t+\overline\lambda_t}{2}=\frac{1}{2}(\overline\lambda_{t}-\underline\lambda_{t})\,.
\end{equation}
\item If $\eta_t=1$, then $\underline\lambda_{t+1}=\underline\lambda_t$ and $\overline\lambda_{t+1}=\lambda_t$, which gives
\begin{equation}
\overline\lambda_{t+1}-\underline\lambda_{t+1}= \lambda_t-\underline\lambda_t=\frac{\underline\lambda_t+\overline\lambda_t}{2}-\underline\lambda_t=\frac{1}{2}(\overline\lambda_{t}-\underline\lambda_{t})\,.
\end{equation}
\end{itemize}
Thus, the result follows.
\end{proof}

\begin{lemma}\label{lem:monotoncity}
Let $(\underline \lambda_{t})_{t=0}^T$ and $(\overline \lambda_{t})_{t=0}^T$ be the sequences generated by Algorithm \ref{alg:Binary.search.HU}.
Then, the following statements hold:
\begin{enumerate}
\item For $t=0,\dots,T-1$, the sequence satisfies 
\begin{equation}
\lambda_{\min}\le\underline \lambda_t\le \underline \lambda_{t+1}\le \overline \lambda_{t+1}\le \overline \lambda_t\le \lambda_{\max}\,.
\end{equation}
\item For $t=0,\dots,T$, we have $\overline \lambda_t -\underline \lambda_t=\frac{1}{2^t}(\lambda_{\max}-\lambda_{\min})$. In particular, 
\begin{equation}
\overline \lambda_T -\underline \lambda_T=\frac{1}{2^T}(\lambda_{\max}-\lambda_{\min})\le \varepsilon\,.
\end{equation}
\end{enumerate}
\end{lemma}

\begin{proof}
The first statement is straightforward to prove, as it follows directly from the binary search update rules.
Now, let us prove the second statement.
By Lemma \ref{lem:induc.ineq}, by induction, we have
\begin{equation}\label{eq:bound.delta.T.2}
\overline \lambda_t-\underline \lambda_t=\frac{1}{2}(\overline \lambda_{t-1}-\underline \lambda_{t-1})=\dots= \frac{1}{2^t}(\overline \lambda_0-\underline\lambda_0)=\frac{1}{2^t}(\lambda_{\max}-\lambda_{\min})\,.
\end{equation}
Thus, the result follows.
\end{proof}

The following lemma establishes a lower bound on the gap between the optimal value of the SDP and the value returned by the binary search using Hamiltonian updates:

\begin{lemma}\label{lem:zero.lower.bound}
Let $C,A_i\in \mathbb S^N$, and $b_i\in\R$, for $i=1,\dots,M$.
Assume that $I_N\succeq C\succeq -I_N$ and $I_N\succeq A_i\succeq -I_N$, for $i=1,\dots,M$.
Let $\lambda^\star$ be the optimal value in \eqref{eq:sdp}, and assume that problem \eqref{eq:sdp} has an optimal solution.
Additionally, let $\lambda_{\min}\le \lambda^\star$, and let $(\underline \lambda_{t})_{t=0}^T$ be the sequence generated by Algorithm \ref{alg:Binary.search.HU}.
Then for all $t=0,\dots,T$,
\begin{equation}\label{eq:zero.lower.bound}
0\le \lambda^\star-\underline \lambda_t\,.
\end{equation}
\end{lemma}
\begin{proof}
We will prove the inequality \eqref{eq:zero.lower.bound} by induction on $t$.
\begin{itemize}
\item Base case ($t=0$): From the initial setup, we have $\underline \lambda_0=\lambda_{\min}$, so $\lambda^\star-\underline \lambda_0=\lambda^\star-\lambda_{\min}\ge 0$, which satisfies the inequality  \eqref{eq:zero.lower.bound} for $t=0$.
\item Inductive step:
Assume that the inequality \eqref{eq:zero.lower.bound} holds for $t=r\le T-1$, i.e., $\lambda^\star-\underline{\lambda}_r>0$.
We claim that \eqref{eq:zero.lower.bound} holds with $t=r+1$.
Indeed, 
Consider Step 2 (c) of Algorithm \ref{alg:Binary.search.HU} with $t=r$.
\begin{itemize}
\item If $\eta_r=0$, then Lemma \ref{lem:Hamilton.updates} states that $S^{(\lambda_r)}=S^{(\lambda_r)}_0=\emptyset$, which implies from Lemma \ref{lem:feas.ineq} that $\lambda^\star> \lambda_r$. Thus, we have $\lambda^\star> \lambda_r=\underline \lambda_{r+1}$,
ensuring that the inequality holds for $t=r+1$.
\item If $\eta_r=1$, then by the inductive hypothesis, $\underline \lambda_{r+1}=\underline \lambda_r\le \lambda^\star$, so the inequality again holds.
\end{itemize}
\end{itemize}
Thus, by induction, the inequality \eqref{eq:zero.lower.bound} holds for all $t=0,\dots,T$.
\end{proof}

The following lemma provides upper bounds on the gap between the optimal value of the SDP and the value returned by the binary search using Hamiltonian updates:
\begin{lemma}\label{lem:upper.bound}
Let $\varepsilon>0$, $C,A_i\in \mathbb S^N$, $b_i\in\R$, for $i=1,\dots,M$, and assume that $I_N\succeq C\succeq -I_N$ and $I_N\succeq A_i\succeq -I_N$, for $i=1,\dots,M$.
Let $\lambda^\star$ be the optimal value in \eqref{eq:sdp} and assume $\lambda_{\max}\ge \lambda^\star$.
Further, assume the following equality holds:
\begin{equation}\label{eq:remove.trace.one}
\lambda^\star=\inf\{\tr( CX)\,:\,X\in\mathbb S^{N}_+\,,\,\tr( A_iX) =b_i\,,\,i=1,\dots,M\}\,.
\end{equation}
Assume strong duality holds for problem \eqref{eq:remove.trace.one} and its dual \eqref{eq:dual.no.tr.cons}, which admits a solution $\xi^\star$.
Let $(\overline \lambda_{t})_{t=0}^T$ be the sequence generated in Algorithm \ref{alg:Binary.search.HU}.
Then for $t=0,\dots,T$, we have the following bound:

\begin{equation}\label{eq:upper.bound}
\lambda^\star-\overline \lambda_t\le \varepsilon (1+\|\xi^\star\|_{\ell_1})\,.
\end{equation}
\end{lemma}
\begin{proof}
We will prove the result by induction, using Lemma \ref{lem:obj.at.X}.
\begin{itemize}
\item Base case: For $t=0$, we have
\begin{equation}
\lambda^\star-\overline \lambda_0=\lambda^\star-\lambda_{\max}\le 0\le \varepsilon\left(1+\|\xi^\star\|_{\ell_1}\right)\,.
\end{equation}
Thus \eqref{eq:upper.bound} holds for $t=0$.
\item Inductive Step:
Assume that \eqref{eq:upper.bound} holds for some $t=r\le T-1$.
We now show that it holds for $t=r+1$.
Consider Step 2 (c) of Algorithm \ref{alg:Binary.search.HU} with $t=r$.
\begin{itemize}
\item Case 1: $\eta_r=0$. In this case, we have $\overline \lambda_{r+1}=\overline \lambda_r$. By the induction assumption, it follows that
\begin{equation}
\overline \lambda_{r+1}=\overline \lambda_r\ge \lambda^\star - \varepsilon \left(1+\|\xi^\star\|_{\ell_1} \right)\,.
\end{equation}
\item Case 2: $\eta_r=1$.
By Lemma \ref{lem:Hamilton.updates},  we know that $X_\varepsilon^{(r)}\in S^{(\lambda_r)}_\varepsilon$.
By Lemma \ref{lem:obj.at.X}, we obtain
\begin{equation}
\tr( C X_\varepsilon^{(r)}) \ge\lambda^\star -\|\xi^\star\|_{\ell_1}  \max\limits_{i=1,\dots,M}|\tr(A_iX_\varepsilon^{(r)})-b_i| \ge \lambda^\star - \varepsilon \|\xi^\star\|_{\ell_1} \,.
\end{equation}
Since $X_\varepsilon^{(r)}\in S^{(\lambda_r)}_\varepsilon$, we have $|\tr(A_iX_\varepsilon^{(r)})-b_i|\le \varepsilon$ for all $i$ and the last inequality holds.
Therefore, we have the following estimate:
\begin{equation}
\overline \lambda_{r+1} +\varepsilon=\lambda_r+\varepsilon\ge \tr( C X_\varepsilon^{(r)})
\ge   \lambda^\star - \varepsilon \|\xi^\star\|_{\ell_1}\,.
\end{equation}
The first inequality is based on $X_\varepsilon^{(r)}\in S^{(\lambda_r)}_\varepsilon$.
It implies that
\begin{equation}\label{eq:ineq.rho}
\overline \lambda_{r+1} \ge \lambda^\star -\varepsilon(1+\|\xi^\star\|_{\ell_1} )\,.
\end{equation}
\end{itemize}
Combining both cases, we conclude that
\begin{equation}
\lambda^\star-\overline \lambda_{r+1}\le \varepsilon\left(1+\|\xi^\star\|_{\ell_1} \right)\,.
\end{equation}
\end{itemize}
Thus, by induction, \eqref{eq:upper.bound} holds for all $t=0,\dots,T$.
\end{proof}

\begin{remark}
Condition \eqref{eq:remove.trace.one} implies that the constraint $\tr(X)=1$ in problem \eqref{eq:sdp} can be omitted without affecting the optimal value of the problem. 
\end{remark}

The following lemma establishes the accuracy of the approximate optimal value for the SDP \eqref{eq:sdp}, as obtained through binary search using Hamiltonian updates:

\begin{lemma}[Accuracy]\label{lem:convergence.sdp}
Let $\varepsilon>0$, $C,A_i\in \mathbb S^N$, $b_i\in\R$, for $i=1,\dots,M$, and assume that $I_N\succeq C\succeq -I_N$ and $I_N\succeq A_i\succeq -I_N$, for $i=1,\dots,M$.
Assume that problem \eqref{eq:sdp} has an optimal solution.
Let $\lambda^\star$ be the optimal value in \eqref{eq:sdp}, and let $\lambda_{\min},\lambda_{\max}\in\R$ such that $\lambda^\star\in [\lambda_{\min},\lambda_{\max}]$.
Further, assume the following equality holds:
\begin{equation}\label{eq:remove.trace.one.1}
\lambda^\star=\inf\{\tr( CX)\,:\,X\in\mathbb S^{N}_+\,,\,\tr( A_i X) =b_i\,,\,i=1,\dots,M\}\,.
\end{equation}
Assume strong duality holds for problem \eqref{eq:remove.trace.one.1} and its dual \eqref{eq:dual.no.tr.cons}, which admits a solution $\xi^\star$.
Let $\underline \lambda_{T}\in\R$ be the value returned by Algorithm \ref{alg:Binary.search.HU}.
Then, the following convergence result holds:
\begin{equation}
0\le \lambda^\star-\overline \lambda_T\le \varepsilon (2+\|\xi^\star\|_{\ell_1})\,.
\end{equation}
\end{lemma}
\begin{proof}
From Lemma \ref{lem:zero.lower.bound}, we have $0\le \lambda^\star-\underline \lambda_T$.
From Lemma \ref{lem:upper.bound}, it follows that $\lambda^\star-\overline\lambda_T< \varepsilon\left(1+\|\xi^\star\|_{\ell_1} \right)$.
By the second statement of Lemma \ref{lem:monotoncity}, we also have $\overline \lambda_T-\underline \lambda_T\le \varepsilon$.
Combining these results, we obtain 
\begin{equation}
\begin{array}{rl}
0\le \lambda^\star-\underline \lambda_T=&(\lambda^\star-\overline \lambda_T)+(\overline\lambda_T-\underline \lambda_T)\\[5pt]
\le& 
\varepsilon\left(1+\|\xi^\star\|_{\ell_1} \right)+\varepsilon\le \varepsilon\left(2+\|\xi^\star\|_{\ell_1}\right)\,.
\end{array}
\end{equation}
Thus, the statement follows.
\end{proof}

The following lemma, which summarizes the classical and quantum complexities of binary search via Hamiltonian updates, is a consequence of Lemma~\ref{lem:MMWU-complexity}.

\begin{lemma}[Complexity]\label{lem:complex.binary.search}
The complexity of running Algorithm \ref{alg:Binary.search.HU} is as follows:
\begin{itemize}
\item {On a classical computer, the algorithm takes
\begin{equation}
O((N^\omega + M s N)\varepsilon^{-2})
\end{equation}
operations, where $N$ is the size of the matrix, $M$ is the number of affine constraints, $\varepsilon$ is the desired accuracy, $\omega\approx 2.373$, and $s$ is the maximum number of nonzero entries in any row of the matrices $C$ and $A_i$ (for $i=1,\dots,M$).}
\item On a quantum computer, the algorithm requires
\begin{equation}
O(s(\sqrt{M}+\sqrt{N}\varepsilon^{-1})\varepsilon^{-4})
\end{equation}
operations.
\end{itemize}
\end{lemma}

\section{Quantum algorithms for solving sum-of-squares relaxations of polynomial optimization}
\label{sec:solve.sos.relax}
\subsection{Unconstrained polynomial optimization}
\label{sec:uncons.pop}
We now turn to how to actually solve the sums-of-squares relaxations of polynomial optimization problems using quantum methods.
To build up to our complexity and accuracy guarantees, we start with the simplest setting—unconstrained polynomial optimization, where the structure of the relaxation and its SDP formulation can be more easily analysed.
Let $f=\sum_{\alpha\in\N^n}f_\alpha x^\alpha\in\R[x]$, and consider the problem
\begin{equation}\label{eq:uncons.pop}
f^\star=\inf\limits_{x\in \R^n} f(x)\,,
\end{equation}
which seeks the global minimum of $f$.

The sum-of-squares (SOS) relaxation of order $k$ for this problem is formulated as
\begin{equation}\label{eq:sos.relaxation}
\begin{array}{rl}
\lambda_k=\sup\limits_{\lambda\in\R}&\lambda\\
\text{s.t.}& f-\lambda\in \Sigma[x]_k\,,
\end{array}
\end{equation}
where $\Sigma[x]_k$ denotes the set of SOS polynomials of degree at most $2k$.
This relaxation satisfies the chain of inequalities $\lambda_k\le\lambda_{k+1}\le f^\star$.

The SOS relaxation \eqref{eq:sos.relaxation} can be equivalently reformulated as the following semidefinite program (SDP):
\begin{equation}\label{eq:sdp.sos.relaxation}
\begin{array}{rl}
\lambda_k=\sup\limits_{\lambda, G}&\lambda\\
\text{s.t.}&\lambda\in\R\,,\, G\succeq 0\,,\, f-\lambda = v_k^\top Gv_k\,,
\end{array}
\end{equation}
where $v_k$ is the vector of monomials up to degree $k$.

The dual of \eqref{eq:sdp.sos.relaxation} is given by
\begin{equation}\label{eq:sdp.mom.relaxation}
\begin{array}{rl}
\tau_k=\min\limits_y& L_y(f)\\
\text{s.t.}& y=(y_\alpha)_{\alpha\in\N^n_{2k}}\subset \R\,,\\
&M_k(y)\succeq 0\,,\,y_0=1\,,\\
\end{array}
\end{equation}
where $M_k(y)$ is the moment matrix of order $k$ associated with $y$, and $L_y$ is the Riesz linear functional on $\R[x]_{2k}$.

According to Lemma \ref{lem:slater.cond}, $\lambda_k=\tau_k$, and problem \eqref{eq:sdp.sos.relaxation} admits an optimal solution.

We make the following coefficient scaling assumption on the objective polynomial $f$ of the polynomial optimization \eqref{eq:uncons.pop}.This will turn out to be crucial to ensure that we can bound the trace of feasible points of the SDP.

\begin{assumption}\label{ass:cond.bounded.trace.1}
There exist $\underline \lambda\in\R$ and positive measure $\mu$ on $\R^n$ such that $\underline \lambda\le \lambda_k$ and 
\begin{equation}
\frac{\int(f-\underline \lambda)d\mu}{\lambda_{\min}(M_k(\mu))}\le 1\,.
\end{equation}
\end{assumption}

\begin{remark}\label{rem:scale.poly.scalar}
If we have a positive lower bound $\delta$ for the smallest eigenvalue of the moment matrix of order $k$ associated with some measure $\mu$, the coefficients of any polynomial $f$ and scalar $\underline{\lambda}\le \lambda_k$ can always be rescaled to produce a new polynomial $\tilde f$ and new scalar $\underline{\tilde \lambda}$ that satisfy Assumption \ref{ass:cond.bounded.trace.1}.
To see this, let $f\in\R[x]_{2k}$, $\underline{\lambda}\le \lambda_k$ and $0< \delta\le \lambda_{\min}(M_k(\mu))$. 
Set $\tau=\frac{\delta}{\int  (f-\underline{\lambda}) d\mu}$.
Since $f\ge f^\star\ge \lambda_k\ge \underline{\lambda}$, we have $f-\underline{\lambda}\ge 0$, which implies $\tau> 0$.
Define $\tilde f=\tau f$ and $\tilde {\underline{\lambda}}=\tau \underline{\lambda}$.
Now consider the rescaled sum-of-squares relaxation. By definition:
\begin{equation}
\tilde \lambda_k=\sup\{\lambda\in\R\,:\,\exists G\succeq 0\,,\, \tilde f-\lambda = v_k^\top Gv_k\}\,.
\end{equation}
Substituting $\tilde f=\tau f$, we have:
\begin{equation}
\begin{array}{rl}
\tilde \lambda_k=&\sup\{\lambda\in\R\,:\,\exists G\succeq 0\,,\, \tau f-\lambda = v_k^\top Gv_k\}\\[3pt]
=&\sup\{\lambda\in\R\,:\,\exists G\succeq 0\,,\, f-\frac{\lambda}\tau = v_k^\top (\frac{1}{\tau}G)v_k\}\\[3pt]
=&\sup\{\tau\tilde\lambda\in\R\,:\,\exists \tilde G\succeq 0\,,\, f-\tilde\lambda = v_k^\top \tilde Gv_k\}\\[3pt]
=&\tau\sup\{\tilde\lambda\in\R\,:\,\exists \tilde G\succeq 0\,,\, f-\tilde\lambda = v_k^\top \tilde Gv_k\}\\
=&\tau \lambda_k
\,.
\end{array}
\end{equation}
Thus, $\tilde f\in\R[x]_{2k}$, $\tilde {\underline{\lambda}}\le \tau \lambda_k=\tilde \lambda_k$, and
\begin{equation}
\frac{\int (\tilde f-\tilde {\underline{\lambda}}) d\mu}{\lambda_{\min}(M_k(\mu))}\le\frac{\int (\tilde f-\tilde {\underline{\lambda}}) d\mu}{\delta}= \frac{\tau \int (f-\underline{\lambda}) d\mu}{\delta} = 1\,.
\end{equation}
In Appendix \ref{sec:lower.bound.eigenval.moment.mat}, we provide an explicit value of $\delta$ when $\mu$ is the Lebesgue measure on $[0,1]^n$ (see Proposition \ref{prop:lower.bound.eigen.val}).
\end{remark}

The following lemma states that, under the coefficient scaling Assumption \ref{ass:cond.bounded.trace.1}, the trace of the Gram matrix that solves the sum-of-squares relaxation is upper bounded by one:

\begin{lemma}\label{lem:cond.bounded.trace.1}
Let $f\in\R[x]_{2k}$, and suppose Assumption \ref{ass:cond.bounded.trace.1} holds.
Let $G^\star$ be an optimal solution to the semidefinite program \eqref{eq:sdp.sos.relaxation}.
Then, $\tr(G^\star)\le 1$.
\end{lemma}
\begin{proof}
From the relation $f-\lambda_k=v_k G^\star v_k^\top$, we have
\begin{equation}
\int(f-\underline \lambda)d\mu\ge \int(f- \lambda_k)d\mu= \tr( G^\star M_k(\mu)) \ge \lambda_{\min}(M_k(\mu))\tr(G^\star)\,.
\end{equation}
From property $M_k(\mu)\succeq 0$, it follows that:
\begin{equation}
\tr(G^\star)\le \frac{\int(f-\underline \lambda)d\mu}{\lambda_{\min}(M_k(\mu))}\le 1\,,
\end{equation}
where the final inequality follows from Assumption \ref{ass:cond.bounded.trace.1}. 
Thus, the result is proven.
\end{proof}

To transform the sum-of-squares relaxation \eqref{eq:sdp.sos.relaxation} into the standard SDP form \eqref{eq:sdp}, we introduce the family of symmetric matrices $(B^{(\gamma)})_{\gamma\in\N^n_k}$, defined as follows.

\begin{definition}\label{def:Bgamma}
Let $k \in \mathbb{N}$ and $\gamma \in \N^{n}_{2k}$.  
The \emph{coefficient-matching matrix} $B^{(\gamma)}$ is the symmetric matrix indexed by $\alpha,\beta \in \N^{n}_{k}$, defined entrywise by
\begin{equation}\label{eq:def.B.gamma}
B^{(\gamma)}_{\alpha,\beta} \;:=\;
\begin{cases}
1 & \text{if } \alpha+\beta = \gamma,\\[2pt]
0 & \text{otherwise.}
\end{cases}
\end{equation}
\end{definition}
It is straightforward to show that $B^{(\gamma)}=M_k(z)$ with $z_\alpha=
1$ if $\alpha=\gamma$ and $z_\alpha=0$ else, $M_k(y)=\sum_{\gamma\in\N^n_{2k}} y_\gamma B^{(\gamma)}$, $v_kv_k^\top=\sum_{\gamma\in\N^n_{2k}} x^\gamma B^{(\gamma)}$ and that
$B^{(0)}=\diag(1,0,\dots,0)$.
\begin{example}
For $n=1$ and $k=2$, we have
\begin{equation}
v_k=\begin{bmatrix}
1\\
x\\
x^2
\end{bmatrix}\,,\,v_k^\top=\begin{bmatrix}
1&x&x^2
\end{bmatrix}\,,\,v_kv_k^\top=\begin{bmatrix}
1&x&x^2\\
x&x^2&x^3\\
x^2&x^3&x^4
\end{bmatrix}\,.
\end{equation}
Thus, we can express $v_kv_k^\top$ as a sum of the form
\begin{equation}
v_kv_k^\top= B^{(0)}+xB^{(1)}+x^2B^{(2)}+x^3B^{(3)}+x^4B^{(4)}
\end{equation}
where the matrices $B^{(i)}$  are given by
\begin{equation}
\begin{array}{rl}
&B^{(0)}=\begin{bmatrix}
1&0&0\\
0&0&0\\
0&0&0
\end{bmatrix}\,,\,B^{(1)}=\begin{bmatrix}
0&1&0\\
1&0&0\\
0&0&0
\end{bmatrix}\,,\,B^{(2)}=\begin{bmatrix}
0&0&1\\
0&1&0\\
1&0&0
\end{bmatrix}\,,\\[18pt]
 &B^{(3)}=\begin{bmatrix}
0&0&0\\
0&0&1\\
0&1&0
\end{bmatrix}\,,\,B^{(4)}=\begin{bmatrix}
0&0&0\\
0&0&0\\
0&0&1
\end{bmatrix}\,.
\end{array}
\end{equation}
\end{example}

The following properties of $B^{(\gamma)}$  will be useful later for analyzing the accuracy of the approximate value returned by the binary search using Hamiltonian updates, applied to solve the sum-of-squares relaxation \eqref{eq:sdp.sos.relaxation}:
\begin{lemma}\label{lem:properties.B.gamma}
Let $B^{(\gamma)}$ be defined as in \eqref{eq:def.B.gamma}. The following properties hold:
\begin{enumerate}[(i)]
\item\label{unique.nonzero.per.row} If $B^{(\gamma)}_{\alpha,\beta}=1$, then $ B^{(\gamma)}_{\alpha,\beta'}=0$, for all $\beta'\ne \beta$, and $B^{(\gamma)}_{\alpha,\beta}=1 $ implies $ B^{(\gamma)}_{\alpha',\beta}=0$, for all $\alpha'\ne \alpha$. 
\item\label{unique.one.each.position} If $B^{(\gamma)}_{\alpha,\beta}=1$, then $ B^{(\gamma')}_{\alpha,\beta}=0$, for all $\gamma'\ne \gamma$.
\item\label{square.mat} $B^{(\gamma)2}$ is a diagonal matrix with diagonal entries that are either $1$ or $0$. Specifically, 
\begin{equation}
B^{(\gamma)2}=\diag(0,\dots,0,1,0,\dots,0,1,0,\dots,0)\,.
\end{equation}
\item\label{bound.eigenvalue} $-I\preceq B^{(\gamma)}\preceq I$.
\item\label{trace.evalution} $\tr(B^{(\gamma)}B^{(\gamma')})\begin{cases}
=0 &\text{if }\gamma\ne \gamma'\,,\\
\ge 1&\text{else.}
\end{cases}$
\item\label{linearly.indepen.mat} The matrices $B^{(\gamma)}$, for $\gamma\in\N^n_k$, are linearly independent in $\mathbb S^{\binom{n+k}n}$.
\end{enumerate}
\end{lemma}

\begin{proof}
\eqref{unique.nonzero.per.row} Suppose $B^{(\gamma)}_{\alpha,\beta}=1= B^{(\gamma)}_{\alpha,\beta'}$ for some $\beta'\ne \beta$. Then, $\alpha+\beta=\gamma=\alpha+\beta'$, which implies $\beta=\beta'$, a contradiction.
A similar argument holds for the columns due to the symmetry of $B^{(\gamma)}$.

\eqref{unique.one.each.position} Suppose $B^{(\gamma)}_{\alpha,\beta}=1= B^{(\gamma')}_{\alpha,\beta}$ for some $\gamma'\ne \gamma$. Then, $\gamma=\alpha+\beta=\gamma'$, implying $\gamma=\gamma'$, a contradiction.

The following implications hold:
\eqref{unique.nonzero.per.row} $\Rightarrow$ \eqref{square.mat} $\Rightarrow$ \eqref{bound.eigenvalue}, and \eqref{unique.one.each.position} $\Rightarrow$ \eqref{trace.evalution} $\Rightarrow$ \eqref{linearly.indepen.mat}.
\end{proof}
In addition, we show in Appendix~\ref{sec:block.encoding}
how to explicittly generate block encodings for $B^{(\gamma)}$ by applying $O(n\log(k))$ gates, showing that this is not a bottleneck for the application of the methods.

We now present our first main algorithm, which provides an approximation of the optimal value to the SOS relaxation of order $k$.

\begin{algorithm}\label{alg:uncon.pop}
Unconstrained polynomial optimization
\begin{itemize}
\item Input: $\varepsilon>0$, $\underline \lambda\in\R$, $a\in\R^n$,\\
\phantom{1.1cm} coefficients $(f_\alpha)_{\alpha\in\N^n_{2k}}$ of $f\in\R[x]_{2k}$.
\item Output: $\lambda_k^{(\varepsilon)}\in\R$ and $G_\varepsilon\in\mathbb S^N$.
\end{itemize}
\begin{enumerate}
\item Convert relaxation to standard semidefinite program:
\begin{enumerate}
\item\label{step:index.powers} Let $\N^n_{2k}\backslash \{0\}=\{\gamma_1,\dots,\gamma_{M}\}$, with $M=|\N^n_{2k}\backslash \{0\}|=\binom{n+2k}n-1$.
\item Define $\tilde C=B^{(0)}$, $\tilde A_i=B^{(\gamma_i)}$, $b_i=f_{\gamma_i}$, for $i=1,\dots,M$, and set $\tilde N=\binom{n+k}n$.

\item Define $C=\diag(\tilde C,0)$, $A_i=\diag(\tilde A_i,0)$, for $i=1,\dots,M$, and set $N=\tilde N+1$.
\end{enumerate}
\item Initialize search bounds: Set $\lambda_{\max}=f_0-\underline \lambda$ and $\lambda_{\min}=f_0-f(a)$.
\item Run Algorithm \ref{alg:Binary.search.HU}  (Binary search using Hamiltonian updates)
to get $\underline \lambda_T\in\R$ and $X_\varepsilon=((X_\varepsilon)_{ij})_{i,j=1,\dots,N}\in\mathbb S^N$.

\item Extract approximate results: Set $\lambda_k^{(\varepsilon)}=f_0-\underline \lambda_T$ and $G_\varepsilon=((X_\varepsilon)_{ij})_{i,j=1,\dots,\tilde N}$.
\end{enumerate}
\end{algorithm}

\begin{remark}
In Step \ref{step:index.powers} of Algorithm \ref{alg:uncon.pop}, elements in $\N^n_{2k}\backslash \{0\}$ can be indexed using lexicographic order.
\end{remark}

The following lemma demonstrates that SOS relaxations can be reformulated as standard-form SDPs satisfying the bounded-trace property.

\begin{lemma}\label{lem:equi.sdp.relax.uncons}
Let $f\in\R[x]$, and suppose that Assumption \ref{ass:cond.bounded.trace.1} holds.
Let $\tilde C$, $\tilde A_i$, for $i=1,\dots,M$, and $\tilde N$ be generated by Algorithm \ref{alg:uncon.pop}.
Then the problem in \eqref{eq:sdp.sos.relaxation} is equivalent to:
\begin{equation}\label{eq:sdp.bounded.trace}
\begin{array}{rl}
f_0-\lambda_k=\inf\limits_{ G}&\tr( \tilde C G)\\
\text{s.t.}&G\in\mathbb S^{\tilde N}_+\,,\,\tr(G)\le 1\,,\\
&\tr( \tilde A_i G) =b_i\,,\,i=1,\dots,M\,.
\end{array}
\end{equation}
Moreover, this is further equivalent to the following form, where the trace constraint is removed:
\begin{equation}\label{eq:sdp.remove.bounded.trace}
\begin{array}{rl}
f_0-\lambda_k=\inf\limits_{ G}&\tr( \tilde CG)\\
\text{s.t.}&G\in\mathbb S^{\tilde N}_+\,,\\
&\tr( \tilde A_i G) =b_i\,,\,i=1,\dots,M\,.
\end{array}
\end{equation}
\end{lemma}

\begin{proof}
By Lemma \ref{lem:cond.bounded.trace.1}, the problem in  \eqref{eq:sdp.sos.relaxation} is equivalent to:
\begin{equation}\label{eq:sdp.sos.relaxation.bounded.trace}
\begin{array}{rl}
\lambda_k=\sup\limits_{\lambda, G}&\lambda\\
\text{s.t.}&\lambda\in\R\,,\, G\succeq 0\,,\,\tr(G)\le 1\,,\, f-\lambda = v_k^\top Gv_k\,.
\end{array}
\end{equation}
To express this in standard SDP form, let $G=(G_{\alpha,\beta})_{\alpha,\beta\in\N^n_k}$. Following the steps in SOS testing, we have:
\begin{equation}
f-\lambda=v_k^\top Gv_k \\
\Leftrightarrow   \begin{cases}
f_0-\lambda=G_{0,0}=\tr( B^{(0)} G)\,,\\
f_\gamma=\sum\limits_{
\arraycolsep=1.4pt\def\arraystretch{.7}
\begin{array}{cc}
\scriptstyle \alpha,\beta\in\N^n_k\\
\scriptstyle \alpha+\beta=\gamma
\end{array}
} G_{\alpha,\beta}=\tr( B^{(\gamma)} G) \,,\,\gamma\in \N^n_{2k}\backslash \{0\}\,.
\end{cases}
\end{equation}
Rewriting, we find:
\begin{equation}
\begin{cases}
\lambda=f_0-\tr( B^{(0)} G)\,,\\
f_{\gamma_i}=\tr( B^{(\gamma_i)} G) \,,\,i=1,\dots,M\,,
\end{cases}
\end{equation}
which leads to:
\begin{equation}
\begin{cases}
\lambda=f_0-\tr( \tilde C G)\,,\\
b_i=\tr( \tilde A_i G) \,,\,i=1,\dots,M\,.
\end{cases}
\end{equation}
Thus, \eqref{eq:sdp.sos.relaxation.bounded.trace} can be reformulated as:
\begin{equation}\label{eq:sdp.relaxation.bounded.trace}
\begin{array}{rl}
\lambda_k=\sup\limits_{ G}&f_0-\tr( \tilde C G)\\
\text{s.t.}&G\in\mathbb S^{\tilde N}_+\,,\,\tr(G)\le 1\,,\\
&\tr( \tilde A_i G) =b_i\,,\,i=1,\dots,{ M}\,.
\end{array}
\end{equation}
This is equivalent to \eqref{eq:sdp.bounded.trace}.

Finally, we observe that removing the trace constraint
$\tr(G)\le 1$
does not affect the optimal value or solutions of \eqref{eq:sdp.sos.relaxation.bounded.trace}, as shown in the proof of Lemma \ref{lem:cond.bounded.trace.1}.
Hence, the problem in \eqref{eq:sdp.sos.relaxation} is also equivalent to \eqref{eq:sdp.remove.bounded.trace}.
Thus, the result follows.
\end{proof}

The next lemma shows that an SDP with the bounded-trace property can be transformed into an equivalent SDP with the constant-trace property by appropriately lifting its matrix variable.

\begin{lemma}\label{lem:equivalent.sol.uncons}
Let $f\in\R[x]_{2k}$ and suppose that Assumption \ref{ass:cond.bounded.trace.1} holds.
Define $\lambda^\star=f_0-\lambda_k$.
Let $C$, $b_i$, $A_i$  (for $i=1,\dots,M$), $\tilde N$, and $N$ be generated by Algorithm \ref{alg:uncon.pop}.
Then problem \eqref{eq:sdp} holds, and it is equivalent to problem \eqref{eq:sdp.bounded.trace} in the following sense:
\begin{itemize}
\item If $G^\star$ is an optimal solution to  \eqref{eq:sdp.bounded.trace}, then the matrix $\diag(G^\star, c)$, with $c=1-\tr(G^\star)$, is an optimal solution to \eqref{eq:sdp}.
\item Conversely, if $X^\star=(X^\star_{ij})_{i,j=1,\dots,N}$ is an optimal solution to \eqref{eq:sdp}, then the leading principal submatrix of $X^\star$, $\tilde X^\star=(X^\star_{ij})_{i,j=1,\dots,\tilde N}$, is an optimal solution to \eqref{eq:sdp.bounded.trace}.
\end{itemize}
Furthermore, the following equality holds:
\begin{equation}\label{eq:remove.trace.one.1.2}
\lambda^\star=\inf\{\tr( C X)\,:\,X\in\mathbb S^{N}_+\,,\,\tr( A_i X) =b_i\,,\,i=1,\dots,M\}\,.
\end{equation}
\end{lemma}
\begin{proof}
We first prove that \eqref{eq:sdp} holds. Let $X$ be a feasible solution to \eqref{eq:sdp}, and define $\tilde X=(X_{ij})_{i,j=1,\dots,\tilde N}$. Since $X\succeq 0$, , it follows that $\tilde X\succeq 0$. Additionally, because $\tr(X)=1$, we have $\tr(\tilde X)\le 1$.

From Lemma \ref{lem:equi.sdp.relax.uncons}, \eqref{eq:sdp.bounded.trace} holds.
Moreover, since $A_i=\diag(\tilde A_i, 0)$, we have $b_i=\tr( A_i X)=\tr( \tilde A_i \tilde X)$, which implies that $\tilde X$ is a feasible solution to \eqref{eq:sdp.bounded.trace}. 
Thus, $\lambda^\star=f_0-\lambda_k\le \tr( \tilde C \tilde X)=\tr( C X)$, using $C=\diag(\tilde C, 0)$.

Now, let $G^\star$ be an optimal solution to problem \eqref{eq:sdp.bounded.trace}.
Define $Y=\diag(G^\star, c)$, where $c=1-\tr(G^\star)$.
Since $\tr(G^\star)\le 1$, it follows that $c\ge 0$, implying $Y\in\mathbb S^N_+$.
Additionally, we have $\tr(Y)=\tr(G^\star)+c=1$, and and for all $i=1,\dots,M$,
$b_i=\tr(\tilde A_i G^\star)=\tr( A Y)$.
This shows that $Y$ is a feasible solution to \eqref{eq:sdp}.
Furthermore, $\lambda^\star=f_0-\lambda_k=\tr( \tilde C G^\star)=\tr( C Y)$ since $C=\diag(\tilde C, 0)$ and $A_i=\diag(\tilde A_i, 0)$, for all $i=1,\dots,M$.
Thus, \eqref{eq:sdp} holds, and $Y$ is an optimal solution to \eqref{eq:sdp}.

Next, let $X^\star=(X^\star_{ij})_{i,j=1,\dots,N}$ be an optimal solution to \eqref{eq:sdp}, and define $\tilde X^\star=(X^\star_{ij})_{i,j=1,\dots,\tilde N}$.
Then $\tilde X^\star\in \mathbb S^{\tilde N}$, and $\tr(\tilde X^\star)=\tr(X^\star)-X^\star_{NN}=1-X^\star_{NN}\le 1$.
Additionally, for all $i=1,\dots,M$,  $b_i=\tr(A_i X^\star)=\tr(\tilde A_i \tilde X^\star)$, and $f_0-\lambda_k=\lambda^\star=\tr(C X^\star)=\tr(\tilde C \tilde X^\star)$, using $C=\diag(\tilde C, 0)$ and $A_i=\diag(\tilde A_i, 0)$, $i=1,\dots,M$.
Therefore, $\tilde X^\star$ is an optimal solution to \eqref{eq:sdp.bounded.trace}.

By the last statement of Lemma \ref{lem:equi.sdp.relax.uncons}, the equality \eqref{eq:remove.trace.one.1.2} holds.
Hence the result follows.
\end{proof}

The following lemma establishes key properties of the input data in the standard form of SDPs. These properties will later enable the application of the error bound for Hamiltonian updates.

\begin{lemma}\label{lem:properties.data.uncons}
Let $\lambda^\star=f_0-\lambda_k$, and let $C$, $b_i$, $A_i$, for $i=1,\dots,M$, $N$, $\lambda_{\max}$, $\lambda_{\min}$ be generated by Algorithm \ref{alg:uncon.pop}.
The following properties hold:
\begin{enumerate}[(i)]
\item \label{bound.minimum} $\lambda^\star\in [\lambda_{\min},\lambda_{\max}]$.
\item \label{bound.data.mat} $-I_N\preceq C\preceq I_N$, $-I_N\preceq A_i\preceq I_N$, for all $i=1,\dots,M$.
\item\label{sparsity} The maximum number of nonzero entries in each row of $C$ and $A_i$, for $i=1,\dots,M$, is $s=1$.
\end{enumerate}
\end{lemma}

\begin{proof}
\eqref{bound.minimum} Since $\underline{\lambda}\le \lambda_k\le f^\star\le f(a)$, it follows that
\begin{equation}
\lambda_{\min}=f_0-f(a)\le f_0-\lambda_k=\lambda^\star\le f_0-\underline{\lambda}=\lambda_{\max}\,.
\end{equation}

\eqref{bound.data.mat} By construction, $C=\diag(\tilde C,0)=\diag(B^{(0)},0)$,  and $A_i=\diag(\tilde A_i,0)=\diag(B^{(\gamma_i)},0)$, for $i=1,\dots,M$.
By Lemma \ref{lem:properties.B.gamma} \eqref{bound.eigenvalue}, it follows that $-I_N\preceq C\preceq I_N$, $-I_N\preceq A_i\preceq I_N$.

\eqref{sparsity} By  Lemma \ref{lem:properties.B.gamma} \eqref{unique.nonzero.per.row}, each row of $C$ and $A_i$, for $i=1,\dots,M$, contains at most $s=1$ nonzero entry.
\end{proof}

The lemma below characterizes the correspondence between the optimal solutions for the moment relaxation and the equivalent standard-form semidefinite program.

\begin{lemma}\label{lem:equivalent.mom.relax.0}
Let $C$, $A_i,b_i$, for  $i=1,\dots,M$, and $N$ be generated by Algorithm \ref{alg:uncon.pop}.
Assume that strong duality holds for problems \eqref{eq:sdp.sos.relaxation}-\eqref{eq:sdp.mom.relaxation},as well as for problem \eqref{eq:remove.trace.one.1.2} and its dual \eqref{eq:dual.no.tr.cons}.
Let $\xi\in\R^M$, and define $y=(y_{\gamma})_{\gamma\in\N^n_{2k}}$ such that $y_{0}=1$ and $y_{\gamma_i}=-\xi_i$, for $i=1,\dots,M$.
Then $y$ is an optimal solution to the moment relaxation \eqref{eq:sdp.mom.relaxation} if and only if $\xi$ is an optimal solution to problem \eqref{eq:dual.no.tr.cons}.
\end{lemma}

\begin{proof}
By definition of $b,\tilde C,\tilde A_i$, we have $L_y(f)=y_0 f_0+\sum_{i=1}^M y_{\gamma_i}f_{\gamma_i}=f_0-b^\top \xi$ and 
$M_k(y)=y_0 B^{(0)}+\sum_{i=1}^M y_{\gamma_i}B^{(\gamma_i)} = \tilde C-\sum_{i=1}^M \xi_i\tilde A_i$, which implies $\diag(M_k(y),0)= C-\sum_{i=1}^M \xi_i A_i=C-\mathcal A^\top \xi$.
From this, the following equivalences hold:
\begin{equation}
\begin{array}{rl}
&y \text{ is an optimal solution to problem \eqref{eq:sdp.mom.relaxation}}\\[5pt]
\Leftrightarrow&  M_k(y)\succeq 0 \text{ and }L_y(f)=\tau_k=\lambda_k \qquad \text{ (by strong duality)}\\[5pt]
  \Leftrightarrow & C-\mathcal A^\top \xi\succeq 0 \text{ and }b^\top \xi =f_0-\lambda_k=\lambda^\star=\tau^\star \qquad\text{ (by strong duality)}\\[5pt]
  \Leftrightarrow &\xi\text{ is an optimal solution for \eqref{eq:dual.no.tr.cons}.}
\end{array}
\end{equation}
\end{proof}

We now present our first main theorem, which establishes an error bound for the approximation produced by our first main algorithm.
\begin{theorem}\label{theo:accuarcy.approximate.val.uncons}
Let $\varepsilon>0$, $f\in\R[x]_{2k}$, and suppose that Assumption \ref{ass:cond.bounded.trace.1} holds.
Assume strong duality holds for primal-dual semidefinite relaxations \eqref{eq:sdp.sos.relaxation}-\eqref{eq:sdp.mom.relaxation}.
Let $\lambda_k^{(\varepsilon)}$ be the value returned by Algorithm \ref{alg:uncon.pop}.
If $\lambda_k=f^\star$ and polynomial optimization problem \eqref{eq:uncons.pop} has an optimal solution in $\ell_1$-ball $\{x\in\R^n\,:\,\|x\|_{\ell_1}\le r\}$ for some $r\in (0,1)$, then
\begin{equation}
0\le \lambda_k^{(\varepsilon)}-f^\star\le\varepsilon\left(2+\frac{r}{1-r}\right)\,.
\end{equation}
\end{theorem}
\begin{proof}
Let $\lambda^\star=f_0-\lambda_k$.
Since strong duality holds for problems  \eqref{eq:sdp.sos.relaxation}-\eqref{eq:sdp.mom.relaxation}, it holds for problem \eqref{eq:remove.trace.one.1.2} and its dual \eqref{eq:dual.no.tr.cons} by Lemma \ref{lem:equivalent.sol.uncons}.

Assume $\lambda_k=f^\star$ and polynomial optimization problem \eqref{eq:uncons.pop} has an optimal solution  $x^\star$ satisfying $\|x^\star\|_{\ell_1}\le r$ for some $r\in (0,1)$.
Let $y^\star=(y^\star_\alpha)_{\alpha\in\N^n_{2k}}$ be the truncated moment sequence with respect to the Dirac measure $\delta_{x^\star}$, i.e., $y^\star_\alpha=\int x^\alpha d \delta_{x^\star}=x^{\star \alpha}$.

We claim that $y^\star$ is an optimal solution to problem \eqref{eq:sdp.mom.relaxation}.
First, observe that $y^{\star}_0=(x^{\star})^0=1$ and $M_{k}(y^{\star})=v_k(x^\star)v_k(x^\star)^\top\succeq 0$.
Then $y^\star$ is a feasible solution for problem \eqref{eq:sdp.mom.relaxation}.
Additionally, $L_{y^\star}(f)=\int f d \delta_{x^\star}=f(x^\star)=f^\star=\lambda_k=\tau_k$. 
The last equality follows from strong duality. Therefore, $y^\star$ is an optimal solution to problem \eqref{eq:sdp.mom.relaxation}.
By Lemma \ref{lem:equivalent.mom.relax.0}, this implies $-\tilde y^\star$, with $\tilde y^\star=(y^\star_\gamma)_{\gamma\in\N^n_{2k}\backslash \{0\}}$, is an optimal solution to problem \eqref{eq:dual.no.tr.cons}.

By Lemma \ref{lem:properties.data.uncons}, \eqref{bound.minimum} and \eqref{bound.data.mat}, $I_N\succeq C\succeq -I_N$, $I_N\succeq A_i\succeq -I_N$, for $i=1,\dots,M$, and $\lambda^\star\in [\lambda_{\min},\lambda_{\max}]$. 
By Lemma \ref{lem:equivalent.sol.uncons}, \eqref{eq:remove.trace.one.1.2} holds.
From these, by Lemma \ref{lem:convergence.sdp}, $0\le \lambda^\star -\underline \lambda_T\le \varepsilon(2+\|\tilde y^\star\|_{\ell_1})$.
Recall that $\lambda_k^{(\varepsilon)}=f_0-\underline \lambda_T$. Substituting, we get:
$\lambda^\star -\underline \lambda_T=\lambda_k^{(\varepsilon)}-\lambda_k=\lambda_k^{(\varepsilon)}-f^\star$.
Thus,
\begin{equation}\label{eq:bound.gap.uncons}
0\le \lambda_k^{(\varepsilon)}-f^\star\le\varepsilon\left(2+\|\tilde y^\star\|_{\ell_1}\right)\,.
\end{equation}
Next, we bound the $\ell_1$-norm of $\tilde y^\star$.
By the multinomial theorem, 
\begin{equation}
(x_1+\dots+x_n)^j =\sum_{\arraycolsep=1.4pt\def\arraystretch{.7}
\begin{array}{cc}
\scriptstyle \gamma\in\N^{n}\\
\scriptstyle |\gamma|=j
\end{array}
} \binom{j}{\gamma} x^\gamma\,,\, \forall x\in \R^n\,,\,\forall j\in \N\,.
\end{equation}
Setting $x_i=|x_i^\star|$, we have
\begin{equation}
\begin{array}{rl}
r^j\ge& \|x^\star\|_{\ell_1}^j=(|x_1^\star|+\dots+|x_n^\star|)^j\\[5pt]
=& \sum_{\arraycolsep=1.4pt\def\arraystretch{.7}
\begin{array}{cc}
\scriptstyle \gamma\in\N^{n}\\
\scriptstyle |\gamma|=j
\end{array}
}\binom{j}{\gamma} |x_1^\star|^{\gamma_1}\dots|x_n^\star|^{\gamma_n}\\
\ge& \sum_{\arraycolsep=1.4pt\def\arraystretch{.7}
\begin{array}{cc}
\scriptstyle \gamma\in\N^{n}\\
\scriptstyle |\gamma|=j
\end{array}
} |x^{\star\gamma}| =\sum_{\arraycolsep=1.4pt\def\arraystretch{.7}
\begin{array}{cc}
\scriptstyle \gamma\in\N^{n}\\
\scriptstyle |\gamma|=j
\end{array}
} |y^\star_\gamma|\,.
\end{array}
\end{equation}
The last inequality holds because $\binom{j}{\gamma}\ge 1$ for all $\gamma\in\N^n$ with $|\gamma|=j$.
Thus,
\begin{equation}
\|\tilde y^\star\|_{\ell_1}=\sum_{\gamma\in\N^n_{2k}\backslash \{0\}} |y^\star_\gamma|=\sum_{j=1}^{2k} \sum_{\arraycolsep=1.4pt\def\arraystretch{.7}
\begin{array}{cc}
\scriptstyle \gamma\in\N^{n}\\
\scriptstyle |\gamma|=j
\end{array}
} |y^\star_\gamma|\le \sum_{j=1}^{2k} r^j\le r\sum_{t=0}^\infty r^t=\frac{r}{1-r}\,.
\end{equation}
From this and \eqref{eq:bound.gap.uncons}, the result  follows.

\end{proof}

\begin{lemma}[Complexity]
The runtime complexity of  Algorithm \ref{alg:uncon.pop} is as follows:
\begin{itemize}
\item  On a classical computer, it requires
\begin{equation}
O\left(\left[\binom{n+2k}n \binom{n+k}n+\binom{n+k}n^\omega\right]\frac{1}{\varepsilon^2}\right)
\end{equation}
operations, where $\omega\approx 2.373$.
\item On a quantum computer, it requires
\begin{equation}
O\left(\left[\binom{n+2k}n^{1/2}+\frac{1}\varepsilon\binom{n+k}n^{1/2}\ \right]\frac{1}{\varepsilon^4}\right)
\end{equation}
operations.
\end{itemize}
\end{lemma}
\begin{proof}
The most computationally intensive part of Algorithm \ref{alg:uncon.pop} is the execution of Algorithm \ref{alg:Binary.search.HU}, which solves a semidefinite program (SDP). The SDP has $M$ affine equality constraints and operates on a positive semidefinite matrix of size $N$.
By Lemma \ref{lem:complex.binary.search}, the runtime complexity for solving this SDP is:
\begin{itemize}
\item On a classical computer: $O((MsN+N^\omega)\varepsilon^{-2})$,
\item On a quantum computer: $O(s(\sqrt{M}+\sqrt{N}\varepsilon^{-1})\varepsilon^{-4})$,
\end{itemize}
where $s$ is the maximum number of nonzero entries in any row of the data matrices $C,A_i$, for $i=1,\dots,M$.
From Lemma \ref{lem:properties.data.uncons} \eqref{sparsity}, we know $s=1$.
Additionally: $M=\binom{n+2k}n-1$ and $N=\binom{n+k}n+1$.
Substituting these values into the complexity expressions gives the result.
\end{proof}

\subsection{Inequality-constrained polynomial optimization}
\label{sec:cons.pop}
Building upon the unconstrained case, we now show similar results for inequality constrained polynomial optimization.
Let $f=\sum_{\alpha\in\N^n}f_\alpha x^\alpha\in\R[x]$ and $g_i=\sum_{\alpha\in\N^n}g_{i,\alpha}x^\alpha\in\R[x]$, for $i=0,\dots,m$, where $g_0=1$.
Define $d=\lceil \deg(f)/2\rceil$ and $d_i=\lceil\deg(g_i)/2 \rceil$, $i=0,\dots,m$.
Consider the polynomial optimization problem
\begin{equation}\label{eq:cons.pop}
f^\star=\inf_{x\in S(g)} f(x)\,,
\end{equation}
where $S(g)$ is the semialgebraic set defined as
\begin{equation}\label{eq:def.semialgebaric.set}
S(g)=\{x\in\R^n\,:\,g_i(x)\ge 0\,,\,i=1,\dots,m\}\,,
\end{equation}
with $g=(g_i)_{i=1}^m$.

The sum-of-squares (SOS) relaxation of order $k$ for problem \eqref{eq:cons.pop} is formulated as:
\begin{equation}\label{eq:sos.relaxation.cons.pop}
\begin{array}{rl}
\lambda_k=\sup\limits_{\lambda,\sigma_i} & \lambda\\
\text{s.t.}& \lambda\in\R\,,\,\sigma_i\in\Sigma[x]_{k-d_i}\,,\,f-\lambda=\sum_{i=0}^m\sigma_i g_i\,.
\end{array}
\end{equation}
Here, $\Sigma[x]_{k-d_i}$ denotes the cone of sum-of-squares polynomials of degree at most $2(k-d_i)$.
It is known that
$\lambda_k\le \lambda_{k+1}\le f^\star$.
Furthermore, if $g_m= R-\|x\|_{\ell_2}^2$ for some $R>0$, then $\lambda_k\to f^\star$ as $k\to \infty$.

Problem \eqref{eq:sos.relaxation.cons.pop} can be reformulated as a semidefinite program (SDP):
\begin{equation}\label{eq:sdp.relaxation.cons.pop}
\begin{array}{rl}
\lambda_k=\sup\limits_{\lambda,G} & \lambda\\
\text{s.t.}& \lambda\in\R\,,\,G=\diag(G_0,\dots,G_m)\succeq 0\,,\\
&f-\lambda=\sum_{i=0}^m g_iv_{k-d_i}^\top G_iv_{k-d_i}\,.
\end{array}
\end{equation}
Here, $v_{k-d_i}$ is the vector of monomials of degree at most $k-d_i$.

The duals of \eqref{eq:sdp.relaxation.cons.pop} reads as:
\begin{equation}\label{eq:sdp.relaxation.cons.pop.dual}
\begin{array}{rl}
\tau_k=\min\limits_y& L_y(f)\\
\text{s.t.}& y=(y_\alpha)_{\alpha\in\N^n_{2k}}\subset \R\,,\\
&M_k(y)\succeq 0\,,\,y_0=1\,,\\
&M_{k-d_i}(g_iy)\succeq 0\,,\,i=1,\dots,m\,,\\
\end{array}
\end{equation}
where $M_k(y)$ is the moment matrix of order $k$ associated with $y$, $M_{k-u}(py)$ is the localizing matrix of order $k-u$ associated with $y$ and a polynomial $p$ of degree at most $2u$, and $L_y$ is the Riesz linear functional on $\R[x]_{2k}$.

From Lemma \ref{lem:strong.duality.ball}, if $g_1=R-\|x\|_{\ell_2}^2$ for some $R>0$, then $\lambda_k=\tau_k$, and problem \eqref{eq:sdp.relaxation.cons.pop.dual} admits an optimal solution.
By Lemma \ref{lem:slater.cond}, if $S(g)$ has nonempty interior, then $\lambda_k=\tau_k$, and problem \eqref{eq:sdp.relaxation.cons.pop} admits an optimal solution.

In this section, we assume that problem \eqref{eq:sdp.relaxation.cons.pop} has an optimal solution.

For $g=(g_i)_{i=1,\dots,m}$ with $g_i\in\R[x]_{2d_i}$, the block-diagonal moment/localizing matrix of order $k$ associated with a sequence $y=(y_\alpha)_{\alpha\in\N^n_{2k}}$ is defined as
\begin{equation}\label{eq:block.diagonal.mom.local}
D_k(gy)=\diag((M_{k-d_i}(g_i y))_{i=0,\dots,m})\,,
\end{equation}
where $g_0=1$ and $d_0=0$.
For a given measure $\mu$, the block-diagonal moment/localizing matrix of order $k$ associated with $\mu$ is defined as $D_k(g\mu)=D_k(gy)$, where $y$ is the truncated moment sequence with respect to $\mu$.

We make the following assumption:
\begin{assumption}\label{ass:cond.bounded.trace.2}
There exist $\underline \lambda\in\R$ and positive measure $\mu$ supported on $S(g)$ such that $\underline \lambda\le \lambda_k$ and 
\begin{equation}
\frac{\int(f-\underline \lambda)d\mu}{\lambda_{\min}(D_k(g\mu))}\le 1\,.
\end{equation}
\end{assumption}

\begin{remark}
If we have a positive lower bound $\delta$ for the smallest eigenvalue of the block-diagonal moment/localizing matrix of order $k$ associated with some measure $\mu$ supported on $S(g)$, the coefficients of any polynomial $f$ and scalar $\underline{\lambda}\le \lambda_k$ can always be rescaled to produce a new polynomial $\tilde f$ and new scalar $\underline{\tilde \lambda}$ that satisfy Assumption \ref{ass:cond.bounded.trace.1}.
To see this, let $f\in\R[x]_{2k}$, $\underline{\lambda}\le \lambda_k$ and $0< \delta\le \lambda_{\min}(D_k(g\mu))$. 
Set $\tau=\frac{\delta}{\int  (f-\underline{\lambda}) d\mu}$.
Since $f\ge f^\star\ge \lambda_k\ge \underline{\lambda}$ on $S(g)$, we have $f-\underline{\lambda}\ge 0$ on $S(g)$, which implies $\tau> 0$.
Define $\tilde f=\tau f$ and $\tilde {\underline{\lambda}}=\tau \underline{\lambda}$.
Now consider the rescaled sum-of-squares relaxation. By definition:
\begin{equation}
\tilde \lambda_k=\sup\{\lambda\in\R\,:\,\exists G_i\succeq 0\,,\,i=0,\dots,m\,,\, \tilde f-\lambda = \sum_{i=0}^m g_iv_{k-d_i}^\top G_iv_{k-d_i}\}\,.
\end{equation}
Substituting $\tilde f=\tau f$, we have:
\begin{equation}
\begin{array}{rl}
\tilde \lambda_k=&\sup\{\lambda\in\R\,:\,\exists G_i\succeq 0\,,\, \tau f-\lambda = \sum_{i=0}^m g_iv_{k-d_i}^\top G_iv_{k-d_i}\}\\[3pt]
=&\sup\{\lambda\in\R\,:\,\exists G_i\succeq 0\,,\, f-\frac{\lambda}\tau = \sum_{i=0}^m g_iv_{k-d_i}^\top (\frac{1}{\tau}G)v_{k-d_i}\}\\[3pt]
=&\sup\{\tau\tilde\lambda\in\R\,:\,\exists \tilde G_i\succeq 0\,,\, f-\tilde\lambda = \sum_{i=0}^m g_iv_{k-d_i}^\top \tilde G_iv_{k-d_i}\}\\[3pt]
=&\tau\sup\{\tilde\lambda\in\R\,:\,\exists \tilde G_i\succeq 0\,,\, f-\tilde\lambda = \sum_{i=0}^m g_iv_{k-d_i}^\top \tilde G_iv_{k-d_i}\}\\
=&\tau \lambda_k
\,.
\end{array}
\end{equation}
Thus, $\tilde f\in\R[x]_{2k}$, $\tilde {\underline{\lambda}}\le \tau \lambda_k=\tilde \lambda_k$, and
\begin{equation}
\frac{\int (\tilde f-\tilde {\underline{\lambda}}) d\mu}{\lambda_{\min}(D_k(g\mu))}\le\frac{\int (\tilde f-\tilde {\underline{\lambda}}) d\mu}{\delta}= \frac{\tau \int (f-\underline{\lambda}) d\mu}{\delta} = 1\,.
\end{equation}
In Appendix \ref{sec:lower.bound.eigenval.moment.mat}, we derive a positive lower bound for the smallest eigenvalue of $D_k(g\mu)$ in case where the coefficients of $g_i$ are rational, and $S(g)$ has nonempty interior (see Proposition \ref{prop:lower.bound.eigenval.constrained}). 
\end{remark}
In Sec.~\ref{sec:app.portfolio} we discuss this assumption and show it is satisfied for portfolio optimization problems.

To express the sum-of-squares relaxation \eqref{eq:sdp.relaxation.cons.pop} in the standard semidefinite programming form \eqref{eq:sdp}, we define the family of matrices $(B^{(\gamma)}_g)_{\gamma\in\N^n_k}$  as outlined below.

For $i=0,\dots,m$, for $\eta\in \N^n_{2(k-d_i)}$, define $B^{(\eta)}_i$ as
\begin{equation}\label{eq:def.B.i.eta}
B^{(\eta)}_i=((B_i^{(\eta)})_{\alpha,\beta})_{\alpha,\beta\in\N^n_{k-d_i}}\,,\text{ where }
(B_i^{(\eta)})_{\alpha,\beta}=\begin{cases}
1&\text{if }\alpha+\beta=\eta\,,\\
0&\text{else}.
\end{cases}
\end{equation}
For $g=(g_i)_{i=1}^m\subset \R[x]$ with $g_0=1$, define
\begin{equation}\label{eq.tilde.B.gamma}
  B_g^{(\gamma)}=\diag\left(\left(\sum\limits_{
\arraycolsep=1.4pt\def\arraystretch{.7}
\begin{array}{cc}
\scriptstyle  \zeta\in\N^n_{2d_i}\\
\scriptstyle \gamma-\zeta\in \N^n
\end{array}
} g_{i,\zeta} B_i^{(\gamma-\zeta)}\right)_{i=0,\dots,m}\right)\,,\,\gamma\in\N^n_{2k}\,.
\end{equation}

It is not hard to show that $D_k(gy)=\sum_{\gamma\in\N^n_{2k}}y_\gamma B_g^{(\gamma)}$, for $y=(y_\gamma)_{\gamma\in \N^n_{2k}}$.
The matrices $(B_g^{(\gamma)})_{\gamma\in\N^n_{2k}}$ are called the coefficient-matching matrices in the expression of the block-diagonal moment/localizing matrix $D_k(gy)$ as linear combination of the entries of $y$.

We denote by $s_g$ the maximum number of nonzero entries per row of $B_g^{(\gamma)}$, $\gamma\in\N^n_{2k}$.

\begin{assumption}\label{ass:rescale.g}
For $i=1,\dots,m$, $\|g_i\|_{\ell_1}=\sum_{\xi\in\N^n_{2d_i}}|g_{i,\xi}|\le 1$.
\end{assumption}
\begin{remark}
The coefficients of any  polynomials $g=(g_i)_{i=1}^m$ can always be rescaled to obtain new polynomials $\tilde g=(\tilde g_i)_{i=1}^m$ such that $S(g)=S(\tilde g)$ and $\tilde g$ satisfies \eqref{ass:rescale.g}. Moreover,  if the coefficients of each $g_i$ are rational, the coefficients of $\tilde g_i$ will also be rational.
This is achieved by defining $\tilde g_i=\frac{g_i}{\|g_i\|_{\ell_1}}$, for $i=1,\dots,m$.
\end{remark}

\begin{algorithm}\label{alg:con.pop.exact}
Inequality-constrained polynomial optimization
\begin{itemize}
\item Input: $\varepsilon>0$, $\underline \lambda\in\R$, $a\in S(g)$,\\
\phantom{1.1cm} coefficients $(f_\alpha)_{\alpha\in\N^n_{2k}}$ of $f\in\R[x]$,\\
\phantom{1.1cm} coefficients $(g_{i,\alpha})_{\alpha\in\N^n_{2d_i}}$ of $g_i\in\R[x]$, for $i=0,\dots,m$.
\item Output: $\lambda_k^{(\varepsilon)}\in\R$ and $G_\varepsilon\in\mathbb S^N$.
\end{itemize}
\begin{enumerate}
\item Convert relaxation to standard semidefinite program
\begin{enumerate}
\item Let $\N^n_{2k}=\{\gamma_0,\gamma_1,\dots,\gamma_{M}\}$ with $\gamma_0=0$, and set $M=|\N^n_{2k}|-1=\binom{n+2k}n -1$.
\item Define $b_i=f_{\gamma_i}$, $\tilde A_i=  B_g^{(\gamma_i)}$, for $i=0,\dots,M$, $\tilde C=\tilde A_0=  B_g^{(\gamma_0)}=  B_g^{(0)}$, $\tilde N=\sum_{i=0}^m \binom{n+k-d_i}n$.
\item Define $C=\diag(\tilde C,0)$, $A_i=\diag(\tilde A_i,0)$, for $i=1,\dots,M$, $N=\tilde N+1$.

\item Set $\lambda_{\min}=f_0-f(a)$ and $\lambda_{\max}=f_0-\underline \lambda$.
\end{enumerate}

\item\label{step:hamilton.update.cons.exact} Run Algorithm \ref{alg:Binary.search.HU} (Binary search using Hamiltonian updates) 
to get $\underline \lambda_T\in\R$ and $X_\varepsilon=((X_\varepsilon)_{ij})_{i,j=1,\dots,N}\in\mathbb S^N$.

\item Extract approximate results: Set $\lambda_k^{(\varepsilon)}=f_0-\underline\lambda_T$ and $G_\varepsilon=((X_\varepsilon)_{ij})_{i,j=1,\dots,\tilde N}$.
\end{enumerate}
\end{algorithm}

\begin{theorem}\label{theo:convergence.ineq.cons}
Let $\varepsilon>0$, $f\in\R[x]_{2d}$ and $g_i\in\R[x]_{2d_i}$, for $i=0,\dots,m$, with $g_0=1$.
Suppose Assumptions \ref{ass:cond.bounded.trace.2} and \ref{ass:rescale.g} are satisfied.
Assume strong duality holds for the semidefinite relaxations \eqref{eq:sdp.relaxation.cons.pop}-\eqref{eq:sdp.relaxation.cons.pop.dual}, and let $\lambda_k=f^\star$.
Additionally, suppose the polynomial optimization problem \eqref{eq:cons.pop} has an optimal solution  $x^\star$ satisfying $\|x^\star\|_{\ell_1}\le r$, where $r\in (0,1)$.
Let $\lambda_k^{(\varepsilon)}$ be the value returned by Algorithm \ref{alg:con.pop.exact}.
Then, the following bound holds:
\begin{equation}
0\le \lambda_k^{(\varepsilon)}-\lambda_k\le\varepsilon\left(2+\frac{r}{1-r} \right)\,.
\end{equation}
\end{theorem}
The proof of Theorem \ref{theo:convergence.ineq.cons} is given in Appendix \ref{proof:theo:convergence.ineq.cons}.

Building on Lemma \ref{lem:finite.conver.ball}, we derive the following result as a corollary of Theorem \ref{theo:convergence.ineq.cons}:
\begin{corollary}\label{coro:sosoc}
Let $\varepsilon>0$, $f\in\R[x]_{2d}$ and $g_i\in\R[x]_{2d_i}$, for $i=0,\dots,m$, with $g_0=1$ and $g_1=c(R-\|x\|_{\ell_2}^2)$ for some $c>0$ and $R>0$.
Suppose that Assumption \ref{ass:rescale.g} is satisfied.
Assume that the polynomial optimization problem \eqref{eq:cons.pop} has an optimal solution $x^\star$ satisfying $\|x^\star\|_{\ell_1}\le r$, where $r\in (0,1)$.
Further, assume that the second-order sufficient optimality conditions hold for problem \eqref{eq:cons.pop} at every optimal solution.
Then, there exists $K\in\N$ such that for all $k\ge K$, $\lambda_k=f^\star$.
In this case, if Assumption \ref{ass:cond.bounded.trace.2} is met, and $\lambda_k^{(\varepsilon)}$ denotes the value returned by Algorithm \ref{alg:con.pop.exact}, it holds that
\begin{equation}
0\le \lambda_k^{(\varepsilon)}-f^\star\le\varepsilon\left(2+\frac{r}{1-r} \right)\,.
\end{equation}
\end{corollary}
\begin{remark}
The constants $c$ in Corollary \ref{coro:sosoc} are selected to ensure that Assumption \ref{ass:rescale.g} is satisfied.
\end{remark}

\begin{theorem}\label{theo:convergence.over.ball}
Let $\varepsilon>0$, $f\in\R[x]_{2d}$ and $g_i\in\R[x]_{2d_i}$, for $i=0,\dots,m$, with $g_0=1$ and $g_1=c(R-\|x\|_{\ell_2}^2)$ for some $R\in(0,1)$ and $c>0$.
Suppose that Assumptions \ref{ass:cond.bounded.trace.2} and \ref{ass:rescale.g} hold.
Further, assume that strong duality is satisfied for the semidefinite relaxations \eqref{eq:sdp.relaxation.cons.pop}-\eqref{eq:sdp.relaxation.cons.pop.dual}.
Let $\lambda_k^{(\varepsilon)}$ be the value returned by Algorithm \ref{alg:con.pop.exact}.
Then, the following bound holds:
\begin{equation}
0\le \lambda_k^{(\varepsilon)}-\lambda_k\le\varepsilon\left(2+\frac{1}{1-R} \sqrt{\frac{1}{2}\binom{n+k}n\left[\binom{n+k}n+1\right]}\ \right)\,.
\end{equation}
\end{theorem}
The proof of Theorem \ref{theo:convergence.over.ball} is given in Appendix \ref{proof:theo:convergence.over.ball}.

\begin{theorem}\label{theo:convergence.over.simplex}
Let $\varepsilon>0$, $f\in\R[x]_{2d}$ and $g_i\in\R[x]_{2d_i}$, for $i=0,\dots,m$, with $g_0=1$.
Fix $r\in(0,1)$.
Assume there exist indices $i_1,i_2,i_\alpha\in\{1,\dots,m\}$, with $\alpha\in\{0,1\}^n\backslash \{0\}$, satisfying the following:
\begin{itemize}
\item $g_{i_j}=c_j[r^j-(x_1+\dots+x_n)^j]$ with $c_j>0$ for $j=1,2$,
\item $g_{i_\alpha}=x^\alpha$ for $\alpha\in\{0,1\}^n\backslash \{0\}$.
\end{itemize}
Suppose that Assumptions \ref{ass:cond.bounded.trace.2} and \ref{ass:rescale.g} hold.
Further, assume that strong duality is satisfied for the semidefinite relaxations \eqref{eq:sdp.relaxation.cons.pop}-\eqref{eq:sdp.relaxation.cons.pop.dual}.
Let $\lambda_k^{(\varepsilon)}$ be the value returned by Algorithm \ref{alg:con.pop.exact}.
Then, the following bound holds:
\begin{equation}
0\le \lambda_k^{(\varepsilon)}-\lambda_k\le\varepsilon\left(2+\frac{r}{1-r} \right)\,.
\end{equation}
\end{theorem}

The proof of Theorem \ref{theo:convergence.over.simplex} is given in Appendix \ref{proof:theo:convergence.over.simplex}.

\begin{remark}
The constants $c$ and $c_j$ in Theorems \ref{theo:convergence.over.ball} and \ref{theo:convergence.over.simplex} are selected to ensure that Assumption \ref{ass:rescale.g} is satisfied.
\end{remark}
\begin{remark}
Under the assumption of Theorem \ref{theo:convergence.over.simplex}, the semialgebraic set $S(g)$ is contained within the simplex $\{x\in\R^n_+\,:\,x_1+\dots+x_n\le r\}$. However, in this case, the number of inequality constraints $m$ in problem \eqref{eq:cons.pop} grows at least exponentially with $n$, specifically $m\ge 2^n+1$.
\end{remark}
\begin{lemma}[Complexity]\label{lem:complexity.con.pop.exact}
Let $k\ge \max\{d,d_0,\dots,d_m\}$.
To run Algorithm \ref{alg:con.pop.exact}, the classical computer requires
\begin{equation}
O\left( \left[ s_g\binom{n+2k}n\sum_{i=0}^m\binom{n+k-d_i}n +\left( \sum_{i=0}^m\binom{n+k-d_i}n\right)^\omega\right] \frac{1}{\varepsilon^2}\right)
\end{equation}
operations, where $\omega\approx 2.373$, while the quantum computer requires
\begin{equation}
O\left(s_g\left[\binom{n+2k}{n}^{1/2}+\left(\sum_{i=0}^m\binom{n+k-d_i}n\right)^{1/2}\frac{1}{\varepsilon}\right]\frac{1}{\varepsilon^4}\right)
\end{equation}
operations, where $s_g$ is the maximum number of nonzero entries per row of $B_g^{(\gamma)}$ (defined as in \eqref{eq.tilde.B.gamma}), $\gamma\in\N^n_{2k}$, which is not larger than the maximum number of nonzero coefficients of each $g_i$, $i=1,\dots,m$.
\end{lemma}

The proof of Lemma \ref{lem:complexity.con.pop.exact} is given in Appendix \ref{proof:lem:complexity.con.pop.exact}.

\subsection{Inequality-constrained polynomial optimization over the simplex}
\label{sec:ineq.cons.pop.simplex}

We relax the equality affine constraints in SDP \eqref{eq:sdp.relaxation.cons.pop} to obtain the following SDP:
\begin{equation}\label{eq:sdp.relaxation.cons.pop.ineq.aff}
\begin{array}{rl}
\hat \lambda_k=\sup\limits_{\lambda,G} & \lambda\\
\text{s.t.}& \lambda\in\R\,,\,G=\diag(G_0,\dots,G_m)\succeq 0\,,\\[5pt]
&f_0-\lambda\ge\sum_{i=0}^m  g_{i,0} (G_i)_{0,0}\,,\\[5pt]
&f_\gamma\ge\sum_{i=0}^m \sum\limits_{
\arraycolsep=1.4pt\def\arraystretch{.7}
\begin{array}{cc}
\scriptstyle \xi\in\N^n_{2d_i}\\
\scriptstyle \gamma-\xi\in \N^n
\end{array}
} g_{i,\xi} \sum\limits_{
\arraycolsep=1.4pt\def\arraystretch{.7}
\begin{array}{cc}
\scriptstyle \alpha,\beta\in\N^n_{k-d_i}\\
\scriptstyle \alpha+\beta=\gamma-\xi
\end{array}
} (G_i)_{\alpha,\beta}\,,\,\gamma\in\N^n_{2k}\backslash \{0\}\,.
\end{array}
\end{equation}
Then $\lambda_k\le \hat \lambda_k$, and if $S(g)\subset \R_+^n$, $\hat \lambda_k\le f^\star$.
Assume that $\underline{\lambda}\le \hat \lambda_k$ is given.

The duals of \eqref{eq:sdp.relaxation.cons.pop.ineq.aff} reads as:
\begin{equation}\label{eq:sdp.relaxation.cons.pop.dual.ineq.aff}
\begin{array}{rl}
\hat \tau_k=\min\limits_y& L_y(f)\\
\text{s.t.}& y=(y_\alpha)_{\alpha\in\N^n_{2k}}\subset \R_+\,,\\
&M_k(y)\succeq 0\,,\,y_0=1\,,\\
&M_{k-d_i}(g_iy)\succeq 0\,,\,i=1,\dots,m\,.\\
\end{array}
\end{equation}

\begin{algorithm}\label{alg:con.pop.exact.ineq.aff}
Inequality-constrained polynomial optimization
\begin{itemize}
\item Input: $\varepsilon>0$, $\underline \lambda\in\R$, $a\in S(g)$,\\
\phantom{1.1cm} coefficients $(f_\alpha)_{\alpha\in\N^n_{2k}}$ of $f\in\R[x]$,\\
\phantom{1.1cm} coefficients $(g_{i,\alpha})_{\alpha\in\N^n_{2d_i}}$ of $g_i\in\R[x]$, for $i=0,\dots,m$.
\item Output: $\hat\lambda_k^{(\varepsilon)}\in\R$ and $G_\varepsilon\in\mathbb S^N$.
\end{itemize}
\begin{enumerate}
\item Convert relaxation to standard semidefinite program
\begin{enumerate}
\item Let $\N^n_{2k}=\{\gamma_0,\gamma_1,\dots,\gamma_{M}\}$ with $\gamma_0=0$, and set $M=|\N^n_{2k}|-1=\binom{n+2k}n -1$.
\item Define $b_i=f_{\gamma_i}$, $\tilde A_i=  B_g^{(\gamma_i)}$, for $i=0,\dots,M$, $\tilde C=\tilde A_0=  B_g^{(\gamma_0)}=  B_g^{(0)}$, $\tilde N=\sum_{i=0}^m \binom{n+k-d_i}n$.
\item Define $C=\diag(\tilde C,0)$, $A_i=\diag(\tilde A_i,0)$, for $i=1,\dots,M$, $N=\tilde N+1$.

\item Set $\lambda_{\min}=f_0-f(a)$ and $\lambda_{\max}=f_0-\underline \lambda$.
\end{enumerate}

\item\label{step:hamilton.update.cons.exact.ineq.aff} Run Algorithm \ref{alg:Binary.search.HU.ineq.aff} (Binary search using Hamiltonian updates) 
to get $\underline \lambda_T\in\R$ and $X_\varepsilon=((X_\varepsilon)_{ij})_{i,j=1,\dots,N}\in\mathbb S^N$.

\item Extract approximate results: Set $\hat \lambda_k^{(\varepsilon)}=f_0-\underline\lambda_T$ and $G_\varepsilon=((X_\varepsilon)_{ij})_{i,j=1,\dots,\tilde N}$.
\end{enumerate}
\end{algorithm}

\begin{theorem}\label{theo:convergence.ineq.cons.ineq.aff}
Let $\varepsilon>0$, $f\in\R[x]_{2d}$ and $g_i\in\R[x]_{2d_i}$, for $i=0,\dots,m$, with $g_0=1$.
Assume $S(g)\subset \R_+^n$.
Suppose Assumptions \ref{ass:cond.bounded.trace.2} and \ref{ass:rescale.g} are satisfied.
Assume strong duality holds for the semidefinite relaxations \eqref{eq:sdp.relaxation.cons.pop.ineq.aff}-\eqref{eq:sdp.relaxation.cons.pop.dual.ineq.aff}, and let $\hat \lambda_k=f^\star$.
Additionally, suppose the polynomial optimization problem \eqref{eq:cons.pop} has an optimal solution  $x^\star$ satisfying $\|x^\star\|_{\ell_1}\le r$, where $r\in (0,1)$.
Let $\hat \lambda_k^{(\varepsilon)}$ be the value returned by Algorithm \ref{alg:con.pop.exact.ineq.aff}.
Then, the following bound holds:
\begin{equation}
0\le \hat \lambda_k^{(\varepsilon)}-\hat \lambda_k\le\varepsilon\left(2+\frac{r}{1-r} \right)\,.
\end{equation}
\end{theorem}
The proof of Theorem \ref{theo:convergence.ineq.cons.ineq.aff} is given in Appendix \ref{proof:theo:convergence.ineq.cons.ineq.aff}.

\begin{theorem}\label{theo:convergence.over.simplex.ineq.aff}
Let $\varepsilon>0$, $f\in\R[x]_{2d}$ and $g_i\in\R[x]_{2d_i}$, for $i=0,\dots,m$, with $g_0=1$.
Assume $S(g)\subset \R_+^n$.
Let $r\in(0,1)$.
Assume there exist indices $i_1,i_2\in\{1,\dots,m\}$ such that $g_{i_j}=c_j[r^j-(x_1+\dots+x_n)^j]$ with $c_j>0$ for $j=1,2$.
Suppose that Assumptions \ref{ass:cond.bounded.trace.2} and \ref{ass:rescale.g} hold.
Further, assume that strong duality is satisfied for the semidefinite relaxations \eqref{eq:sdp.relaxation.cons.pop.ineq.aff}-\eqref{eq:sdp.relaxation.cons.pop.dual.ineq.aff}.
Let $\hat \lambda_k^{(\varepsilon)}$ be the value returned by Algorithm \ref{alg:con.pop.exact.ineq.aff}.
Then, $\lambda_k\le \hat \lambda_k\le f^\star$ and the following bound holds:
\begin{equation}
0\le \hat \lambda_k^{(\varepsilon)}-\hat \lambda_k\le\varepsilon\left(2+\frac{r}{1-r} \right)\,.
\end{equation}
\end{theorem}
The proof of Theorem \ref{theo:convergence.over.simplex.ineq.aff} is given in Appendix \ref{proof:theo:convergence.over.simplex.ineq.aff}.

\begin{lemma}[Complexity]\label{lem:complexity.con.pop.exact.ineq.aff}
Let $k\ge \max\{d,d_0,\dots,d_m\}$.
To run Algorithm \ref{alg:con.pop.exact.ineq.aff}, the classical computer requires
\begin{equation}
O\left( \left[ s_g\binom{n+2k}n\sum_{i=0}^m\binom{n+k-d_i}n +\left( \sum_{i=0}^m\binom{n+k-d_i}n\right)^\omega\right] \frac{1}{\varepsilon^2}\right)
\end{equation}
operations, where $\omega\approx2.373$, while the quantum computer requires
\begin{equation}
O\left(s_g\left[\binom{n+2k}{n}^{1/2}+\left(\sum_{i=0}^m \binom{n+k-d_i}n\right)^{1/2}\frac{1}{\varepsilon}\right]\frac{1}{\varepsilon^4}\right)
\end{equation}
operations, where $s_g$ is the maximum number of nonzero entries per row of $B_g^{(\gamma)}$ (defined as in \eqref{eq.tilde.B.gamma}), $\gamma\in\N^n_{2k}$, which is not larger than the maximum number of nonzero coefficients of each $g_i$, $i=1,\dots,m$.
\end{lemma}
The proof of Lemma \ref{lem:complexity.con.pop.exact.ineq.aff} is similar to the proof of Lemma \ref{lem:complexity.con.pop.exact}.

\subsection{Application for portfolio optimization}
\label{sec:app.portfolio}
Portfolio optimization is a fundamental problem in quantitative finance, where the goal is to allocate wealth among a set of assets in order to balance expected return and risk. A common model for the risk is based on the covariance matrix of asset returns. The following optimization problem considers a normalized version of the mean-variance formulation:
\begin{equation}\label{eq:portfolio.prob}
\begin{array}{rl}
\eta=\min\limits_{z\in\R^n}& \frac{1}{\|Q\|_{1,1}}z^\top Q z\\
\text{s.t.} &w^\top z\ge \kappa\,,\\
&z_1+\dots+z_n\le \zeta\,,\\
&z_j\ge 0\,,\,j=1,\dots,n\,,
\end{array}
\end{equation}
where $Q=(Q_{ij})_{i,j=1,\dots,n}$ denotes the covariance matrix of the asset returns, and $\|Q\|_{1,1}=\sum_{i,j=1}^n |Q_{ij}|$ is its entrywise 1-norm used for normalization.

In theory, $Q$ is positive semidefinite (psd).
However, in practice, due to sampling noise, estimation errors, or regularization (e.g., shrinkage techniques), $Q$ may fail to be psd. In that case, problem \eqref{eq:portfolio.prob} becomes nonconvex.

Throughout, we assume that $\zeta$ is a constant independent of $n$ and that problem \eqref{eq:portfolio.prob} admits at least one optimal solution.

When $Q$ is psd, \eqref{eq:portfolio.prob} is a convex quadratic program, and hence a polynomial optimization problem. To the best of our knowledge, interior-point methods is the fastest known classical algorithms for solving convex quadratic programs \eqref{eq:portfolio.prob}  with complexity $O(n^{\omega+1} \log (1/\varepsilon))$, $\omega\approx 2.373$ (see \cite{kerenidis2019quantum}).

In contrast, our quantum algorithm presented in the previous section can approximate the optimal value $\eta$ with accuracy $\varepsilon>0$  in complexity $O(n\varepsilon^{-4}+\sqrt{n}\varepsilon^{-5})$ when $Q$ is psd and assuming we have a block encoding of $Q$ in both cases. 
If $Q$ is not psd, the method still yields an approximate optimal value for the first-order relaxation of \eqref{eq:portfolio.prob} with the same complexity on a quantum computer. Thus, we obtain a quantum speedup over rigorous classical algorithms in the dimension, at the cost of a worse dependency on the precision.

Next, by setting $x=\frac{r}\zeta z$ with $r\in (0,1)$, problem \eqref{eq:portfolio.prob} can be equivalently written as
\begin{equation}\label{eq:portfolio.prob2}
\begin{array}{rl}
\frac{r^2\eta}{\zeta^2}=\min\limits_{x\in\R^n}& \frac{1}{\|Q\|_{1,1}}x^\top Q x\\
\text{s.t.} &\zeta w^\top x \ge r\kappa\\
& x_1+\dots+x_n\le r\,,\\
&x_j\ge 0\,,\,j=1,\dots,n\,,
\end{array}
\end{equation}

For our method to be applicable, we must construct a measure supported on the feasible domain of problem \eqref{eq:portfolio.prob2} such that its moments can be efficiently computed. These moments will be used to verify Assumption~\ref{ass:cond.bounded.trace.2}.

To this end, we select the Lebesgue measure on a small hyper-rectangular box $[t_1,t_2]^n\subset \R^n$, where the bounds $t_1,t_2$ are defined as
\begin{equation}
t_1=\frac{r}{3n}\qquad\text{ and }\qquad t_2=\frac{r}{2n}\,.
\end{equation}
Clearly, $0<t_1<t_2<\frac{r}{n}$, and we have
\begin{equation}
t_1+t_2=\frac{5r}{6n}\quad\text{ and }\quad t_2-t_1=\frac{r}{6n}\,.
\end{equation}
It is straightforward to verify that every point in the box $[t_1,t_2]^n$ satisfies the last $n+1$ affine constraints of problem \eqref{eq:portfolio.prob2}.
To ensure that these points also satisfy the first affine constraint, we introduce the following assumption.
\begin{assumption}\label{ass:contain.box}
The box $[t_1,t_2]^n$ lies within the feasible region of problem~\eqref{eq:portfolio.prob2}, i.e.,
    \begin{equation}\label{eq:feasible.box}
\frac{\zeta w^\top x - r\kappa}{\zeta \|w\|_1+r\kappa}\ge \frac{1}{n^2} \,,\,\forall x\in [t_1,t_2]^n.
\end{equation}
\end{assumption}
Figure~\ref{fig:contain.box} illustrates this box construction in the two-dimensional case.

If $w\in \R^n_+$ and the inequality
\begin{equation}
\frac{\zeta t_1 \|w\|_1 - r\kappa}{\zeta \|w\|_1+r\kappa}\ge \frac{1}{n^2}
\end{equation}
holds, the condition~\eqref{eq:feasible.box} is satisfied.
In particular, if $w\in \R^n_+$ satisfies $n^2\ge \|w\|_1\ge n$ and $\zeta-3\kappa>0$, then \eqref{eq:feasible.box} holds with a constant factor independent of $n$.

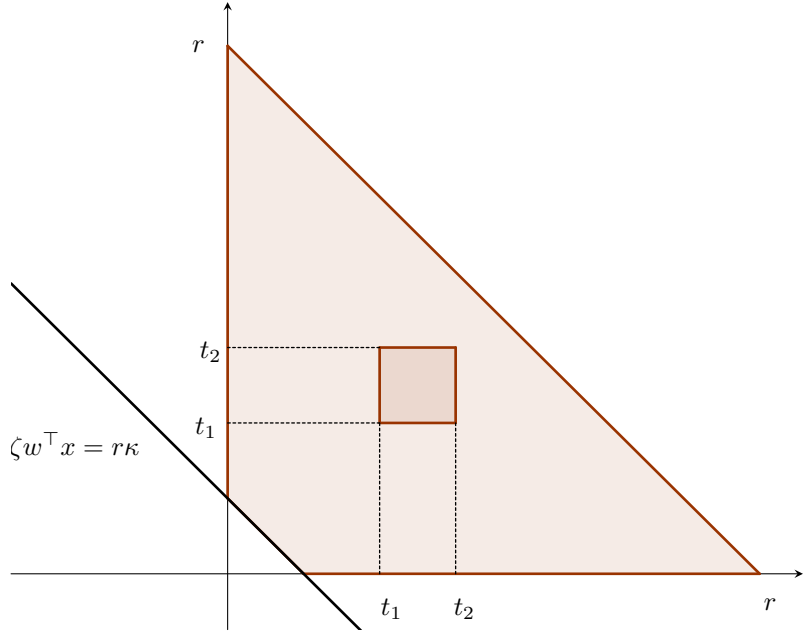
\begin{figure}
    \centering
    \definecolor{zzttqq}{rgb}{0.6,0.2,0}
\begin{tikzpicture}[line cap=round,line join=round,>=triangle 45,x=1cm,y=1cm]
\begin{axis}[
x=1cm,y=1cm,
axis lines=middle,
xmin=-2.8500000000000036,
xmax=7.570000000000002,
ymin=-0.7399999999999991,
ymax=7.580000000000004,
xtick={-3,8},
ytick={-1,8},]
\clip(-2.85,-0.74) rectangle (7.57,7.58);
\fill[line width=1pt,color=zzttqq,fill=zzttqq,fill opacity=0.10000000149011612] (0,7) -- (0,1) -- (1,0) -- (7,0) -- cycle;
\fill[line width=1pt,color=zzttqq,fill=zzttqq,fill opacity=0.10000000149011612] (2,3) -- (2,2) -- (3,2) -- (3,3) -- cycle;
\draw [line width=1pt,color=zzttqq] (0,7)-- (0,1);
\draw [line width=1pt,color=zzttqq] (0,1)-- (1,0);
\draw [line width=1pt,color=zzttqq] (1,0)-- (7,0);
\draw [line width=1pt,color=zzttqq] (7,0)-- (0,7);
\draw [line width=1pt,domain=-2.85:7.57] plot(\x,{(--1-1*\x)/1});
\draw [line width=1pt,color=zzttqq] (2,3)-- (2,2);
\draw [line width=1pt,color=zzttqq] (2,2)-- (3,2);
\draw [line width=1pt,color=zzttqq] (3,2)-- (3,3);
\draw [line width=1pt,color=zzttqq] (3,3)-- (2,3);
\draw [line width=0.5pt,dash pattern=on 1pt off 1pt] (2,3)-- (0,3);
\draw [line width=0.5pt,dash pattern=on 1pt off 1pt] (0,2)-- (2,2);
\draw [line width=0.5pt,dash pattern=on 1pt off 1pt] (2,2)-- (2,0);
\draw [line width=0.5pt,dash pattern=on 1pt off 1pt] (3,2)-- (3,0);
\draw (1.89,-0.2) node[anchor=north west] {$t_1$};
\draw (2.85,-0.2) node[anchor=north west] {$t_2$};
\draw (-0.55,2.2) node[anchor=north west] {$t_1$};
\draw (-0.49,3.2) node[anchor=north west] {$t_2$};
\draw (-3,2) node[anchor=north west] {$\zeta w^\top x = r\kappa$};
\draw (-0.59,7.2) node[anchor=north west] {$r$};
\draw (6.93,-0.2) node[anchor=north west] {$r$};
\end{axis}
\end{tikzpicture}
    \caption{Plot of Assumption \ref{ass:contain.box}.}
    \label{fig:contain.box}
\end{figure}

The following theorem establishes the error bound and computational complexity of our proposed method when applied to the portfolio optimization problem.

\begin{theorem}\label{theo:apply.portfolio}
    Consider portfolio optimization problem \eqref{eq:portfolio.prob} with covariance matrix $Q$ and optimal value $\eta\ge 0$. 
Let $\lambda_1$ denote the optimal value of the first order SOS relaxation  for problem \eqref{eq:portfolio.prob}, and  assume that $\lambda_1\ge 0$.
    Suppose that Assumption~\ref{ass:contain.box} holds. Then, for any given accuracy parameter $\varepsilon>0$, the following statements are true:
\begin{enumerate}[(i)]
\item\label{general.case}(General case) Our Algorithm~\ref{alg:con.pop.exact.ineq.aff} returns an approximate value in the interval $[\lambda_1,\eta+\varepsilon]$ with computational complexity $O(n^{3}\varepsilon^{-2})$ on a classical computer, and $O(n\varepsilon^{-4}+\sqrt{n}\varepsilon^{-5})$ on a quantum computer.
\item \label{convex.case} (Convex case)  If $Q$ is positive semidefinite, our Algorithm~\ref{alg:con.pop.exact.ineq.aff} provides an approximate value in the tighter range $[\eta,\eta + \varepsilon]$ with computational complexity  
$O(n^{3}\varepsilon^{-2})$
on a classical computer, and  
$O(n\varepsilon^{-4}+\sqrt{n}\varepsilon^{-5})$
on a quantum computer.

\end{enumerate}    
\end{theorem}
\begin{proof}
\eqref{general.case} Let $\delta=\frac{1}{152}$, $f=\frac{\delta}{\|Q\|_{1,1}} x^\top Q x$, $m=n+3$, $g_{j}=x_j$, for $j=1,\dots,n$, $g_{n+1}=\frac{1}{n+1}(r-\sum_{j=1}^n x_j)$, and $g_{n+2}=\frac{\zeta w^\top x - r\kappa}{\zeta \|w\|_1+r\kappa}$, $g_{n+3}=\frac{1}{n^2+1}[r^2-(\sum_{j=1}^n x_j)^2]$.
Then problem \eqref{eq:portfolio.prob2} is equivalent to POP \eqref{eq:cons.pop} with $f^\star=\frac{\delta r^2\eta}{\zeta^2}\ge 0$ and
\begin{equation}\label{eq:simplex.domain}
S(g)=\{x\in\R^n_+\,:\,\sum_{j=1}^n x_j\le r\}\,.
\end{equation}
It is not hard to verify that  Assumption \ref{ass:rescale.g} holds.

Define $\lambda_k$ and $\hat\lambda_k$ as in \eqref{eq:sdp.relaxation.cons.pop} and \eqref{eq:sdp.relaxation.cons.pop.ineq.aff}, respectively.
Since $S(g)\subset \R^n_+$, 
\begin{equation}\label{lem:ineq.first.order}
\lambda_1\le \hat \lambda_1\le f^\star\,.
\end{equation}

Let $\underline{\lambda}=0$.
Then $\lambda_1\ge 0= \underline{\lambda}$.
We now prove that Assumption \ref{ass:cond.bounded.trace.2} holds with $k=1$ and $\mu$ as a Lebesgue measure on $[t_1,t_2]^n$. 
Take $y=(y_\alpha)_{\alpha\in\N^n_{2k}}$ as the truncated moment vector with respect to $\mu$.
Then 
\begin{equation}
\begin{array}{rl}
y_{\alpha}=&\int x^\alpha d\mu=\int_{[t_1,t_2]^n} x^\alpha dx=\prod_{j=1}^n \int_{t_1}^{t_2} x_j^{\alpha_j} dx_j\\[10pt]
 =&\prod_{j=1}^n \left.\frac{x_j^{\alpha_j+1}}{\alpha_j+1}\right|_{t_1}^{t_2} = \prod_{j=1}^n \frac{t_2^{\alpha_j+1}-t_1^{\alpha_j+1}}{\alpha_j+1}\,.
\end{array}
\end{equation}
With $(e_j)_{j=1}^n$ being the standard basis of $\R^n$, we have
\begin{equation}
\begin{array}{rl}
y_0=&(t_2-t_1)^n\,,\\[5pt]
y_{e_j}=&\frac{1}{2}(t_2-t_1)^n(t_1+t_2)\,,\\[5pt]
y_{e_i+e_j}=&\frac{1}{4}(t_2-t_1)^n(t_1+t_2)^2\,,\\[5pt]
y_{2e_j}=&\frac{1}{3}(t_2-t_1)^n(t_1^2+t_1t_2+t_2^2)\,,\,i\ne j
\,.
\end{array}
\end{equation}
This implies
\begin{equation}
\begin{array}{rl}
M_1(\mu)=&\begin{bmatrix}
y_0 & y_{e_1} & y_{e_2}&\dots& y_{e_n}\\
y_{e_1} & y_{2e_1} & y_{e_1+e_2}& \dots & y_{e_1+e_n}\\
y_{e_2} & y_{e_2+e_1} & y_{2e_2}& \dots & y_{e_2+e_n}\\
.&.&.&\dots &.\\
y_{e_n}& y_{e_n+e_1}& y_{e_n+e_2}& \dots& y_{2e_n}
\end{bmatrix}=(t_2-t_1)^nW\,,
\end{array}
\end{equation}
where 
\begin{equation}
W=\begin{bmatrix}
1 & a & a&\dots& a\\
a & b & a^2& \dots & a^2\\
a & a^2 & b& \dots & a^2\\
.&.&.&\dots &.\\
a& a^2&a^2& \dots& b
\end{bmatrix}
\end{equation}
with $a=\frac{1}{2}(t_1+t_2)$ and $b=\frac{1}{3}(t_1^2+t_1t_2+t_2^2)$.
By Lemma \ref{lem:cayley-haminton}, with $d=b-a^2$, we have
\begin{equation}
\begin{array}{rl}
\det(W-\lambda I_{n+1})=&\left(b-a^2-\lambda\right)^{n-1}\left[(1-\lambda)\left(b-a^2-\lambda\right)-n\left(a^2-(1-\lambda)a^2\right)\right]\\
=&\left(d-\lambda\right)^{n-1}\left[\lambda^2-\left(1+d+na^2\right)\lambda+d\right]\,.
\end{array}
\end{equation}
The zeros of the above polynomial in $\lambda$ are
\begin{equation}
d=\frac{(t_2-t_1)^2}{12}=\frac{r^2}{432n^2}\qquad\text{and}\qquad \frac{1+d+na^2\pm \sqrt{\Delta}}{2}\,,
\end{equation}
where $\Delta=(1+d+na^2)^2-4d\ge 1-4d> 0$.
We have 
\begin{equation}
\begin{array}{rl}
\frac{1+d+na^2- \sqrt{\Delta}}{2}=&\frac{(1+d+na^2)^2-\Delta}{2(1+d+na^2+ \sqrt{\Delta})}=\frac{4d}{2(1+d+na^2+ \sqrt{\Delta})}\ge \frac{4d}{4(1+d+na^2)}=\frac{d}{1+d+na^2}\ge \frac{d}{2}=\frac{r^2}{864n^2}\,.\\
\end{array}
\end{equation}
The first inequality is due to $\Delta\le (1+d+na^2)^2$.
The second inequality is because 
\begin{equation}
    d+na^2=\frac{r^2}{432n^2}+\frac{n(5r)^2}{4(6n)^2}=\frac{r^2}{432n^2}+\frac{25r^2}{144n}\le 1\,.
\end{equation}
This implies that $\lambda_{\min}(W)\ge \frac{r^2}{864n^2}$, yielding 
\begin{equation}
\lambda_{\min}(M_1(\mu))\ge (t_2-t_1)^n\lambda_{\min}(W)\ge \frac{r^2(t_2-t_1)^n}{864n^2}\,.
\end{equation}

Next, we have that for $j=1,\dots,n$,
\begin{equation}
M_0(g_j\mu)=\int g_j d\mu=\int_{[t_1,t_2]^n}x_jdx=y_{e_j}=\frac{1}{2}(t_2-t_1)^n(t_1+t_2)= \frac{5r(t_2-t_1)^n}{12n}\,.
\end{equation}
In addition, we get
\begin{equation}
\begin{array}{rl}
M_0(g_{n+1}\mu)=&\int g_{n+1} d\mu=\int_{[t_1,t_2]^n}\frac{1}{n}(r-\sum_{j=1}^n x_j)dx\\
\ge& \int_{[t_1,t_2]^n}\frac{1}{n}(r-nt_2)dx=\int_{[t_1,t_2]^n}\frac{r}{2n}dx=\frac{r(t_2-t_1)^n}{2n}\,,\\[15pt]
M_0(g_{n+2}\mu)=&\int g_{n+2} d\mu=\int_{[t_1,t_2]^n}\frac{1}{n^2+1}[r^2-(\sum_{j=1}^n x_j)^2]dx\\[5pt]
\ge& \int_{[t_1,t_2]^n}\frac{1}{n^2+1}[r^2-(nt_2)^2]dx=\int_{[t_1,t_2]^n}\frac{3r^2}{4(n^2+1)}dx\\[5pt]
=&\frac{3r^2(t_2-t_1)^n}{4(n^2+1)}\,,\\[10pt]
M_0(g_{n+3}\mu)=&\int g_{n+3} d\mu=\int_{[t_1,t_2]^n}\frac{\zeta w^\top x - r\kappa}{\zeta \|w\|_1+r\kappa}dx\\[5pt]
\ge& \int_{[t_1,t_2]^n}\frac{1}{n^2}dx=\frac{(t_2-t_1)^n}{n^2}\,.
\end{array}
\end{equation}
The last inequality is based on \eqref{eq:feasible.box}.

Since $D_1(g\mu)=\diag(M_1(\mu),M_0(g_1\mu),\dots,M_0(g_{n+3}\mu))$, we have
 \begin{equation}
 \lambda_{\min}(D_1(g\mu))=\min\{\lambda_{\min}(M_1(\mu)),M_0(g_1\mu),\dots,M_0(g_{n+3}\mu) \}\ge \frac{r^2(t_2-t_1)^n}{864n^2}\,.
\end{equation}  

Besides, we have 
\begin{equation}
\begin{array}{rl}
\int (f-\underline{\lambda}) d\mu=& \frac{\delta}{\|Q\|_{1,1}} \int \sum_{i,j=1}^n Q_{ij} x_ix_j d\mu\\[5pt]
=& \frac{\delta}{\|Q\|_{1,1}} \sum_{i,j=1}^n Q_{ij} y_{e_i+e_j}\\[5pt]
\le  & \frac{\delta}{\|Q\|_{1,1}} \sum_{i,j=1}^n |Q_{ij}| b (t_2-t_1)^n\\
= & \delta b (t_2-t_1)^n=\frac{\delta(t_2-t_1)^n}{3}(\frac{r^2}{9n^2}+\frac{r^2}{6n^2}+\frac{r^2}{4n^2})=\frac{19\delta r^2(t_2-t_1)^n}{108n^2}\,.
\end{array}
\end{equation}
The first inequality is because $y_{e_i+e_j}\le b (t_2-t_1)^n$.
From these, it follows that
\begin{equation}
\frac{\int (f-\underline{\lambda}) d\mu}{\lambda_{\min}(D_1(g\mu))}\le \frac{19\delta r^2(t_2-t_1)^n}{108n^2} \times \frac{864n^2}{r^2(t_2-t_1)^n}=152\delta =1\,.
\end{equation}
Thus Assumption \ref{ass:cond.bounded.trace.2} holds.

Let $\hat \lambda_1^{(\varepsilon)}$ be the value returned by Algorithm \ref{alg:con.pop.exact.ineq.aff} with $k=1$ and  $\underline\lambda = 0$.
By Theorem  \ref{theo:convergence.over.simplex.ineq.aff}, $\lambda_1\le \hat \lambda_1\le f^\star=\frac{\delta r^2\eta}{\zeta^2}$ and 
$0\le \hat \lambda_1^{(\varepsilon)}-\hat \lambda_1\le\varepsilon(2+\frac{r}{1-r} )$.
This implies that $\lambda_1\le \hat \lambda_1^{(\varepsilon)}\le \varepsilon(2+\frac{r}{1-r} )+\frac{\delta r^2\eta}{\zeta^2}$.
Note that $\delta=\frac{1}{152}$.
Let us fix $r=\frac{1}{2}$ and $\zeta$ as a constant.
We claim that $s_g=1$. Indeed, by definition \eqref{eq.tilde.B.gamma} with $k=d_1=\dots=d_m=1$, the last $m$ blocks of $B_g^{(\gamma)}$ are scalar, which implies $s_g$ is the maximum number of nonzero entries of the first block $B_0^{(\gamma)}=B^{(\gamma)}$. By Lemma \ref{lem:properties.B.gamma} \eqref{unique.nonzero.per.row}, $s_g=1$.

By Lemma \ref{lem:complexity.con.pop.exact.ineq.aff} (with $s_g=1$, $k=d_1=\dots=d_m=1$ and $m=O(n)$), the complexities of Algorithm \ref{alg:con.pop.exact.ineq.aff} to provide a value in $[\lambda_1,\eta+\varepsilon]$ on classical and quantum computers are $O(n^{3}\varepsilon^{-2})$ and $O(n\varepsilon^{-4}+\sqrt{n}\varepsilon^{-5})$, respectively.

\eqref{convex.case} Now we assume that $Q$ is psd.
We do the same process as in the proof of \eqref{general.case}, but with $m=n+2$ (the quadratic constraint $g_{n+3}$ is removed).
We also have $s_g=1$.
 We claim that 
\begin{equation}\label{lem:exact.first.order}
    f^\star=\lambda_1=\hat \lambda_1\,.
\end{equation}
Indeed,  by Lemma \ref{lem:exact.quadratic}, $\lambda_1=f^\star$.
By \eqref{lem:ineq.first.order}, the equalities hold.
 
Let $\hat \lambda_1^{(\varepsilon)}$ be the value returned by Algorithm \ref{alg:con.pop.exact.ineq.aff} with $k=1$ and  $\underline\lambda = 0$.
By \eqref{eq:simplex.domain}, all optimal solutions  $x^\star$ to $\min_{x\in S(g)}f(x)$ satisfy $\|x^\star\|_{\ell_1}=x_1^\star+\dots+x_n^\star\le r$.
By Theorem  \ref{theo:convergence.ineq.cons.ineq.aff}, 
$0\le \hat \lambda_1^{(\varepsilon)}-\hat \lambda_1\le\varepsilon(2+\frac{r}{1-r} )$.
By \eqref{lem:exact.first.order}, $\hat \lambda_1=f^\star=\frac{\delta r^2\eta}{\zeta^2}$.
This implies that $\frac{\delta r^2\eta}{\zeta^2}\le \hat \lambda_1^{(\varepsilon)}\le\frac{\delta r^2\eta}{\zeta^2}+\varepsilon(2+\frac{r}{1-r} )$.
Let us fix $r=\frac{1}{2}$ and $\zeta$ as a constant.
By Lemma \ref{lem:complexity.con.pop.exact.ineq.aff} (with $s_g=1$, $k=d_1=\dots=d_m=1$ and $m=O(n)$), the complexities of Algorithm \ref{alg:con.pop.exact.ineq.aff} to provide a value in $[\eta,\eta+\varepsilon]$ on classical and quantum computers are $O(n^{3}\varepsilon^{-2})$ and $O(n\varepsilon^{-4}+\sqrt{n}\varepsilon^{-5})$, respectively.
\end{proof}

In \cite{kerenidis2019quantum}, Kerenidis \textit{et al.} propose a quantum algorithm based on interior-point approaches for solving a portfolio optimization problem with affine equality constraints. The portfolio optimization problem~\eqref{eq:portfolio.prob} can be reformulated into their framework by introducing nonnegative slack variables. Their algorithm solves the resulting problem with a quantum complexity of 
$O\left(n^{1.5}\frac{\zeta \kappa}{\delta^2}\log(1/\varepsilon)\right)$,
where $\zeta$, $\kappa$, and $\delta$ are problem-dependent parameters related to the conditioning of intermediate solutions. In particular, $\kappa$ denotes the maximum condition number of the Newton matrices arising across all iterations of their method. However, since the authors do not provide an explicit characterization of how $\kappa$ scales with $n$, the precise dependence of their quantum complexity on the problem dimension remains unclear. This lack of explicit scaling makes it difficult to accurately compare the quantum complexities of their approach with ours. In contrast, the quantum complexity of our Hamiltonian-update-based method does not depend on any unknown parameters. A similar limitation is observed in the interior-point-based quantum algorithm of Dalzell \textit{et al.}~\cite{dalzell2023end} for portfolio optimization.

\section{Conclusion} 
\label{sec:conclusion} 

In this work, we have shown how binary search combined with Hamiltonian updates can be applied to solve sums-of-squares relaxations of polynomial optimization problems (POPs), both in the unconstrained and the inequality-constrained setting.  
By carefully rescaling the problem data and exploiting structural properties of the relaxation, we obtained explicit runtime and accuracy bounds under strong duality and other natural optimality conditions.  
Our analysis identifies concrete regimes in which quantum SDP solvers can offer super-quadratic speedups in the problem dimension compared to the best known classical methods, even though they retain a worse dependence on the target precision~$\varepsilon$.  

A natural direction for future work is to extend these techniques to \emph{sparse} polynomial optimization problems \cite{wang2021tssos,wang2020cs,
lasserre2006convergent,waki2006sums}.  
In many practical instances, each constraint and the objective function depend only on a small subset of variables.  
Formally, suppose the $n$ variables are indexed by $[n]$ and there exists a collection of index sets $I_1,\dots,I_r \subseteq [n]$ such that:  
\begin{itemize}
\item the objective function is a sum of polynomials, each supported on $x(I_\ell) \in \mathbb{R}^{n_\ell}$,  
\item the inequality constraints $(g_i)_{i \in J_\ell}$ associated with block $I_\ell$ depend only on $x(I_\ell)$ and have cardinality $m_\ell$.  
\end{itemize}

Such \emph{structured sparsity} allows one to replace the dense moment/SOS relaxation with a block-structured one whose PSD constraints act only on the monomial variables of size $n_\ell$ for each block.  
This dramatically reduces the size of the semidefinite constraints from $\binom{n+k}{n}$ to $\binom{n_\ell+k}{n_\ell}$, and similarly for their degree-$2k$ analogues.  

If our quantum Hamiltonian-updates approach is applied directly to the block-diagonal SDP arising from such a sparse POP, the runtime bound is expected to improve to  
\[
O\!\left( s_g\left\{ \left[ \sum_{\ell=1}^r \binom{n_\ell+2k}{n_\ell} \right]^{1/2} + \left[ \sum_{\ell=1}^r \sum_{i\in J_\ell} \binom{n_\ell+k-d_i}{n_\ell} \right]^{1/2} \frac{1}{\varepsilon} \right\} \frac{1}{\varepsilon^4} \right),
\]  
where $s_g$ is the maximal number of nonzero coefficients in any constraint polynomial.  
This bound reflects the fact that sparsity reduces both the total number of constraints and the size of each PSD block, which in turn lowers the cost of the quantum subroutines.  

Another promising avenue is to investigate alternative semidefinite relaxations that are weaker than the standard Lasserre hierarchy but cheaper to solve, such as diagonally dominant sum-of-squares (DSOS/SDSOS) relaxations \cite{ahmadi2019dsos} or bounded-degree moment relaxations \cite{weisser2018sparse}.  
Combining these with quantum solvers could further enlarge the range of POP instances for which quantum speedups are practically achievable.   

Finally, we note that optimization over the simplex \cite{de2006ptas,de2017convergence}, which naturally arises in combinatorial optimization and probability distribution optimization, is a particularly compelling application of our framework.  
In this case, the problem structure enables sharp bounds on the dual norm, which directly improves the accuracy guarantees of the quantum algorithm.  
In summary, this work provides an investigation into the potential of quantum SDP solvers to give speedups for POP relaxations, identifying regimes of superquadratic speedups.

\section*{Acknowledgements}
DSF and HNM acknowledge funding from the European Union under Grant Agreement 101080142 and the project EQUALITY. D.S.F. acknowledges financial support from the Novo Nordisk Foundation (Grant No. NNF20OC0059939 Quantum for Life) and by the ERC grant GIFNEQ 101163938.
\bibliography{references} 

\begin{thebibliography}{10}

\bibitem{ahmadi2019dsos}
A.~A. Ahmadi and A.~Majumdar.
\newblock {DSOS and SDSOS optimization: more tractable alternatives to sum of squares and semidefinite optimization}.
\newblock {\em SIAM Journal on Applied Algebra and Geometry}, 3(2):193--230, 2019.

\bibitem{arora2005fast}
S.~Arora, E.~Hazan, and S.~Kale.
\newblock Fast algorithms for approximate semidefinite programming using the multiplicative weights update method.
\newblock In {\em 46th Annual IEEE Symposium on Foundations of Computer Science (FOCS'05)}, pages 339--348. IEEE, 2005.

\bibitem{arora2012multiplicative}
S.~Arora, E.~Hazan, and S.~Kale.
\newblock {The multiplicative weights update method: a meta-algorithm and applications}.
\newblock {\em Theory of computing}, 8(1):121--164, 2012.

\bibitem{augustino2023quantum}
B.~Augustino, G.~Nannicini, T.~Terlaky, and L.~F. Zuluaga.
\newblock {Quantum interior point methods for semidefinite optimization}.
\newblock {\em Quantum}, 7:1110, 2023.

\bibitem{baldi2023effective}
L.~Baldi and B.~Mourrain.
\newblock {On the effective Putinar's Positivstellensatz and moment approximation}.
\newblock {\em Mathematical Programming}, 200(1):71--103, 2023.

\bibitem{barthel2012solving}
T.~Barthel and R.~H{\"u}bener.
\newblock Solving condensed-matter ground-state problems by semidefinite relaxations.
\newblock {\em Physical review letters}, 108(20):200404, 2012.

\bibitem{brandao2017quantum}
F.~G. Brandao and K.~M. Svore.
\newblock {Quantum speed-ups for solving semidefinite programs}.
\newblock In {\em 2017 IEEE 58th Annual Symposium on Foundations of Computer Science (FOCS)}, pages 415--426. IEEE, 2017.

\bibitem{Camps2024}
D.~Camps, L.~Lin, R.~Van~Beeumen, and C.~Yang.
\newblock Explicit quantum circuits for block encodings of certain sparse matrices.
\newblock {\em SIAM Journal on Matrix Analysis and Applications}, 45(1):801–827, Mar. 2024.

\bibitem{dahl2012semidefinite}
J.~Dahl.
\newblock Semidefinite optimization using mosek.
\newblock {\em ISMP, Berlin}, 2012.

\bibitem{dalzell2023end}
A.~M. Dalzell, B.~D. Clader, G.~Salton, M.~Berta, C.~Y.-Y. Lin, D.~A. Bader, N.~Stamatopoulos, M.~J. Schuetz, F.~G. Brand{\~a}o, H.~G. Katzgraber, et~al.
\newblock End-to-end resource analysis for quantum interior-point methods and portfolio optimization.
\newblock {\em Prx Quantum}, 4(4):040325, 2023.

\bibitem{de2006ptas}
E.~De~Klerk, M.~Laurent, and P.~A. Parrilo.
\newblock {A PTAS for the minimization of polynomials of fixed degree over the simplex}.
\newblock {\em Theoretical Computer Science}, 361(2-3):210--225, 2006.

\bibitem{de2017convergence}
E.~De~Klerk, M.~Laurent, and Z.~Sun.
\newblock {Convergence analysis for Lasserre’s measure-based hierarchy of upper bounds for polynomial optimization}.
\newblock {\em Mathematical Programming}, 162(1):363--392, 2017.

\bibitem{doherty2012convergence}
A.~C. Doherty and S.~Wehner.
\newblock {Convergence of SDP hierarchies for polynomial optimization on the hypersphere}.
\newblock {\em arXiv preprint arXiv:1210.5048}, 2012.

\bibitem{Gilyn2019}
A.~Gilyén, Y.~Su, G.~H. Low, and N.~Wiebe.
\newblock Quantum singular value transformation and beyond: exponential improvements for quantum matrix arithmetics.
\newblock In {\em Proceedings of the 51st Annual ACM SIGACT Symposium on Theory of Computing}, STOC ’19, page 193–204. ACM, June 2019.

\bibitem{goemans1995improved}
M.~X. Goemans and D.~P. Williamson.
\newblock Improved approximation algorithms for maximum cut and satisfiability problems using semidefinite programming.
\newblock {\em Journal of the ACM (JACM)}, 42(6):1115--1145, 1995.

\bibitem{GSLBrando2022}
F.~G.S L.~Brandão, R.~Kueng, and D.~Stilck~Fran\c{c}a.
\newblock Faster quantum and classical sdp approximations for quadratic binary optimization.
\newblock {\em Quantum}, 6:625, Jan. 2022.

\bibitem{helmberg2000semidefinite}
C.~Helmberg.
\newblock {\em Semidefinite programming for combinatorial optimization}.
\newblock PhD thesis, 2000.

\bibitem{helmberg2014spectral}
C.~Helmberg, M.~L. Overton, and F.~Rendl.
\newblock The spectral bundle method with second-order information.
\newblock {\em Optimization Methods and Software}, 29(4):855--876, 2014.

\bibitem{helmberg2000spectral}
C.~Helmberg and F.~Rendl.
\newblock A spectral bundle method for semidefinite programming.
\newblock {\em SIAM Journal on Optimization}, 10(3):673--696, 2000.

\bibitem{helmberg1996interior}
C.~Helmberg, F.~Rendl, R.~J. Vanderbei, and H.~Wolkowicz.
\newblock An interior-point method for semidefinite programming.
\newblock {\em SIAM Journal on Optimization}, 6(2):342--361, 1996.

\bibitem{helton2012semidefinite}
J.~W. Helton and J.~Nie.
\newblock {A Semidefinite Approach for Truncated K-Moment Problems}.
\newblock {\em Foundations of Computational Mathematics}, 12(6):851--881, 2012.

\bibitem{henrion2020}
D.~Henrion, M.~Korda, and J.~B. Lasserre.
\newblock {\em Moment-sos Hierarchy, The: Lectures In Probability, Statistics, Computational Geometry, Control And Nonlinear Pdes}, volume~4.
\newblock World Scientific, 2020.

\bibitem{henrion2005detecting}
D.~Henrion and J.-B. Lasserre.
\newblock {Detecting global optimality and extracting solutions in GloptiPoly}.
\newblock In {\em Positive polynomials in control}, pages 293--310. Springer, 2005.

\bibitem{2502.15426}
F.~Henze, V.~Tran, B.~Ostermann, R.~Kueng, T.~de~Wolff, and D.~Gross.
\newblock {S}olving quadratic binary optimization problems using quantum {SDP} methods: {N}on-asymptotic running time analysis, 2025.
\newblock arXiv:2502.15426v1.

\bibitem{jiang2020faster}
H.~Jiang, T.~Kathuria, Y.~T. Lee, S.~Padmanabhan, and Z.~Song.
\newblock {A faster interior point method for semidefinite programming}.
\newblock In {\em 2020 IEEE 61st annual symposium on foundations of computer science (FOCS)}, pages 910--918. IEEE, 2020.

\bibitem{josz2016strong}
C.~Josz and D.~Henrion.
\newblock {Strong duality in Lasserre's hierarchy for polynomial optimization}.
\newblock {\em Optimization Letters}, 10(1):3--10, 2016.

\bibitem{kale2007efficient}
S.~Kale.
\newblock {\em Efficient algorithms using the multiplicative weights update method}.
\newblock Princeton University, 2007.

\bibitem{kerenidis2019quantum}
I.~Kerenidis, A.~Prakash, and D.~Szil{\'a}gyi.
\newblock Quantum algorithms for portfolio optimization.
\newblock In {\em Proceedings of the 1st ACM Conference on Advances in Financial Technologies}, pages 147--155, 2019.

\bibitem{lasserre2001global}
J.~B. Lasserre.
\newblock Global optimization with polynomials and the problem of moments.
\newblock {\em SIAM Journal on optimization}, 11(3):796--817, 2001.

\bibitem{lasserre2006convergent}
J.~B. Lasserre.
\newblock {Convergent SDP-relaxations in polynomial optimization with sparsity}.
\newblock {\em SIAM Journal on Optimization}, 17(3):822--843, 2006.

\bibitem{lasserre2009convex}
J.~B. Lasserre.
\newblock Convex sets with semidefinite representation.
\newblock {\em Mathematical programming}, 120(2):457--477, 2009.

\bibitem{lasserre2009moments}
J.~B. Lasserre.
\newblock {\em Moments, positive polynomials and their applications}, volume~1.
\newblock World Scientific, 2009.

\bibitem{lasserre2010moments}
J.-B. Lasserre.
\newblock {\em Moments, positive polynomials and their applications}, volume~1.
\newblock World Scientific, 2010.

\bibitem{lasserre2015introduction}
J.~B. Lasserre.
\newblock {\em An introduction to polynomial and semi-algebraic optimization}, volume~52.
\newblock Cambridge University Press, 2015.

\bibitem{laurent2009sums}
M.~Laurent.
\newblock Sums of squares, moment matrices and optimization over polynomials.
\newblock In {\em Emerging applications of algebraic geometry}, pages 157--270. Springer, 2009.

\bibitem{laurent2023effective}
M.~Laurent and L.~Slot.
\newblock {An effective version of Schm{\"u}dgen’s Positivstellensatz for the hypercube}.
\newblock {\em Optimization Letters}, 17(3):515--530, 2023.

\bibitem{lilecture2018}
T.~Li and X.~Wu.
\newblock {Lecture notes for quantum semidefinite programming (SDP) solvers}.
\newblock {\em CMSC 657, Intro to Quantum Information Processing, \href{https://www.cs.umd.edu/class/fall2018/cmsc657/note/SDP.pdf}{link to download}}, pages 1--6, 2018.

\bibitem{magron2023sparse}
V.~Magron and J.~Wang.
\newblock {\em Sparse polynomial optimization: theory and practice}.
\newblock World Scientific, 2023.

\bibitem{nie2013exact}
J.~Nie.
\newblock {An exact Jacobian SDP relaxation for polynomial optimization}.
\newblock {\em Mathematical Programming}, 137(1):225--255, 2013.

\bibitem{nie2014optimality}
J.~Nie.
\newblock {Optimality conditions and finite convergence of Lasserre's hierarchy}.
\newblock {\em Mathematical programming}, 146(1-2):97--121, 2014.

\bibitem{nocedal2006numerical}
J.~Nocedal and S.~Wright.
\newblock {\em Numerical optimization}.
\newblock Springer Science \& Business Media, 2006.

\bibitem{parrilo2000structured}
P.~A. Parrilo.
\newblock {\em Structured semidefinite programs and semialgebraic geometry methods in robustness and optimization}.
\newblock PhD thesis, California Institute of Technology, 2000.

\bibitem{van2018improvements}
J.~Van~Apeldoorn and A.~Gily{\'e}n.
\newblock {Improvements in quantum SDP-solving with applications}.
\newblock {\em arXiv preprint arXiv:1804.05058}, 2018.

\bibitem{van2020convex}
J.~van Apeldoorn, A.~Gily{\'e}n, S.~Gribling, and R.~de~Wolf.
\newblock Convex optimization using quantum oracles.
\newblock {\em Quantum}, 4:220, 2020.

\bibitem{vandenberghe1996semidefinite}
L.~Vandenberghe and S.~Boyd.
\newblock Semidefinite programming.
\newblock {\em SIAM review}, 38(1):49--95, 1996.

\bibitem{waki2006sums}
H.~Waki, S.~Kim, M.~Kojima, and M.~Muramatsu.
\newblock {Sums of squares and semidefinite program relaxations for polynomial optimization problems with structured sparsity}.
\newblock {\em SIAM Journal on Optimization}, 17(1):218--242, 2006.

\bibitem{wang2021tssos}
J.~Wang, V.~Magron, and J.-B. Lasserre.
\newblock {TSSOS: A Moment-SOS hierarchy that exploits term sparsity}.
\newblock {\em SIAM Journal on Optimization}, 31(1):30--58, 2021.

\bibitem{wang2020cs}
J.~Wang, V.~Magron, J.~B. Lasserre, and N.~H.~A. Mai.
\newblock {CS-TSSOS: Correlative and term sparsity for large-scale polynomial optimization}.
\newblock {\em arXiv preprint arXiv:2005.02828}, 2020.

\bibitem{weisser2018sparse}
T.~Weisser, J.~B. Lasserre, and K.-C. Toh.
\newblock {Sparse-BSOS: a bounded degree SOS hierarchy for large scale polynomial optimization with sparsity}.
\newblock {\em Mathematical Programming Computation}, 10(1):1--32, 2018.

\bibitem{yuan2025exponentialspeedupsstructuredgoemanswilliamson}
H.~Yuan, D.~S. França, I.~Luchnikov, E.~Tiunov, T.~Haug, and L.~Aolita.
\newblock Exponential speed-ups for structured goemans-williamson relaxations via quantum gibbs states and pauli sparsity, 2025.

\bibitem{yurtsever2019conditional}
A.~Yurtsever, O.~Fercoq, and V.~Cevher.
\newblock A conditional-gradient-based augmented lagrangian framework.
\newblock In {\em International Conference on Machine Learning}, pages 7272--7281. PMLR, 2019.

\end{thebibliography}
\bibliographystyle{abbrv}
\appendix

\section{Proofs}

\subsection{Inequality-constrained polynomial optimization}

\begin{lemma}\label{lem:cond.bounded.trace.2}
Let $f\in\R[x]_{2d}$ and $g_i\in\R[x]_{2d_i}$, $i=0,\dots,m$, with $g_0=1$.
Suppose Assumption \ref{ass:cond.bounded.trace.2} holds.
Let $G^\star$ be an optimal solution to the SOS relaxation \eqref{eq:sdp.relaxation.cons.pop}.
Then $\tr(G^\star)\le 1$.
\end{lemma}
\begin{proof}
We write $G^\star=\diag(G_0^\star,\dots,G_m^\star)$ and assume $\underline \lambda\le \lambda_k$.
From the SOS relaxation \eqref{eq:sdp.relaxation.cons.pop}, we have:
\begin{equation}
f-\lambda_k=\sum_{i=0}^m g_iv_{k-d_i}^\top G_i^\star v_{k-d_i}\,.
\end{equation}
This implies: 
\begin{equation}
\begin{array}{rl}
\int (f-\underline \lambda)d\mu\ge &\int (f-\lambda_k)d\mu \\[5pt]
=&\sum_{i=0}^m \tr( G^\star_i M_{k-d_i}(g_i\mu) )\\[5pt]
=&\tr( G^\star D_k( g\mu))\ge \lambda_{\min}(D_k(g\mu )) \tr(G^\star)\,.
\end{array}
\end{equation}
Since $\mu$ is supported on $S(g)$, it follows that $\lambda_{\min}(D_k(g\mu ))\ge 0$. Thus, 
\begin{equation}
\tr(G^\star)\le \frac{\int (f-\underline \lambda)d\mu}{\lambda_{\min}(D_k(g\mu ))}\le 1\,.
\end{equation}
The last equality is due to Assumption \ref{ass:cond.bounded.trace.2}.
Hence, the result follows.
\end{proof}

\begin{lemma}\label{lem:bounded.eigenvalues}
Let Assumption \ref{ass:rescale.g} hold.
Then, for $\gamma\in\N^n_{2k}$, the matrix $  B_g^{(\gamma)}$ defined in \eqref{eq.tilde.B.gamma} satisfies $-I\preceq   B_g^{(\gamma)}\preceq I$, where $I$ is the identity matrix of the same size as $  B_g^{(\gamma)}$.
\end{lemma}
\begin{proof}
With $i=0$, note that $g_0=1$ satisfies $\|g_0\|_{\ell_1}=1$.
Fix $i\in\{0,1,\dots,m\}$.
By Lemma \ref{lem:properties.B.gamma} \eqref{bound.eigenvalue}, the matrix $B_i^{(\eta)}$ satisfies $-I_i\preceq B_i^{(\eta)}\preceq I_i$ for all $\eta\in \N^n_{2(k-d_i)}$, where $I_i$ is the identity matrix of the same size as $B_i^{(\eta)}$.
The $i$-th block of $  B_g^{(\gamma)}$ is given by
\begin{equation}\label{eq:ith.block.tilde.B.gamma}
B_{\gamma,i}=\sum\limits_{
\arraycolsep=1.4pt\def\arraystretch{.7}
\begin{array}{cc}
\scriptstyle \zeta\in\N^n_{2d_i}\\
\scriptstyle \gamma-\zeta\in \N^n
\end{array}
} g_{i,\zeta} B_i^{(\gamma-\zeta)}\,.
\end{equation}
Using the bounds on $B_i^{(\gamma-\zeta)}$, we find
\begin{equation}
 B_{\gamma,i}\preceq \sum\limits_{
\arraycolsep=1.4pt\def\arraystretch{.7}
\begin{array}{cc}
\scriptstyle \zeta\in\N^n_{2d_i}\\
\scriptstyle \gamma-\zeta\in \N^n
\end{array}
} |g_{i,\zeta}| I_i\preceq \|g_i\|_{\ell_1} I_i\preceq I_i
\end{equation}
and similarly, 
\begin{equation}
 B_{\gamma,i}\succeq- \sum\limits_{
\arraycolsep=1.4pt\def\arraystretch{.7}
\begin{array}{cc}
\scriptstyle \zeta\in\N^n_{2d_i}\\
\scriptstyle \gamma-\zeta\in \N^n
\end{array}
} |g_{i,\zeta}| I_i\succeq- \|g_i\|_{\ell_1} I_i\succeq - I_i\,.
\end{equation}
Thus, the block-diagonal matrix  $  B_g^{(\gamma)}=\diag(B_{\gamma,0},\dots,B_{\gamma,m})$ satisfies
\begin{equation}
-I=\diag(-I_0,\dots,-I_m)\preceq   B_g^{(\gamma)}\preceq \diag(I_0,\dots,I_m)=I\,.
\end{equation}
\end{proof}

\begin{lemma}\label{lem:equi.sdp.relax.cons.exact}
Let $f\in\R[x]_{2d}$ and $g_i\in\R[x]_{2d_i}$, for $i=0,\dots,m$, with $g_0=1$.
Suppose Assumption \ref{ass:cond.bounded.trace.2} holds.
Define $\lambda^\star=f_0-\lambda_k$.
Let $C$, $A_i$, for $i=1,\dots,M$, $N$ be generated by Algorithm \ref{alg:con.pop.exact}.
Then problem \eqref{eq:sdp.relaxation.cons.pop} is equivalent to problem \eqref{eq:sdp}.
Moreover, this is equivalent to
\begin{equation}\label{eq:sdp.trace.one.cons.pop.2.1.exact}
\lambda^\star=\inf\{\tr( C X)\,:\,X\in\mathbb S^{N}_+\,,\,\tr( A_i X) =b_i\,,\,i=1,\dots,M\}\,.
\end{equation}
\end{lemma}
\begin{proof}
By Lemma \ref{lem:cond.bounded.trace.2}, problem \eqref{eq:sdp.relaxation.cons.pop} can be written as 
\begin{equation}\label{eq:sdp.bounded.trace.cons.pop.exact}
\begin{array}{rl}
\lambda_k=\sup\limits_{\lambda,G} & \lambda\\
\text{s.t.}& \lambda\in\R\,,\,G=\diag(G_0,\dots,G_m)\succeq 0\,,\,\tr(G)\le 1\,,\\
&f-\lambda=\sum_{i=0}^m g_iv_{k-d_i}^\top G_iv_{k-d_i}\,.
\end{array}
\end{equation}
With $G_i=((G_i)_{\alpha,\beta})_{\alpha,\beta\in\N^n_{k-d_i}}$, the affine constraint $f-\lambda=\sum_{i=0}^m g_iv_{k-d_i}^\top G_iv_{k-d_i}$ expands as:
\begin{equation}
\begin{array}{rl}
f-\lambda=&\sum_{i=0}^m (\sum_{\zeta\in\N^n_{2d_i}} g_{i,\zeta}x^\zeta)\sum_{\alpha\in\N^n_{k-d_i}}\sum_{\beta\in\N^n_{k-d_i}} (G_i)_{\alpha,\beta} x^{\alpha+\beta}\\[10pt]
=&\sum_{i=0}^m \sum_{\zeta\in\N^n_{2d_i}} \sum_{\alpha\in\N^n_{k-d_i}}\sum_{\beta\in\N^n_{k-d_i}} g_{i,\zeta}(G_i)_{\alpha,\beta} x^{\alpha+\beta+\zeta}\,.
\end{array}
\end{equation}
Equating the coefficients of $x^\gamma$, we obtain:
\begin{equation}
\begin{cases}
f_0-\lambda=\sum_{i=0}^m \sum\limits_{
\arraycolsep=1.4pt\def\arraystretch{.7}
\begin{array}{cc}
\scriptstyle \alpha,\beta\in\N^n_{k-d_i}\\
\scriptstyle \zeta\in\N^n_{2d_i}\\
\scriptstyle \alpha+\beta+\zeta=0
\end{array}
} g_{i,\zeta} (G_i)_{\alpha,\beta}\,,\\

f_\gamma=\sum_{i=0}^m \sum\limits_{
\arraycolsep=1.4pt\def\arraystretch{.7}
\begin{array}{cc}
\scriptstyle \alpha,\beta\in\N^n_{k-d_i}\\
\scriptstyle \zeta\in\N^n_{2d_i}\\
\scriptstyle \alpha+\beta+\zeta=\gamma
\end{array}
} g_{i,\zeta} (G_i)_{\alpha,\beta}\,,\,\gamma\in\N^n_{2k}\backslash \{0\}\,.
\end{cases}
\end{equation}
This simplifies to: 
\begin{equation}
\begin{cases}
f_0-\lambda=\sum_{i=0}^m  g_{i,0} (G_i)_{0,0}\,,\\

f_\gamma=\sum_{i=0}^m \sum\limits_{
\arraycolsep=1.4pt\def\arraystretch{.7}
\begin{array}{cc}
\scriptstyle \zeta\in\N^n_{2d_i}\\
\scriptstyle \gamma-\zeta\in \N^n
\end{array}
} g_{i,\zeta} \sum\limits_{
\arraycolsep=1.4pt\def\arraystretch{.7}
\begin{array}{cc}
\scriptstyle \alpha,\beta\in\N^n_{k-d_i}\\
\scriptstyle \alpha+\beta=\gamma-\zeta
\end{array}
} (G_i)_{\alpha,\beta}\,,\,\gamma\in\N^n_{2k}\backslash \{0\}\,.
\end{cases}
\end{equation}
Using $B^{(\eta)}_i$ as defined in \eqref{eq:def.B.i.eta}, the second constraint becomes:
\begin{equation}
f_\gamma=\sum_{i=0}^m \sum\limits_{
\arraycolsep=1.4pt\def\arraystretch{.7}
\begin{array}{cc}
\scriptstyle \zeta\in\N^n_{2d_i}\\
\scriptstyle \gamma-\zeta\in \N^n
\end{array}
} g_{i,\zeta} \tr( B_i^{(\gamma-\xi)} G_i)=\tr(   B_g^{(\gamma)} G)\,,\,\gamma\in\N^n_{2k}\backslash \{0\}\,,
\end{equation}
where $  B_g^{(\gamma)}$ is defined as in \eqref{eq.tilde.B.gamma}.
Similarly, $\lambda=f_0-\tr(   B_g^{(0)} G)$.
By the definition of $\tilde A_i$ and $\tilde C$, we have:
\begin{equation}
\begin{cases}
\lambda=f_0-\tr(   B_g^{(0)} G) \Leftrightarrow\lambda =f_0-\tr( \tilde C G)\,,\\
f_{\gamma_i}=\tr(  B_g^{(\gamma_i)} G)\Leftrightarrow \tilde b_i=\tr(\tilde A_i G)\,,\,i=1,\dots,M\,.
\end{cases}
\end{equation}
Then problem \eqref{eq:sdp.bounded.trace.cons.pop.exact} becomes
\begin{equation}\label{eq:sdp.bounded.trace.cons.pop.2.exact}
\begin{array}{rl}
f_0-\lambda_k=\inf\limits_{G} & \tr( \tilde C G)\\
\text{s.t.}& G\in\mathbb S_+^{\tilde N}\,,\,\tr(G)\le 1\,,\\
&\tr( \tilde A_i G)=\tilde b_i\,,\,i=1,\dots,M\,.
\end{array}
\end{equation}
By the definition of $C$ and $A_i$, this is equivalent to \eqref{eq:sdp}.

Similarly, \eqref{eq:sdp.trace.one.cons.pop.2.1.exact} follows because removing the constraint $\tr(G)\le 1$ from \eqref{eq:sdp.bounded.trace.cons.pop.exact} does not affect the optimal value of the problem.
Thus, the result is established.
\end{proof}

\begin{lemma}\label{lem:properties.data.A.cons.exact}
Let $f\in\R[x]_{2d}$ and $g_i\in\R[x]_{2d_i}$, for $i=0,\dots,m$, with $g_0=1$.
Suppose Assumptions \ref{ass:cond.bounded.trace.2} and \ref{ass:rescale.g} holds.
Let $C$, $A_i,b_i$, for  $i=1,\dots,M$, and $N$ be generated by  Algorithm \ref{alg:con.pop.exact}.
Set $\lambda^\star=f_0-\lambda_k$.
Then the following properties hold:
\begin{enumerate}[(i)]
\item\label{sparsity.A} The maximum number of nonzero entries in any row of $C$ and $A_i$, for $i=1,\dots,M$, is $s_g\le \max\limits_{i=0,\dots,m} |\supp(g_i)|\le \max\limits_{i=0,\dots,m}\binom{n+2d_i}n$.
\item\label{bounded.eigenvalues} $-I_N\preceq C\preceq I_N$ and $-I_N\preceq A_i\preceq I_N$, for $i=1,\dots,M$.
\item\label{bound.opt.val.cons} $\lambda^\star\in[\lambda_{\min},\lambda_{\max}]$.
\item\label{remove.tr.cons.pop} $\lambda^\star=\inf\{\tr( CX)\,:\,X\in\mathbb S^{N}_+\,,\,\tr( A_i X) =b_i\,,\,i=1,\dots,M\}$.
\end{enumerate}
\end{lemma}
\begin{proof}
\eqref{sparsity.A} By the definition of $  B_g^{(\gamma)}$ in \eqref{eq.tilde.B.gamma}, for $i=0,\dots,m$, each row of the $i$-th block of $  B_g^{(\gamma)}$ contains at most $|\supp(g_i)|\le |\N^n_{2d_i}|=\binom{n+2d_i}n$ nonzero entries.
This is because $B_i^{(\eta)}$ has at most one nonzero entry on each row by Lemma \ref{lem:properties.B.gamma} \eqref{unique.nonzero.per.row}.
Consequently, each row of $  B_g^{(\gamma)}$ has at most $s_g\le \max\limits_{i=0,\dots,m}|\supp(g_i)|\le \max\limits_{i=0,\dots,m}\binom{n+2d_i}n$ nonzero entries.
The same property holds for $C=\diag(  B_g^{(0)},0)$ and $A_i=\diag(  B_g^{(\gamma_i)},0)$.

\eqref{bounded.eigenvalues} follows from Lemma \ref{lem:bounded.eigenvalues} and the definitions of $C$ and $A_i$.

\eqref{bound.opt.val.cons} Since $\lambda_{\min}=f_0-f(a)$ for $a\in S(g)$, and $\lambda_{\max}=f_0-\underline \lambda$ with $\underline \lambda\le \lambda_k\le f^\star\le f(a)$, we conclude that $\lambda^\star\in[\lambda_{\min},\lambda_{\max}]$.

\eqref{remove.tr.cons.pop} follows from the last statement of Lemma \ref{lem:equi.sdp.relax.cons.exact}.

\end{proof}

\begin{lemma}\label{lem:equivalent.mom.relax}
Let $C$, $A_i,b_i$, for  $i=1,\dots,M$, and $N$ be generated by Algorithm \ref{alg:con.pop.exact}.
Assume that strong duality holds for problems \eqref{eq:sdp.relaxation.cons.pop}-\eqref{eq:sdp.relaxation.cons.pop.dual},as well as for problem \eqref{eq:sdp.trace.one.cons.pop.2.1.exact} and its dual \eqref{eq:dual.no.tr.cons}.
Let $\xi\in\R^M$, and define $y=(y_{\gamma})_{\gamma\in\N^n_{2k}}$ such that $y_{0}=1$ and $y_{\gamma_i}=-\xi_i$, for $i=1,\dots,M$.
Then $y$ is an optimal solution to the moment relaxation \eqref{eq:sdp.relaxation.cons.pop.dual} if and only if $\xi$ is an optimal solution to problem \eqref{eq:dual.no.tr.cons}.
\end{lemma}

\begin{proof}
By strong duality, it holds that $\tau_k=\lambda_k$ and $\lambda^\star=\tau^\star$.
From the definition of $b$, the Riesz functional $L_y(f)$ can be expessed as
\begin{equation}\label{eq:rep.risz.linear.functional}
L_y(f)=y_0 f_0+\sum_{j=1}^M y_{\gamma_j}f_{\gamma_j}=f_0-b^\top \xi
\end{equation}
Next, we claim the following representation of the localizing matrix:
\begin{equation}\label{eq:rep.localizing.mat}
M_{k-d_i}(g_iy)=y_0 B_{0,i}+\sum_{j=1}^M y_{\gamma_j}B_{\gamma_j,i}\,,
\end{equation}
where 
$B_{\gamma,i}$, the $i$-th block of $  B_g^{(\gamma)}$, is defined as in \eqref{eq:ith.block.tilde.B.gamma}.
To show this, note that:
\begin{equation}
y_0 B_{0,i}+\sum_{j=1}^M y_{\gamma_j}B_{\gamma_j,i}=
\sum_{\gamma\in \N^n_{2k}} y_\gamma B_{\gamma,i}=\sum_{\gamma\in\N^n_{2k}}y_\gamma \sum\limits_{
\arraycolsep=1.4pt\def\arraystretch{.7}
\begin{array}{cc}
\scriptstyle \zeta\in\N^n_{2d_i}\\
\scriptstyle \gamma-\zeta\in \N^n
\end{array}
} g_{i,\zeta} B_i^{(\gamma-\zeta)}\,.
\end{equation}
Since $(B_i^{(\gamma-\zeta)})_{\alpha,\beta}$ equals $1$ if $\alpha+\beta=\gamma-\zeta$, and $0$ otherwise, the $(\alpha,\beta)$-entry of matrix $y_0 B_{0,i}+\sum_{j=1}^M y_{\gamma_j}B_{\gamma_j,i}$ is given by
\begin{equation}
\begin{array}{rl}
&\sum_{\gamma\in\N^n_{2k}}y_\gamma \sum\limits_{
\arraycolsep=1.4pt\def\arraystretch{.7}
\begin{array}{cc}
\scriptstyle \zeta\in\N^n_{2d_i}\\
\scriptstyle \gamma-\zeta\in \N^n
\end{array}
} g_{i,\zeta} (B_i^{(\gamma-\zeta)})_{\alpha,\beta}\\
=&\sum_{\gamma\in\N^n_{2k}}y_\gamma \sum\limits_{
\arraycolsep=1.4pt\def\arraystretch{.7}
\begin{array}{cc}
\scriptstyle \zeta\in\N^n_{2d_i}\\
\scriptstyle \gamma-\zeta\in \N^n\\
\scriptstyle \alpha+\beta=\gamma-\zeta
\end{array}
} g_{i,\zeta} \\
= &\sum\limits_{\zeta\in\N^n_{2d_i}} g_{i,\zeta}
 \sum\limits_{
\arraycolsep=1.4pt\def\arraystretch{.7}
\begin{array}{cc}
\scriptstyle \gamma\in\N^n_{2k}\\
\scriptstyle \gamma=\alpha+\beta+\zeta
\end{array}
}y_\gamma\\
=&\sum\limits_{\zeta\in\N^n_{2d_i}} g_{i,\zeta}
 y_{\alpha+\beta+\zeta}\,.
\end{array}
\end{equation}
This matches the $(\alpha,\beta)$-entry of $M_{k-d_i}(g_iy)$, confirming \eqref{eq:rep.localizing.mat}.

From \eqref{eq:rep.localizing.mat} and the definition of $\tilde C,\tilde A_i$, we have
\begin{equation}
\begin{array}{rl}
D_k(gy)=& \diag((M_{k-d_i}(g_iy))_{i=0,\dots,m})\\[5pt]
=& y_0 \diag((B_{0,i})_{i=0,\dots,m})+\sum_{j=1}^M y_{\gamma_j}\diag((B_{\gamma_j,i})_{i=0,\dots,m})\\[5pt]
=& y_0   B_g^{(0)}+\sum_{j=1}^M y_{\gamma_j}  B_g^{(\gamma_j)}\\[5pt]
=&\tilde C-\sum_{j=1}^M \xi_j\tilde A_j\,,
\end{array}
\end{equation}
where $D_k(gy)$ is defined as in \eqref{eq:block.diagonal.mom.local}.

Thus, we obtain
 $\diag(D_k(gy),0)= C-\sum_{i=1}^M \xi_i A_i=C-\mathcal A^\top \xi$.
 Using this result alongside \eqref{eq:rep.risz.linear.functional}, we can conclude the following equivalences:
 \begin{equation}
 \begin{array}{rl}
 &y\text{ is an optimal solution for \eqref{eq:sdp.relaxation.cons.pop.dual}}\\[5pt]
 \Leftrightarrow & D_k(gy)\succeq 0 \text{ and }L_y(f)=\tau_k=\lambda_k \\[5pt]
  \Leftrightarrow & C-\mathcal A^\top \xi\succeq 0 \text{ and }b^\top \xi =f_0-\lambda_k=\lambda^\star=\tau^\star \\[5pt]
  \Leftrightarrow &\xi\text{ is an optimal solution for \eqref{eq:dual.no.tr.cons}.}
 \end{array}
 \end{equation}
\end{proof}

\subsubsection{Proof of Theorem \ref{theo:convergence.ineq.cons}}
\label{proof:theo:convergence.ineq.cons}
\begin{proof}
Let $\lambda^\star=f_0-\lambda_k$.
Since strong duality holds for problems  \eqref{eq:sdp.relaxation.cons.pop}-\eqref{eq:sdp.relaxation.cons.pop.dual}, it also holds for problem \eqref{eq:sdp.trace.one.cons.pop.2.1.exact} and its dual \eqref{eq:dual.no.tr.cons} by Lemma \ref{lem:equi.sdp.relax.cons.exact}.

Let $y^\star=(y^\star_\alpha)_{\alpha\in\N^n_{2k}}$ be the truncated moment sequence with respect to the Dirac measure $\delta_{x^\star}$, where $y^\star_\alpha=\int x^\alpha d \delta_{x^\star}=x^{\star \alpha}$.

We claim that $y^\star$ is an optimal solution to problem \eqref{eq:sdp.relaxation.cons.pop.dual}.
First, observe that $y^{\star}_0=(x^{\star})^0=1$ and $M_{k-d_i}(g_iy^{\star})=g_i(x^\star)v_{k-d_i}(x^\star)v_{k-d_i}(x^\star)^\top\succeq 0$.
Then $y^\star$ is a feasible solution for problem \eqref{eq:sdp.relaxation.cons.pop.dual}.
Additionally, $L_{y^\star}(f)=\int f d \delta_{x^\star}=f(x^\star)=f^\star=\lambda_k=\tau_k$,where the last equality follows from strong duality. Therefore, $y^\star$ is an optimal solution to problem \eqref{eq:sdp.relaxation.cons.pop.dual}.
Define $\tilde y^\star=(y^\star_\gamma)_{\gamma\in\N^n_{2k}\backslash \{0\}}$.
By Lemma \ref{lem:equivalent.mom.relax}, $-\tilde y^\star$ is optimal for problem \eqref{eq:dual.no.tr.cons}.
 
The remainder follows similarly to the proof of Theorem \ref{theo:accuarcy.approximate.val.uncons} using Lemmas \ref{lem:properties.data.A.cons.exact} and \ref{lem:convergence.sdp}.
\end{proof}
\subsubsection{Proof of Theorem \ref{theo:convergence.over.ball}}
\label{proof:theo:convergence.over.ball}
\begin{proof}
Let $y=(y_\gamma)_{\gamma\in\N^n_{2k}}$ be a feasible solution for problem \eqref{eq:sdp.relaxation.cons.pop.dual}.

We first prove by induction that for $j=1,\dots,k$,
\begin{equation}\label{eq:ineq.ball}
L_y(R^j-\|x\|_{\ell_2}^{2j})\ge 0\,.
\end{equation}
For $j=1$, we have $L_y(R^j-\|x\|_{\ell_2}^{2j})=L_y(g_1)\ge 0$ since it corresponds to the $(0,0)$-entry  of the localizing matrix $M_{k-d_1}(g_1y)\succeq 0$.
Assume that \eqref{eq:ineq.ball} holds for $j=1,\dots,t$ with $1\le t\le k$.
To prove for $j=t+1$, use the identity:
\begin{equation}\label{eq:rep.r.power.ball}
\begin{array}{rl}
R^{t+1}-\|x\|_{\ell_2}^{2(t+1)}=R(R^{t}-\|x\|_{\ell_2}^{2t})+\|x\|_{\ell_2}^{2t}(R-\|x\|_{\ell_2}^{2})\,.
\end{array}
\end{equation}
Since $\|x\|_{\ell_2}^{2t}$ is a sum-of-squares  when $t$ is even, it can be expressed as:
\begin{equation}
\|x\|_{\ell_2}^{2t}=(u^\top v_{k-d_1})^2=u^\top v_{k-d_1}v_{k-d_1}^\top u\,,
\end{equation}
for some real vector $u$.
Applying $L_y$ to the identity \eqref{eq:rep.r.power.ball}:
\begin{equation}
\begin{array}{rl}
L_y(R^{t+1}-\|x\|_{\ell_2}^{2(t+1)})
=&L_y(R(R^t-\|x\|_{\ell_2}^{2t}))+L_y(\|x\|_{\ell_2}^{2t}(R-\|x\|_{\ell_2}^2))\\[5pt]
=&RL_y(R^t-\|x\|_{\ell_2}^{2t})+\frac{1}{c}L_y(u^\top v_{k-d_1}v_{k-d_1}^\top u g_1)\\[5pt]
\ge &\frac{1}{c}L_y(u^\top (g_1v_{k-d_1}v_{k-d_1}^\top) u)\\[5pt]
=& \frac{1}{c}\tr( uu^\top M_{k-d_1}(g_1y))\ge 0\,,
\end{array}
\end{equation}
where the first inequality uses the induction hypothesis, and the second follows from the positive semidefiniteness of $uu^\top$ and $M_{k-d_{1}}(g_{1}y)$.
Thus, the induction is complete.

Next, we find the upper bound on the trace of the moment matrix $M_k(y)$.
From \eqref{eq:ineq.ball}, 
\begin{equation}
\begin{array}{rl}
R^j=R^jy_0=L_y(R^j)\ge& L_y((x_1^2+\dots+x_n^2)^j)\\[5pt]
=&L_y(\sum_{\arraycolsep=1.4pt\def\arraystretch{.7}
\begin{array}{cc}
\scriptstyle \gamma\in\N^{n}\\
\scriptstyle |\gamma|=j
\end{array}
} \binom{j}{\gamma} x^{2\gamma})\\[15pt]
=&\sum_{\arraycolsep=1.4pt\def\arraystretch{.7}
\begin{array}{cc}
\scriptstyle \gamma\in\N^{n}\\
\scriptstyle |\gamma|=j
\end{array}
} \binom{j}{\gamma} y_{2\gamma}\\[15pt]
\ge &\sum_{\arraycolsep=1.4pt\def\arraystretch{.7}
\begin{array}{cc}
\scriptstyle \gamma\in\N^{n}\\
\scriptstyle |\gamma|=j
\end{array}
} y_{2\gamma}\,,
\end{array}
\end{equation}
where the last step follows because $\binom{j}{\gamma}\ge 1$.
Summing over $j=1,\dots,k$, we get
\begin{equation}
\tr(M_{k}(y))=y_0+\sum_{\gamma\in\N^n_k\backslash\{0\}}y_{2\gamma}=1+ \sum_{j=1}^{k} \sum_{\arraycolsep=1.4pt\def\arraystretch{.7}
\begin{array}{cc}
\scriptstyle \gamma\in\N^{n}\\
\scriptstyle |\gamma|=j
\end{array}
} y_{2\gamma} \le 1+ \sum_{j=1}^{k} R^j\le \frac{1}{1-R}\,.
\end{equation}
From this, it follows that
\begin{equation}
\begin{array}{rl}
\frac{1}{1-R}\ge& \tr(M_{k}(y))=\|M_{k}(y)\|_*\qquad\quad\text{(since $M_{k}(y)\succeq 0$)}\\[5pt]
\ge& \|M_{k}(y)\|_F=\|\text{vec}(M_{k}(y))\|_{\ell_2}\\[5pt]
\ge&\frac{1}{\sqrt{L}} \|\text{vec}(M_{k}(y))\|_{\ell_1}\\[5pt]
=&\frac{1}{\sqrt{L}}\left(\sum_{\alpha\in\N^n_k} |y_{2\alpha}|+\sum_{\arraycolsep=1.4pt\def\arraystretch{.7}
\begin{array}{cc}
\scriptstyle \alpha,\beta\in\N^n_k\\
\scriptstyle \alpha\ne \beta
\end{array}
} \sqrt{2}|y_{\alpha+\beta}|\right)\\[15pt]
\ge& \frac{1}{\sqrt{L}}\sum_{\gamma\in\N^n_{2k}}|y_\gamma|\,,
\end{array}
\end{equation}
where $L=\frac{1}{2}\binom{n+k}n[\binom{n+k}n+1]$ is the length of $\text{vec}(M_{k}(y))$.
This implies that with $\tilde y=(y_\alpha)_{\alpha\in\N^n_{2k}\backslash \{0\}}$, 
\begin{equation}
\|\tilde y\|_{\ell_1}\le \sum_{\gamma\in\N^n_{2k}}|y_\gamma| \le \frac{\sqrt{L}}{1-R}\,.
\end{equation}
Thus, for $y^\star$, the optimal solution to problem \eqref{eq:sdp.relaxation.cons.pop.dual}, we have $\|\tilde y^\star\|_{\ell_1}\le \frac{\sqrt{L}}{1-R}$, where $\tilde y^\star=(y^\star_\gamma)_{\gamma\in\N^n_{2k}\backslash \{0\}}$.
Moreover, $-\tilde y^\star$ is optimal for problem \eqref{eq:dual.no.tr.cons}.
The remaining proof follows similarly to the argument in Theorem \ref{theo:convergence.ineq.cons}.
\end{proof}

\subsubsection{Proof of Theorem \ref{theo:convergence.over.simplex}}
\label{proof:theo:convergence.over.simplex}
\begin{proof}
Let $y=(y_\gamma)_{\gamma\in\N^n_{2k}}$ be a feasible solution for problem \eqref{eq:sdp.relaxation.cons.pop.dual}.

We first prove that $y_\gamma\ge 0$ for all $\gamma\in\N^n_{2k}$.
\begin{itemize}
\item If $\gamma=0$, then $y_\gamma=y_0=1\ge 0$.
\item If $\gamma\ne 0$, there exists $\alpha\in\{0,1\}^n\backslash \{0\}$ such that $\gamma-\alpha\in 2\N^n$.
Define $\beta=\frac{\gamma-\alpha}{2}$. Then $\beta\in\N^n$, and we have $y_\gamma=L_y(x^{\gamma})=L_y(x^{\alpha+2\beta})=L_y(g_{i_\alpha}x^{2\beta})$.
This implies $y_\gamma$ is the $(\beta,\beta)$-entry of the localizing matrix $M_{k-d_{i_{\alpha}}}(g_{i_\alpha}y)$.
Since $M_{k-d_{i_{\alpha}}}(g_{i_\alpha}y)\succeq 0$, it follows that $y_\gamma\ge 0$. 
\end{itemize}

Next, we prove by induction that for $j=1,\dots,2k$,
\begin{equation}\label{eq:ineq.simplex}
L_y(r^j-(x_1+\dots+x_n)^j)\ge 0\,.
\end{equation}
For $j=1,2$, $L_y(r^j-(x_1+\dots+x_n)^j)=L_y(g_{i_j})\ge 0$ as it corresponds to the $(0,0)$-entry  of the localizing matrix $M_{k-d_{i_j}}(g_{i_j}y)\succeq 0$.
Assume that \eqref{eq:ineq.simplex} holds for $j=1,2,\dots,t$ with $3\le t\le 2k$.
We prove it for $j=t+1$.
Consider the following two cases:
\begin{itemize}
\item Case 1: $t$ is even. 
Using the identity
\begin{equation}\label{eq:rep.r.power}
\begin{array}{rl}
r^{t+1}-(x_1+\dots+x_n)^{t+1}=&r(r^t-(x_1+\dots+x_n)^t)\\
&+(x_1+\dots+x_n)^t(r-(x_1+\dots+x_n))\,,
\end{array}
\end{equation}
and noting that $(x_1+\dots+x_n)^t$ is a sum of squares when $t$ is even, we write
\begin{equation}
(x_1+\dots+x_n)^t=(u^\top v_{k-d_{i_1}})^2=u^\top v_{k-d_{i_1}}v_{k-d_{i_1}}^\top u\,,
\end{equation}
for some real vector $u$.
Applying $L_y$ for \eqref{eq:rep.r.power}, we get
\begin{equation}
\begin{array}{rl}
&L_y(r^{t+1}-(x_1+\dots+x_n)^{t+1})\\[5pt]
=&L_y(r(r^t-(x_1+\dots+x_n)^t))\\[5pt]
&+L_y((x_1+\dots+x_n)^t(r-(x_1+\dots+x_n)))\\[5pt]
=&rL_y(r^t-(x_1+\dots+x_n)^t)+\frac{1}{c_1}L_y(u^\top v_{k-d_{i_1}}v_{k-d_{i_1}}^\top u g_{i_1})\\[5pt]
\ge &\frac{1}{c_1}L_y(u^\top (g_{i_1}v_{k-d_{i_1}}v_{k-d_{i_1}}^\top) u)= \frac{1}{c_1}\tr( uu^\top M_{k-d_{i_1}}(g_{i_1}y))\ge 0\,,
\end{array}
\end{equation}
where the first inequality uses the induction hypothesis, and the second follows from the positive semidefiniteness of $uu^\top$ and $M_{k-d_{i_1}}(g_{i_1}y)$.
\item Case 2: $t$ is odd. 
Using the identity 
\begin{equation}\label{eq:rep.r.power.2}
\begin{array}{rl}
r^{t+1}-(x_1+\dots+x_n)^{t+1}=&r^2(r^{t-1}-(x_1+\dots+x_n)^{t-1})\\
&+(x_1+\dots+x_n)^{t-1}(r^2-(x_1+\dots+x_n)^2)\,,
\end{array}
\end{equation}
and noting that $(x_1+\dots+x_n)^{t-1}$ is a sum-of-squares when $t$ is odd, we write
\begin{equation}
(x_1+\dots+x_n)^{t-1}=(u^\top v_{k-d_{i_2}})^2=u^\top v_{k-d_{i_2}}v_{k-d_{i_2}}^\top u\,,
\end{equation}
for some real vector $u$.
Applying $L_y$ for \eqref{eq:rep.r.power.2}, we similarly obtain
\begin{equation}
\begin{array}{rl}
&L_y(r^{t+1}-(x_1+\dots+x_n)^{t+1})\\[5pt]
=&L_y(r^2(r^{t-1}-(x_1+\dots+x_n)^{t-1}))\\[5pt]
&+L_y((x_1+\dots+x_n)^{t-1}(r^2-(x_1+\dots+x_n)^2))\\[5pt]
=&r^2L_y(r^{t-1}-(x_1+\dots+x_n)^{t-1})+\frac{1}{c_2}L_y(u^\top v_{k-d_{i_2}}v_{k-d_{i_2}}^\top u g_{i_2})\\[5pt]
\ge &\frac{1}{c_2}L_y(u^\top (g_{i_2}v_{k-d_{i_2}}v_{k-d_{i_2}}^\top) u)= \frac{1}{c_2}\tr(uu^\top M_{k-d_{i_2}}(g_{i_2}y))\ge 0\,.
\end{array}
\end{equation}
\end{itemize}
In both cases, $L_y(r^{t+1}-(x_1+\dots+x_n)^{t+1})\ge 0$, completing the induction.

We now show that $\sum_{\gamma\in\N^n_{2k}\backslash \{0\}} y_\gamma\le \frac{r}{1-r}$.
From the inequality \eqref{eq:ineq.simplex}, 
\begin{equation}
\begin{array}{rl}
r^j=r^jy_0=L_y(r^j)\ge& L_y((x_1+\dots+x_n)^j)\\[5pt]
=&L_y(\sum_{\arraycolsep=1.4pt\def\arraystretch{.7}
\begin{array}{cc}
\scriptstyle \gamma\in\N^{n}\\
\scriptstyle |\gamma|=j
\end{array}
} \binom{j}{\gamma} x^\gamma)\\[15pt]
=&\sum_{\arraycolsep=1.4pt\def\arraystretch{.7}
\begin{array}{cc}
\scriptstyle \gamma\in\N^{n}\\
\scriptstyle |\gamma|=j
\end{array}
} \binom{j}{\gamma} y_\gamma\\[10pt]
\ge &\sum_{\arraycolsep=1.4pt\def\arraystretch{.7}
\begin{array}{cc}
\scriptstyle \gamma\in\N^{n}\\
\scriptstyle |\gamma|=j
\end{array}
} y_\gamma\,,
\end{array}
\end{equation}
where the last step follows because $\binom{j}{\gamma}\ge 1$ and $y_\gamma\ge 0$.
Summing over $j=1,\dots,2k$, we get
\begin{equation}
\sum_{\gamma\in\N^n_{2k}\backslash \{0\}} y_\gamma=\sum_{j=1}^{2k} \sum_{\arraycolsep=1.4pt\def\arraystretch{.7}
\begin{array}{cc}
\scriptstyle \gamma\in\N^{n}\\
\scriptstyle |\gamma|=j
\end{array}
} y_\gamma \le \sum_{j=1}^{2k} r^j\le \frac{r}{1-r}\,.
\end{equation}
With $\tilde y=(y_\gamma)_{\gamma\in\N^n_{2k}\backslash \{0\}}$, this implies that $\|\tilde y\|_{\ell_1}=\sum_{\gamma\in\N^n_{2k}\backslash \{0\}} y_\gamma \le \frac{r}{1-r}$
Thus, for $y^\star$, the optimal solution to problem \eqref{eq:sdp.relaxation.cons.pop.dual}, we have $\|\tilde y^\star\|_{\ell_1}\le \frac{r}{1-r}$, where $\tilde y^\star=(y^\star_\gamma)_{\gamma\in\N^n_{2k}\backslash \{0\}}$.
Moreover, $-\tilde y^\star$ is optimal for problem \eqref{eq:dual.no.tr.cons}.
The remainder of the proof follows similarly to the argument in Theorem \ref{theo:convergence.ineq.cons}.
\end{proof}

\subsubsection{Proof of Lemma \ref{lem:complexity.con.pop.exact}}
\label{proof:lem:complexity.con.pop.exact}
\begin{proof}
The most computationally expensive part of Algorithm \ref{alg:con.pop.exact} is the binary search using Hamiltonian updates in Step \ref{step:hamilton.update.cons.exact}.
By Lemma \ref{lem:complex.binary.search}, the classical computer requires $O((MsN+N^\omega)\varepsilon^{-2})$ operations, and the quantum computer requires $O(s(\sqrt{M}+\sqrt{N}\varepsilon^{-1})\varepsilon^{-4})$ operations,
where $s$ is the maximum number of nonzero entries in a row of the data matrices $C$ and $A_i$, for $i=1,\dots,M$.
By Lemma \ref{lem:properties.data.A.cons.exact} \eqref{sparsity.A}, we have $s=s_g\le\max\limits_{i=0,\dots,m} |\supp(g_i)|$, where the last value is the maximum number of nonzero coefficients of each $g_i$, for $i=1,\dots,m$.
Here, $M=\binom{n+2k}n-1$ and 
\begin{equation}\label{eq:bound.N}
N=1+\tilde N=1+\sum_{i=0}^m\binom{n+k-d_i}n \,.
\end{equation}
Hence the result follows.
\end{proof}

\subsection{Binary search using Hamiltonian updates for semidefinite programs with affine inequality constraints}
\begin{lemma}\label{lem:obj.at.X.ineq.aff}
Let $b\in\R^M$, $C\in \mathbb S^N$, and let $\mathcal A:\mathbb S^N\to \R^M$ be a linear operator defined by 
\begin{equation}
\mathcal AY=(\tr( A_1 Y),\dots,\tr( A_MY))\,,
\end{equation}
where $A_i\in\mathbb S^N$, for $i=1,\dots,M$.
Consider primal-dual semidefinite programs:
\begin{equation}\label{eq:primal.no.tr.cons.ineq.aff}
\begin{array}{rl}
\lambda^\star=\inf\limits_{Y}&\tr( CY)\\
\text{s.t.}&Y\in\mathbb S^{N}_+\,,\,\mathcal A Y\le b\,,
\end{array}
\end{equation}
\begin{equation}\label{eq:dual.no.tr.cons.ineq.aff}
\begin{array}{rl}
\tau^\star=\sup\limits_{\lambda}&b^\top \xi\\
\text{s.t.}&\xi\in\R^M\,,\,\lambda\le 0\,,\,C-\mathcal A^\top \xi \succeq 0\,,
\end{array}
\end{equation}
where $\mathcal A^\top:\R^M\to \mathbb S^N$ is the adjoint operator of $\mathcal A$, given by
\begin{equation}
\mathcal A^\top \xi=\sum_{i=1}^M \xi_i A_i\,.
\end{equation}
Assume strong duality holds for problems \eqref{eq:primal.no.tr.cons.ineq.aff}-\eqref{eq:dual.no.tr.cons.ineq.aff}: $\lambda^\star=\tau^\star$, and both problems admit solutions $Y^\star$  and $\xi^\star$, respectively.
Then, for all $X\in \mathbb S^N_+$, the following hold:
\begin{equation}
\tr( C X) \ge\lambda^\star - \xi^{\star\top}(b-\mathcal{A}X)\,.
\end{equation}
\end{lemma}
\begin{proof}
Set $Z^\star=C-\mathcal A ^\top \xi^\star$.
Then $Z^\star\succeq 0$.
The Lagrangian has the form
\begin{equation}
L(Y,Z,\xi)=\tr( CY) -\tr( Z Y)-\xi^\top (\mathcal A X-b)\,.
\end{equation}
By the Karush–Kuhn–Tucker (KKT) conditions, we have
\begin{equation}\label{eq:KKT1.ineq.aff}
0=\frac{\partial L}{\partial Y}(Y^\star,Z^\star,\xi^\star)=C-Z^\star-\mathcal A^\top \xi^\star\,.
\end{equation}
\begin{equation}\label{eq:KKT2.ineq.aff}
\tr( Z^\star Y^\star)=0\,,
\end{equation}
\begin{equation}\label{eq:KKT3.ineq.aff}
\xi^{\star\top} (\mathcal A Y^\star-b)=0\,.
\end{equation}
Let $X\in \mathbb S^N_+$.
By \eqref{eq:KKT1.ineq.aff}, $C=Z^\star+\mathcal{A}^\top \xi^\star$.
Then
\begin{equation}\label{eq:bound4.ineq.aff}
\begin{array}{rl}
\tr( C X)-\lambda^\star=&\tr( C( X-Y^\star))\\
=&\tr(( Z^\star+\mathcal{A}^\top \xi^\star)( X-Y^\star))  \\
=& \tr( Z^\star X) + \tr( \mathcal{A}^\top \xi^\star (X-Y^\star))\\
=& \tr( Z^\star X) +   \xi^{\star\top}\mathcal{A}(X-Y^\star )\\
=& \tr( Z^\star X) +   \xi^{\star\top}(\mathcal{A}X-b )\\
\ge &   \xi^{\star\top}(\mathcal{A}X-b )\,. 
\end{array}
\end{equation}
The first equality follows from $\tr( C Y^\star) =\lambda^\star$, the third equality is due to \eqref{eq:KKT2.ineq.aff}, the last equality uses \eqref{eq:KKT3.ineq.aff}, and the last inequality is based on the positive semidefiniteness of $Z^\star$ and $X$.
Hence, the result follows.
\end{proof}

Consider the following semidefinite program with trace one:
\begin{equation}\label{eq:sdp.ineq.aff}
\begin{array}{rl}
\lambda^\star=\inf\limits_{X}&\tr( CX)\\
\text{s.t.}&X\in\mathbb S^{N}_+\,,\,\tr(X)= 1\,,\\
&\tr( A_iX) \le b_i\,,\,i=1,\dots,M\,,
\end{array}
\end{equation}
where $C,A_i\in\mathbb S^N$, and $b_i\in\R$, for $i=1,\dots,M$, are given.
Define the set $S^{(\lambda)}$ as
\begin{equation}\label{eq:S.rho.ineq.aff}
S^{(\lambda)}=\left\{X\in \mathbb S^N_+\left|\begin{array}{rl}
& \tr( A_i X) \le b_i\,,\,i=1,\dots,M\,,\\
&\tr( C X) \le \lambda\,,\,\tr(X)=1
\end{array} \right. \right\}\,.
\end{equation}
We propose the following algorithm to numerically solve the SDP \eqref{eq:sdp.ineq.aff}:
\begin{algorithm}\label{alg:Binary.search.HU.ineq.aff}
Binary search using Hamiltonian updates
\begin{itemize}
\item Input: $\lambda_{\min}\in\R$, $\lambda_{\max}>\lambda_{\min}$, $\varepsilon>0$, $A_i\in \mathbb S^N$, $b_i\in\R$, $i=1,\dots,M$.
\item Output: $\underline \lambda_T\in\R$ and $X_\varepsilon\in\mathbb S^N$.
\end{itemize}
\begin{enumerate}
\item Set $\underline\lambda_0=\lambda_{\min}$, $\overline\lambda_0=\lambda_{\max}$ and $T=\left\lceil \log_2\left(\frac{\lambda_{\max}-\lambda_{\min}}{\varepsilon}\right)\right\rceil$.
\item For $t=0,\dots,T-1$, do
\begin{enumerate}
\item Set $\lambda_t=\frac{\underline\lambda_t+\overline \lambda_t}{2}$.
\item Run Algorithm \ref{alg:HU} (Hamiltonian updates) to test the feasibility of $S^{(\lambda_t)}$ (defined in \eqref{eq:S.rho.ineq.aff}) and obtain $\eta_t\in\{0,1\}$ and $X_\varepsilon^{(t)}\in \mathbb S^N$.
\item Based on $\eta_t$:
\begin{enumerate}
\item If $\eta_t=0$, set $\underline \lambda_{t+1}=\lambda_t$ and $\overline \lambda_{t+1}=\overline \lambda_t$.
\item If $\eta_t=1$, set $\underline \lambda_{t+1}=\underline\lambda_t$, $\overline \lambda_{t+1}= \lambda_t$ and $X_\varepsilon=X_\varepsilon^{(t)}$. 
\end{enumerate}
\end{enumerate}
\end{enumerate}
\end{algorithm}
Define the set:
\begin{equation}
S^{(\lambda)}_\varepsilon=\left\{X\in \mathbb S^N_+\left|\begin{array}{rl}
& \tr( A_i X) \le b_i+\varepsilon\,,\,i=1,\dots,M\,,\\
&\tr( C X) \le \lambda+\varepsilon\,,\,\tr(X)=1
\end{array} \right. \right\}\,.
\end{equation}
Note that $S^{(\lambda)}_0=S^{(\lambda)}$.

The following lemma provides a necessary and sufficient condition for $S^{(\lambda)}$ to be nonempty:

\begin{lemma}\label{lem:feas.ineq.ineq.aff}
Let $C,A_i\in \mathbb S^N$, and $b_i\in\R$, for $i=1,\dots,M$.
Let $\lambda^\star$ be the optimal value of the semidefinite program in \eqref{eq:sdp.ineq.aff}.
Assume that the problem \eqref{eq:sdp.ineq.aff} has an optimal solution.
Define $S^{(\lambda)}$ as in \eqref{eq:S.rho.ineq.aff}.
Then,
\begin{equation}
S^{(\lambda)}\ne\emptyset\quad\Leftrightarrow\quad \lambda\ge \lambda^\star\,.
\end{equation}
\end{lemma}
The proof of Lemma \ref{lem:feas.ineq.ineq.aff} is similar to that of Lemma \ref{lem:feas.ineq}.

We present the basic properties of the values returned by binary search in the following two lemmas:

\begin{lemma}\label{lem:induc.ineq.ineq.aff}
Let $(\underline \lambda_{t})_{t=0}^T$ and $(\overline \lambda_{t})_{t=0}^T$ be the sequences generated by Algorithm \ref{alg:Binary.search.HU.ineq.aff}.
Then, for $t=0,\dots,T-1$, the following holds:
\begin{equation}\label{eq:induc.ineq.ineq.aff}
\overline\lambda_{t+1}-\underline\lambda_{t+1}=\frac{1}{2}(\overline\lambda_{t}-\underline\lambda_{t})\,.
\end{equation}
\end{lemma}
The proof of Lemma \ref{lem:induc.ineq.ineq.aff} is similar to Lemma \ref{lem:induc.ineq}.

\begin{lemma}\label{lem:monotoncity.ineq.aff}
Let $(\underline \lambda_{t})_{t=0}^T$ and $(\overline \lambda_{t})_{t=0}^T$ be the sequences generated by Algorithm \ref{alg:Binary.search.HU.ineq.aff}.
Then, the following statements hold:
\begin{enumerate}
\item For $t=0,\dots,T-1$, the sequence satisfies 
\begin{equation}
\lambda_{\min}\le\underline \lambda_t\le \underline \lambda_{t+1}\le \overline \lambda_{t+1}\le \overline \lambda_t\le \lambda_{\max}\,.
\end{equation}
\item For $t=0,\dots,T$, we have $\overline \lambda_t -\underline \lambda_t=\frac{1}{2^t}(\lambda_{\max}-\lambda_{\min})$. In particular, 
\begin{equation}
\overline \lambda_T -\underline \lambda_T=\frac{1}{2^T}(\lambda_{\max}-\lambda_{\min})\le \varepsilon\,.
\end{equation}
\end{enumerate}
\end{lemma}
The proof of Lemma \ref{lem:monotoncity.ineq.aff} is based on Lemma \ref{lem:induc.ineq.ineq.aff} and similar to Lemma \ref{lem:monotoncity}.

\begin{lemma}\label{lem:zero.lower.bound.ineq.aff}
Let $C,A_i\in \mathbb S^N$, and $b_i\in\R$, for $i=1,\dots,M$.
Assume that $I_N\succeq C\succeq -I_N$ and $I_N\succeq A_i\succeq -I_N$, for $i=1,\dots,M$.
Let $\lambda^\star$ be the optimal value in \eqref{eq:sdp.ineq.aff}, and assume that problem \eqref{eq:sdp.ineq.aff} has an optimal solution.
Additionally, let $\lambda_{\min}\le \lambda^\star$, and let $(\underline \lambda_{t})_{t=0}^T$ be the sequence generated by Algorithm \ref{alg:Binary.search.HU.ineq.aff}.
Then for all $t=0,\dots,T$,
\begin{equation}\label{eq:zero.lower.bound.ineq.aff}
0\le \lambda^\star-\underline \lambda_t\,.
\end{equation}
\end{lemma}
The proof of Lemma \ref{lem:zero.lower.bound.ineq.aff} is similar to Lemma \ref{lem:zero.lower.bound} and based on Lemmas \ref{lem:Hamilton.updates} and \ref{lem:feas.ineq.ineq.aff}

\begin{lemma}\label{lem:upper.bound.ineq.aff}
Let $C,A_i\in \mathbb S^N$, $b_i\in\R$, for $i=1,\dots,M$, and assume that $I_N\succeq C\succeq -I_N$ and $I_N\succeq A_i\succeq -I_N$, for $i=1,\dots,M$.
Let $\lambda^\star$ be the optimal value in \eqref{eq:sdp.ineq.aff} and assume $\lambda_{\max}\ge \lambda^\star$.
Further, assume the following equality holds:
\begin{equation}\label{eq:remove.trace.one.ineq.aff}
\lambda^\star=\inf\{\tr( C X)\,:\,X\in\mathbb S^{N}_+\,,\,\tr( A_i X) \le b_i\,,\,i=1,\dots,M\}\,.
\end{equation}
Assume strong duality holds for problem \eqref{eq:remove.trace.one.ineq.aff} and its dual \eqref{eq:dual.no.tr.cons.ineq.aff}, which admits a solution $\xi^\star\ge 0$.
Let $(\overline \lambda_{t})_{t=0}^T$ be the sequence generated in Algorithm \ref{alg:Binary.search.HU.ineq.aff}.
Then for $t=0,\dots,T$, we have the following bound:
\begin{equation}\label{eq:upper.bound.ineq.aff}
\lambda^\star-\overline \lambda_t\le \varepsilon (1-\sum_{i=1}^M\xi^\star_i)\,.
\end{equation}
\end{lemma}
\begin{proof}
We will prove statement \eqref{eq:upper.bound.ineq.aff} by induction, using Lemma \ref{lem:obj.at.X.ineq.aff}. 
\begin{itemize}
\item Base case: For $t=0$, we have
\begin{equation}
\lambda^\star-\overline \lambda_0=\lambda^\star-\lambda_{\max}\le 0\le \varepsilon\left(1-\sum_{i=1}^M\xi_i^\star\right)\,.
\end{equation}
Thus \eqref{eq:upper.bound.ineq.aff} holds for $t=0$.
\item Inductive Step:
Assume that \eqref{eq:upper.bound.ineq.aff} holds for some $t=r\le T-1$.
We now show that it holds for $t=r+1$.
Consider Step 2 (c) of Algorithm \ref{alg:Binary.search.HU.ineq.aff} with $t=r$.
\begin{itemize}
\item Case 1: $\eta_r=0$. In this case, we have $\overline \lambda_{r+1}=\overline \lambda_r$. By the induction assumption, it follows that
\begin{equation}
\overline \lambda_{r+1}=\overline \lambda_r\ge \lambda^\star - \varepsilon \left(1-\sum_{i=1}^M\xi^\star_i \right)\,.
\end{equation}
\item Case 2: $\eta_r=1$.
By Lemma \ref{lem:Hamilton.updates},  we know that $X_\varepsilon^{(r)}\in S^{(\lambda_r)}_\varepsilon$.
By Lemma \ref{lem:obj.at.X.ineq.aff}, we obtain
\begin{equation}
\tr( C X_\varepsilon^{(r)}) \ge\lambda^\star -\sum_{i=1}^M\lambda^\star_i  [b_i-\tr(A_iX_\varepsilon^{(r)})] \ge \lambda^\star + \varepsilon \sum_{i=1}^M\lambda^\star_i \,.
\end{equation}
Since $X_\varepsilon^{(r)}\in S^{(\lambda_r)}_\varepsilon$, we have $b_i-\tr(A_iX_\varepsilon^{(r)})\ge -\varepsilon$ for all $i=1,\dots,M$ and the last inequality holds.
Therefore, we have the following estimate:
\begin{equation}
\overline \lambda_{r+1} +\varepsilon=\lambda_r+\varepsilon\ge \tr( C X_\varepsilon^{(r)})
\ge   \lambda^\star + \varepsilon \sum_{i=1}^M\xi^\star_i\,.
\end{equation}
The first inequality is based on $X_\varepsilon^{(r)}\in S^{(\lambda_r)}_\varepsilon$.
It implies that
\begin{equation}
\overline \lambda_{r+1} \ge \lambda^\star -\varepsilon(1-\sum_{i=1}^M\lambda^\star_i )\,.
\end{equation}
\end{itemize}
Combining both cases, we conclude that
\begin{equation}
\lambda^\star-\overline \lambda_{r+1}\le \varepsilon\left(1-\sum_{i=1}^M\xi^\star_i \right)\,.
\end{equation}
\end{itemize}
Thus, by induction, \eqref{eq:upper.bound.ineq.aff} holds for all $t=0,\dots,T$.
\end{proof}

\begin{lemma}[Accuracy]\label{lem:convergence.sdp.ineq.aff}
Let $C,A_i\in \mathbb S^N$, $b_i\in\R$, for $i=1,\dots,M$, and assume that $I_N\succeq C\succeq -I_N$ and $I_N\succeq A_i\succeq -I_N$, for $i=1,\dots,M$.
Assume that problem \eqref{eq:sdp.ineq.aff} has an optimal solution.
Let $\lambda^\star$ be the optimal value in \eqref{eq:sdp.ineq.aff}, and let $\lambda_{\min},\lambda_{\max}\in\R$ such that $\lambda^\star\in [\lambda_{\min},\lambda_{\max}]$.
Further, assume the following equality holds:
\begin{equation}\label{eq:remove.trace.one.1.ineq.aff}
\lambda^\star=\inf\{\tr(CX)\,:\,X\in\mathbb S^{N}_+\,,\,\tr( A_i X) \le b_i\,,\,i=1,\dots,M\}\,.
\end{equation}
Assume strong duality holds for problem \eqref{eq:remove.trace.one.1.ineq.aff} and its dual \eqref{eq:dual.no.tr.cons.ineq.aff}, which admits a solution $\xi^\star\ge 0$.
Let $\underline \lambda_{T}\in\R$ be the value returned by Algorithm \ref{alg:Binary.search.HU.ineq.aff}.
Then, the following convergence result holds:
\begin{equation}
0\le \lambda^\star-\overline \lambda_T\le \varepsilon (2-\sum_{i=1}^M\xi^\star_i)\,.
\end{equation}
\end{lemma}
The proof of Lemma \ref{lem:convergence.sdp.ineq.aff} is  similar to that of Lemma \ref{lem:convergence.sdp} and based on Lemmas \ref{lem:zero.lower.bound.ineq.aff}, \ref{lem:upper.bound.ineq.aff}, and \ref{lem:monotoncity.ineq.aff}.

The following lemma presents the complexity of binary search using Hamiltonian updates, as derived from the work of Apeldoorn and Gily\'en \cite{van2018improvements}:
\begin{lemma}[Complexity]
\label{lem:complex.binary.search.ineq.aff}
The complexity of running Algorithm \ref{alg:Binary.search.HU.ineq.aff} is as follows:
\begin{itemize}
\item On a classical computer, the algorithm takes
\begin{equation}
O((MsN+N^\omega)\varepsilon^{-2})
\end{equation}
operations, where $\omega\approx 2.373$, $N$ is the size of the matrix, $M$ is the number of affine constraints, and $\varepsilon$ is the desired accuracy.
\item On a quantum computer, the algorithm requires
\begin{equation}
O(s(\sqrt{M}+\sqrt{N}\varepsilon^{-1})\varepsilon^{-4})
\end{equation}
operations, where $s$ is the maximum number of nonzero entries in any row of the matrices $C$ and $A_i$ (for $i=1,\dots,M$).
\end{itemize}
\end{lemma}

\subsection{Inequality constrained polynomial optimization over the simplex}

\begin{lemma}\label{lem:cond.bounded.trace.2.ineq.aff}
Let $f\in\R[x]_{2d}$ and $g_i\in\R[x]_{2d_i}$, $i=0,\dots,m$, with $g_0=1$.
Assume $S(g)\subset \R_+^n$.
Suppose Assumption \ref{ass:cond.bounded.trace.2} holds.
Let $G^\star$ be an optimal solution to \eqref{eq:sdp.relaxation.cons.pop.ineq.aff}.
Then $\tr(G^\star)\le 1$.
\end{lemma}
\begin{proof}
We write $G^\star=\diag(G_0^\star,\dots,G_m^\star)$ and assume $\underline \lambda\le \lambda_k$.
From the SOS relaxation, we have:
\begin{equation}\label{eq:coeff.identities}
\begin{cases}
f_0-\lambda\ge\sum_{i=0}^m  g_{i,0} (G_i^*)_{0,0}\,,\\[5pt]
f_\gamma\ge\sum_{i=0}^m \sum\limits_{
\arraycolsep=1.4pt\def\arraystretch{.7}
\begin{array}{cc}
\scriptstyle \xi\in\N^n_{2d_i}\\
\scriptstyle \gamma-\xi\in \N^n
\end{array}
} g_{i,\xi} \sum\limits_{
\arraycolsep=1.4pt\def\arraystretch{.7}
\begin{array}{cc}
\scriptstyle \alpha,\beta\in\N^n_{k-d_i}\\
\scriptstyle \alpha+\beta=\gamma-\xi
\end{array}
} (G_i^*)_{\alpha,\beta}\,,\,\gamma\in\N^n_{2k}\backslash \{0\}\,.
\end{cases}
\end{equation}
Let $y=(y_\gamma)_{\gamma\in\N^n_{2k}}$ be the truncated moment sequence with respect to $\mu$. 
Since $\mu$ is supported on $S(g)\subset \R_+^n$, for all $\gamma\in\N^n_{2k}$, $y_\gamma\ge 0$.
From \eqref{eq:coeff.identities}, 
\begin{equation}
\begin{array}{rl}
0\le&  y_0 (f_0-\lambda-\sum_{i=0}^m  g_{i,0} (G_i^*)_{0,0})\\[5pt]
&+ \sum\limits_{\gamma\in\N^n_{2k}\backslash\{0\}} y_\gamma \left( f_\gamma-\sum_{i=0}^m \sum\limits_{
\arraycolsep=1.4pt\def\arraystretch{.7}
\begin{array}{cc}
\scriptstyle \xi\in\N^n_{2d_i}\\
\scriptstyle \gamma-\xi\in \N^n
\end{array}
} g_{i,\xi} \sum\limits_{
\arraycolsep=1.4pt\def\arraystretch{.7}
\begin{array}{cc}
\scriptstyle \alpha,\beta\in\N^n_{k-d_i}\\
\scriptstyle \alpha+\beta=\gamma-\xi
\end{array}
} (G_i^*)_{\alpha,\beta} \right)\\[25pt]
=& L_y(f-\lambda-\sum_{i=0}^m g_iv_{k-d_i}^\top G_i^\star v_{k-d_i})\\[5pt]
=& \int (f-\lambda_k)d\mu -\sum_{i=0}^m \tr( G^\star_i M_{k-d_i}(g_i\mu) )
\end{array}
\end{equation}
This implies that 
\begin{equation}
\int (f-\lambda_k)d\mu \ge \sum_{i=0}^m \tr( G^\star_i  M_{k-d_i}(g_i\mu) )\,.
\end{equation}
The remaining is similar to the proof of Lemma \ref{lem:cond.bounded.trace.2}.
\end{proof}

\begin{lemma}\label{lem:equi.sdp.relax.cons.exact.ineq.aff}
Let $f\in\R[x]_{2d}$ and $g_i\in\R[x]_{2d_i}$, for $i=0,\dots,m$, with $g_0=1$.
Suppose Assumption \ref{ass:cond.bounded.trace.2} holds.
Define $\lambda^\star=f_0-\lambda_k$.
Let $C$, $A_i$, for $i=1,\dots,M$, $N$ be generated by Algorithm \ref{alg:con.pop.exact.ineq.aff}.
Then problem \eqref{eq:sdp.relaxation.cons.pop.ineq.aff} is equivalent to problem \eqref{eq:sdp.ineq.aff}.
Moreover, this is equivalent to
\begin{equation}\label{eq:sdp.trace.one.cons.pop.2.1.exact.ineq.aff}
\lambda^\star=\inf\{\tr( CX)\,:\,X\in\mathbb S^{N}_+\,,\,\tr( A_i X) \le b_i\,,\,i=1,\dots,M\}\,.
\end{equation}
\end{lemma}
The proof of Lemma \ref{lem:equi.sdp.relax.cons.exact.ineq.aff}  is based on Lemma \ref{lem:cond.bounded.trace.2.ineq.aff} similar to that of Lemma \ref{lem:equi.sdp.relax.cons.exact}.

\begin{lemma}\label{lem:properties.data.A.cons.exact.ineq.aff}
Let $f\in\R[x]_{2d}$ and $g_i\in\R[x]_{2d_i}$, for $i=0,\dots,m$, with $g_0=1$.
Suppose Assumption \ref{ass:rescale.g} holds.
Let $C$, $A_i,b_i$, for  $i=1,\dots,M$, and $N$ be generated by Algorithm \ref{alg:con.pop.exact.ineq.aff}.
Set $\lambda^\star=f_0-\lambda_k$.
Then the following properties hold:
\begin{enumerate}[(i)]
\item\label{sparsity.A.ineq.aff} The maximum number of nonzero entries in any row of $C$ and $A_i$, for $i=1,\dots,M$, is $s_g\le\max\limits_{i=0,\dots,m} |\supp(g_i)|\le \max\limits_{i=0,\dots,m}\binom{n+2d_i}n$.
\item\label{bounded.eigenvalues.ineq.aff} $-I_N\preceq C\preceq I_N$ and $-I_N\preceq A_i\preceq I_N$, for $i=1,\dots,M$.
\item\label{bound.opt.val.cons.ineq.aff} $\lambda^\star\in[\lambda_{\min},\lambda_{\max}]$.
\item\label{remove.tr.cons.pop.ineq.aff} $\lambda^\star=\inf\{\tr( CX)\,:\,X\in\mathbb S^{N}_+\,,\,\tr( A_iX) \le b_i\,,\,i=1,\dots,M\}$.
\end{enumerate}
\end{lemma}

The proof of Lemma \ref{lem:properties.data.A.cons.exact.ineq.aff}  is based on Lemma \ref{lem:equi.sdp.relax.cons.exact.ineq.aff} similar to that of Lemma \ref{lem:properties.data.A.cons.exact}.

\begin{lemma}\label{lem:equivalent.mom.relax.ineq.aff}
Let $C$, $A_i,b_i$, for  $i=1,\dots,M$, and $N$ be generated by Algorithm \ref{alg:con.pop.exact.ineq.aff}.
Assume that strong duality holds for problems \eqref{eq:sdp.relaxation.cons.pop.ineq.aff}-\eqref{eq:sdp.relaxation.cons.pop.dual.ineq.aff},as well as for problem \eqref{eq:sdp.trace.one.cons.pop.2.1.exact.ineq.aff} and its dual \eqref{eq:dual.no.tr.cons.ineq.aff}.
Let $\xi\in\R^M$, and define $y=(y_{\gamma})_{\gamma\in\N^n_{2k}}$ such that $y_{0}=1$ and $y_{\gamma_i}=-\xi_i$, for $i=1,\dots,M$.
Then $y$ is an optimal solution to the moment relaxation \eqref{eq:sdp.relaxation.cons.pop.dual} if and only if $\xi$ is an optimal solution to problem \eqref{eq:dual.no.tr.cons.ineq.aff}.
\end{lemma}

The proof of Lemma \ref{lem:equivalent.mom.relax.ineq.aff}  is similar to that of Lemma \ref{lem:equivalent.mom.relax}.

\subsubsection{Proof of Theorem \ref{theo:convergence.ineq.cons.ineq.aff}}
\label{proof:theo:convergence.ineq.cons.ineq.aff}
\begin{proof}
Let $\lambda^\star=f_0-\hat\lambda_k$.
Since strong duality holds for problems  \eqref{eq:sdp.relaxation.cons.pop.ineq.aff}-\eqref{eq:sdp.relaxation.cons.pop.dual.ineq.aff}, it also holds for problem \eqref{eq:sdp.trace.one.cons.pop.2.1.exact} and its dual \eqref{eq:dual.no.tr.cons.ineq.aff} by Lemma \ref{lem:equi.sdp.relax.cons.exact.ineq.aff}.

Let $y^\star=(y^\star_\alpha)_{\alpha\in\N^n_{2k}}$ be the truncated moment sequence with respect to the Dirac measure $\delta_{x^\star}$, where $y^\star_\alpha=\int x^\alpha d \delta_{x^\star}=x^{\star \alpha}$.
Since $x^\star\in S(g)\subset\R^n_+$, $y^\star_\alpha\ge 0$.

We claim that $y^\star$ is an optimal solution to problem \eqref{eq:sdp.relaxation.cons.pop.dual.ineq.aff}.
First, observe that $y^{\star}_0=(x^{\star})^0=1$ and $M_{k-d_i}(g_iy^{\star})=g_i(x^\star)v_{k-d_i}(x^\star)v_{k-d_i}(x^\star)^\top\succeq 0$.
Then $y^\star$ is a feasible solution for problem \eqref{eq:sdp.relaxation.cons.pop.dual.ineq.aff}.
Additionally, $L_{y^\star}(f)=\int f d \delta_{x^\star}=f(x^\star)=f^\star=\lambda_k=\tau_k$,where the last equality follows from strong duality. Therefore, $y^\star$ is an optimal solution to problem \eqref{eq:sdp.relaxation.cons.pop.dual.ineq.aff}.
Define $\tilde y^\star=(y^\star_\gamma)_{\gamma\in\N^n_{2k}\backslash \{0\}}$.
By Lemma \ref{lem:equivalent.mom.relax.ineq.aff}, $-\tilde y^\star$ is optimal for problem \eqref{eq:dual.no.tr.cons.ineq.aff}.
 
The remainder follows similarly to the proof of Theorem \ref{theo:accuarcy.approximate.val.uncons} by replacing $\lambda_k$ (resp. $\lambda_k^{(\varepsilon)}$) by $\hat \lambda_k$ (resp. $\hat \lambda_k^{(\varepsilon)}$) and using Lemmas \ref{lem:properties.data.A.cons.exact.ineq.aff} and \ref{lem:convergence.sdp.ineq.aff}.
\end{proof}

\subsubsection{Proof of Theorem \ref{theo:convergence.over.simplex.ineq.aff}}
\label{proof:theo:convergence.over.simplex.ineq.aff}
\begin{proof}
Let $y=(y_\gamma)_{\gamma\in\N^n_{2k}}$ be a feasible solution for problem \eqref{eq:sdp.relaxation.cons.pop.dual.ineq.aff}.
Then $y_\gamma\ge 0$ for all $\gamma\in\N^n_{2k}$.
Similar to the proof of Theorem \ref{theo:convergence.over.simplex}, we obtain $\sum_{\gamma\in\N^n_{2k}\backslash \{0\}} y_\gamma \le \frac{r}{1-r}$.
Thus, for $y^\star$, the optimal solution to problem \eqref{eq:sdp.relaxation.cons.pop.dual.ineq.aff}, we have $\sum_{\gamma\in\N^n_{2k}\backslash \{0\}} y_\gamma^\star\le \frac{r}{1-r}$, where $\tilde y^\star=(y^\star_\gamma)_{\gamma\in\N^n_{2k}\backslash \{0\}}$.
Moreover, by Lemma \ref{lem:equivalent.mom.relax.ineq.aff}, $-\tilde y^\star$ is optimal for problem \eqref{eq:dual.no.tr.cons.ineq.aff}.
By Lemmas \ref{lem:properties.data.A.cons.exact.ineq.aff} and \ref{lem:convergence.sdp.ineq.aff}, $0\le \lambda^\star -\underline \lambda_T\le \varepsilon(2+\sum_{\gamma\in\N^n_{2k}\backslash \{0\}} y_\gamma^\star)\le \varepsilon(2+\frac{r}{1-r})$.
Recall that $\lambda_k^{(\varepsilon)}=f_0-\underline \lambda_T$. Substituting, we get:
$\lambda^\star -\underline \lambda_T=\lambda_k^{(\varepsilon)}-\lambda_k=\lambda_k^{(\varepsilon)}-f^\star$.
Thus, the result follows.
\end{proof}

\section{Additional quantum complexities  for solving SDP relaxations of polynomial optimization}
\label{sec:order.quantum.complexities}

In this appendix, we present some quantum complexities for solving the semidefinite relaxations \eqref{eq:sos.sdp.general.general.pop}-\eqref{eq:mom.relax.pop} of the polynomial optimization problem \eqref{eq:pop.general} under Slater's condition. Our approach in this appendix also utilizes binary search and Hamiltonian updates.

This method applies to the following cases, under the assumption that strong duality holds for the primal-dual semidefinite relaxations \eqref{eq:sos.sdp.general.general.pop}-\eqref{eq:mom.relax.pop}.

For the first two cases, we use the same rescaling coefficient on the objective polynomial $f$ as described in the main body of the paper to derive an upper bound on the trace of the matrix variable in the semidefinite program corresponding to the sum-of-squares relaxation \eqref{eq:sos.sdp.general.general.pop}.
We assume that Slater's condition holds for the sum-of-squares relaxation \eqref{eq:sos.sdp.general.general.pop}, meaning there exists a strictly feasible solution $\diag(G_0^,\dots,G_m)\succ 0$ for the sum-of-squares  relaxation \eqref{eq:sos.sdp.general.general.pop} with conditional number $\kappa $.

\begin{itemize}
\item \textbf{Case 1: Unconstrained polynomial optimization} (see Appendix \ref{sec:uncons.pop.2}).

Consider the unconstrained case of problem \eqref{eq:pop.general}, where $m=l=0$. 

Algorithm \ref{alg:uncon.pop} provides an approximate value for  $\lambda_k$ with an accuracy of the minimum of
\begin{equation}\label{eq:accuracy.case1}
\varepsilon\kappa  \binom{n+k}n \quad\text{ and }\quad\varepsilon\kappa \binom{n+2k}n^{1/2} 
\end{equation}
(see Theorem \ref{theo:accuarcy.approximate.val.uncons.2}).
\item \textbf{Case 2: Inequality-constrained polynomial optimization problem} (see Appendix \ref{sec:ineq.cons.add.complex}).

Now consider the inequality-constrained case of problem \eqref{eq:pop.general}, where $l=0$. 

Our quantum Algorithm \ref{alg:con.pop} achieves a runtime of
\begin{equation}\label{eq:complexity.ineq.pop.intro}
\footnotesize
O\left(\sum_{i=0}^m \binom{n+k-d_i}n^2+\left[\binom{n+k}n\binom{n+2k}n\sum_{i=0}^m \binom{n+k-d_i}n+\left(\sum_{i=0}^m \binom{n+k-d_i}n\right)^\omega\right] \frac{1}{\varepsilon^2}\right)
\end{equation}
to compute an approximate value for $\lambda_k$ with an accuracy of
\begin{equation}\label{eq:accuracy.case2}
\varepsilon \kappa  \binom{n+2k}n^{1/2} \sum_{i=0}^m|g_{i,0}| \,,
\end{equation}
where $g_{i,0}$ is the constant coefficient of $g_i$ (see Theorem \ref{theo:accuray.with.gram.schimidt.ineqpop}).
Note that the term without $\varepsilon$ in the complexity \eqref{eq:complexity.ineq.pop.intro} represents the time required to prepare the data using the Gram-Schmidt process.
\end{itemize}
 
In the following two cases, we formulate the moment relaxation \eqref{eq:mom.relax.pop} as a standard semidefinite program with a constant or bounded trace property, specifically when the polynomial optimization problem  \eqref{eq:pop.general} includes a sphere or ball constraint.
We also assume that Slater's condition holds for the moment relaxation \eqref{eq:mom.relax.pop}, that is, there exists a feasible solution $y$ such that
\begin{equation}
\diag(M_k(y),M_{k-d_1}(g_1y),\dots,M_{k-d_m}(g_my))\succ 0
\end{equation}
with a condition number denoted by $\kappa $.
Note that in these last two cases, the coefficients of the objective polynomial $f$ and the constraint polynomials $g_i$ are not rescaled, unlike in the first two cases.
\begin{itemize}
\item \textbf{Case 3: Equality-constrained polynomial optimization over the unit sphere} (see Appendix \ref{sec:mom.relax.eq.cons.sphere}).

Consider the polynomial optimization problem \eqref{eq:pop.general} with $m=0$, $l\ge 1$, and $h_1=1-\|x\|_{\ell_2}^2$. 
In this scenario, our quantum Algorithm \ref{alg:eq.con.pop.sphere} for solving the moment relaxation \eqref{eq:mom.relax.pop} has a complexity of
\begin{equation}\label{eq:complexity.eq.pop.sphere.intro}
\footnotesize
O\left(\binom{n+k}{n}^{1.5}\left\{(l-1) \binom{n+2k}n \binom{n+k}{n}^{2.5}+\binom{n+k}{n}^{4.5}+\left[\binom{n+k}{n}^{0.5}+\frac{1}{\varepsilon} \right]\frac{1}{\varepsilon^{4}}\right\}\right)\,.
\end{equation}
This algorithm computes an approximate value for $\tau_k$ with accuracy 
\begin{equation}
\varepsilon 2^k\|f\|_{\ell_1}  \kappa  \binom{n+k}n \max\limits_{\alpha\in\N^n_k} \binom{k}{(k-|\alpha|,\alpha)} \,,
\end{equation}
where $\binom{k}{(\alpha_0,\alpha)}$, with $\alpha_0+\alpha_1+\dots+\alpha_n=k$, denotes the multinomial coefficient (see Theorem \ref{theo:accuracy.moment.relax.eqpop.sphere}).
\item \textbf{Case 4: (In)equality-constrained polynomial optimization over the unit ball} (see Appendix \ref{sec:mom.relax.ineq.cons.all}).

Consider the polynomial optimization problem \eqref{eq:pop.general} with $m\ge 1$, and $g_1=1-\|x\|_{\ell_2}^2$. 
In this case, our quantum Algorithm \ref{alg:eq.con.pop.ball} for solving the moment relaxation \eqref{eq:mom.relax.pop} has a complexity of
\begin{equation}\label{eq:complexity.ineq.pop.ball.intro}
\tiny
O\left(\binom{n+k}{n}^{1.5}\left\{m(l-1) \binom{n+k}{n}^{0.5}\binom{n+2k}n +m^3\binom{n+k}{n}^{4.5}+m^{0.5}\left[\binom{n+k}{n}^{0.5}+\frac{1}{\varepsilon}\right]\frac{1}{\varepsilon^4}\right\}\right)\,.
\end{equation}
This algorithm computes an approximate value for $\tau_k$  with accuracy 
\begin{equation}
\varepsilon 2^{k}\|f\|_{\ell_1}  \kappa r_{g,n,k}m^{1/2}\binom{n+k}n \,,
\end{equation}
where $r_{g,n,k}$ is a positive constant depending on $g_i,n,k$ (see Theorem \ref{theo:accuracy.moment.relax.ineqpop.ball}).
\end{itemize}
Here, $\|f\|_{\ell_1}$ represents the $\ell_1$-norm of the vector of coefficients of polynomial $f$.
Note that the terms without $\varepsilon$ in the complexities \eqref{eq:complexity.eq.pop.sphere.intro} and \eqref{eq:complexity.ineq.pop.ball.intro} correspond to the time required for data preparation using Gaussian elimination and the Gram-Schmidt process.

In Table \ref{tab:summary.appendix} we summarize the results for the cases considered above.

\paragraph{Limitations.}
The above results are obtained under the Slater condition, and in particular, the error bounds on the optimal values of the SDP relaxations depend on the condition number $\kappa$ of the Slater point in the feasible region of these relaxations.
However, finding such a Slater point is generally difficult; in fact, identifying one is equivalent to solving a semidefinite feasibility problem of the same size as the SDP under consideration. Moreover, the condition number $\kappa$ of the Slater point may itself depend on the problem size. These assumptions therefore make the results somewhat less convincing than those presented in the main part of the paper.

\begin{table}
\tiny
\caption{Summary of additional quantum algorithms based on Hamiltonian updates for solving order-$k$ sum-of-squares and moment relaxations \eqref{eq:sos.sdp.general.general.pop}-\eqref{eq:mom.relax.pop} of the polynomial optimization problem (POP) \eqref{eq:pop.general} with $n$ variables, $m$ inequality constraints and $l$ equality constraints. Here $g$ (resp. $h$) is the set of inequality (resp. equality) constraints in \eqref{eq:pop.general}; $\kappa $ is the condition number of some strictly feasible solution of relaxations; $n_k$ stands for $\binom{n+k}n$;  $n_{m,k}$ stands for $\sum_{i=0}^m\binom{n+k-d_i}n$; and $n_{m,k}^{(2)}$ stands for $\sum_{i=0}^m\binom{n+k-d_i}n^2$.}
\label{tab:summary.appendix}
\begin{center}
\begin{tabular}{ |p{1.7cm}|p{0.9cm}|p{1.1cm}|p{3.5cm}|p{6.7cm}| } 
\hline
Type of POP & Quantum 

algorithm &Conditions &  Accuracy for approximate optimal value & Quantum complexity \\
\hline
Unconstrained ($m=l=0$) &\ref{alg:uncon.pop}& Slater's conditions & $\varepsilon\kappa \min\{ n_k ,n_{2k}^{1/2}\}$

(Theorem \ref{theo:accuarcy.approximate.val.uncons.2})&$O\left(\left(n_{2k}^{1/2}+\frac{1}\varepsilon n_k^{1/2}\ \right)\frac{1}{\varepsilon^4}\right)$\\\cline{1-2}\cline{4-5}
Inequality-constrained ($l=0$)& \ref{alg:con.pop} & &  $\varepsilon \kappa  n_k^{1/2} \sum_{i=0}^m|g_{i,0}|$ 

(Theorem \ref{theo:accuray.with.gram.schimidt.ineqpop})& $O\left(n_{m,k}^{(2)}+\left[n_kn_{2k}n_{m,k}+n_{m,k}^\omega\right] \varepsilon^{-2}\right)$  \\ 
\cline{1-2}\cline{4-5}
 Equality-constrained over the unit sphere 
 
 ($m=0$, $l\ge 1$, $h_1=1-\|x\|_{\ell_2}^2$)& \ref{alg:eq.con.pop.sphere} &  &  $\varepsilon 2^k\|f\|_{\ell_1}\kappa  n_k\max\limits_{\alpha\in\N^n_k} \binom{k}{(k-|\alpha|,\alpha)}$ 

(Theorem \ref{theo:accuracy.moment.relax.eqpop.sphere})& $O\left(n_k^{1.5}\left[(l-1) n_{2k} n_k^{2.5}+n_k^{4.5}+\left(n_k^{0.5}+\frac{1}{\varepsilon} \right)\frac{1}{\varepsilon^{4}}\right]\right)$  \\
\cline{1-2}\cline{4-5}
 (In)equality-constrained over the unit ball
 
 ($m\ge 1$,
 
  $g_1=1-\|x\|_{\ell_2}^2$)& \ref{alg:eq.con.pop.ball} &  &  $\varepsilon 2^{k}\|f\|_{\ell_1} \kappa r_{g,n,k}m^{1/2}n_k$ 
  
 (Theorem \ref{theo:accuracy.moment.relax.ineqpop.ball}) & $O\left(n_k^{1.5}\left[m(l-1) n_k^{0.5}n_{2k} +m^3n_k^{4.5}+m^{0.5}\left(n_k^{0.5}+\frac{1}{\varepsilon}\right)\frac{1}{\varepsilon^4}\right]\right)$  \\ 
\hline
\end{tabular}
\end{center}
\end{table}

\subsection{Other quantum complexities  for solving sum-of-squares relaxations}
\label{sec:other.quan.complex.SOS}

\begin{lemma}\label{lem:obj.at.X.2}
Let all assumptions and notations in Lemma \ref{lem:obj.at.X} hold.
Define the matrix $A$ associated with $\mathcal A$ as 
\begin{equation}
A=\begin{bmatrix}
\text{vec}(A_1)^\top \\
\dots\\
\text{vec}(A_M)^\top
\end{bmatrix}\,.
\end{equation}
If $\lambda^\star\ge 0$, $A$ has full rank $M$ and Slater's condition holds for the primal problem \eqref{eq:primal.no.tr.cons} (i.e., there exists $U\succ 0$ such that $\mathcal A U=b$), then:
\begin{enumerate}
\item[(i)] $\tr(C X)\ge \lambda^\star -(1+\sqrt{2} N \kappa(U))\|\text{vec}(C)\|_{\ell_1}\|(AA^\top)^{-1}A\|_{\ell_1\to\ell_1}\|\mathcal{A}X-b \|_{\ell_\infty}$.
\item[(ii)] $\tr(C X)\ge \lambda^\star -(1+\sqrt{2} \kappa(U)) \|\text{vec}(C)\|_{\ell_1} \sigma_{\min}(A)^{-1}\sqrt{M}\|\mathcal{A}X-b \|_{\ell_\infty}$.
\end{enumerate}
\end{lemma}

\begin{proof}

(i) Define $c=\text{vec}(C)$ and $z^\star=\text{vec}(Z^\star)$.
Rewriting the first KKT condition \eqref{eq:KKT1} in vector form gives:
$c-z^\star-A^\top \xi^\star=0$,
which implies:
\begin{equation}\label{eq:def.lambd}
\xi^\star=(AA^\top)^{-1}A(c-z^\star)\,.
\end{equation}
From this, we obtain:
\begin{equation}
\|\xi^\star\|_{\ell_1}\le \|(AA^\top)^{-1}A\|_{\ell_1\to\ell_1}\|c-z^\star\|_{\ell_1}\le \|(AA^\top)^{-1}A\|_{\ell_1\to\ell_1}(\|c\|_{\ell_1}+\|z^\star\|_{\ell_1})\,.
\end{equation}
Combining this result with the bound \eqref{eq:bound2}, it follows:
\begin{equation}\label{eq:bound3}
\tr(C X)\ge \lambda^\star -(\|c\|_{\ell_1}+\|z^\star\|_{\ell_1})\|(AA^\top)^{-1}A\|_{\ell_1\to\ell_1}\|\mathcal{A}X-b \|_{\ell_\infty}\,.
\end{equation}
Assume $\lambda^\star\ge 0$, $A$ has full rank $M$, and that Slater's condition holds for problem \eqref{eq:primal.no.tr.cons}, i.e., there exists $U\succ 0$ such that $\mathcal A U=b$.
Multiplying the first KKT condition \eqref{eq:KKT1} by $U-Y^\star$ and taking the trace, we have:
\begin{equation}
\tr( C (U-Y^\star)) -\tr( Z^\star ( U-Y^\star)) -\tr( \mathcal A^\top \xi^\star( U-Y^\star))=0
\end{equation}
Since $\tr(Z^\star  Y^\star)=0$ and $\tr( \mathcal A^\top \xi^\star( U-Y^\star))=\xi^{\star\top}\mathcal{A}(U-Y^\star )=\xi^{\star\top}(b-b)=0$, we deduce:
\begin{equation}
\tr( C U) -\lambda^\star -\tr( Z^\star  U) =0
\end{equation}
Since $\lambda^\star>0$, it follows that:
\begin{equation}\label{eq:bound1}
\begin{array}{rl}
\lambda_{\min}(U)\tr(Z^\star)\le&\tr( Z^\star U) = \tr( CU) -\lambda^\star\\
\le& c^\top u\le \|c\|_{\ell_1}\|u\|_{\ell_\infty}\,,
\end{array}
\end{equation}
where $u=\text{vec}(U)$.
The first inequality is derived from the positive semidefiniteness of $Z^\star$.
Let $U=(U_{ij})_{i,j=1,\dots,N}$.
Since $U\succeq 0$, we have $0\le U_{ii}\le \lambda_{\max}(U)$.
This follows from $0\le e_i^\top U e_i\le \lambda_{\max}(U) e_i^\top I e_i$, where $e_i=(0,\dots,0,1,0,\dots,0)$, since $0\preceq U\preceq \lambda_{\max}(U) I$.
Since $U\succeq 0$, it holds that: $\begin{bmatrix}
U_{ii}& U_{ij}\\
U_{ij}& U_{jj}
\end{bmatrix}\succeq 0$. 
This implies $U_{ij}^2\le U_{ii}U_{jj}\le \lambda_{\max}(U)^2$, and thus $|U_{ij}|\le \lambda_{\max}(U)$.
Consequently:
\begin{equation}
\|u\|_{\ell_\infty} =\max_{i\le j}|U_{ij}|\sqrt{2-\delta_{ij}}\le \lambda_{\max}(U)\sqrt{2}\,.
\end{equation}
From \eqref{eq:bound1}, we deduce:
\begin{equation}
\lambda_{\min}(U)\tr(Z^\star)\le\sqrt{2} \|c\|_{\ell_1}\lambda_{\max}(U)\,.
\end{equation}
This implies: 
\begin{equation}\label{eq:bound5}
\|z^\star\|_{\ell_2}=\|Z^\star\|_F\le\|Z^\star\|_* =\tr(Z^\star)\le \sqrt{2} \|c\|_{\ell_1} \kappa(U)\,.
\end{equation}
With $L=\frac{N(N+1)}{2}\le N^2$, the length of $z^\star$, it follows that:
\begin{equation}
\|z^\star\|_{\ell_1}\le \sqrt{L}\|z^\star\|_{\ell_2}\le N\|z^\star\|_{\ell_2} \le\sqrt{2} N \kappa(U)\|c\|_{\ell_1} \,.
\end{equation}
Using \eqref{eq:bound3}, we obtain
\begin{equation}
\tr( C X)\ge \lambda^\star -(1+\sqrt{2} N \kappa(U))\|c\|_{\ell_1}\|(AA^\top)^{-1}A\|_{\ell_1\to\ell_1}\|\mathcal{A}X-b \|_{\ell_\infty}\,.
\end{equation}
(ii) From \eqref{eq:bound4}, we have:
\begin{equation}\label{eq:bound6}
\begin{array}{rl}
\tr( C X)\ge& \lambda^\star -\|\xi^{\star}\|_{\ell_2}\|\mathcal{A}X-b \|_{\ell_2}\\[5pt]
=& \lambda^\star -\|(AA^\top)^{-1}A(c-z^\star)\|_{\ell_2} \sqrt{M}\|\mathcal{A}X-b \|_{\ell_\infty}\\[5pt]
\ge& \lambda^\star -\|(AA^\top)^{-1}A\|_{\ell_2\to\ell_2}(\|c\|_{\ell_2}+\|z^\star\|_{\ell_2})\sqrt{M}\|\mathcal{A}X-b \|_{\ell_\infty}\,.
\end{array}
\end{equation}
The second inequality follows from \eqref{eq:def.lambd}.
Using \eqref{eq:bound5}, it follows that:
\begin{equation}\label{eq:bound7}
\|c\|_{\ell_2}+\|z^\star\|_{\ell_2}\le \|c\|_{\ell_1}+\|z^\star\|_{\ell_2}\le (1+\sqrt{2} \kappa(U)) \|c\|_{\ell_1}\,.
\end{equation}
To compute $\|A^\top(AA^\top)^{-1}\|_{\ell_2\to\ell_2}$, note:
\begin{equation}\label{eq:operator2.2.norm}
\begin{array}{rl}
\|(AA^\top)^{-1}A\|_{\ell_2\to\ell_2}=&\sigma_{\max}((AA^\top)^{-1}A)=\sqrt{\lambda_{\max}((AA^\top)^{-1}AA^\top(AA^\top)^{-1})}\\[5pt]
=&\sqrt{\lambda_{\max}((AA^\top)^{-1})}=\frac{1}{\sqrt{\lambda_{\min}(AA^\top)}}=\frac{1}{\sigma_{\min}(A)}\,.
\end{array}
\end{equation}
Combining this with \eqref{eq:bound6} and \eqref{eq:bound7}, we conclude:
\begin{equation}
\tr( C X)\ge \lambda^\star -\sqrt{M}(1+\sqrt{2} \kappa(U)) \|c\|_{\ell_1} \sigma_{\min}(A)^{-1}\|\mathcal{A}X-b \|_{\ell_\infty}\,.
\end{equation}
\end{proof}

\begin{remark}
Lemma~\ref{lem:obj.at.X.2} provides lower bounds on the objective function of an SDP, which depend on the condition number of a Slater point $U$, assumed to exist. In certain cases, such a point $U$ can be explicitly constructed. However, the resulting bounds tend to involve significantly worse dimensional factors. We will derive these bounds explicitly for the problems of interest later.
\end{remark}

\begin{lemma}\label{lem:upper.bound.2}
Let all assumptions and notations in Lemma \ref{lem:upper.bound} hold.
If  $\lambda^\star\ge 0$, $A$ has full rank $M$, and Slater's condition holds for the primal problem \eqref{eq:remove.trace.one} (i.e., there exists $U\succ 0$ such that $\mathcal A U=b$), then:
\begin{enumerate}
\item[(i)] $\lambda^\star-\overline \lambda_t\le \varepsilon\left[1+(1+\sqrt{2} N \kappa(U))\|\text{vec}(C)\|_{\ell_1}\|(AA^\top)^{-1}A\|_{\ell_1\to\ell_1} \right]$.
\item[(ii)] $\lambda^\star-\overline \lambda_t\le \varepsilon\left[1+(1+\sqrt{2} \kappa(U)) \|\text{vec}(C)\|_{\ell_1} \sigma_{\min}(A)^{-1}\sqrt{M} \right]$.
\end{enumerate}
\end{lemma}
The proof of Lemma \ref{lem:upper.bound.2} is similar to the proof of Lemma \ref{lem:upper.bound} and based on Lemma \ref{lem:obj.at.X.2}.

\begin{lemma}[Accuracy]\label{lem:convergence.sdp.2}
Let all assumptions and notations in Lemma \ref{lem:convergence.sdp} hold.
If $\lambda^\star\ge 0$, $A$ has full rank $M$ and Slater's condition holds for the primal problem \eqref{eq:remove.trace.one} (i.e., there exists $U\succ 0$ such that $\mathcal A U=b$), then:
\begin{enumerate}
\item[(i)] $0\le \lambda^\star-\overline \lambda_T\le \varepsilon\left[2+(1+\sqrt{2} N \kappa(U))\|\text{vec}(C)\|_{\ell_1}\|(AA^\top)^{-1}A\|_{\ell_1\to\ell_1} \right]$.
\item[(ii)]
$0\le \lambda^\star-\overline \lambda_T\le \varepsilon\left[2+(1+\sqrt{2} \kappa(U)) \|\text{vec}(C)\|_{\ell_1} \sigma_{\min}(A)^{-1}\sqrt{M} \right]$.
\end{enumerate}
\end{lemma}
The proof of Lemma \ref{lem:convergence.sdp.2} is similar to the proof of Lemma \ref{lem:convergence.sdp} and based on Lemma \ref{lem:upper.bound.2}.
\subsubsection{Unconstrained polynomial optimization}
\label{sec:uncons.pop.2}

\begin{lemma}\label{lem:properties.data.uncons.2}
Let all assumptions and notations in Lemma \ref{lem:properties.data.uncons} hold.
Define 
\begin{equation}
A=\begin{bmatrix}
\text{vec}(A_1)^\top \\
\dots\\
\text{vec}(A_M)^\top
\end{bmatrix}\,.
\end{equation}
The following properties hold:
\begin{enumerate}[(i)]
\item\label{fullrank} $A$ has rank $M$.
\item\label{lower.bound.sing.val} $\sigma_{\min}(A)\ge 1$.
\item \label{norm.AAA.T} $\|( A {A}^\top)^{-1}A\|_{\ell_1\to\ell_1}\le 1$.
\item\label{norm.C} $\|\text{vec}(C)\|_{\ell_1}=1$. 
\end{enumerate}
\end{lemma}

\begin{proof}

\eqref{fullrank} To show $A_1,\dots,A_M$ are linearly independent in $\mathbb S^N$, assume $w\in\R^M$ such that $0=w_1 A_1+\dots+w_M A_M$.
Substituting $A_i=\diag(B^{(\gamma_i)},0)$, we have $0=w_1 B^{(\gamma_1)}+\dots+w_{M} B^{(\gamma_{M})}$.
By Lemma \ref{lem:properties.B.gamma} \eqref{linearly.indepen.mat}, this implies $w_1=\dots=w_{M}=0$.
Hence, $A$ has rank $M$.

\eqref{lower.bound.sing.val} Writing $A$ as
\begin{equation}
A=\begin{bmatrix}
\text{vec}(\diag(B^{(\gamma_1)},0))^\top\\
\dots\\
\text{vec}(\diag(B^{(\gamma_{M})},0))^\top
\end{bmatrix}\,,
\end{equation}
and its transpose as
\begin{equation}
A^\top=\begin{bmatrix}
\text{vec}(\diag(B^{(\gamma_1)},0))&
\dots&
\text{vec}(\diag(B^{(\gamma_{M})},0))
\end{bmatrix}\,.
\end{equation}
For $i=1,\dots,M$, the inner product
\begin{equation}
a_i=\text{vec}(\diag(B^{(\gamma_j},0))^T\text{vec}(\diag(B^{(\gamma_j},0))=\tr(B^{(\gamma_j)2})\,.
\end{equation}
satisfies $a_i\ge 1$ by Lemma \ref{lem:properties.B.gamma} \eqref{trace.evalution}.  For $i\ne j$,
\begin{equation}
\text{vec}(\diag(B^{(\gamma_i},0))^T\text{vec}(\diag(B^{(\gamma_j},0))=\tr(B^{(\gamma_i)}B^{(\gamma_j)})=0\,.
\end{equation}
Thus, $AA^\top=\diag(a_1,\dots,a_M)$, and $\lambda_{\min}(AA^\top)\ge \min\{a_1,\dots,a_M\}\ge 1$.
It follows that $\sigma_{\min}(A)=\sqrt{\lambda_{\min}(AA^\top)}\ge 1$.

\eqref{norm.AAA.T}  From the diagonal structure of $A {A}^\top$, we have
\begin{equation}
(A {A}^\top)^{-1}=\diag(a_1^{-1},\dots,a_M^{-1})
\end{equation}
Thus,
\begin{equation}
(A {A}^\top)^{-1}A= \begin{bmatrix}
a_1^{-1}\text{vec}(B^{(\gamma_1)})^\top\\
a_2^{-1}\text{vec}(B^{(\gamma_2)})^\top\\
\dots\\ a_M^{-1}\text{vec}(B^{(\gamma_M)})^\top
\end{bmatrix}= \begin{bmatrix}
a_1^{-1}B^{(\gamma_1)}_{\alpha_1,\beta_1}&a_1^{-1}B^{(\gamma_1)}_{\alpha_2,\beta_2}&\dots& a_1^{-1}B^{(\gamma_1)}_{\alpha_L,\beta_L}\\
a_2^{-1}B^{(\gamma_2)}_{\alpha_1,\beta_1}&a_2^{-1}B^{(\gamma_2)}_{\alpha_2,\beta_2}&\dots&a_2^{-1}B^{(\gamma_2)}_{\alpha_L,\beta_L}\\
.&.&\dots&.\\
a_M^{-1}B^{(\gamma_M)}_{\alpha_1,\beta_1}&a_M^{-1}B^{(\gamma_M)}_{\alpha_2,\beta_2}&\dots&a_M^{-1}B^{(\gamma_M)}_{\alpha_L,\beta_L}
\end{bmatrix}\,.
\end{equation}
By Lemma \ref{lem:properties.B.gamma}  \eqref{unique.one.each.position}, each column $i$ of this matrix has a unique nonzero entry, $a_j^{-1}B_{\alpha_i,\beta_i}^{(\gamma_j)}$, for some $j\in \{1,\dots,M\}$, and  this entry is in the range $(0,1]$.
Therefore, we have $\|{A}^\top( A {A}^\top)^{-1}\|_{\ell_1\to\ell_1}\le 1$.

\eqref{norm.C} Since $C=\diag(B^{(0)},0)=\diag(1,0,\dots,0)$, $\|\text{vec}(C)\|_{\ell_1}=1$.

\end{proof}

\begin{theorem}\label{theo:accuarcy.approximate.val.uncons.2}
Let $f\in\R[x]_{2k}$, and suppose that Assumption \ref{ass:cond.bounded.trace.1} holds.
Assume strong duality holds for primal-dual semidefinite relaxations \eqref{eq:sdp.sos.relaxation}-\eqref{eq:sdp.mom.relaxation}.
Let $\lambda_k^{(\varepsilon)}$ be the value returned by Algorithm \ref{alg:uncon.pop}.
If Slater's condition holds for the primal relaxation \eqref{eq:sdp.sos.relaxation} (i.e., it has a feasible solution  $G_{n,k}\succ 0$ with condition number $\kappa $), then:
\begin{enumerate}
\item[(i)] $0\le \lambda_k^{(\varepsilon)}-\lambda_k\le \varepsilon\left[3+\sqrt{2}\kappa  \left(\binom{n+k}n+1\right) \right]$.
\item[(ii)]
$0\le \lambda_k^{(\varepsilon)}-\lambda_k\le \varepsilon\left[2+(1+\sqrt{2} \kappa )  \sqrt{\binom{n+2k}n-1} \ \right]$.
\end{enumerate}
\end{theorem}
\begin{proof}
Let $\lambda^\star=f_0-\lambda_k$.
Since strong duality holds for problems  \eqref{eq:sdp.sos.relaxation}-\eqref{eq:sdp.mom.relaxation}, it holds for problem \eqref{eq:remove.trace.one.1.2} and its dual \eqref{eq:dual.no.tr.cons} by Lemma \ref{lem:equivalent.sol.uncons}.

(i) Assume Slater's condition holds for the primal relaxation \eqref{eq:sdp.sos.relaxation} (i.e., it has a feasible solution  $G_{n,k}\succ 0$ with condition number $\kappa $).
Define $U_{n,k}=\diag(G_{n,k},\lambda_{\min}(G_{n,k}))$. Since $\lambda_{\min}(G_{n,k})>0$, it follows that $U_{n,k}\succ 0$ and also has the same condition number $\kappa $.
Since $\tr( A_i U_{n,k})=\tr( B^{(\gamma_i)} G_{n,k})=f_{\gamma_i}=b_i$, $U_{n,k}$ is a feasible solution for problem \eqref{eq:remove.trace.one.1.2}.
This implies that Slater's condition holds for  problem \eqref{eq:remove.trace.one.1.2}.
By Lemma \ref{lem:properties.data.uncons.2} \eqref{fullrank}, $A$ has full rank $M$.
Since $\lambda_k\le f^\star\le f(0)=f_0$, it follows that $\lambda^\star=f_0-\lambda_k\ge 0$.
By Lemma \ref{lem:properties.data.uncons}, \eqref{bound.minimum} and \eqref{bound.data.mat}, $I_N\succeq C\succeq -I_N$, $I_N\succeq A_i\succeq -I_N$, for $i=1,\dots,M$, and $\lambda^\star\in [\lambda_{\min},\lambda_{\max}]$. 
By Lemma \ref{lem:equivalent.sol.uncons}, \eqref{eq:remove.trace.one.1.2} holds.
By Lemma \ref{lem:convergence.sdp.2} (i), we have 
\begin{equation}
0\le \lambda^\star-\overline \lambda_T\le \varepsilon\left[2+(1+\sqrt{2} N \kappa )\|\text{vec}(C)\|_{\ell_1}\|(AA^\top)^{-1}A\|_{\ell_1\to\ell_1} \right]\,.
\end{equation}
Using properties of $C$ and $A$ from Lemma \ref{lem:properties.data.uncons.2} \eqref{norm.AAA.T} and \eqref{norm.C},  this simplifies to: $0\le \lambda^\star -\underline \lambda_T\le \varepsilon\left(3+\sqrt{2} N \kappa  \right)$.
Since $N=\binom{n+k}n+1$, we have
\begin{equation}
0\le \lambda^\star -\underline \lambda_T\le \varepsilon\left[3+\sqrt{2} \left(\binom{n+k}n+1\right) \kappa \right]\,.
\end{equation}
Recall that $\lambda_k^{(\varepsilon)}=f_0-\underline \lambda_T$. Substituting,
$\lambda^\star -\underline \lambda_T=\lambda_k^{(\varepsilon)}-\lambda_k$, which proves (i).

The proof of (ii) follows similarly.
\end{proof}

\subsubsection{Inequality-constrained polynomial optimization}
\label{sec:ineq.cons.add.complex}
\begin{algorithm}\label{alg:con.pop}
Inequality-constrained polynomial optimization
\begin{itemize}
\item Input: $\varepsilon>0$, $\underline \lambda\in\R$, $a\in S(g)$,\\
\phantom{1.1cm} coefficients $(f_\alpha)_{\alpha\in\N^n_{2k}}$ of $f\in\R[x]$,\\
\phantom{1.1cm} coefficients $(g_{i,\alpha})_{\alpha\in\N^n_{2d_i}}$ of $g_i\in\R[x]$, for $i=0,\dots,m$.
\item Output: $\lambda_k^{(\varepsilon)}\in\R$ and $G_\varepsilon\in\mathbb S^N$.
\end{itemize}
\begin{enumerate}
\item\label{step:convert.to.sdp} Convert relaxation to standard semidefinite program
\begin{enumerate}
\item Let $\N^n_{2k}=\{\gamma_0,\gamma_1,\dots,\gamma_{M}\}$ with $\gamma_0=0$, and set $M=|\N^n_{2k}|-1=\binom{n+2k}n -1$.
\item Define $\tilde b_i=f_{\gamma_i}$, $\tilde A_i=  B_g^{(\gamma_i)}$, for $i=0,\dots,M$, $\tilde C=\tilde A_0=  B_g^{(\gamma_0)}=  B_g^{(0)}$, $\tilde N=\sum_{i=0}^m \binom{n+k-d_i}n$.
\item Define $\check C=\diag(\tilde C,0)$, $\check A_i=\diag(\tilde A_i,0)$, $\check b_i= \tilde b_i$, for $i=1,\dots,M$, $N=\tilde N+1$.
\item\label{step:Gram-Schmidt} Apply Algorithm \ref{alg:Gram-Schmidt} (Gram-Schmidt process) to obtain $A_i\in\mathbb S^N$ and $b_i\in\R$, $i=1,\dots,M$.

\item Set $c_g=\sqrt{1+g_{1,0}^2+\dots+g_{m,0}^2}$, where $g_{i,0}$ is the constant coefficient of $g_i$.

\item Set $C=\frac{\check C}{c_g}$, $\lambda_{\min}=\frac{f_0-f(a)}{c_g}$ and $\lambda_{\max}=\frac{f_0-\underline \lambda}{c_g}$.
\end{enumerate}

\item\label{step:HU.GS} Run Algorithm \ref{alg:Binary.search.HU} (Binary search using Hamiltonian updates) 
to get $\underline \lambda_T\in\R$ and $X_\varepsilon=((X_\varepsilon)_{ij})_{i,j=1,\dots,N}\in\mathbb S^N$.

\item Extract approximate results: Set $\lambda_k^{(\varepsilon)}=f_0-\underline\lambda_T c_g$ and $G_\varepsilon=((X_\varepsilon)_{ij})_{i,j=1,\dots,\tilde N}$.
\end{enumerate}
\end{algorithm}

\begin{lemma}\label{lem:equi.sdp.relax.cons}
Let $f\in\R[x]_{2d}$ and $g_i\in\R[x]_{2d_i}$, for $i=0,\dots,m$, with $g_0=1$.
Suppose Assumption \ref{ass:cond.bounded.trace.2} holds.
Let $\check C$, $\check A_i$, for $i=1,\dots,M$, $N$ be generated by Algorithm \ref{alg:con.pop}.
Then problem \eqref{eq:sdp.relaxation.cons.pop} is equivalent to
\begin{equation}\label{eq:sdp.trace.one.cons.pop.2}
\begin{array}{rl}
f_0-\lambda_k=\inf\limits_{X} & \tr( \check C X)\\
\text{s.t.}& X\in\mathbb S_+^{N}\,,\,\tr(X)= 1\,,\\
&\tr( \check A_i X)=\check b_i\,,\,i=1,\dots,M\,.
\end{array}
\end{equation}
Moreover, this is equivalent to
\begin{equation}\label{eq:sdp.trace.one.cons.pop.2.1}
\begin{array}{rl}
f_0-\lambda_k=\inf\limits_{X} & \tr( \check C X)\\
\text{s.t.}& X\in\mathbb S_+^{N}\,,\\
&\tr( \check A_i X)=\check b_i\,,\,i=1,\dots,M\,.
\end{array}
\end{equation}
\end{lemma}
The proof of Lemma \ref{lem:equi.sdp.relax.cons} follows the same reasoning as the proof of Lemma \ref{lem:equi.sdp.relax.cons.exact}.

\begin{lemma}\label{lem:properties.check.C}
Let $f\in\R[x]_{2d}$ and $g_i\in\R[x]_{2d_i}$, for $i=0,\dots,m$, with $g_0=1$.
Let $\check C$, $\check A_i$,for  $i=1,\dots,M$, $N$, and $c_g$ be generated by Algorithm \ref{alg:con.pop}.
Then the following properties hold:
\begin{enumerate}[(i)]
\item\label{linear.inde.mat.data} $\check A_1,\dots, \check A_M$ are linearly independent in $\mathbb S^N$.
\item\label{frobenious.norm.C} $\|\check C\|_F=c_g$.
\item\label{sparse.C} The maximum number of nonzero entries in any row of $\check C$ is $1$.
\end{enumerate}
\end{lemma}
\begin{proof}
\eqref{linear.inde.mat.data} 
 Suppose $u\in\R^M$ satisfies $u_1 \check A_1+\dots+u_{M} \check A_M=0$.
Since $\check A_i=\diag(\tilde A_i,0)$, this implies $u_1 \tilde A_1+\dots+u_{M} \tilde A_{M}=0$.
From the construction of  $\tilde A_i=  B_g^{(\gamma_i)}$, this reduces to $u_1   B_g^{(\gamma_1)}+\dots+u_{M}   B_g^{(\gamma_{M})}=0$.
Since $g_0=1$, the first block of $  B_g^{(\gamma_j)}$ is $B_0^{(\gamma_i)}=B^{(\gamma_i)}$, as defined in \eqref{eq:def.B.gamma}.
Thus, $u_1 B^{(\gamma_1)}+\dots+u_{M} B^{(\gamma_{M})}=0$.
By Lemma \ref{lem:properties.B.gamma} \eqref{linearly.indepen.mat}, $u_1=\dots=u_{M}=0$.
This establishes linear independence.

\eqref{frobenious.norm.C} Recall that $\check C=\diag(\tilde C,0)$, where
\begin{equation}\label{eq:def.tilde.C}
\tilde C=  B_g^{(0)}=\diag(g_{0,0},0,\dots,0,g_{1,0},0,\dots,0,g_{m,0},0,\dots,0)\,,
\end{equation}
with $g_{i,0}$ denoting the constant coefficient of $g_i$. 
Since $g_{0,0}=1$, it follows that $\|\check C\|_F=c_g$.

\eqref{sparse.C} Additionally, the diagonal structure of $\tilde C$ implies that at most one nonzero entry exists per row of $\check C=\diag(\tilde C,0)$.
\end{proof}

\begin{remark}
By Lemma \ref{lem:properties.check.C} \eqref{linear.inde.mat.data}, the matrix 
\begin{equation}
\check A =\begin{bmatrix}
\text{vec}(\check A_1)\\
\dots\\
\text{vec}(\check A_M)
\end{bmatrix}
\end{equation}
has rank $M$.
However, unlike the unconstrained case, estimating a positive lower bound for the smallest singular value of $\check A$ is challenging. This bound is crucial for analyzing the accuracy of the algorithm's output.
The difficulty arises because the entries of $\check A_i$ depend on the coefficients of $g_i$, which can take arbitrary values.

To address this, we apply the Gram-Schmidt process to transform $\check A_1,\dots,\check A_M$ into an orthogonal basis $A_1,\dots,A_M$  (see Step 1 (d) of Algorithm \ref{alg:con.pop}).
\end{remark}
\begin{remark}
By Lemma \ref{lem:properties.data.A.cons.exact} \eqref{sparsity.A}, each row of $\check A_i=\diag(  B_g^{(\gamma_i)},0)$ has at most $\max\limits_{i=0,\dots,m} |\supp(g_i)|\le \max\limits_{i=0,\dots,m}\binom{n+2d_i}n$ nonzero entries.

However, this sparsity is generally lost for $A_j$, as the Gram-Schmidt process transforms $\check A_i$ into a new orthogonal basis $A_j$, making each $A_j$ a linear combination of the $\check A_i$ matrices.
\end{remark}

\begin{lemma}\label{lem:properties.data.A.cons}
Let $f\in\R[x]_{2d}$ and $g_i\in\R[x]_{2d_i}$, for $i=0,\dots,m$, with $g_0=1$.
Suppose that Assumption \ref{ass:cond.bounded.trace.2} holds.
Let $C$, $A_i$, $i=1,\dots,M$, $N$, $c_g$, $\lambda_{\min}$, and $\lambda_{\max}$ be generated by Algorithm \ref{alg:con.pop}.
Define 
\begin{equation}
A =\begin{bmatrix}
\text{vec}(A_1)\\
\dots\\
\text{vec}(A_M)
\end{bmatrix}\,.
\end{equation}
Set $\lambda^\star=\frac{f_0-\lambda_k}{c_g}$.
Then the following hold:
\begin{enumerate}[(i)]
\item\label{bound.eig} $-I_N\preceq C\preceq I_N$ and $\|\text{vec}(C)\|_{\ell_1}=\frac{|g_{0,0}|+\dots+|g_{m,0}|}{c_g}$.
\item\label{sparsi.A} The number of nonzero entries on each row of $A_i$, $i=1,\dots,M$, is no larger than $\binom{n+k}n$. 
\item\label{sparsi.C} The maximum number of nonzero entries on a row of $C$ is $1$.
\item\label{full.rank.A} $A$ has rank $M$, $\sigma_{\min}(A)=1$, and $-I_N\preceq A_i\preceq I_N$, for $i=1,\dots,M$
\item\label{bound.opt.val.sdp} $\lambda^\star\in[\lambda_{\min},\lambda_{\max}]$.
\item\label{equiva.prob} Problem \eqref{eq:sdp.relaxation.cons.pop}  is equivalent to \eqref{eq:sdp}.
\item\label{no.tr.constr} 
\begin{equation}\label{eq:sdp.trace.one.cons.pop.with.GS}
\lambda^\star=\inf\{\tr( C X)\,:\,X\in\mathbb S^{N}_+\,,\,\tr( A_i X) =b_i\,,\,i=1,\dots,M\}\,.
\end{equation}
\end{enumerate}
\end{lemma}
\begin{proof}
\eqref{full.rank.A} By Lemma \ref{lem:properties.check.C} \eqref{linear.inde.mat.data}, the matrices $\check A_1,\dots,\check A_M$ are linearly independent in $\mathbb S^N$.
By Lemma \ref{lem:Gram-Schmidt}, we have $\tr( A_i A_j)=\delta_{ij}$, which implies $-I_N\preceq A_i\preceq I_N$, $A$ has rank $M$, and $\sigma_{\min}(A)=1$.

\eqref{sparsi.A} Since each $A_i$ is a linear combination of $\check A_j$'s, it retains the block-diagonal structure of $\check A_j$.
Recall that $\check A_j=\diag(  B_g^{(\gamma_j)},0)$.
By the block-diagonal structure of $  B_g^{(\gamma)}$ in \eqref{eq.tilde.B.gamma}, for $i=0,\dots,m$, the $i$-th block of $  B_g^{(\gamma)}$ has size $\binom{n+k-d_i}n$.
Thus, each row of $A_i$ contains no more than $\binom{n+k}n$ nonzero entries.

\eqref{bound.eig} By Lemma \ref{lem:properties.check.C} \eqref{frobenious.norm.C}, we have $C=\frac{\check C}{c_g}=\frac{\check C}{\|\check C\|_F}$.
Thus,  $-I_N\preceq C\preceq I$ and by \eqref{eq:def.tilde.C}, 
\begin{equation}
\|\text{vec}(C)\|_{\ell_1}=\frac{\|\text{vec}(\check C)\|_{\ell_1}}{c_g}=\frac{|g_{0,0}|+\dots+|g_{m,0}|}{c_g}=\frac{|g_{0,0}|+\dots+|g_{m,0}|}{c_g}\,.
\end{equation}

\eqref{sparsi.C} follows from Lemma \ref{lem:properties.check.C} \eqref{sparse.C} and $C=\check C/c_g$.

\eqref{equiva.prob} By Lemma \ref{lem:equi.sdp.relax.cons}, problem \eqref{eq:sdp.relaxation.cons.pop} is equivalent to \eqref{eq:sdp.trace.one.cons.pop.2}.
Additionally, by Lemma \ref{lem:Gram-Schmidt}, we have
 \begin{equation}\label{eq:same.polytope}
\tr( \check A_i X)=\check b_i\,,\,i=1,\dots,M \Leftrightarrow \tr(  A_i X)= b_i\,,\,i=1,\dots,M
\end{equation}
Therefore, problem \eqref{eq:sdp.trace.one.cons.pop.2} is equivalent to
\begin{equation}\label{eq:sdp.trace.one.cons.pop.3}
\begin{array}{rl}
f_0-\lambda_k=\inf\limits_{X} & \tr( \check C X)\\
\text{s.t.}& X\in\mathbb S_+^{N}\,,\,\tr(X)= 1\,,\\
&\tr( A_iX)=b_i\,,\,i=1,\dots,M\,.
\end{array}
\end{equation}
Since $C=\check C/c_g$, this problem can be written as
\begin{equation}\label{eq:sdp.trace.one.cons.pop.4}
\begin{array}{rl}
\frac{f_0-\lambda_k}{c_g}=\inf\limits_{X} & \tr(C X)\\
\text{s.t.}& X\in\mathbb S_+^{N}\,,\,\tr(X)= 1\,,\\
&\tr( A_i X)=b_i\,,\,i=1,\dots,M\,.
\end{array}
\end{equation}
This matches the form of problem \eqref{eq:sdp}, where $\lambda^\star=\frac{f_0-\lambda_k}{c_g}$.

\eqref{bound.opt.val.sdp} Since $\lambda_{\min}=\frac{f_0-f(a)}{c_g}$ for $a\in S(g)$, and $\lambda_{\max}=\frac{f_0-\underline \lambda}{c_g}$ with $\underline \lambda\le \lambda_k\le f^\star\le f(a)$, we conclude that $\lambda^\star\in[\lambda_{\min},\lambda_{\max}]$.

\eqref{no.tr.constr} follows from the last statement of Lemma \ref{lem:equi.sdp.relax.cons}.
\end{proof}

\begin{theorem}\label{theo:accuray.with.gram.schimidt.ineqpop}
Let $f\in\R[x]_{2d}$ and $g_i\in\R[x]_{2d_i}$, for $i=0,\dots,m$, with $g_0=1$.
Suppose $0\in S(g)$, Assumption \ref{ass:cond.bounded.trace.2} holds, and Slater's condition holds for the  sum-of-squares relaxation \eqref{eq:sdp.relaxation.cons.pop}, i.e., there exists a feasible solution  $G_{n,k}\succ 0$ for problem \eqref{eq:sdp.relaxation.cons.pop} with condition number $\kappa $.
Let $\lambda_k^{(\varepsilon)}$ be the value returned by Algorithm \ref{alg:con.pop}.
Then, the following bound holds:
\begin{equation}
0\le \lambda_k^{(\varepsilon)}-\lambda_k\le \varepsilon  \left[2+(1+\sqrt{2} \kappa ) \sqrt{\binom{n+2k}n}\ \right]\sum_{i=0}^m|g_{i,0}|\,.
\end{equation}
\end{theorem}
\begin{proof}
Let $\lambda^\star=\frac{f_0-\lambda_k}{c_g}$. Since Slater's condition holds for the  sum-of-squares relaxation \eqref{eq:sdp.relaxation.cons.pop}, strong duality is satisfied for problems  \eqref{eq:sdp.relaxation.cons.pop}-\eqref{eq:sdp.relaxation.cons.pop.dual}.
By Lemma \ref{lem:properties.data.A.cons} \eqref{no.tr.constr}, strong duality also holds for problem \eqref{eq:sdp.trace.one.cons.pop.with.GS} and its dual \eqref{eq:dual.no.tr.cons} .

Define $U_{n,k}=\diag(G_{n,k},\lambda_{\min}(G_{n,k}))$. Since $\lambda_{\min}(G_{n,k})>0$, it follows that $U_{n,k}\succeq 0$ and has the same condition number $\kappa $.
Moreover, $\tr( \check A_i U_{n,k})=\tr( B^{(\gamma_i)} G_{n,k})=f_{\gamma_i}=\check b_i$ and by \eqref{eq:same.polytope},  $\tr( A_i U_{n,k})=b_i$.
This implies $U_{n,k}$ is a feasible solution for problem \eqref{eq:sdp.trace.one.cons.pop.with.GS}, and thus Slater's condition holds for  problem \eqref{eq:sdp.trace.one.cons.pop.with.GS}.
From Lemma \ref{lem:properties.data.A.cons} \eqref{full.rank.A}, $A$ has full rank $M$.
Since $0\in S(g)$, we have $\lambda_k\le f^\star\le f(0)=f_0$, which implies $\lambda^\star=\frac{f_0-\lambda_k}{c_g}\ge 0$.
By Lemma \ref{lem:properties.data.A.cons} \eqref{bound.eig}, \eqref{full.rank.A} and \eqref{bound.opt.val.sdp}, $I_N\succeq C\succeq -I_N$, $I_N\succeq A_i\succeq -I_N$, for $i=1,\dots,M$, and  $\lambda^\star\in [\lambda_{\min},\lambda_{\max}]$.
Using Lemma \ref{lem:convergence.sdp.2} (ii), we obtain
\begin{equation}
0\le \lambda^\star-\overline \lambda_T\le \varepsilon\left[2+(1+\sqrt{2} \kappa ) \|\text{vec}(C)\|_{\ell_1} \sigma_{\min}(A)^{-1}\sqrt{M} \right]\,.
\end{equation}
Applying properties of $C$ and $A$ from Lemma \ref{lem:properties.data.A.cons} \eqref{bound.eig}  and  \eqref{full.rank.A},  this simplifies to
\begin{equation}
0\le \lambda^\star-\overline \lambda_T\le \varepsilon\left[2+(1+\sqrt{2} \kappa )\frac{b_g}{c_g}  \sqrt{M}  \right]\,,
\end{equation}
where $b_g=|g_{0,0}|+\dots+|g_{m,0}|$.
Since $\lambda_k^{(\varepsilon)}=f_0-\underline\lambda_T c_g$ and $\lambda_k=f_0-\lambda^\star c_g$, we have 
\begin{equation}
\begin{array}{rl}
0\le \lambda_k^{(\varepsilon)}-\lambda_k=c_g(\lambda^\star-\underline \lambda_T)\le &\varepsilon\left[2c_g+(1+\sqrt{2} \kappa )b_g  \sqrt{M}  \right]\\
\le& \varepsilon b_g \left[2+(1+\sqrt{2} \kappa ) \sqrt{M} \right]\,.
\end{array}
\end{equation}
The last inequality is due to $c_g\le b_g$.
Finally, substituting $M= \binom{n+2k}n-1$ completes the proof.
\end{proof}

\begin{lemma}[Complexity]
Let $k\ge \max\{d,d_0,\dots,d_m\}$.
To run Algorithm \ref{alg:con.pop}, the classical computer requires
\begin{equation}
O\left(\sum_{i=0}^m \binom{n+k-d_i}n^2+\left[\binom{n+k}n\binom{n+2k}n\sum_{i=0}^m \binom{n+k-d_i}n+\left(\sum_{i=0}^m \binom{n+k-d_i}n\right)^\omega\right] \frac{1}{\varepsilon^2}\right)
\end{equation}
operations, where $\omega\approx2.373$, while the quantum computer requires
\begin{equation}
\footnotesize
O\left(\sum_{i=0}^m \binom{n+k-d_i}n^2+\binom{n+k}{n}\left[\binom{n+2k}{n}^{1/2}+\left(\sum_{i=0}^m \binom{n+k-d_i}n\right)^{1/2}\frac{1}{\varepsilon}\right]\frac{1}{\varepsilon^4}\right)
\end{equation}
operations.
\end{lemma}
\begin{proof}
The most computationally expensive parts of Algorithm \ref{alg:con.pop} are the Gram-Schmidt process in Step \ref{step:Gram-Schmidt} and the binary search using Hamiltonian updates in Step \ref{step:HU.GS}.

For the Gram-Schmidt process, by Lemma \ref{lem:complexity.Gram-Schmidt}, both classical and quantum computers have the same complexity of $O(LM^2)$, where
\begin{equation}
L=\frac{1}{2}\sum_{i=0}^m \binom{n+k-d_i}n\left[\binom{n+k-d_i}n+1\right]+1\,.
\end{equation}

For the binary search using Hamiltonian updates, by Lemma \ref{lem:complex.binary.search}, the classical computer requires $O((MsN+N^\omega)\varepsilon^{-2})$ operations, and the quantum computer requires $O(s(\sqrt{M}+\sqrt{N}\varepsilon^{-1})\varepsilon^{-4})$ operations,
where $s$ is the maximum number of nonzero entries in a row of the data matrices $C$ and $A_i$, for $i=1,\dots,M$.
By Lemma \ref{lem:properties.data.A.cons} \eqref{sparsi.A} and \eqref{sparsi.C}, we have $s=\binom{n+k}n$.

Here, $M=\binom{n+2k}n-1$ and $N=\sum_{i=0}^m \binom{n+k-d_i}n+1$.
Hence the result follows.
\end{proof}

\subsection{Quantum algorithms for solving moment relaxations}
\label{app:moment.relax}

\subsubsection{Equality-constrained polynomial optimization over the unit sphere}
\label{sec:mom.relax.eq.cons.sphere}
Let $f,h_j\in\R[x]$, for $j=1,\dots,l$, with $h_1=1-\|x\|_{\ell_2}^2$.
Consider polynomial optimization problem
\begin{equation}\label{eq:pop.eq.constrains}
f^\star=\min\limits_{x\in V(h)} f(x)\,,
\end{equation}
where $V(h)$ is the algebraic variety defined as:
\begin{equation}\label{eq:def.variety}
V(h)=\{x\in\R^n\,:\,h_j(x)=0\,,\,j=1,\dots,l\}\,.
\end{equation}
Define $d=\lceil \deg(f)/2 \rceil$ and $w_j=\lceil \deg(h_j)/2 \rceil$, $j=1,\dots,l$, noting that $w_1=1$.

The moment relaxation of order $k$ for for this problem is:
\begin{equation}\label{eq:mom.relax.pop.eq.constrains}
\begin{array}{rl}
\tau_k=\min\limits_y& L_y(f)\\
\text{s.t.}& y=(y_\alpha)_{\alpha\in\N^n_{2k}}\subset \R\,,\\
&M_k(y)\succeq 0\,,\,y_0=1\,,\\
&M_{k-w_j}(h_jy)=0\,,\,j=1,\dots,l\,,
\end{array}
\end{equation}
where $L_y$ is the Riez linear functional, $M_k(y)$ denotes the moment matrix, and $M_{k-w_j}(h_jy)$ represents the localizing matrix.
 
The relaxation sequence satisfies:
$\tau_k\le \tau_{k+1}\le f^\star$ and $\tau_k\to f^\star$ as $k\to \infty$.

Assume that problem \eqref{eq:mom.relax.pop.eq.constrains} has an optimal solution and that a lower bound $\underline \tau\le \tau_k$ is known.

\begin{algorithm}\label{alg:eq.con.pop.sphere}
Equality constrained polynomial optimization over the unit sphere
\begin{itemize}
\item Input: $\varepsilon>0$, , $\underline \tau\in\R$, $a\in V(h)$,\\
\phantom{1.1cm} coefficients $(f_\alpha)_{\alpha\in\N^n_{2k}}$ of $f\in\R[x]$,\\
\phantom{1.1cm} coefficients $(h_{j,\alpha})_{\alpha\in\N^n_{2w_j}}$ of $h_j\in\R[x]$, for $j=1,\dots,l$.
\item Output: $\tau_k^{(\varepsilon)}\in\R$ and $Z_\varepsilon\in\mathbb S^N$.
\end{itemize}
\begin{enumerate}
\item Convert relaxation to standard semidefinite program: 
\begin{enumerate}
\item Set $u_\alpha=\binom{k}{(k-|\alpha|,\alpha)}$, $\alpha\in \N^n_k$.
\item For $\gamma\in\N^n_{2k}$, let $\alpha_\gamma,\beta_\gamma\in\N^n_k$ be such that $\alpha_\gamma\le \beta_\gamma$ and $\alpha_\gamma+\beta_\gamma=\gamma$.
\item For $\gamma\in\N^n_{2k}$, for $\alpha,\beta\in\N^n_{k}$ satisfying $\alpha\le \beta$, $(\alpha,\beta)\ne (\alpha_\gamma,\beta_\gamma)$, and $\alpha+\beta=\gamma$, let $\tilde A^{(\alpha,\beta,\gamma)}=(\tilde A^{(\alpha,\beta,\gamma)}_{\xi,\zeta})_{\xi,\zeta\in\N^n_k}$ be the matrix defined by 
\begin{equation}\label{eq:def.mat.A.alp.beta.gam}
\tilde A^{(\alpha,\beta,\gamma)}_{\xi,\zeta}=
\begin{cases}
\frac{1}{u_{\xi}^{1/2}u_{\zeta}^{1/2}}&\text{if }(\xi,\zeta)=(\alpha_\gamma,\beta_\gamma)\,,\, \alpha_\gamma=\beta_\gamma\,,\\
\frac{1}{2u_{\xi}^{1/2}u_{\zeta}^{1/2}}&\text{if }(\xi,\zeta)\in\{(\alpha_\gamma,\beta_\gamma),(\alpha_\gamma,\beta_\gamma)\}\,,\,\alpha_\gamma\ne \beta_\gamma \,,\\
-\frac{1}{u_{\xi}^{1/2}u_{\zeta}^{1/2}}&\text{if }(\xi,\zeta)=(\alpha,\beta)\,,\, \alpha=\beta\,,\\
-\frac{1}{2u_{\xi}^{1/2}u_{\zeta}^{1/2}}&\text{if }(\xi,\zeta)\in\{(\alpha,\beta),(\alpha,\beta)\}\,,\,\alpha\ne \beta \,,\\
0&\text{else}.
\end{cases}
\end{equation}
\item For $j=1,\dots,l$, for $\alpha\in\N^n_{2(k-w_j)}$, let $\tilde A^{(j,\alpha)}=(\tilde A^{(j,\alpha)}_{\xi,\zeta})_{\xi,\zeta\in\N^n_k}$ be the matrix defined by 
\begin{equation}\label{eq:mat.A.j.alp}
\tilde A^{(j,\alpha)}_{\xi,\zeta}=
\begin{cases}
\frac{  h_{j,\gamma-\alpha}}{2u_{\xi}^{1/2}u_{\zeta}^{1/2}}&\text{if }(\xi,\zeta)\in \left\{(\alpha_\gamma,\beta_\gamma),(\beta_\gamma,\alpha_\gamma)\,:\,\alpha_\gamma\ne \beta_\gamma\,,\right.\\
&\qquad\qquad\qquad\qquad\qquad\qquad\quad   \left.\gamma\in\alpha+\N^n_{2w_j}\right\}\,,\\
\frac{ h_{j,\gamma-\alpha}}{u_{\xi}^{1/2}u_{\zeta}^{1/2}}&\text{if }(\xi,\zeta)\in \{(\alpha_\gamma,\beta_\gamma),\,:\,\alpha_\gamma= \beta_\gamma\,,\,\gamma\in\alpha+\N^n_{2w_j}\}\,,\\
0&\text{else.}
\end{cases}
\end{equation}
\item Let $\tilde A^{(0)}=(\tilde A^{(0)}_{\alpha,\beta})_{\alpha,\beta\in\N^n_k}$ be the matrix defined by 
\begin{equation}\label{eq:mat.A.0}
\tilde A^{(0)}_{\alpha,\beta}=
\begin{cases}
2^k&\text{if }\alpha=\beta=0\,,\\
0&\text{else}.
\end{cases}
\end{equation}
\item Let $\tilde C=(\tilde C_{\alpha,\beta})_{\alpha,\beta\in\N^n_{k}}$ be the matrix defined by 
\begin{equation}\label{eq:def.C}
\tilde C_{\alpha,\beta}=\begin{cases}
 \frac{f_\gamma 2^k}{u_{\alpha}^{1/2}u_{\beta}^{1/2}}&\text{if }(\alpha,\beta)=(\alpha_\gamma,\beta_\gamma)\,,\,\alpha_\gamma=\beta_\gamma\,,\\
  \frac{f_\gamma 2^{k-1}}{u_{\alpha}^{1/2}u_{\beta}^{1/2}}&\text{if }(\alpha,\beta)\in \{(\alpha_\gamma,\beta_\gamma),(\beta_\gamma,\alpha_\gamma)\}\,,\,\alpha_\gamma\ne\beta_\gamma\,,\\
  0&\text{else}.
\end{cases}
\end{equation}
\item Set $N=\binom{n+k}n$.
\item Let $\tilde A_1,\dots,\tilde A_{\tilde M}$ be the following matrices:
\begin{enumerate}
\item $\tilde A^{(\alpha,\beta,\gamma)}$, for $\alpha,\beta\in\N^n_{k}$ satisfying $\alpha\le \beta$, $(\alpha,\beta)\ne (\alpha_\gamma,\beta_\gamma)$, and $\alpha+\beta=\gamma$, for $\gamma\in\N^n_{2k}$;
\item $\tilde A^{(j,\alpha)}$, for $\alpha\in\N^n_{2(k-w_j)}$, for $j=1,\dots,l$;
\item $\tilde A^{(0)}$.
\end{enumerate}
In particular, $\tilde A_{\tilde M}=\tilde A^{(0)}$.
\item Let $\tilde b_1,\dots,\tilde b_{\tilde M}$ be  defined by $\tilde b_i=\begin{cases}
1&\text{if } i=\tilde M\,,\\
0&\text{else}.
\end{cases}$
\item\label{step:gauss.elimi.mom.relax} Run Algorithm \ref{alg:Gaussian.elimination} (Gaussian elimination) to get $\check A_i\in \mathbb S^N$, $\check b_i\in\R$, $i=1,\dots,M$.
\item\label{step:gram-schmidt.mom.relax} Run Algorithm \ref{alg:Gram-Schmidt} (Gram-Schmidt process) to get $A_i\in\mathbb S^N$ and $b_i\in\R$, $i=1,\dots,M$.

\item Set $c=\|\tilde C\|_F$, $C=\frac{\tilde C}{c}$, $\lambda_{\max}=\frac{f(a)}{c}$ and $\lambda_{\min}=\frac{\underline \tau}{c}$.
\end{enumerate}

\item\label{step:Hamiltonian.update.mom.relax} Run Algorithm \ref{alg:Binary.search.HU} (Binary search using Hamiltonian updates) 
to obtain $\underline \lambda_T\in\R$ and $X_\varepsilon\in\mathbb S^N$.

\item Extract approximate results: Set $\tau_k^{(\varepsilon)}=\underline\lambda_T c$ and $Z_\varepsilon=X_\varepsilon$.
\end{enumerate}
\end{algorithm}

\begin{lemma}\label{lem:trce.1.moment.relax}
Let $f\in\R[x]_{2k}$ and $h_j\in\R[x]_{2w_j}$ for $j=1,\dots,l$, with $h_1=1-\|x\|_{\ell_2}^2$.
Let $u_\alpha$ for $\alpha\in \N^n_k$ be the coefficients  generated by Algorithm \ref{alg:eq.con.pop.sphere}.
Define
\begin{equation}\label{eq:change.coordinate}
U=\diag\left(\left(\frac{u_\alpha^{1/2}}{2^{k/2}}\right)_{\alpha\in\N^n_k}\right)\,.
\end{equation}
Then, for any feasible solution $y$ to problem \eqref{eq:mom.relax.pop.eq.constrains}, $\tr(UM_k(y)U)=1$.
\end{lemma}
\begin{proof}
Using the identity $a^k-b^k=(a-b)\sum_{j=0}^k a^{k-1-j}b^j$ with $a=2$ and $b=1+\|x\|_{\ell_2}^2$, we derive
\begin{equation}\label{eq:equality.pol}
2^k=(1+\|x\|_{\ell_2}^2)^k+(1-\|x\|_{\ell_2}^2)p\,,
\end{equation}
where $p=\sum_{j=0}^{k-1} 2^{k-1-j}(1+\|x\|_{\ell_2}^2)^j\in\R[x]_{2(k-1)}$.
For $\alpha=(\alpha_1,\dots,\alpha_n)$ and $\bar\alpha=(\alpha_0,\alpha)$, we expand $(1+\|x\|_{\ell_2}^2)^k$  as follows:
\begin{equation}
\begin{array}{rl}
(1+\|x\|_{\ell_2}^2)^k=&(1+x_1^2+\dots+x_n^2)^k\\[10pt]
=& \sum\limits_{\arraycolsep=1.4pt\def\arraystretch{.7}
\begin{array}{cc}
\scriptstyle \bar\alpha\in\N^{n+1}\\
\scriptstyle |\bar\alpha|=k
\end{array}
} \binom{k}{\bar \alpha} 1^{\alpha_0} x_1^{2\alpha_1}\dots x_n^{2\alpha_n}=\sum\limits_{\alpha\in\N^n_k} \binom{k}{(k-|\alpha|,\alpha)} x^{2\alpha} \,.
\end{array}
\end{equation}
With $u_\alpha=\binom{k}{(k-|\alpha|,\alpha)}$, , it follows that $(1+\|x\|_{\ell_2}^2)^k=\sum_{\alpha\in\N^n_k}u_\alpha x^{2\alpha}$.
Substituting this back into the equality \eqref{eq:equality.pol}, we obtain:
\begin{equation}\label{eq:rep.eq.pol}
2^k=\sum_{\alpha\in\N^n_k}u_\alpha x^{2\alpha}+h_1p
\end{equation}
Now let $y$ be a feasible solution to problem \eqref{eq:mom.relax.pop.eq.constrains}.
Applying the linear operator $L_y$ for both sides, we get:
\begin{equation}
L_y(2^k)=L_y\left(\sum_{\alpha\in\N^n_k}u_\alpha x^{2\alpha}\right)+L_y(h_1p)\,.
\end{equation}
Since $L_y(2^k)=2^ky_0=2^k$, and $L_y(\sum_{\alpha\in\N^n_k}u_\alpha x^{2\alpha})=\sum_{\alpha\in\N^n_k}u_\alpha y_{2\alpha}$, it remains to evaluate $L_y(h_1p)$.
Note that $L_y(h_1p)=0$ because $M_{k-w_j}(h_jy)=0$ and $\deg(p)=2(k-u_1)$.
Thus, $2^k=\sum_{\alpha\in\N^n_k}u_\alpha y_{2\alpha}$.
Finally, by the definition of $U$, 
\begin{equation}
\tr(UM_k(y)U)=\sum_{\alpha\in\N^n_k}\left(\frac{u_\alpha^{1/2}}{2^{k/2}}\right)^2y_{2\alpha}=\frac{1}{2^k}\sum_{\alpha\in\N^n_k}u_\alpha y_{2\alpha}\,.
\end{equation}
Since $2^k=\sum_{\alpha\in\N^n_k}u_\alpha y_{2\alpha}$, it follows that $\tr(UM_k(y)U)=1$, completing the proof.
\end{proof}

\begin{lemma}\label{lem:properties.tilde.prob.eq.pop}
Let $f\in\R[x]_{2k}$ and $h_j\in\R[x]_{2w_j}$ for $j=1,\dots,l$, with $h_1=1-\|x\|_{\ell_2}^2$.
Let $\tilde C$, $\tilde A_i$, $\tilde b_i$, $i=1,\dots,\tilde M$, $N$ be generated by Algorithm \ref{alg:eq.con.pop.sphere}.
Then:
\begin{enumerate}[(i)]
\item\label{equival.prob} Problem \eqref{eq:mom.relax.pop.eq.constrains} is equivalent to the SDP
\begin{equation}\label{eq:sdp.moment.relax.tr.1}
\begin{array}{rl}
\tau_k=\inf\limits_Z&\tr( \tilde CZ)\\
\text{s.t.}&Z\in\mathbb S^N_+\,,\,\tr(Z)=1\,,\\
&\tr( \tilde A_i Z) =\tilde b_i\,,\,i=1,\dots,\tilde M\,.	
\end{array}
\end{equation}
\item\label{remove.trace.cons} The constraint $tr(Z)=1$ in the above SDP can be omitted without changing the optimal value:
\begin{equation}\label{eq:sdp.moment.relax.tr.1.2}
\begin{array}{rl}
\tau_k=\inf\limits_Z&\tr( \tilde CZ)\\
\text{s.t.}&Z\in\mathbb S^N_+\,,\\
&\tr( \tilde A_i Z) =\tilde b_i\,,\,i=1,\dots,\tilde M\,.	
\end{array}
\end{equation}
Moreover, if $y$ is a feasible solution for problem \eqref{eq:mom.relax.pop.eq.constrains}, then $UM_k(y)U$ is a feasible solution for \eqref{eq:sdp.moment.relax.tr.1.2}, where $U$ is defined as in \eqref{eq:change.coordinate}.
\item\label{norm.mat.C}  $\|\tilde C\|_F\le 2^k \|f\|_{\ell_2}$ and $\|\text{vec}(\tilde C)\|_{\ell_1}\le 2^k \|f\|_{\ell_1}$.
\item\label{bound.num.aff.cons} If $k\ge \max\{w_1,\dots,w_l\}$, then the total number of constraints $\tilde M$ is 
\begin{equation}\label{eq:num.tilde.M}
\tilde M=1+\sum_{j=1}^l \binom{n+2(k-w_j)}n+\frac{1}{2}\binom{n+k}n\left[\binom{n+k}n+1\right]-\binom{n+2k}n\,,
\end{equation}
and it satisfies: 
\begin{equation}
\tilde M\le 1+ (l-1) \binom{n+2k}n+\frac{1}{2}\binom{n+k}n\left[\binom{n+k}n+1\right]\,.
\end{equation}
\end{enumerate}
\end{lemma}
\begin{proof}
\eqref{equival.prob} Let $y$ be a feasible solution to problem \eqref{eq:mom.relax.pop.eq.constrains}.
Define $Z=UM_k(y)U$, where $U$ is constructed as in Lemma \ref{lem:trce.1.moment.relax}. For $Z=(z_{\alpha,\beta})_{\alpha,\beta\in\N^n_k}$, the trace condition from Lemma \ref{lem:trce.1.moment.relax} ensures $\tr(Z)=1$.
For $\alpha,\beta\in\N^n_k$, using $z_{\alpha,\beta}=\frac{1}{2^k}u_\alpha^{1/2}u_\beta^{1/2} y_{\alpha+\beta}$, we have $y_{\alpha+\beta}=2^k \frac{z_{\alpha,\beta}}{u_\alpha^{1/2}u_\beta^{1/2}}$.
Recall that for $\gamma\in\N^n_{2k}$, $\alpha_\gamma,\beta_\gamma\in\N^n_k$ satisfy $\alpha_\gamma\le \beta_\gamma$ and $\alpha_\gamma+\beta_\gamma=\gamma$.
Thus, $y_\gamma=2^k \frac{z_{\alpha_\gamma,\beta_\gamma}}{u_{\alpha_\gamma}^{1/2}u_{\beta_\gamma}^{1/2}}$.
Additionaly, for all $\alpha,\beta\in\N^n_{k}$ satisfying $\alpha\le \beta$, $(\alpha,\beta)\ne (\alpha_\gamma,\beta_\gamma)$, and $\alpha+\beta=\gamma$, 
\begin{equation}
2^k\frac{z_{\alpha_\gamma,\beta_\gamma}}{u_{\alpha_\gamma}^{1/2}u_{\beta_\gamma}^{1/2}}=y_\gamma=y_{\alpha+\beta}=2^k\frac{z_{\alpha,\beta}}{u_\alpha^{1/2}u_\beta^{1/2}}\,,
\end{equation}
which gives
\begin{equation}\label{eq:y.gamma.in.mom.mat}
\frac{z_{\alpha_\gamma,\beta_\gamma}}{u_{\alpha_\gamma}^{1/2}u_{\beta_\gamma}^{1/2}}=\frac{z_{\alpha,\beta}}{u_\alpha^{1/2}u_\beta^{1/2}}\,.
\end{equation}
With $\tilde A^{(\alpha,\beta,\gamma)}$ defined as in  \eqref{eq:def.mat.A.alp.beta.gam}, this condition is equivalent to: 
\begin{equation}\label{eq:cons.y.gamma.in.mom.mat}
\tr( \tilde A^{(\alpha,\beta,\gamma)} Z )=0\,.
\end{equation}
Let $y_\gamma$ exist $l_\gamma$ times in the upper triangular portion of $M_k(y)$.
It implies that there are $l_\gamma-1$ affine constraints of form \eqref{eq:cons.y.gamma.in.mom.mat} for each $\gamma$.
Note that the number of entries in the upper triangular portion of $M_k(y)$ is $\sum_{\gamma\in\N^n_{2k}}l_\gamma=\frac{1}{2}\binom{n+k}n[\binom{n+k}n+1]$.
Thus, the number of $\tilde A^{(\alpha,\beta,\gamma)}$, for  $\alpha,\beta\in\N^n_{k}$ satisfying $\alpha\le \beta$, $(\alpha,\beta)\ne (\alpha_\gamma,\beta_\gamma)$, and $\alpha+\beta=\gamma$, for $\gamma\in\N^n_{2k}$, is
\begin{equation}
\sum_{\gamma\in \N^n_{2k}}(l_\gamma-1)=\frac{1}{2}\binom{n+k}n\left[\binom{n+k}n+1\right]-\binom{n+2k}n\,.
\end{equation}

Next, constraint $M_{k-w_j}(h_jy)=0$ is equivalent to 
\begin{equation}
\begin{array}{rl}
\forall \alpha\in\N^n_{2(k-w_j)}\,,\,0=&L_y(x^\alpha h_j)=L_y(x^\alpha \sum_{\beta\in\N^n_{2w_j}}h_{j,\beta}x^\beta)\\
=&L_y(\sum_{\beta\in\N^n_{2w_j}}h_{j,\beta}x^{\alpha+\beta})\\
=&\sum_{\beta\in\N^n_{2w_j}}h_{j,\beta}y_{\alpha+\beta}\\
=&\sum_{\gamma\in\alpha+\N^n_{2w_j}}h_{j,\gamma-\alpha} y_\gamma\\
=&\sum_{\gamma\in\alpha+\N^n_{2w_j}}h_{j,\gamma-\alpha} \left(2^k \frac{z_{\alpha_\gamma,\beta_\gamma}}{u_{\alpha_\gamma}^{1/2}u_{\beta_\gamma}^{1/2}}\right)\\
=&2^k\sum_{\gamma\in\alpha+\N^n_{2w_j}} \frac{  h_{j,\gamma-\alpha}}{u_{\alpha_\gamma}^{1/2}u_{\beta_\gamma}^{1/2}}z_{\alpha_\gamma,\beta_\gamma}\,.
\end{array}
\end{equation}
Then constraint $M_{k-w_j}(h_jy)=0$ can be reformulated as
\begin{equation}\label{eq:eq.cons.moment}
\forall \alpha\in\N^n_{2(k-w_j)}\,,\,\sum_{\gamma\in\alpha+\N^n_{2w_j}} \frac{ h_{j,\gamma-\alpha}}{u_{\alpha_\gamma}^{1/2}u_{\beta_\gamma}^{1/2}}z_{\alpha_\gamma,\beta_\gamma}=0\,.
\end{equation}
With $\tilde A^{(j,\alpha)}$ defined as in \eqref{eq:mat.A.j.alp}, constraints \eqref{eq:eq.cons.moment} becomes
\begin{equation}
\tr( \tilde A^{(j,\alpha)} Z)=0\,,\,\forall \alpha\in\N^n_{2(k-w_j)}\,.
\end{equation}
Note that the number of $\tilde A^{(j,\alpha)}$, for $\alpha\in\N^n_{2(k-w_j)}$, for $j=1,\dots,l$, is
\begin{equation}
\sum_{j=1}^l |\N^n_{2(k-w_j)}|=  \sum_{j=1}^l \binom{n+2(k-w_j)}n\,.
\end{equation}
The last constrain $y_0=1$ translates into affine constraints of the form
\begin{equation}
1=y_0=\frac{2^k z_{0,0}}{u_0^{1/2}u_0^{1/2}}=2^k z_{0,0}=\tr( \tilde A^{(0)} Z)\,,
\end{equation}
where $\tilde A^{(0)}$ is defined as in \eqref{eq:mat.A.0}.
The objective function of \eqref{eq:mom.relax.pop.eq.constrains} can be written as
\begin{equation}
L_y(f)=\sum_{\gamma\in\N^n_{2k}} f_\gamma y_\gamma =\sum_{\gamma\in\N^n_{2k}} f_\gamma 2^k \frac{z_{\alpha_\gamma,\beta_\gamma}}{u_{\alpha_\gamma}^{1/2}u_{\beta_\gamma}^{1/2}} =\tr( \tilde C Z)\,,
\end{equation}
where $\tilde C=(\tilde C_{\alpha,\beta})_{\alpha,\beta\in\N^n_{k}}$ is the matrix defined as in \eqref{eq:def.C}.
Thus, \eqref{eq:mom.relax.pop.eq.constrains} is equivalent to \eqref{eq:sdp.moment.relax.tr.1}.

\eqref{norm.mat.C} From the definition of $\tilde C$, we have
\begin{equation}
\|\tilde C\|_F\le \max_{\gamma\in\N^n_{2k}}\frac{2^k}{u_{\alpha_\gamma}^{1/2}u_{\beta_\gamma}^{1/2}}\|f\|_{\ell_1}\le 2^k \|f\|_{\ell_1}\,,
\end{equation}
\begin{equation}
\|\text{vec}(\tilde C)\|_{\ell_1}\le \max_{\gamma\in\N^n_{2k}}\frac{2^k}{u_{\alpha_\gamma}^{1/2}u_{\beta_\gamma}^{1/2}}\|f\|_{\ell_1}\le 2^k \|f\|_{\ell_1}\,.
\end{equation}

\eqref{remove.trace.cons} Since removing the trace constraint $\tr(G)= 1$ from \eqref{eq:sdp.moment.relax.tr.1} does not affect the optimal value, the same reasoning leads directly to the formulation in \eqref{eq:sdp.moment.relax.tr.1.2}.

\eqref{bound.num.aff.cons} Recall that the total number of affine constraints in the SDP \eqref{eq:sdp.moment.relax.tr.1} includes:
\begin{itemize}
\item Constraints due to the structure of the moment matrix $M_k(y)$, which contribute
\begin{equation}
\frac{1}{2}\binom{n+k}n\left[\binom{n+k}n+1\right]-\binom{n+2k}n\,.
\end{equation}
\item Constraints arising from $M_{k-w_j}(h_jy)=0$, for $j=1,\dots,l$, which contribute
\begin{equation}
\sum_{j=1}^l \binom{n+2(k-w_j)}n\,.
\end{equation}
\item A single constraint from $y_0=1$.
\end{itemize}
Summing these contributions gives the total number of affine constraints $\tilde M$. From this, the stated bound on $\tilde M$ follows.
\end{proof}

\begin{remark}
The maximum number of nonzero entries in a row of 
$\tilde C, \tilde A_i$ is bounded by
\begin{equation}
\begin{array}{rl}
\max\limits_{j=1,\dots,l}\{|\supp(f)|,|\supp(h_j)|\}\le& \max\limits_{j=1,\dots,l}\{|\N^n_{2d}|,|\N^n_{2w_j}|\}\\[5pt]
=&\max\limits_{j=1,\dots,l}\left\{\binom{n+2d}n,\binom{n+2w_j}n\right\}\,.
\end{array}
\end{equation}
However, due to Gaussian elimination and the Gram-Schmidt process, each $A_j$ can be expressed as a linear combination of the $\tilde A_i$'s, causing this sparsity structure to be disrupted.
\end{remark}

\begin{lemma}\label{lem:properties.data.A.eq.cons}
Let $f\in\R[x]_{2k}$ and $h_j\in\R[x]_{2w_j}$, for $j=1,\dots,l$, with $h_1=1-\|x\|_{\ell_2}^2$.
Assume $k\ge \max\{w_1,\dots,w_l\}$.
Let $C$, $A_i$, $i=1,\dots,M$, $N$, $c$, $\lambda_{\min}$, $\lambda_{\max}$ be generated by Algorithm \ref{alg:eq.con.pop.sphere}.
Define
\begin{equation}
A =\begin{bmatrix}
\text{vec}(A_1)\\
\dots\\
\text{vec}(A_M)
\end{bmatrix}\,.
\end{equation}
Let $\lambda^\star=\frac{\tau_k}{c}$.
The following properties hold:
\begin{enumerate}[(i)]
\item\label{bound.eigenvalue.C} $-I_N\preceq C\preceq I_N$ and $\|\text{vec}(C)\|_{\ell_1}\le\frac{2^k \|f\|_{\ell_1}}{c}$.
\item\label{A.full.rank} $A$ has rank $M$, $\sigma_{\min}(A)=1$, and $-I_N\preceq A_i\preceq I_N$, for $i=1,\dots,M$.
\item\label{bound.opt.val.sdp.rho} $\lambda^\star\in[\lambda_{\min},\lambda_{\max}]$.
\item\label{equivalent.probsdp} Problem \eqref{eq:mom.relax.pop.eq.constrains} is equivalent to \eqref{eq:sdp}.
\item\label{upper.bound.on.c} $c\le 2^k \|f\|_{\ell_2}$.
\item\label{bound.of.M.N} $M\le \frac{1}{2}\binom{n+k}n[\binom{n+k}n+1]$ and $N=\binom{n+k}n$.
\item\label{sdp.without.trace.const} 
\begin{equation}\label{eq:sdp.without.trace.const} 
\lambda^\star=\inf\{\tr( C X)\,:\,X\in\mathbb S^{N}_+\,,\,\tr( A_i X) =b_i\,,\,i=1,\dots,M\}\,.
\end{equation}
Moreover, if $y$ is a feasible solution for problem \eqref{eq:mom.relax.pop.eq.constrains}, then  $UM_k(y)U$ is a feasible solution for \eqref{eq:sdp.without.trace.const}, where $U$ is defined as in \eqref{eq:change.coordinate}.
\end{enumerate}
\end{lemma}
\begin{proof}
\eqref{equivalent.probsdp} By Lemma \ref{lem:Gaussian.elimination}, the matrices $\check A_1,\dots,\check A_M$ are linearly independent in $\mathbb S^N$ and
\begin{equation}
\tr( \tilde A_i Z)=\tilde b_i\,,\,i=1,\dots,\tilde M \Leftrightarrow \tr( \check A_i Z)=b_i\,,\,i=1,\dots,M\,.
\end{equation}
By Lemma \ref{lem:Gram-Schmidt}, we have $\tr(  A_i A_j)=\delta_{ij}$ and
 \begin{equation}
\tr( \check A_i X)=\check b_i\,,\,i=1,\dots,M \Leftrightarrow \tr(  A_i X)= b_i\,,\,i=1,\dots,M
\end{equation}
Thus, the problem \eqref{eq:sdp.moment.relax.tr.1} is equivalent to
\begin{equation}\label{eq:sdp.trace.one.eq.pop.sphere}
\begin{array}{rl}
\tau_k=\inf\limits_{Z} & \tr( \tilde C Z)\\
\text{s.t.}& Z\in\mathbb S_+^{N}\,,\,\tr(Z)= 1\,,\\
&\tr( A_i Z)=b_i\,,\,i=1,\dots,M\,.
\end{array}
\end{equation}
Since $C=\frac{\tilde C}{c}$, this problem can be written as
\begin{equation}\label{eq:sdp.trace.one.eq.pop.sphere2}
\begin{array}{rl}
\frac{\tau_k}{c}=\inf\limits_{Z} & \tr( C Z)\\
\text{s.t.}& Z\in\mathbb S_+^{N}\,,\,\tr(Z)= 1\,,\\
&\tr( A_i Z)=b_i\,,\,i=1,\dots,M\,.
\end{array}
\end{equation}
With $\lambda^\star=\frac{\tau_k}{c}$, this is equivalent to \eqref{eq:sdp}.

\eqref{sdp.without.trace.const}  follows from Lemma \ref{lem:properties.tilde.prob.eq.pop} \eqref{remove.trace.cons}.

\eqref{A.full.rank} Since $\tr(  A_i A_j)=\delta_{ij}$, $-I_N\preceq A_i\preceq I_N$, $A$ has rank $M$ and $\sigma_{\min}(A)=1$.

\eqref{upper.bound.on.c} follows from $c=\|\tilde C\|_F$ and Lemma \ref{lem:properties.tilde.prob.eq.pop} \eqref{norm.mat.C}.

\eqref{bound.eigenvalue.C} Since $C=\frac{\tilde C}{c}=\frac{\tilde C}{\|\tilde C\|_F}$, we have $-I_N\preceq C\preceq I$ and $\|\text{vec}(C)\|_{\ell_1}=\frac{\|\text{vec}(\tilde C)\|_{\ell_1}}c\le \frac{2^k \|f\|_{\ell_1}}{c}$.

\eqref{bound.opt.val.sdp.rho} Since $\lambda_{\max}=\frac{f(a)}{c}$ for $a\in V(h)$, $\lambda_{\min}=\frac{\underline \tau}{c}$ and $\underline \tau\le \tau_k\le f^\star\le f(a)$, we obtain $\lambda^\star\in[\lambda_{\min},\lambda_{\max}]$.

\eqref{bound.of.M.N} follows from $M\le \min\{\tilde M,L\}$, with $L=\frac{N(N+1)}{2}$, and Lemma \ref{lem:properties.tilde.prob.eq.pop} \eqref{bound.num.aff.cons}.
\end{proof}

\begin{theorem}\label{theo:accuracy.moment.relax.eqpop.sphere}
Let $f\in\R[x]_{2k}$ and $h_j\in\R[x]_{2w_j}$, for $j=1,\dots,l$, with $h_1=1-\|x\|_{\ell_2}^2$.
Assume $k\ge \max\{d,w_1,\dots,w_l\}$, $\tau_k\ge 0$, and Slater's condition holds for the  moment relaxation \eqref{eq:mom.relax.pop.eq.constrains}, i.e., there exists a feasible solution  $\bar y$ for problem \eqref{eq:mom.relax.pop.eq.constrains} such that $M_k(\bar y)\succ 0$ and has condition number $\kappa $.
Let $\tau_k^{(\varepsilon)}$ be the value returned by Algorithm \ref{alg:eq.con.pop.sphere}.
Then, the following inequality holds:
\begin{equation}
\begin{array}{rl}
0\le& \tau_k^{(\varepsilon)}-\tau_k\\
\le& \varepsilon 2^k\|f\|_{\ell_1} \left[2+\left(1+\sqrt{2} \kappa  \max\limits_{\alpha\in\N^n_k} \binom{k}{(k-|\alpha|,\alpha)}\right) \sqrt{\frac{1}{2}\binom{n+k}n\left[\binom{n+k}n+1\right]}\ \right]\,.
\end{array}
\end{equation}
\end{theorem}
\begin{proof}
Let $\lambda^\star=\frac{\tau_k}{c}$. Since Slater's condition holds for the  moment relaxation \eqref{eq:mom.relax.pop.eq.constrains}, strong duality is satisfied for problem  \eqref{eq:mom.relax.pop.eq.constrains} and its dual.
By Lemma \ref{lem:properties.data.A.eq.cons} \eqref{sdp.without.trace.const}, strong duality also holds for problem \eqref{eq:sdp.trace.one.cons.pop.with.GS} and its dual \eqref{eq:dual.no.tr.cons}.

Define $U_{n,k}=UM_k(\bar y)U$, where $U$ is defined as in \eqref{eq:change.coordinate}. Since $U \succ 0$, it follows that $U_{n,k}\succ 0$ and has condition number
\begin{equation}
\kappa(U_{n,k})\le \kappa \max\limits_{\alpha\in\N^n_k} u_\alpha\,.
\end{equation}
By Lemma \ref{lem:properties.data.A.eq.cons} \eqref{sdp.without.trace.const}, $U_{n,k}$ is a feasible solution for problem \eqref{eq:sdp.without.trace.const}, and thus Slater's condition holds for  problem \eqref{eq:sdp.without.trace.const}.
From Lemma \ref{lem:properties.data.A.eq.cons} \eqref{A.full.rank}, $A$ has full rank $M$.
Since $\tau_k\ge 0$, we have $\lambda^\star=\frac{\tau_k}{c}\ge 0$.
By Lemma \ref{lem:properties.data.A.eq.cons} \eqref{bound.eigenvalue.C}, \eqref{A.full.rank} and \eqref{bound.opt.val.sdp.rho}, $I_N\succeq C\succeq -I_N$, $I_N\succeq A_i\succeq -I_N$, for $i=1,\dots,M$, and  $\lambda^\star\in [\lambda_{\min},\lambda_{\max}]$.
Using Lemma \ref{lem:convergence.sdp.2} (ii), we obtain
\begin{equation}
0\le \lambda^\star-\overline \lambda_T\le \varepsilon\left[2+(1+\sqrt{2} \kappa(U_{n,k})) \|\text{vec}(C)\|_{\ell_1} \sigma_{\min}(A)^{-1}\sqrt{M} \right]\,.
\end{equation}
Applying properties of $C$ and $A$ from Lemma \ref{lem:properties.data.A.eq.cons} \eqref{bound.eigenvalue.C}  and  \eqref{A.full.rank},  this simplifies to
\begin{equation}
0\le \lambda^\star-\overline \lambda_T\le \varepsilon\left[2+(1+\sqrt{2} \kappa \max\limits_{\alpha\in\N^n_k} u_\alpha)\frac{2^k\|f\|_{\ell_1}}{c}  \sqrt{M}  \right]\,.
\end{equation}
Since $\tau_k^{(\varepsilon)}=\underline\lambda_T c$ and $\tau_k=\lambda^\star c$, we have 
\begin{equation}
\begin{array}{rl}
0\le \tau_k^{(\varepsilon)}-\tau_k=&c(\lambda^\star-\underline \lambda_T)\\[5pt]
\le &\varepsilon\left[2c+(1+\sqrt{2} \kappa \max\limits_{\alpha\in\N^n_k} u_\alpha)2^k\|f\|_{\ell_1}\sqrt{M}  \right]\\[10pt]
\le& \varepsilon \left[2^{k+1}\|f\|_{\ell_2}+(1+\sqrt{2} \kappa \max\limits_{\alpha\in\N^n_k} u_\alpha)2^k\|f\|_{\ell_1} \sqrt{M} \right]\\[10pt]
\le& \varepsilon 2^k\|f\|_{\ell_1} \left[2+(1+\sqrt{2} \kappa \max\limits_{\alpha\in\N^n_k} u_\alpha) \sqrt{M} \right]\,.
\end{array}
\end{equation}
The third inequality is due to Lemma \ref{lem:properties.data.A.eq.cons} \eqref{upper.bound.on.c}.
Finally, using equality $u_\alpha=\binom{k}{(k-|\alpha|,\alpha)}$, for $\alpha\in \N^n_k$, and inequality $M\le \frac{1}{2}\binom{n+k}n[\binom{n+k}n+1]$ completes the proof.
\end{proof}

\begin{lemma}[Complexity]
Assume $k\ge \max\{d,w_1,\dots,w_l\}$.
To run Algorithm \ref{alg:eq.con.pop.sphere}, the time complexity on a classical computer is
\begin{equation}
O\left(\binom{n+k}n^4\left[(l-1) \binom{n+2k}n+\binom{n+k}n^2 \right]+ \binom{n+k}n^4\frac{1}{\varepsilon^2}\right)\,,
\end{equation}
while on a quantum computer, the time complexity is
\begin{equation}
\footnotesize
O\left(\binom{n+k}{n}^{1.5}\left\{(l-1) \binom{n+2k}n \binom{n+k}{n}^{2.5}+\binom{n+k}{n}^{4.5}+\left[\binom{n+k}{n}^{0.5}+\frac{1}{\varepsilon} \right]\frac{1}{\varepsilon^{4}}\right\}\right)\,.
\end{equation}
\end{lemma}
\begin{proof}
The most computationally expensive steps in Algorithm \ref{alg:eq.con.pop.sphere} are performing Gaussian elimination in Step \ref{step:gauss.elimi.mom.relax}, the Gram-Schmidt process in Step \ref{step:gram-schmidt.mom.relax}, and running binary search using Hamiltonian updates in Step \ref{step:Hamiltonian.update.mom.relax}.

For Gaussian elimination, by Lemma \ref{lem:complexity.Gaussian.elimination}, both classical and quantum computers have the same complexity: $O(\tilde M N^4)$.

For the Gram-Schmidt process, by Lemma \ref{lem:complexity.Gram-Schmidt}, both classical and quantum computers have the same complexity: $O(N^2M^2)$.

For binary search using Hamiltonian updates, by Lemma \ref{lem:complex.binary.search}, a classical computer takes $O((MsN+N^\omega)\varepsilon^{-2})$, a classical computer takes $O(s(\sqrt{M}+\sqrt{N}\varepsilon^{-1})\varepsilon^{-4})$ times on a quantum computer,
, where $s$ is the maximum number of nonzero entries in a row of the data matrices $C,A_i$, for $i=1,\dots,M$.
We take $s=N$, which is the size of $C,A_i$.

By Lemma \ref{lem:properties.data.A.eq.cons} \eqref{bound.of.M.N}, we have $M\le \frac{1}2\binom{n+k}n[\binom{n+k}n+1]$, $N= \binom{n+k}n$.
Furthermore, by Lemma \ref{lem:properties.tilde.prob.eq.pop} \eqref{bound.num.aff.cons},
$\tilde M\le 1+ (l-1) \binom{n+2k}n+\frac{1}{2}\binom{n+k}n[\binom{n+k}n+1]$.
Thus, the result follows.
\end{proof}

\subsubsection{(In)equality-constrained polynomial optimization over the unit ball}
\label{sec:mom.relax.ineq.cons.all}
Let $f,g_i,h_j\in\R[x]$ for $i=1,\dots,m$ and $j=1,\dots,l$,  with $g_1=1-\|x\|_{\ell_2}^2$.
Consider polynomial optimization problem
\begin{equation}\label{eq:pop.in.eq.constrains}
f^\star=\min\limits_{x\in S(g)\cap V(h)} f(x)\,,
\end{equation}
where $S(g)$ s the semialgebraic set defined as in \eqref{eq:def.semialgebaric.set} and $V(h)$ is the algebraic variety defined as in \eqref{eq:def.variety}.

Define $d=\lceil \deg(f)/2 \rceil$, $d_i=\lceil \deg(g_i)/2 \rceil$, for $i=1,\dots,m$, and $w_j=\lceil \deg(h_j)/2 \rceil$, for $j=1,\dots,l$.
Here, $d_1=1$.

The moment relaxation of order $k$ for problem \eqref{eq:pop.in.eq.constrains} is expressed as:
\begin{equation}\label{eq:mom.relax.pop.in.eq.constrains}
\begin{array}{rl}
\tau_k=\min\limits_y& L_y(f)\\
\text{s.t.}& y=(y_\alpha)_{\alpha\in\N^n_{2k}}\subset \R\,,\\
&M_k(y)\succeq 0\,,\,y_0=1\,,\\
&M_{k-d_i}(g_iy)\succeq 0\,,\,i=1,\dots,m\,,\\
&M_{k-w_j}(h_jy)=0\,,\,j=1,\dots,l\,.
\end{array}
\end{equation}
For this relaxation, we have the following properties:
$\tau_k\le \tau_{k+1}\le f^\star$ and  $\tau_k\to f^\star$ as $k\to \infty$.

Assume that problem \eqref{eq:mom.relax.pop.in.eq.constrains} has an optimal solution and that a lower bound $\underline \tau\le \tau_k$ is known.

\begin{algorithm}\label{alg:eq.con.pop.ball}
(In)equality constrained polynomial optimization over the unit ball
\begin{itemize}
\item Input: $\varepsilon>0$, $\underline \tau\in\R$, $a\in S(g)\cap V(h)$,\\ 
\phantom{1.1cm} coefficients $(f_\alpha)_{\alpha\in\N^n_{2k}}$ of $f\in\R[x]$,\\
\phantom{1.1cm} coefficients $(g_{i,\alpha})_{\alpha\in\N^n_{2d_i}}$ of $g_i\in\R[x]$, for $i=0,\dots,m$,\\
\phantom{1.1cm} coefficients $(h_{j,\alpha})_{\alpha\in\N^n_{2w_j}}$ of $h_j\in\R[x]$, for $j=1,\dots,l$.
\item Output: $\tau_k^{(\varepsilon)}\in\R$ and $Z_\varepsilon\in\mathbb S^{\bar N}$.
\end{itemize}
\begin{enumerate}
\item\label{step:convert.to.sdp.mom.relax.genpop} Convert relaxation to standard SDP:
\begin{enumerate}
\item Define $u_\alpha=\binom{k}{(k-|\alpha|,\alpha)}$, for $\alpha\in \N^n_k$, and  $p_\alpha=\sum_{j=|\alpha|}^{k-1} 2^{k-1-j} \binom{j}{(j-|\alpha|,\alpha)}$, for $\alpha\in\N^n_{k-1}$.
\item Let $W\in \mathbb S^{\binom{n+k}n}$ satisfy
\begin{equation}\label{eq:constraint.W}
\sum_{i=2}^m g_i\sum_{\alpha\in\N^n_{k-d_i}}x^{2\alpha}=v_k^\top Wv_k\,.
\end{equation}
\item Let 
\begin{equation}\label{eq:def.delta}
\delta=\frac{\min_{\alpha\in\N^n_k}u_\alpha}{2\sigma_{\max}(W)+1}\,.
\end{equation}
\item Set 
\begin{equation}\label{eq:def.p.i.alp}
\begin{array}{rl}
&p_\alpha^{(0)}=\frac{u_\alpha^{1/2}}{2^{(k+1)/2}}\,, \,\forall \alpha\in\N^n_k\,,\\[5pt]
&p_\alpha^{(1)}=\frac{p_\alpha^{1/2}}{2^{k/2}}\,,\,\forall \alpha\in\N^n_{k-d_1}\,,\\[5pt]
&p_\alpha^{(i)}=\delta^{1/2}\,,\,\forall \alpha\in\N^n_{k-d_i}\,,\,i=2,\dots,m\,.
\end{array}
\end{equation}
\item For $\gamma\in\N^n_{2k}$, let $\alpha_\gamma,\beta_\gamma\in\N^n_k$ be such that $\alpha_\gamma\le \beta_\gamma$ and $\alpha_\gamma+\beta_\gamma=\gamma$.
\item For $\gamma\in\N^n_{2k}$, for $\alpha,\beta\in\N^n_{k}$ satisfying $\alpha\le \beta$, $(\alpha,\beta)\ne (\alpha_\gamma,\beta_\gamma)$, and $\alpha+\beta=\gamma$, let $\bar A^{(\alpha,\beta,\gamma,0)}=(\bar A^{(\alpha,\beta,\gamma,0)}_{\xi,\zeta})_{\xi,\zeta\in\N^n_k}$ be the matrix defined by
\begin{equation}\label{eq:mat.A.alpha.beta.gamma.0}
\bar A^{(\alpha,\beta,\gamma,0)}_{\xi,\zeta}=
\begin{cases}
\frac{1}{p_{\xi}^{(0)}p_{\zeta}^{(0)}}&\text{if }(\xi,\zeta)=(\alpha_\gamma,\beta_\gamma)\,,\, \alpha_\gamma=\beta_\gamma\,,\\
\frac{1}{2p_{\xi}^{(0)}p_{\zeta}^{(0)}}&\text{if }(\xi,\zeta)\in\{(\alpha_\gamma,\beta_\gamma),(\alpha_\gamma,\beta_\gamma)\}\,,\,\alpha_\gamma\ne \beta_\gamma \,,\\
-\frac{1}{p_{\xi}^{(0)}p_{\zeta}^{(0)}}&\text{if }(\xi,\zeta)=(\alpha,\beta)\,,\, \alpha=\beta\,,\\
-\frac{1}{2p_{\xi}^{(0)}p_{\zeta}^{(0)}}&\text{if }(\xi,\zeta)\in\{(\alpha,\beta),(\alpha,\beta)\}\,,\,\alpha\ne \beta \,,\\
0&\text{else},
\end{cases}
\end{equation}
and set $\bar A^{(\alpha,\beta,\gamma)}=\diag(\bar A^{(\alpha,\beta,\gamma,0)},0,\dots,0)$.
\item For $i=1,\dots,m$, for $\alpha,\beta\in\N^n_{k-d_i}$, $\alpha\le \beta$, let $\bar A^{(i,\alpha,\beta,0)}=(\bar A^{(i,\alpha,\beta,0)}_{\xi,\zeta})_{\xi,\zeta\in\N^n_k}$ be the matrix defined by 
\begin{equation}\label{eq:mat.A.i.alp.beta.0}
\bar A^{(i,\alpha,\beta,0)}_{\xi,\zeta}=
\begin{cases}
\frac{  g_{i,\gamma-\alpha-\beta}}{2p_{\xi}^{(0)}p_{\zeta}^{(0)}}&\text{if }(\xi,\zeta)\in \left\{(\alpha_\gamma,\beta_\gamma),(\beta_\gamma,\alpha_\gamma)\,:\,\alpha_\gamma\ne \beta_\gamma\,,\right.\\
&\qquad\qquad\qquad\qquad\qquad\qquad\left.\gamma\in\alpha+\beta+\N^n_{2d_i}\right\}\,,\\
\frac{ g_{i,\gamma-\alpha-\beta}}{p_{\xi}^{(0)}p_{\zeta}^{(0)}}&\text{if }(\xi,\zeta)\in \left\{(\alpha_\gamma,\beta_\gamma),\,:\,\alpha_\gamma= \beta_\gamma\,,\right.\\
&\qquad\qquad\qquad\qquad\qquad\qquad\left.\gamma\in\alpha+\beta+\N^n_{2d_i}\right\}\,,\\
0&\text{else,}
\end{cases}
\end{equation}
let $\bar A^{(i,\alpha,\beta,i)}=(\bar A^{(i,\alpha,\beta,i)}_{\xi,\zeta})_{\xi,\zeta\in\N^n_{k-d_i}}$ be the matrix defined by 
\begin{equation}\label{eq:mat.A.i.alp.beta.i}
\bar A^{(i,\alpha,\beta,i)}_{\xi,\zeta}=
\begin{cases}
-\frac{1}{2p_{\xi}^{(0)}p_{\zeta}^{(0)}}&\text{if }(\xi,\zeta)\in \{(\alpha,\beta),(\beta,\alpha)\}\,,\,\alpha\ne \beta\,,\\
-\frac{1}{p_{\xi}^{(0)}p_{\zeta}^{(0)}}&\text{if }(\xi,\zeta)=(\alpha_\gamma,\beta_\gamma)\,,\,\alpha= \beta\,,\\
0&\text{else,}
\end{cases}
\end{equation}
and set $\bar A^{(i,\alpha,\beta)}=\diag(\bar A^{(i,\alpha,\beta,0)},0,\dots,0,\bar A^{(i,\alpha,\beta,i)},0,\dots,0)$.
\item For $j=1,\dots,l$, for $\alpha\in\N^n_{2(k-w_j)}$, let $\bar A^{(j,\alpha,0)}=(\bar A^{(j,\alpha,0)}_{\xi,\zeta})_{\xi,\zeta\in\N^n_k}$ be the matrix defined by 
\begin{equation}\label{eq:def.mat.A.j.alpha.0}
\bar A^{(j,\alpha,0)}_{\xi,\zeta}=
\begin{cases}
\frac{  h_{j,\gamma-\alpha}}{2p_{\xi}^{(0)}p_{\zeta}^{(0)}}&\text{if }(\xi,\zeta)\in \left\{(\alpha_\gamma,\beta_\gamma),(\beta_\gamma,\alpha_\gamma)\,:\,\alpha_\gamma\ne \beta_\gamma\,,\right.\\
&\qquad\qquad\qquad\qquad\qquad\qquad\qquad\left.\gamma\in\alpha+\N^n_{2w_j}\right\}\,,\\
\frac{ h_{j,\gamma-\alpha}}{p_{\xi}^{(0)}p_{\zeta}^{(0)}}&\text{if }(\xi,\zeta)\in \{(\alpha_\gamma,\beta_\gamma),\,:\,\alpha_\gamma= \beta_\gamma\,,\,\gamma\in\alpha+\N^n_{2w_j}\}\,,\\
0&\text{else,}
\end{cases}
\end{equation}
and set $\bar A^{(j,\alpha)}=\diag(\bar A^{(j,\alpha,0)},0,\dots,0)$.
\item Let $\bar A^{(0,0)}=(\bar A_{\alpha,\beta}^{(0,0)})_{\alpha,\beta\in\N^n_k}$ be the matrix defined by
\begin{equation}\label{eq:def.mat.A.0.0}
\bar A_{\alpha,\beta}^{(0,0)}=
\begin{cases}
\frac{1}{p_0^{(0)2}}&\text{if }\alpha=\beta=0\,,\\
0&\text{else},
\end{cases}
\end{equation}
and set $\bar A^{(0)}=\diag(\bar A^{(0,0)},0,\dots,0)$.
\item Let $\bar C^{(0)}=(\bar C^{(0)}_{\alpha,\beta})_{\alpha,\beta\in\N^n_{k}}$ be the matrix defined by
\begin{equation}\label{eq:def.mat.C.ineq.pop}
\bar C^{(0)}_{\alpha,\beta}=\begin{cases}
 \frac{f_\gamma }{p_{\alpha}^{(0)}p_{\beta}^{(0)}}&\text{if }(\alpha,\beta)=(\alpha_\gamma,\beta_\gamma)\,,\,\alpha_\gamma=\beta_\gamma\,,\\
  \frac{f_\gamma }{2p_{\alpha}^{(0)}p_{\beta}^{(0)}}&\text{if }(\alpha,\beta)\in \{(\alpha_\gamma,\beta_\gamma),(\beta_\gamma,\alpha_\gamma)\}\,,\,\alpha_\gamma\ne\beta_\gamma\,,\\
  0&\text{else},
\end{cases}
\end{equation}
and set $\bar C=\diag(\bar C^{(0)},0,\dots,0)$.
\item Set $\bar N=\sum_{i=0}^m\binom{n+k-d_i}n$.
\item Let $\bar A_1,\dots,\bar A_{\tilde M}$ be the following matrices:
\begin{enumerate}
\item $\bar A^{(\alpha,\beta,\gamma)}$, for $\alpha,\beta\in\N^n_{k}$ satisfying $\alpha\le \beta$, $(\alpha,\beta)\ne (\alpha_\gamma,\beta_\gamma)$, and $\alpha+\beta=\gamma$, for $\gamma\in\N^n_{2k}$;
\item $\bar A^{(i,\alpha,\beta)}$, for $\alpha,\beta\in\N^n_{k-d_i)}$, for $i=1,\dots,m$;
\item $\bar A^{(j,\alpha)}$, for $\alpha\in\N^n_{2(k-w_j)}$, for $j=1,\dots,l$;
\item $\bar A^{(0)}$.
\end{enumerate}
In particular, $\bar A_{\tilde M}=\bar A^{(0)}$.
\item Let $\bar b_1,\dots,\bar b_{\tilde M}$ be  defined by $\bar b_i=\begin{cases}
1&\text{if } i=\tilde M\,,\\
0&\text{else}.
\end{cases}$
\item Set $\tilde C=\diag(\bar C,0)$, $\tilde A_i=\diag(\bar A_i,0)$, $\tilde b_i= \bar b_i$, $i=1,\dots,\tilde M$, $N=\bar N+1$.
\item\label{step:gauss.elimination.mom.relax.genpop} Run Algorithm \ref{alg:Gaussian.elimination} (Gaussian elimination) to get $\check A_i\in \mathbb S^N$, $\check b_i\in\R$, $i=1,\dots,M$.
\item\label{step:gram.schmidt.mom.relax.genpop} Run Algorithm \ref{alg:Gram-Schmidt} (Gram-Schmidt process) to get $A_i\in\mathbb S^N$ and $b_i\in\R$, $i=1,\dots,M$.

\item Set $c=\|\tilde C\|_F$, $C=\frac{\tilde C}{c}$, $\lambda_{\max}=\frac{f(a)}{c}$ and $\lambda_{\min}=\frac{\underline \tau}{c}$.
\end{enumerate}

\item\label{step:hamilnion.updates.mom.relax.genpop} Run Algorithm \ref{alg:Binary.search.HU} (Binary search using Hamiltonian updates) 
to obtain $\underline \lambda_T\in\R$ and $X_\varepsilon=((X_\varepsilon)_{ij})_{i,j=1,\dots,N}\in\mathbb S^N$.

\item Extract approximate results: Set $\tau_k^{(\varepsilon)}=\underline\lambda_T c$ and $Z_\varepsilon=((X_\varepsilon)_{ij})_{i,j=1,\dots,\bar N}$.
\end{enumerate}
\end{algorithm}

\begin{lemma}\label{lem:trce.1.moment.relax.ineq.pop}
Let $f\in\R[x]_{2k}$ and $g_i\in\R[x]_{2d_i}$, for $i=1,\dots,m$, $h_j\in\R[x]_{2w_j}$, for $j=1,\dots,l$, with $g_1=1-\|x\|_{\ell_2}^2$.
Let $p_\alpha^{(i)}$, for $\alpha\in \N^n_{k-d_i}$ and $i=0,\dots,m$, be generated by Algorithm \ref{alg:eq.con.pop.ball}, and define $P^{(i)}=\diag(p_\alpha^{(i)})_{\alpha\in\N^n_{k-d_i}}$, for $i=0,\dots,m$.
Then, for any feasible solution $y$ to problem \eqref{eq:mom.relax.pop.in.eq.constrains},
\begin{equation} 
\tr(PD_k(gy)P)\le 1\,,
\end{equation}
where $D_k(gy)=\diag((M_{k-d_i}(g_iy))_{i=0,\dots,m})$ and $P=\diag((P^{(i)})_{i=0,\dots,m})$.
\end{lemma}
\begin{proof}
Using the identity $a^k-b^k=(a-b)\sum_{j=0}^k a^{k-1-j}b^j$ with $a=2$ and $b=1+\|x\|_{\ell_2}^2$, we have
\begin{equation}\label{eq:equality.pol.ineq.pop}
2^k=(1+\|x\|_{\ell_2}^2)^k+(1-\|x\|_{\ell_2}^2)p\,,
\end{equation}
where $p=\sum_{j=0}^{k-1} 2^{k-1-j}(1+\|x\|_{\ell_2}^2)^j$.
For $\alpha=(\alpha_1,\dots,\alpha_n)$ and $\bar\alpha=(\alpha_0,\alpha)$, we express $(1+\|x\|_{\ell_2}^2)^k$ in term  of monomials:
\begin{equation}
\begin{array}{rl}
(1+\|x\|_{\ell_2}^2)^k=&(1+x_1^2+\dots+x_n^2)^k\\[5pt]
=& \sum\limits_{\arraycolsep=1.4pt\def\arraystretch{.7}
\begin{array}{cc}
\scriptstyle \bar\alpha\in\N^{n+1}\\
\scriptstyle |\bar\alpha|=k
\end{array}} \binom{k}{\bar \alpha} 1^{\alpha_0} x_1^{2\alpha_1}\dots x_n^{2\alpha_n}=\sum\limits_{\alpha\in\N^n_k} \binom{k}{(k-|\alpha|,\alpha)} x^{2\alpha} \,.
\end{array}
\end{equation}
With $u_\alpha=\binom{k}{(k-|\alpha|,\alpha)}$, we obtan $(1+\|x\|_{\ell_2}^2)^k=\sum_{\alpha\in\N^n_k}u_\alpha x^{2\alpha}$.
Similarly, $p$ can be written as
\begin{equation}
\begin{array}{rl}
p=&\sum_{j=0}^{k-1} 2^{k-1-j}\sum_{\alpha\in\N^n_j} \binom{j}{(j-|\alpha|,\alpha)} x^{2\alpha}\\[10pt]
=& \sum_{\alpha\in\N^n_{k-1}} \sum_{j=|\alpha|}^{k-1} 2^{k-1-j} \binom{j}{(j-|\alpha|,\alpha)} x^{2\alpha}\,.
\end{array}
\end{equation}
With $p_\alpha=\sum_{j=|\alpha|}^{k-1} 2^{k-1-j} \binom{j}{(j-|\alpha|,\alpha)}$, we have $p=\sum_{\alpha\in\N^n_{k-d_1}} p_\alpha x^{2\alpha}$.
Equality \eqref{eq:equality.pol.ineq.pop} becomes
\begin{equation}
2^k=\sum_{\alpha\in\N^n_k}u_\alpha x^{2\alpha}+g_1\sum_{\alpha\in\N^n_{k-d_1}} p_\alpha x^{2\alpha}
\end{equation}
With $\delta>0$, we write
\begin{equation}\label{eq:equality.pol.ineq.pop2}
\begin{array}{rl}
2^k=&\frac{1}{2}\sum_{\alpha\in\N^n_k}u_\alpha x^{2\alpha}
+g_1\sum_{\alpha\in\N^n_{k-d_1}} p_\alpha x^{2\alpha}
+\delta\sum_{i=2}^m g_i\sum_{\alpha\in\N^n_{k-d_i}}  x^{2\alpha}\\[10pt]
&+ v_k^\top G v_k\,,
\end{array}
\end{equation}
where $W\in \mathbb S^{\binom{n+k}n}$ satisfies \eqref{eq:constraint.W} and $G=\frac{1}{2}[\diag((u_\alpha)_{\alpha\in\N^n_k})-2\delta W]$.
Choosing $\delta$ as in \eqref{eq:def.delta} gives $G\succeq 0$.
Dividing both sides of \eqref{eq:equality.pol.ineq.pop2} implies
\begin{equation}
\begin{array}{rl}
1=&\sum_{\alpha\in\N^n_k}\frac{u_\alpha}{2^{k+1}} x^{2\alpha}
+g_1\sum_{\alpha\in\N^n_{k-d_1}} \frac{p_\alpha }{2^k}x^{2\alpha}
+\delta(\sum_{i=2}^m g_i\sum_{\alpha\in\N^n_{k-d_i}}  x^{2\alpha})\\[10pt]
&+\frac{1}{2^k}v_k^\top G v_k\\[10pt]
=&\sum_{i=0}^mg_i\sum_{\alpha\in\N^n_{k-d_i}}p_\alpha^{(i)2} x^{2\alpha}+\frac{1}{2^k}v_k^\top G v_k\,,
\end{array}
\end{equation}
where $p_\alpha^{(i)}$ is defined as in \eqref{eq:def.p.i.alp}
With $P^{(i)}=\diag(p_\alpha^{(i)})_{\alpha\in\N^n_{k-d_i}}$, we obtain
\begin{equation}
1=\sum_{i=0}^mg_i v_{k-d_i}^\top P^{(i)2}v_{k-d_i}+\frac{1}{2^k}v_k^\top G v_k\,.
\end{equation}
Let $y$ be a feasible solution of problem \eqref{eq:mom.relax.pop.eq.constrains}.
Applying $L_y$ for both sides yields
\begin{equation}
\begin{array}{rl}
1=&\sum_{i=0}^m\tr(M_{k-d_i}(g_iy)P^{(i)2})+\frac{1}{2^k}\tr(GM_k(y))\\[5pt]
\ge &\sum_{i=0}^m\tr(M_{k-d_i}(g_iy)P^{(i)2})\\[5pt]
=&\sum_{i=0}^m\tr(P^{(i)}M_{k-d_i}(g_iy)P^{(i)})=\tr(PD_k(y)P)\,.
\end{array}
\end{equation}
The last inequality follows from $G\succeq 0$ and $M_k(y)\succeq 0$.
Hence, the result is proved.
\end{proof}
\begin{lemma}\label{lem:properties.tilde.prob.ineq.pop}
Let $f\in\R[x]_{2k}$ and $g_i\in\R[x]_{2d_i}$, for $i=1,\dots,m$, $h_j\in\R[x]_{2w_j}$, for $j=1,\dots,l$, with $g_1=1-\|x\|_{\ell_2}^2$.
Let $\bar C$, $\bar A_i$, for $i=1,\dots,\tilde M$, $\bar N$, $c$, $\lambda_{\min}$, and $\lambda_{\max}$ be generated by Algorithm \ref{alg:eq.con.pop.ball}.
The following properties hold:
\begin{enumerate}[(i)]
\item \label{equi.proble.mom.relax.gen.pop} Problem \eqref{eq:mom.relax.pop.in.eq.constrains} is equivalent to  the semidefinite program (SDP):
\begin{equation}\label{eq:sdp.moment.relax.tr.1.ineq.pop}
\begin{array}{rl}
\tau_k=\min\limits_Z&\tr( \bar C Z)\\
\text{s.t.}&Z\in\mathbb S^{\bar N}_+\,,\,\tr(Z)\le 1\,,\\
&\tr( \bar A_i Z) =\bar b_i\,,\,i=1,\dots,\tilde M\,.	
\end{array}
\end{equation}
\item\label{norm.C.mom.relax.gen.pop} $\|\bar C\|_F\le 2^{k+1} \|f\|_{\ell_2}$ and $\|\text{vec}(\bar C)\|_{\ell_1}\le 2^{k+1} \|f\|_{\ell_1}$.
\item\label{bound.size.mat.mom.relax.gen.pop} If $k\ge \max\{d_1,\dots,d_m\}$, then the size $\bar N$ is bounded as: $\bar N\le (m+1)\binom{n+k}n$.
\item\label{bound.num.aff.cons.mom.relax.gen.pop} If $k\ge \max\{d_1,\dots,d_m,w_1,\dots,w_l\}$, then the total number of constraints $\tilde M$ is
\begin{equation}\label{eq:num.aff.cons.M.tilde.ineq.pop}
\begin{array}{rl}
\tilde M=&1+\sum_{j=1}^l \binom{n+2(k-w_j)}n+\frac{1}{2}\binom{n+k}n[\binom{n+k}n+1]-\binom{n+2k}n\\[5pt]
&+\frac{1}2\sum_{i=1}^m \binom{n+k-d_i}n[\binom{n+k-d_i}n+1]\,.
\end{array}
\end{equation}
\item\label{no.trace.mom.relax.genpop} The optimal value of \eqref{eq:sdp.moment.relax.tr.1.ineq.pop} is unchanged when the trace constraint is omitted:
\begin{equation}\label{eq:sdp.mom.relax.no.tr.gen.pop}
\begin{array}{rl}
\tau_k=\min\limits_Z&\tr( \bar C Z)\\
\text{s.t.}&Z\in\mathbb S^{\bar N}_+\,,\\
&\tr( \bar A_i Z) =\bar b_i\,,\,i=1,\dots,\tilde M\,.	
\end{array}
\end{equation}
Moreover, if $y$ is a feasible solution for problem \eqref{eq:mom.relax.pop.in.eq.constrains}, then $PD_k(gy)P$ is a feasible solution for \eqref{eq:sdp.mom.relax.no.tr.gen.pop}, where $P$ is constructed as in Lemma \ref{lem:trce.1.moment.relax.ineq.pop}.
\end{enumerate}
\end{lemma}
\begin{proof}
\eqref{equi.proble.mom.relax.gen.pop} Let $y$ be a feasible solution of problem \eqref{eq:mom.relax.pop.eq.constrains}.
Define $Z=PD_k(y)P$.
Since $D_k(y)\in\mathbb S^{\bar N}_+$, it follows that $Z\in\mathbb S^{\bar N}_+$.
Using Lemma \ref{lem:trce.1.moment.relax.ineq.pop}, we have $\tr(Z)\le1$.
We express $Z$ in block-diagonal form as 
\begin{equation}
Z=\diag(Z^{(0)},\dots,Z^{(m)})\,,
\end{equation}
where $Z^{(i)}=(z^{(i)}_{\alpha,\beta})_{\alpha,\beta\in\N^n_{k-d_i}}$, for $i=0,\dots,m$.
In particular, $Z^{(0)}=P^{(0)}M_k(y)P^{(0)}$ with $z^{(0)}_{\alpha,\beta}=p^{(0)}_\alpha  y_{\alpha+\beta}p^{(0)}_\beta$.
Consequently, $y_{\alpha+\beta}= \frac{z^{(0)}_{\alpha,\beta}}{p^{(0)}_\alpha p^{(0)}_\beta}$.
By the definition of $\alpha_\gamma$ and $\beta_\gamma$, we have $y_\gamma= \frac{z^{(0)}_{\alpha_\gamma,\beta_\gamma}}{p^{(0)}_{\alpha_\gamma}p^{(0)}_{\beta_\gamma}}$.
Moreover, for all $\alpha,\beta\in\N^n_{k}$ satisfying $\alpha\le \beta$, $(\alpha,\beta)\ne (\alpha_\gamma,\beta_\gamma)$, and $\alpha+\beta=\gamma$,
\begin{equation}\label{eq:constrain.moment.mat}
\frac{z^{(0)}_{\alpha_\gamma,\beta_\gamma}}{p_{\alpha_\gamma}^{(0)}p_{\beta_\gamma}^{(0)}}=y_\gamma=\frac{z^{(0)}_{\alpha,\beta}}{p_\alpha^{(0)}p_\beta^{(0)}}\,.
\end{equation}
Defining $\bar A^{(\alpha,\beta,\gamma)}=\diag(\bar A^{(\alpha,\beta,\gamma,0)},0,\dots,0)$ with $\bar A^{(\alpha,\beta,\gamma,0)}$ as in \eqref{eq:mat.A.alpha.beta.gamma.0}, constraint \eqref{eq:constrain.moment.mat} becomes 
\begin{equation}
\tr( \bar A^{(\alpha,\beta,\gamma)} Z)=\tr(\bar A^{(\alpha,\beta,\gamma,0)} Z^{(0)})=0\,.
\end{equation}
The number of such $\bar A^{(\alpha,\beta,\gamma)}$, for $\alpha,\beta\in\N^n_{k}$ satisfying $\alpha\le \beta$, $(\alpha,\beta)\ne (\alpha_\gamma,\beta_\gamma)$, and $\alpha+\beta=\gamma$, $\gamma\in\N^n_{2k}$, is
 $\frac{1}{2}\binom{n+k}n[\binom{n+k}n+1]-\binom{n+2k}n$.

For $i=1,\dots,m$, the constraint $Z^{(i)}=P^{(i)}M_{k-d_i}(g_iy)P^{(i)}$ is equivalent to $P^{(i)-1}Z^{(i)}P^{(i)-1}=M_{k-d_i}(g_iy)$.
For $\alpha,\beta\in\N^n_{k-d_i}$, $\alpha\le \beta$, this leads to
\begin{equation}\label{eq.cons.localizing.mat}
\begin{array}{rl}
\frac{z^{(i)}_{\alpha,\beta}}{p_{\alpha}^{(i)}p_{\beta}^{(i)}}
=&\sum_{\eta\in\N^n_{2d_i}}g_{i,\eta}y_{\alpha+\beta+\eta}\\
=&\sum_{\gamma\in\alpha+\beta+\N^n_{2d_i}}g_{i,\gamma-\alpha-\beta} y_\gamma\\
=&\sum_{\gamma\in\alpha+\beta+\N^n_{2d_i}}g_{i,\gamma-\alpha-\beta} \left(\frac{z^{(0)}_{\alpha_\gamma,\beta_\gamma}}{p_{\alpha_\gamma}^{(0)}p_{\beta_\gamma}^{(0)}}\right)\\[10pt]
=&\sum_{\gamma\in\alpha+\beta+\N^n_{2d_i}} \frac{ g_{i,\gamma-\alpha-\beta}}{p_{\alpha_\gamma}^{(0)}p_{\beta_\gamma}^{(0)}} z^{(0)}_{\alpha_\gamma,\beta_\gamma}\,.
\end{array}
\end{equation}
Introducing $\bar A^{(i,\alpha,\beta)}=\diag(\bar A^{(i,\alpha,\beta,0)},0,\dots,0,\bar A^{(i,\alpha,\beta,i)},0,\dots,0)$, where $\bar A^{(i,\alpha,\beta,0)}$ are defined as in \eqref{eq:mat.A.i.alp.beta.0}
and $\bar A^{(i,\alpha,\beta,i)}$ defined as in \eqref{eq:mat.A.i.alp.beta.i}, the constraint  \eqref{eq.cons.localizing.mat} becomes
\begin{equation}
\tr( \bar A^{(i,\alpha,\beta)} Z)=\tr(\bar A^{(i,\alpha,\beta,0)} Z^{(0)})+\tr( \bar A^{(i,\alpha,\beta,i)} Z^{(i)})=0\,.
\end{equation}
The number of such $\bar A^{(i,\alpha,\beta)}$, for $\alpha,\beta\in\N^n_{k-d_i}$, $\alpha\le \beta$, $i=1,\dots,m$, is
\begin{equation}
\sum_{i=1}^m \frac{1}2|\N^n_{k-d_i)}|(|\N^n_{k-d_i)}|+1)=  \frac{1}2\sum_{i=1}^m \binom{n+k-d_i}n\left[\binom{n+k-d_i}n+1\right]\,.
\end{equation}
For $j=1,\dots,l$, the constraint $M_{k-w_j}(h_jy)=0$ is equivalent to 
\begin{equation}\label{eq:localizing.mat.cons.eq}
\begin{array}{rl}
\forall \alpha\in\N^n_{2(k-w_j)}\,,\,0=&L_y(x^\alpha h_j)=L_y(x^\alpha \sum_{\beta\in\N^n_{2w_j}}h_{j,\beta}x^\beta)\\
=&L_y(\sum_{\beta\in\N^n_{2w_j}}h_{j,\beta}x^{\alpha+\beta})\\
=&\sum_{\beta\in\N^n_{2w_j}}h_{j,\beta}y_{\alpha+\beta}\\
=&\sum_{\gamma\in\alpha+\N^n_{2w_j}}h_{j,\gamma-\alpha} y_\gamma\\
=&\sum_{\gamma\in\alpha+\N^n_{2w_j}}h_{j,\gamma-\alpha} \left(\frac{z^{(0)}_{\alpha_\gamma,\beta_\gamma}}{p_{\alpha_\gamma}^{(0)}p_{\beta_\gamma}^{(0)}}\right)\\[10pt]
=&\sum_{\gamma\in\alpha+\N^n_{2w_j}} \frac{ h_{j,\gamma-\alpha}}{p_{\alpha_\gamma}^{(0)}p_{\beta_\gamma}^{(0)}}z^{(0)}_{\alpha_\gamma,\beta_\gamma}\,.
\end{array}
\end{equation}
Defining $\bar A^{(j,\alpha)}=\diag(\bar A^{(j,\alpha,0)},0,\dots,0)$, where $\bar A^{(j,\alpha,0)}$  is as in \eqref{eq:def.mat.A.j.alpha.0},
the constraint \eqref{eq:localizing.mat.cons.eq} becomes
\begin{equation}
\tr( \bar A^{(j,\alpha)} Z)=\tr(\bar A^{(j,\alpha,0)} Z^{(0)})=0\,,\,\forall \alpha\in\N^n_{2(k-w_j)}\,.
\end{equation}
The number of such $\bar A^{(j,\alpha)}$, for $\alpha\in\N^n_{2(k-w_j)}$, $j=1,\dots,l$, is
\begin{equation}
\sum_{j=1}^l |\N^n_{2(k-w_j)}|=  \sum_{j=1}^l \binom{n+2(k-w_j)}n\,.
\end{equation}
Additionally, we have
\begin{equation}
1=y_0=\frac{z_{0,0}^{(0)}}{p_0^{(0)}p_0^{(0)}}=\tr( \bar A^{(0,0)} Z^{(0)})=\tr( \bar A^{(0)} Z)\,,
\end{equation}
where $\bar A^{(0)}=\diag(\bar A^{(0,0)},0,\dots,0)$, with $\bar A^{(0,0)}$ as in \eqref{eq:def.mat.A.0.0}.

The objective function of the problem \eqref{eq:mom.relax.pop.in.eq.constrains} is expressed as
\begin{equation}
L_y(f)=\sum_{\gamma\in\N^n_{2k}} f_\gamma y_\gamma =\sum_{\gamma\in\N^n_{2k}} f_\gamma  \frac{z^{(0)}_{\alpha_\gamma,\beta_\gamma}}{p_{\alpha_\gamma}^{(0)}p_{\beta_\gamma}^{(0)}} =\tr( C^{(0)} Z^{(0)})=\tr( C Z)\,,
\end{equation}
where $\bar C=\diag(\bar C^{(0)},0,\dots,0)$, with $\bar C^{(0)}$ as in \eqref{eq:def.mat.C.ineq.pop}.
Thus, \eqref{eq:mom.relax.pop.in.eq.constrains} is equivalent to \eqref{eq:sdp.moment.relax.tr.1.ineq.pop}.

\eqref{no.trace.mom.relax.genpop} follows, as the optimal value of \eqref{eq:mom.relax.pop.in.eq.constrains} remains unchanged when the constraint $\tr(PD_k(y)P)\le 1$ is added or removed.

\eqref{norm.C.mom.relax.gen.pop} From the definition of $\bar C$, it follows that
\begin{equation}
\|\bar C\|_F\le \max_{\gamma\in\N^n_{2k}}\frac{2^{k+1}}{p_{\alpha_\gamma}^{(0)}p_{\beta_\gamma}^{(0)}}\|f\|_{\ell_2}=\max_{\gamma\in\N^n_{2k}}\frac{2^{k+1}}{u_{\alpha_\gamma}^{1/2}u_{\beta_\gamma}^{1/2}}\|f\|_{\ell_2}\le 2^{k+1} \|f\|_{\ell_2}\,,
\end{equation}
\begin{equation}
\|\text{vec}(\bar C)\|_{\ell_1}\le \max_{\gamma\in\N^n_{2k}}\frac{2^{k+1}}{p_{\alpha_\gamma}^{(0)}p_{\beta_\gamma}^{(0)}}\|f\|_{\ell_1}=\max_{\gamma\in\N^n_{2k}}\frac{2^{k+1}}{u_{\alpha_\gamma}^{1/2}u_{\beta_\gamma}^{1/2}}\|f\|_{\ell_1}\le 2^{k+1} \|f\|_{\ell_1}\,.
\end{equation}

\eqref{bound.size.mat.mom.relax.gen.pop} follows directly from the definition of $\bar N$: $\bar N=\sum_{i=0}^m\binom{n+k-d_i}n$.

\eqref{bound.num.aff.cons.mom.relax.gen.pop} Recall that the total number of affine constraints in the SDP \eqref{eq:sdp.moment.relax.tr.1.ineq.pop} includes:
\begin{itemize}
\item Constraints from $Z^{(0)}=P^{(0)}M_k(y)P^{(0)}$, which contribute
\begin{equation}
\frac{1}{2}\binom{n+k}n\left[\binom{n+k}n+1\right]-\binom{n+2k}n\,.
\end{equation}
\item Constraints from $Z^{(i)}=P^{(i)}M_{k-d_i}(g_iy)P^{(i)}$ for $i=1,\dots,m$, which contribute
\begin{equation}
\frac{1}2\sum_{i=1}^m \binom{n+k-d_i}n\left[\binom{n+k-d_i}n+1\right]\,.
\end{equation}
\item Constraints arising from $M_{k-w_j}(h_jy)=0$, for $j=1,\dots,l$, which contribute
\begin{equation}
\sum_{j=1}^l \binom{n+2(k-w_j)}n\,.
\end{equation}
\item A single constraint from $y_0=1$.
\end{itemize}
Summing these contributions gives the total number of affine constraints $\tilde M$. 

\end{proof}

\begin{lemma}\label{lem:properties.data.A.ineq.cons}
Let $f\in\R[x]_{2k}$ and $g_i\in\R[x]_{2d_i}$, for $i=1,\dots,m$, $h_j\in\R[x]_{2w_j}$, for $j=1,\dots,l$, with $g_1=1-\|x\|_{\ell_2}^2$.
Assume $k\ge \max\{d_1,\dots,d_m,w_1,\dots,w_l\}$.
Let $C$, $A_i$, $i=1,\dots,M$, $N$, $c$, $\lambda_{\min}$, and $\lambda_{\max}$ be generated by Algorithm \ref{alg:eq.con.pop.ball}.
Define 
\begin{equation}
A =\begin{bmatrix}
\text{vec}(A_1)\\
\dots\\
\text{vec}(A_M)
\end{bmatrix}\,.
\end{equation}
Set $\lambda^\star=\frac{\tau_k}{c}$.
Then, the following properties hold: 
\begin{enumerate}[(i)]
\item\label{bound.eig.val.mom.relax.genpop} $-I_N\preceq C\preceq I_N$, and $\|\text{vec}(C)\|_{\ell_1}\le \frac{2^{k+1}\|f\|_{\ell_1}}c$.
\item\label{A.full.rank.mom.relax.genpop} $A$ has rank $M$, $\sigma_{\min}(A)=1$, and $-I_N\preceq A_i\preceq I_N$, for $i=1,\dots,M$.
\item\label{bound.opt.val.sdp.mom.relax.genpop} $\lambda^\star\in[\lambda_{\min},\lambda_{\max}]$.
\item\label{equivalen.prob.mom.relax.genpop} Problem \eqref{eq:mom.relax.pop.in.eq.constrains} is equivalent to \eqref{eq:sdp}.
\item\label{remove.tr.constrain.mom.relax.genpop}
\begin{equation}\label{eq:no.trace.cons.mom.relax.generalpop}
\lambda^\star=\inf\{\tr( C X)\,:\,X\in\mathbb S^{N}_+\,,\,\tr( A_i X) =b_i\,,\,i=1,\dots,M\}\,.
\end{equation}
Moreover, if $y$ is a feasible solution for problem \eqref{eq:mom.relax.pop.in.eq.constrains}, then for any $\omega\ge 0$, $\diag(PD_k(gy)P,\omega)$ is a feasible solution for \eqref{eq:no.trace.cons.mom.relax.generalpop}, where $P$ is constructed as in Lemma \ref{lem:trce.1.moment.relax.ineq.pop}.

\item\label{bound.c.mom.relax.generalpop} $c\le 2^{k+1} \|f\|_{\ell_2}$.
\item\label{bound.num.aff.cons.mom.relax.generalpop} $M\le \frac{m+1}{2}\binom{n+k}n[\binom{n+k}n+1]+1$, $N\le (m+1)\binom{n+k}n+1$, and \\\\
$\tilde M\le 1+ (l-1) \binom{n+2k}n+\frac{m+1}{2}\binom{n+k}n[\binom{n+k}n+1]$.
\item\label{sparsity.mom.relax.generalpop} The maximum number of nonzero entries in a row of $C,A_i$ is at most $s=\binom{n+k}n$.
\end{enumerate}
\end{lemma}
\begin{proof}
\eqref{equivalen.prob.mom.relax.genpop} By Lemma \ref{lem:properties.tilde.prob.ineq.pop} \eqref{equi.proble.mom.relax.gen.pop}, the problem \eqref{eq:mom.relax.pop.in.eq.constrains} is equivalent to \eqref{eq:sdp.moment.relax.tr.1.ineq.pop}.
Using the definitions of $\tilde C,\tilde A_i, \tilde b_i$, the problem \eqref{eq:sdp.moment.relax.tr.1.ineq.pop} can be reformulated as:
\begin{equation}\label{eq:sdp.trace.one.cons.ineq.pop}
\begin{array}{rl}
\tau_k=\inf\limits_{X} & \tr( \tilde C X)\\
\text{s.t.}& X\in\mathbb S_+^{N}\,,\,\tr(X)= 1\,,\\
&\tr( \tilde A_i X)=\tilde b_i\,,\,i=1,\dots,\tilde M\,.
\end{array}
\end{equation}
Applying Lemma \ref{lem:Gaussian.elimination}, the matrices $\check A_1,\dots,\check A_M$ are linearly independent in $\mathbb S^N$, and
\begin{equation}\label{eq:same.polytope.ineq.pop}
\tr( \tilde A_i Z)=\tilde b_i\,,\,i=1,\dots,\tilde M \Leftrightarrow \tr( \check A_iZ)=b_i\,,\,i=1,\dots,M\,.
\end{equation}
Further, by Lemma \ref{lem:Gram-Schmidt}, the matrices $A_1,\dots,A_M$ satisfy $\tr(  A_i A_j)=\delta_{ij}$, and and the equivalence holds:
 \begin{equation}\label{eq:same.polytope.ineq.pop2}
\tr( \check A_i X)=\check b_i\,,\,i=1,\dots,M \Leftrightarrow \tr(  A_i X)= b_i\,,\,i=1,\dots,M
\end{equation}
This implies that problem \eqref{eq:sdp.trace.one.cons.ineq.pop} is equivalent to
\begin{equation}\label{eq:sdp.trace.one.ineq.pop.ball}
\begin{array}{rl}
\tau_k=\inf\limits_{X} & \tr( \tilde C X)\\
\text{s.t.}& X\in\mathbb S_+^{N}\,,\,\tr(X)= 1\,,\\
&\tr( A_i X)=b_i\,,\,i=1,\dots,M\,.
\end{array}
\end{equation}
Rewriting this problem using $C=\frac{\tilde C}{c}$, we obtain:
\begin{equation}\label{eq:sdp.trace.one.ineq.pop.ball2}
\begin{array}{rl}
\frac{\tau_k}{c}=\inf\limits_{X} & \tr( C X)\\
\text{s.t.}& X\in\mathbb S_+^{N}\,,\,\tr(X)= 1\,,\\
&\tr( A_i X)=b_i\,,\,i=1,\dots,M\,.
\end{array}
\end{equation}

With $\lambda^\star=\frac{\tau_k}{c}$, this problem is of the form \eqref{eq:sdp}.

\eqref{remove.tr.constrain.mom.relax.genpop} follows from Lemma \ref{lem:properties.tilde.prob.ineq.pop} \eqref{no.trace.mom.relax.genpop}.

\eqref{A.full.rank.mom.relax.genpop} Since $\tr(  A_i A_j)=\delta_{ij}$, $-I_N\preceq A_i\preceq I_N$, $A$ has rank $M$, and $\sigma_{\min}(A)=1$.

\eqref{bound.c.mom.relax.generalpop} derives from the fact that $c=\|\tilde C\|_F=\|\bar C\|_F$ (from Lemma \ref{lem:properties.tilde.prob.ineq.pop} \eqref{norm.C.mom.relax.gen.pop}).

\eqref{bound.eig.val.mom.relax.genpop} Consequently, with $C=\frac{\tilde C}{c}=\frac{\tilde C}{\|\tilde C\|_F}$, it follows that $-I_N\preceq C\preceq I$ and by Lemma \ref{lem:properties.tilde.prob.ineq.pop} \eqref{norm.C.mom.relax.gen.pop},
\begin{equation}
\|\text{vec}(C)\|_{\ell_1}=\frac{\|\text{vec}(\tilde C)\|_{\ell_1}}{c}=\frac{\|\text{vec}(\bar C)\|_{\ell_1}}{c}\le \frac{2^{k+1}\|f\|_{\ell_1}}{c}\,.
\end{equation}

\eqref{bound.opt.val.sdp.mom.relax.genpop} Since $\lambda_{\max}=\frac{f(a)}{c}$ for $a\in S(g)\cap V(h)$, and $\lambda_{\min}=\frac{\underline \tau}{c}$ and $\underline \tau\le \tau_k\le f^\star\le f(a)$, it follows that $\lambda^\star\in[\lambda_{\min},\lambda_{\max}]$.

\eqref{bound.num.aff.cons.mom.relax.generalpop} is obtained from the bound $M\le \min\{\tilde M,L\}$ with $N=\bar N+1$, $\tilde M=\bar M+1$,
\begin{equation}
L=\frac{1}{2}\sum_{i=0}^m \binom{n+k-d_i}n\left[\binom{n+k-d_i}n+1\right]+1\,,
\end{equation}
and Lemma \ref{lem:properties.tilde.prob.ineq.pop} \eqref{bound.size.mat.mom.relax.gen.pop} and  \eqref{bound.num.aff.cons.mom.relax.gen.pop}.

\eqref{sparsity.mom.relax.generalpop} The diagonal-block structure of $C,A_i$ yields \eqref{sparsity.mom.relax.generalpop}.
\end{proof}
\begin{remark}
The maximum number of nonzero entries in a row of 
$\tilde C, \tilde A_i$ is at most:
\begin{equation}
\begin{array}{rl}
&\max\limits_{\arraycolsep=1.4pt\def\arraystretch{.7}
\begin{array}{cc}
\scriptstyle i=1,\dots,m,\\
\scriptstyle j=1,\dots,l\\
\end{array}}\{|\supp(f)|,|\supp(g_i)|,|\supp(h_j)|\}\\
\le &\max\limits_{\arraycolsep=1.4pt\def\arraystretch{.7}
\begin{array}{cc}
\scriptstyle i=1,\dots,m,\\
\scriptstyle j=1,\dots,l\\
\end{array}}\{|\N^n_{2d}|,|\N^n_{2d_i}|,|\N^n_{2w_j}|\}\\
=&\max\limits_{\arraycolsep=1.4pt\def\arraystretch{.7}
\begin{array}{cc}
\scriptstyle i=1,\dots,m,\\
\scriptstyle j=1,\dots,l\\
\end{array}}\left\{\binom{n+2d}n,\binom{n+2d_i}n,\binom{n+2w_j}n\right\}\,.
\end{array}
\end{equation}
However, the application of Gaussian elimination and the Gram-Schmidt process transforms each $A_i$ into a linear combination of the $\tilde A_i$  matrices. As a result, the sparsity structure of $\tilde A_i$ is generally lost.
\end{remark}

\begin{theorem}\label{theo:accuracy.moment.relax.ineqpop.ball}
Let $f\in\R[x]_{2k}$ and $g_i\in\R[x]_{2d_i}$, for $i=1,\dots,m$, $h_j\in\R[x]_{2w_j}$, for $j=1,\dots,l$, with $g_1=1-\|x\|_{\ell_2}^2$.
Assume $k\ge \max\{d,d_1,\dots,d_m,w_1,\dots,w_l\}$, $\tau_k\ge 0$, and Slater's condition holds for the  moment relaxation \eqref{eq:mom.relax.pop.in.eq.constrains}, i.e., there exists a feasible solution  $\bar y$ for problem \eqref{eq:mom.relax.pop.in.eq.constrains} such that $D_k(g\bar y)\succ 0$ and has condition number $\kappa $.
Let $\tau_k^{(\varepsilon)}$ be the value returned by Algorithm \ref{alg:eq.con.pop.ball}.
Then, the following inequality holds:
\begin{equation}
\begin{array}{rl}
0\le& \tau_k^{(\varepsilon)}-\tau_k\\
\le& \varepsilon 2^{k+1}\|f\|_{\ell_1} \left[2+(1+\sqrt{2} \kappa r_{g,n,k}) \sqrt{\frac{m+1}{2}\binom{n+k}n[\binom{n+k}n+1]+1}\ \right]\,,
\end{array}
\end{equation}
where $r_{g,n,k}=\frac{\max\Lambda_{g,n,k}}{\min\Lambda_{g,n,k}}$ with $\Lambda_{g,n,k}=\{p_\alpha^{(i)2}\,:\,i=0,\dots,m\,,\,\alpha\in\N^n_{k-d_i}\}$
and $p_\alpha^{(i)}$ defined as in Step \ref{step:convert.to.sdp.mom.relax.genpop}a-d of Algorithm \ref{alg:eq.con.pop.ball}.
\end{theorem}
\begin{proof}
Let $\lambda^\star=\frac{\tau_k}{c}$. Since Slater's condition holds for the  moment relaxation \eqref{eq:mom.relax.pop.in.eq.constrains}, strong duality is satisfied for problem  \eqref{eq:mom.relax.pop.in.eq.constrains} and its dual.
By Lemma \ref{lem:properties.data.A.eq.cons} \eqref{sdp.without.trace.const}, strong duality also holds for problem \eqref{eq:no.trace.cons.mom.relax.generalpop} and its dual \eqref{eq:dual.no.tr.cons}.

Define $U_{n,k}=\diag(PD_k(g\bar y)P,\xi)$, where $P$ is constructed as in Lemma \ref{lem:trce.1.moment.relax.ineq.pop} and $\xi=\lambda_{\min}(PD_k(g\bar y)P)$. Since $P \succ 0$, it follows that $U_{n,k}\succ 0$ and has condition number
\begin{equation}\label{eq:bound.on.kappa}
\kappa(U_{n,k})=\kappa(PD_k(g\bar y)P)\le \kappa r_{g,n,k}\,.
\end{equation}
By Lemma \ref{lem:properties.data.A.ineq.cons} \eqref{remove.tr.constrain.mom.relax.genpop}, $U_{n,k}$ is a feasible solution for problem \eqref{eq:no.trace.cons.mom.relax.generalpop}, and thus Slater's condition holds for  problem \eqref{eq:no.trace.cons.mom.relax.generalpop}.
From Lemma \ref{lem:properties.data.A.ineq.cons} \eqref{A.full.rank.mom.relax.genpop}, $A$ has full rank $M$.
Since $\tau_k\ge 0$, we have $\lambda^\star=\frac{\tau_k}{c}\ge 0$.
By Lemma \ref{lem:properties.data.A.ineq.cons} \eqref{bound.eig.val.mom.relax.genpop}, \eqref{A.full.rank.mom.relax.genpop} and \eqref{bound.opt.val.sdp.mom.relax.genpop}, $I_N\succeq C\succeq -I_N$, $I_N\succeq A_i\succeq -I_N$, for $i=1,\dots,M$, and  $\lambda^\star\in [\lambda_{\min},\lambda_{\max}]$.
Using Lemma \ref{lem:convergence.sdp.2} (ii), we obtain
\begin{equation}
0\le \lambda^\star-\overline \lambda_T\le \varepsilon\left[2+(1+\sqrt{2} \kappa(U_{n,k})) \|\text{vec}(C)\|_{\ell_1} \sigma_{\min}(A)^{-1}\sqrt{M} \right]\,.
\end{equation}
Applying properties of $C$ and $A$ from Lemma \ref{lem:properties.data.A.ineq.cons} \eqref{bound.eig.val.mom.relax.genpop}  and  \eqref{A.full.rank.mom.relax.genpop},  this simplifies to
\begin{equation}
0\le \lambda^\star-\overline \lambda_T\le \varepsilon\left[2+(1+\sqrt{2} \kappa r_{g,n,k})\frac{2^{k+1}\|f\|_{\ell_1}}{c}  \sqrt{M}  \right]\,.
\end{equation}
Since $\tau_k^{(\varepsilon)}=\underline\lambda_T c$ and $\tau_k=\lambda^\star c$, we have 
\begin{equation}
\begin{array}{rl}
0\le \tau_k^{(\varepsilon)}-\tau_k=&c(\lambda^\star-\underline \lambda_T)\\[5pt]
\le &\varepsilon\left[2c+(1+\sqrt{2} \kappa r_{g,n,k})2^{k+1}\|f\|_{\ell_1}  \sqrt{M}  \right]\\[10pt]
\le& \varepsilon \left[2^{k+2}\|f\|_{\ell_2}+(1+\sqrt{2} \kappa r_{g,n,k})2^{k+1}\|f\|_{\ell_1} \sqrt{M} \right]\\[10pt]
\le& \varepsilon 2^{k+1}\|f\|_{\ell_1} \left[2+(1+\sqrt{2} \kappa r_{g,n,k}) \sqrt{M} \right]\,.
\end{array}
\end{equation}
The third inequality is due to Lemma \ref{lem:properties.data.A.ineq.cons} \eqref{bound.c.mom.relax.generalpop}.
Using $M\le \frac{m+1}{2}\binom{n+k}n[\binom{n+k}n+1]+1$,  we conclude the result. 
\end{proof}

\begin{lemma}[Complexity]
Let $k\ge \max\{d,d_1,\dots,d_m,w_1,\dots,w_l\}$.
The time complexity to run Algorithm \ref{alg:eq.con.pop.ball} on a classical computer is:
\begin{equation}
O\left(m\binom{n+k}n^2\left[(l-1) \binom{n+2k}n+m^2\binom{n+k}n^4 + m^2\binom{n+k}n^2\operatorname{poly}\left(\frac{1}\varepsilon\right)\right]\right)\,,
\end{equation}
while on a quantum computer, the time complexity is:
\begin{equation}
\footnotesize
O\left(\binom{n+k}{n}^{1.5}\left\{m(l-1) \binom{n+k}{n}^{0.5}\binom{n+2k}n +m^3\binom{n+k}{n}^{4.5}+m^{0.5}\left[\binom{n+k}{n}^{0.5}+\frac{1}{\varepsilon}\right]\frac{1}{\varepsilon^4}\right\}\right)\,.
\end{equation}
\end{lemma}
\begin{proof}
The most computationally intensive parts of Algorithm \ref{alg:eq.con.pop.ball} include Gaussian elimination in Step \ref{step:gauss.elimination.mom.relax.genpop}, the Gram-Schmidt process in Step \ref{step:gram.schmidt.mom.relax.genpop}, and the binary search with Hamiltonian updates in Step \ref{step:hamilnion.updates.mom.relax.genpop}.

For Gaussian elimination, by Lemma \ref{lem:complexity.Gaussian.elimination}, both the classical and quantum computers have the same time complexity of $O(\tilde M L^2)$, with 
\begin{equation}
L=\frac{1}{2}\sum_{i=0}^m \binom{n+k-d_i}n\left[\binom{n+k-d_i}n+1\right]+1\,.
\end{equation}
This implies $L\le \frac{m+1}{2}\binom{n+k}n[\binom{n+k}n+1]+1$.

For Gram-Schmidt process, by Lemma \ref{lem:complexity.Gram-Schmidt}, both classical and quantum computers have the same time complexity of $O(LM^2)$.

Regarding the binary search with Hamiltonian updates, by Lemma, by Lemma \ref{lem:complex.binary.search}, the classical computer requires $O((MsN+N^\omega)\varepsilon^{-2})$ operations, and the quantum computer requires $O(s(\sqrt{M}+\sqrt{N}\varepsilon^{-1})\varepsilon^{-4})$ operations, where $s$ is the maximum number of nonzero entries in any row of the data matrices $C,A_i$, for $i=1,\dots,M$.
By Lemma \ref{lem:properties.data.A.ineq.cons} \eqref{sparsity.mom.relax.generalpop}, we have $s=\binom{n+k}n$.

Additionally, by Lemma \ref{lem:properties.data.A.ineq.cons} \eqref{bound.num.aff.cons.mom.relax.generalpop}, we know that $M\le \frac{m+1}2\binom{n+k}n[\binom{n+k}n+1]+1$, $N\le (m+1)\binom{n+k}n$ and 
$\tilde M\le 1+ (l-1) \binom{n+2k}n+\frac{m+1}{2}\binom{n+k}n[\binom{n+k}n+1]$.
Thus, the result follows.
\end{proof}

\section{Positive lower bounds for the smallest eigenvalues of moment and localizing matrices associated with a positive measure}
\label{sec:lower.bound.eigenval.moment.mat}

\begin{lemma}\label{lem:lower.bound.smallest.eigenval}
Let $A\in \mathbb S^N_+$. 
Then
\begin{equation}
\lambda_{\min}(A)\ge \left( \frac{N-1}{\tr(A)}\right)^{N-1}\det(A)\,.
\end{equation}
\end{lemma}
\begin{proof}
Let $\lambda_1\ge \lambda_2\ge \dots\ge \lambda_N$ be the eigenvalues of $A$, so that $\lambda_{\min}(A)=\lambda_N$. 
Since $A\succeq 0$, we have $\lambda_N\ge 0$.
By the arithmetic-geometric mean inequality, this implies
\begin{equation}
\begin{array}{rl}
\det(A)=\lambda_1\dots \lambda_{N-1}\lambda_N\le& \left(\frac{\lambda_1+\dots+\lambda_{N-1}}{N-1} \right)^{N-1}\lambda_N\\[5pt]
\le&\left(\frac{\lambda_1+\dots+\lambda_{N-1}+\lambda_N}{N-1} \right)^{N-1}\lambda_N\\[5pt]
=& \left(\frac{\tr(A)}{N-1} \right)^{N-1}\lambda_N\,.
\end{array}
\end{equation}
Hence the result follows.
\end{proof}

\begin{lemma}\label{lem:det.larger.1}
Let $A\in \mathbb Z^{N\times N}$ be a symmetric positive definite matrix ($A\succ 0$). Then $\det(A)\ge 1$.
\end{lemma}
\begin{proof}
Let $A=(a_{ij})_{i,j=1,\dots,N}$. By definition, the determinant of $A$ is expressed as:
\begin{equation}
\det(A)=\sum_{\sigma\in S^N}\sign(\sigma) a_{1\sigma(1)}\dots a_{N\sigma(N)}\,,
\end{equation}
where $S^N$ is the symmetric group of all bijections $\sigma:\{1,\dots,N\}\to \{1,\dots,N\}$, and $\sign(\sigma)$ is the signature of $\sigma$, which equals  $1$ for even permutations and $-1$ for odd permutations.
Since $a_{ij}\in\mathbb Z$ for all $i,j$, it follows that $\det(A)\in\mathbb Z$.
Additionally, because $A\succ 0$, we know $\det(A)>0$. Combining these facts, we conclude $\det(A)\ge 1$.
\end{proof}

\begin{proposition}\label{prop:lower.bound.eigen.val}
Let $\mu$ be the Lebesugue measure on $[0,1]^n$.
Assume that $n\ge 2$.
Then, the smallest eigenvalue of the moment matrix $M_k(\mu)$ satisfies:
\begin{equation}
\lambda_{\min}(M_k(\mu))\ge [(2k+1)!]^{-n \binom{n+k}n}\,.
\end{equation}
\end{proposition}
\begin{proof}
Define $\nu=C_k\mu$ for a constant $C_k>0$, and let $y=(y_\alpha)_{\alpha\in\N^n_{2k}}$ be the truncated moment sequence of $\nu$.
The moment matrix $M_k(\nu)$ is given by $M_k(\nu)=M_k(y)=(y_{\alpha+\beta})_{\alpha,\beta\in\N^n_k}$.

We first show that $M_k(\nu)\succ 0$.
Consider any $u\in\R^{\binom{n+k}n}\backslash \{0\}$.
Then,
\begin{equation}
u^\top M_k(\nu) u=\int u^\top v_kv_k^\top ud\nu=\int (u^\top v_k)^2d\mu=C_k\int_{[0,1]^n} p^2dx\ge 0\,,
\end{equation}
where $p=u^\top v_k$
This implies that $M_k(\nu)\succeq 0$.
Suppose, for contradiction, that $u^\top M_k(\nu) u=0$. Then, $\int_{[0,1]^n} p^2dx=0$
Since $p^2\ge 0$, this implies $p=0$ on $[0,1]^n$, so that $p$ is the zero polynomial. Thus, $u=0$, a contradiction. Hence, $M_k(\nu)\succ 0$.

Next, calculate $y_\alpha$:
\begin{equation}\label{eq:moment.lebesgue.on.0.1}
\begin{array}{rl}
y_\alpha=C_k\int_{[0,1]^n} x^\alpha dx=&C_k\prod\limits_{i=1}^n \int_{0}^1 x_i^{\alpha_i} dx_i\\[10pt]
=&C_k\prod\limits_{i=1}^n\frac{1}{\alpha_i+1}=\frac{C_k}{(\alpha_1+1)\dots(\alpha_n+1)}\,.
\end{array}
\end{equation}
To ensure $y_\alpha\in\N$ for all $\alpha\in\N^n_{2k}$ choose $C_k=[(2k+1)!]^n$. Since $\alpha_i+1\in\{1,\dots,2k+1\}$ for $\alpha\in\N^n_{2k}$,
\begin{equation}
\frac{C_k}{\prod_{i=1}^n(\alpha_i+1)}=\prod_{i=1}^n \frac{(2k+1)!}{\alpha_i+1}\in\N\,.
\end{equation}
By Lemma \ref{lem:det.larger.1}, $\det(M_k(\nu))\ge 1$.
To estimate $\tr(M_k(\nu))$, note
 \begin{equation}
 \tr(M_k(\nu))=\sum_{\alpha\in\N^n} y_{2\alpha}=\sum_{i=1}^n y_{2e_i}+\sum_{\alpha\in\N^n_k\backslash \{e_i\}_{i=1}^n} y_{2\alpha}\,,
 \end{equation}
 where $\{e_i\}_{i=1}^n$ is the standard basis in $\R^n$.
 Then,
 \begin{equation}
 \begin{array}{rl}
 \tr(M_k(\nu))\le& \frac{nC_k}{3}+C_k\left[ \binom{n+k}n -n\right]\\[10pt]
 =&C_k\left[ \binom{n+k}n-\frac{2n}{3}\right]\\[10pt]
 \le& C_k\left[ \binom{n+k}n-1\right]\,,
 \end{array}
 \end{equation}
where the last inequality holds because $n\ge 2$.
Using Lemma \ref{lem:lower.bound.smallest.eigenval}, 
\begin{equation}
\begin{array}{rl}
C_k\lambda_{\min}(M_k(\mu))=&\lambda_{\min}(C_k M_k(\mu))=\lambda_{\min}( M_k(C_k\mu))\\[7pt]
=&\lambda_{\min}(M_k(\nu))\\[7pt]
\ge &\left(\frac{\binom{n+k}n-1}{\tr(M_k(\nu))} \right)^{\binom{n+k}n-1}\det(M_k(\nu))\,.
\end{array}
\end{equation}
Since $\det(M_k(\nu))\ge 1$ and $\tr(M_k(\nu))\le C_k\left[ \binom{n+k}n-1\right]$, 
\begin{equation}
\lambda_{\min}(M_k(\nu))\ge C_k^{-\binom{n+k}n+1}\,.
\end{equation}
Thus, $\lambda_{\min}(M_k(\mu))\ge C_k^{-\binom{n+k}n}$.
This completes the proof.
\end{proof}

\begin{remark}
Using the Stirling approximation $m! \sim m^me^{-m}\sqrt{2\pi m}$, we have:
\begin{equation}
C_k=[(2k+1)!]^n\sim (2k+1)^{(2k+1)n} e^{-(2k+1)n}[2\pi(2k+1)]^{\frac{n}{2}}\,,
\end{equation}
which grows rapidly as $n$ and $k$ increase.

To improve the bound, we can instead define $C_k$ as the least common multiple (LCM):
\begin{equation}
L_{k,n}=\text{LCM}\left(\left\{\prod_{i=1}^n(\alpha_i+1)\,:\,\alpha\in\N^n_{2k}\right\}\right)\,.
\end{equation}
Using this approach provides a tighter positive lower bound for $\lambda_{\min}(M_k(\mu))$.

For $n=1$, $L_{k,1}$ is asymptotically equivalent to the exponential of the Chebyshev function for $k$.
However, the asymptotic behavior of $L_{k,n}$ for $n\ge 2$ remains unknown.
\end{remark}
Let $\mathbb K[x]_d$  denote the set of polynomials with coefficients in $\mathbb K$ of degree at most $d$, where $\mathbb K$ is $\mathbb Z$ or $\mathbb Q$.
\begin{proposition}\label{prop:lower.bound.eigenval.constrained}
Let $\mu$ be the Lebesugue measure on $[0,1]^n$, and let $g_i \in\mathbb Q[x]_{2d_i}$, for $i=0,\dots,m$, such that $g_i>0$ on $[0,1]^n$, for $i=1,\dots,m$.
Let $t_i$ be the smallest positive integer such that $t_ig_i\in \mathbb Z[x]_{2d_i}$.
Then, 
\begin{equation}
\lambda_{\min}(D_k(\mu g))\ge \frac{1}{[(2k+1)!]^n}\min_{i=0,\dots,m} \frac{1}{t_i}\left( \frac{\binom{n+k-d_i}{n}-1}{[(2k+1)!]^n \binom{n+k-d_i}nt_i\|g_i\|_{\ell_1}}\right)^{\binom{n+k-d_i}n-1}\,.
\end{equation}
\end{proposition}
\begin{proof}
Define $\nu=C_k\mu$ with $C_k=[(2k+1)!]^n$.
Let $y=(y_\alpha)_{\alpha\in\N^n_{2k}}$ denote the truncated moment sequence of $\nu$.
Define $\tilde g_i=t_i g_i$.
Since $\tilde g_i=\sum_{\gamma\in\N^n_{2d_i}} g_{i,\gamma} x^\gamma\in\mathbb Z[x]_{2d_i}$ and $y\subset \mathbb Z$ (as shown in the proof of Proposition \ref{prop:lower.bound.eigen.val}), it follows that
\begin{equation}
\begin{array}{rl}
M_{k-d_i}(\tilde g_i\nu)=M_{k-d_i}(\tilde g_iy)=&\left(\sum_{\gamma\in\N^n_{2d_i}}t_ig_{i,\gamma} y_{\alpha+\beta+\gamma}\right)_{\alpha,\beta\in\N^n_{k-d_i}}\\[10pt]
&\in\mathbb Z^{\binom{n+k-d_i}n\times \binom{n+k-d_i}n}\,.
\end{array}
\end{equation} 
We claim that
$M_{k-d_i}(\tilde g_i\nu)\succ 0$.
Indeed, we take $w\in\R^{\binom{n+k-d_i}n}\backslash \{0\}$.
Since $\tilde g_i=t_ig_i>0$ on $[0,1]^n$,
\begin{equation}
\begin{array}{rl}
w^\top M_{k-d_i}(\tilde g_i\nu)w=&w^\top\left(\int \tilde g_i  v_{k-d_i}v_{k-d_i}^\top  d\nu\right)w\\[10pt]
=&\int \tilde g_i w^\top v_{k-d_i}v_{k-d_i}^\top w d\nu\\[10pt]
 =&C_k\int_{[0,1]^n} \tilde g_ip^2dx\ge 0\,,
\end{array}
\end{equation}
where $p=w^\top v_{k-d_i}$.
This implies $ M_{k-d_i}(\tilde g_i\nu)\succeq 0$.
 Assume by contradiction that $0=w^\top M_{k-d_i}(\tilde g_i\nu)w$. Then,
\begin{equation}
0=w^\top\left(\int \tilde g_i  v_{k-d_i}v_{k-d_i}^\top  d\nu\right)w=\int \tilde g_i w^\top v_{k-d_i}v_{k-d_i}^\top w d\nu =C_k\int_{[0,1]^n} \tilde g_ip^2dx\,.
\end{equation}
It implies that $\tilde g_ip^2=0$ on $[0,1]^n$, which gives  $p=0$ on $[0,1]^n$ since $\tilde g_i=t_ig_i>0$ on $[0,1]^n$. Thus, $p$ is zero polynomial, which contradicts $w\ne 0$.
By Lemma \ref{lem:det.larger.1}, $\det(M_{k-d_i}(\nu \tilde g_i))\ge 1$.
To bound $\tr(M_{k-d_i}(\nu \tilde g_i))$, note:
\begin{equation}
\begin{array}{rl}
\tr(M_{k-d_i}(\nu \tilde g_i))=&\sum_{\alpha\in\N^n_{k-d_i}} \sum_{\gamma\in\N^n_{2d_i}} t_ig_{i,\gamma} y_{2\alpha+\gamma}\\[7pt]
\le& \sum_{\alpha\in\N^n_{k-d_i}} t_i\sum_{\gamma\in\N^n_{2d_i}} |g_{i,\gamma}| y_{2\alpha+\gamma}\\[7pt]
\le& C_k\sum_{\alpha\in\N^n_{k-d_i}} t_i\| g_i\|_{\ell_1} \\[7pt]
\le& C_k\binom{n+k-d_i}n t_i\|g_i\|_{\ell_1} \,.
\end{array}
\end{equation}
The second inequality follows from $0\le y_{2\alpha+\gamma}\le C_k$.

Next, we estimate a lower bound for $\lambda_{\min}(M_{k-d_i}(\nu g_i))$.
Using Lemma \ref{lem:lower.bound.smallest.eigenval}, 
\begin{equation}
\begin{array}{rl}
\lambda_{\min}(M_{k-d_i}(\nu g_i))= &\lambda_{\min}(M_{k-d_i}(\nu \frac{1}{t_i}\tilde g_i))\\
= &\lambda_{\min}(\frac{1}{t_i}M_{k-d_i}(\nu \tilde g_i))\\
= &\frac{1}{t_i}\lambda_{\min}(M_{k-d_i}(\nu \tilde g_i))\\
\ge &\frac{1}{t_i}\left( \frac{\binom{n+k-d_i}{n}-1}{\tr(M_{k-d_i}(\nu \tilde g_i))}\right)^{\binom{n+k-d_i}n-1}\\
\ge &\frac{1}{t_i}\left( \frac{\binom{n+k-d_i}{n}-1}{[(2k+1)!]^n \binom{n+k-d_i}nt_i\|g_i\|_{\ell_1}}\right)^{\binom{n+k-d_i}n-1}\,.
\end{array}
\end{equation}
Let us find a positive lower bound for $\lambda_{\min}(D_k(\mu g))$.
Since 
\begin{equation}
D_k(\nu g)=\diag((M_{k-d_i}(\nu g_i))_{i=0}^m)\,,
\end{equation}
we have
\begin{equation}
\begin{array}{rl}
C_k\lambda_{\min}(D_k(\mu g))=&\lambda_{\min}(C_k D_k(\mu g))=\lambda_{\min}( D_k((C_k\mu) g))\\[10pt]
=&\lambda_{\min}(D_k(\nu g))=\min\limits_{i=0,\dots,m}\lambda_{\min}(M_{k-d_i}(\nu g_i))\\[10pt]
\ge &\min\limits_{i=0,\dots,m} \frac{1}{t_i}\left( \frac{\binom{n+k-d_i}{n}-1}{[(2k+1)!]^n \binom{n+k-d_i}nt_i\|g_i\|_{\ell_1}}\right)^{\binom{n+k-d_i}n-1}\,,
\end{array}
\end{equation}
yielding the result.
\end{proof}

\section{Basic algorithms}
\subsection{Gaussian elimination}

\begin{algorithm}\label{alg:Gaussian.elimination}
Gaussian elimination
\begin{itemize}
\item Input: $\tilde A_i\in \prod_{j=1}^r\mathbb S^{N_j}$, $\tilde b_i\in\R$, for $i=1,\dots,\tilde M$.
\item Output: $\check A_i\in \prod_{j=1}^r\mathbb S^{N_j}$, $\check b_i\in\R$, for $i=1,\dots,M$.
\end{itemize}
\begin{enumerate}
\item Convert each $\tilde A_i$  into a row vector: 
\begin{equation}
a_i=\begin{bmatrix}
a_{11}&a_{12}&\dots&a_{1L}
\end{bmatrix}=\text{Vec}(\tilde A_i)^\top\,,\,b_i=\tilde b_i\,,\text{ for } i=1,\dots,\tilde M\,.
\end{equation}
\item Construct the augmented matrix 
$[A|b]$ by concatenating the rows of 
$a_i$  and $b_i$:
\begin{equation}
[A|b]=\begin{bmatrix}
a_{1}&b_1\\
a_{2}&b_2\\
\dots&\dots\\
a_{\tilde M}&b_{\tilde M}\\
\end{bmatrix}=\begin{bmatrix}
a_{11}&a_{12}&\dots&a_{1L}&b_1\\
a_{21}&a_{22}&\dots&a_{2L}&b_2\\
.&.&\dots&.&.\\
a_{\tilde M1}&a_{\tilde M2}&\dots&a_{\tilde ML}&b_{\tilde M}\\
\end{bmatrix}\,.
\end{equation}
\item Initialize $m=0$ and $l=0$.

\item Increment $m$ by one unit.

\item Increment $l$ by one unit.

\item If $l>L$, stop the algorithm. Otherwise, proceed to the next step.

\item If $a_{il}=0$, for all $i=m,\dots ,\tilde M$, go back to Step 5. Otherwise, proceed to the next step.

\item Interchange the $m$-th row with any row $i>m$ in $[A|b]$ such that $a_{il}
\ne 0$ (if $i=m$ there is no need to perform an interchange).

\item For $i=m+1,\dots ,\tilde M$, add $-a_{il}/a_{ml}$ times the $m$-th row to the $i$-th row of the augmented matrix $[A|b]$.

\item If $m<\tilde M-1$, return to Step 4. 
Otherwise, proceed to the next step.
\item Set
\begin{equation}
M=\max\{i\in\{1,\dots,\tilde M\}\,:\,\begin{bmatrix}a_i&b_i\end{bmatrix}\ne 0\}\,.
\end{equation}
Then, extract the matrices: $\check A_i=\text{Mat}(a_i)$, $\check b_i=b_i$, for $i=1,\dots,M$.
\end{enumerate}
\end{algorithm}

\begin{lemma}
\label{lem:Gaussian.elimination}
Let $\tilde A_i\in \prod_{j=1}^r\mathbb S^{N_j}$, $\tilde b_i\in\R$, $i=1,\dots,\tilde M$, such that there exists $\bar X\in\mathbb S^N$ with $N=\sum_{i=1}^r N_i$ satisfying $\tr( \tilde A_i \bar X)=\tilde b_i$, $i=1,\dots,\tilde M$.
Let $\check A_i\in \prod_{j=1}^r\mathbb S^{N_j}$, $\check b_i\in\R$, $i=1,\dots,M$, be the outputs of Algorithm \ref{alg:Gaussian.elimination}.
Then, the set $\check A_1,\dots,\check A_M$ is linearly independent in $\prod_{j=1}^r\mathbb S^{N_j}$, with $M\le \min\{\tilde M, L\}$, where $L=\frac{1}{2}\sum_{i=1}^rN_i(N_i+1)$, and the equivalence
\begin{equation}
\tr( \check A_i X)=\check b_i\,,\,i=1,\dots,M \Leftrightarrow \tr( \tilde A_i X)=\tilde b_i\,,\,i=1,\dots,\tilde M
\end{equation}
holds.
\end{lemma}

\begin{lemma}\label{lem:complexity.Gaussian.elimination}
To execute Algorithm \ref{alg:Gaussian.elimination}, both classical and quantum computers exhibit the same complexity of $O(\tilde M L^2)$, where $L=\frac{1}{2}\sum_{i=1}^rN_i(N_i+1)$.
\end{lemma}

\subsection{Gram-Schmidt process}
\label{sec:Gram-Schmidt}

\begin{algorithm}\label{alg:Gram-Schmidt}
Gram-Schmidt process
\begin{itemize}
\item Input: $\check A_i\in \prod_{j=1}^r\mathbb S^{N_j}$, $\check b_i\in\R$, for $i=1,\dots,M$.
\item Output: $A_i\in \prod_{j=1}^r\mathbb S^{N_j}$, $b_i\in\R$, for $i=1,\dots,M$.
\end{itemize}
\begin{enumerate}
\item Initialize the first pair as $\hat A_1=\check A_1$, $\hat b_1=\check b_1$, $A_1=\frac{\hat A_1}{\|\hat A_1\|_F}$, $b_1=\frac{\hat b_1}{\|\hat A_1\|_F}$.
\item For $j=2,\dots,M$:
\begin{enumerate}
\item Orthogonalize $\hat A_j$ by subtracting the projections onto the previous $\hat A_i$s
\begin{equation}
\hat A_j=\check A_j-\sum_{i=1}^{j-1} \frac{\tr( \check A_j \hat A_i)}{\|\hat A_i\|_F^2}\hat A_i\,.
\end{equation}
\item Adjust the corresponding $\hat b_j$ by the same projections:
\begin{equation}
\hat b_j=\check b_j-\sum_{i=1}^{j-1} \frac{\tr( \check A_j \hat A_i)}{\|\hat A_i\|_F^2}\hat b_i\,.
\end{equation}
\item Normalize the result: $A_j=\frac{\hat A_j}{\|\hat A_i\|_F}$, $b_j=\frac{\hat b_j}{\|\hat A_j\|_F}$.
\end{enumerate}
\end{enumerate}
\end{algorithm}

\begin{lemma}
\label{lem:Gram-Schmidt}
Let $\check A_i\in \prod_{j=1}^r\mathbb S^{N_j}$ and $\check b_i\in\R$, for $i=1,\dots,M$, such that $\check A_1,\dots,\check A_M$ are linearly independent in $\prod_{j=1}^r\mathbb S^{N_j}$.
Let $A_i\in \prod_{j=1}^r\mathbb S^{N_j}$ and $b_i\in\R$, for $i=1,\dots,M$ be matrices and real numbers returned by Algorithm \ref{alg:Gram-Schmidt}.
Then, $\tr( A_i A_j)=\delta_{ij}$, and we have equivalent
\begin{equation}
\tr( \check A_i X)=\check b_i\,,\,\text{for }i=1,\dots,M \Leftrightarrow \tr( A_i X)=b_i\,,\,\text{for }i=1,\dots,M\,.
\end{equation}
\end{lemma}

\begin{lemma}\label{lem:complexity.Gram-Schmidt}
To execute Algorithm \ref{alg:Gram-Schmidt}, both classical computer and quantum computer have the same complexity of $O(LM^2)$, where $L=\frac{1}{2}\sum_{i=1}^rN_i(N_i+1)$.
\end{lemma}
\section{Useful lemmas}

\begin{lemma}\label{lem:cayley-haminton}
Let $a,b,c\in (0,\infty)$. Define square matrix of size $n+1$
\begin{equation}
W=\begin{bmatrix}
1& a &a &\dots& a\\
a&b & c&\dots&c\\
a&c&b&\dots&c\\
. & .&.&\dots&.\\
a&c&c&\dots&b
\end{bmatrix}\,.
\end{equation}
Then $\det(W-\lambda I_{n+1})=(b-c-\lambda)^{n-1}[(1-\lambda)(b-c-\lambda)-n(a^2-(1-\lambda)c)]$.
\end{lemma}
\begin{proof}
We have
\begin{equation}
\begin{array}{rl}
\det(W-\lambda I_{n+1})= \begin{vmatrix}
1-\lambda& a &a &\dots& a\\
a&b-\lambda & c&\dots&c\\
a&c&b-\lambda&\dots&c\\
. & .&.&\dots&.\\
a&c&c&\dots&b-\lambda
\end{vmatrix}
\end{array}
\end{equation}
Subtracting the $j$-th column from the last column, $j=2,\dots,n$, we get
\begin{equation}
\begin{array}{rl}
\det(W-\lambda I_{n+1})= \begin{vmatrix}
1-\lambda & 0 &0 &\dots & 0& a\\
a&b-c-\lambda & 0&\dots& 0& c\\
a&0&b-c-\lambda&\dots&0&c\\
. & .&.&\dots&.&.\\
a&0&0&\dots&b-c-\lambda & c\\
a&c-b+\lambda&c-b+\lambda&\dots&c-b+\lambda & b-\lambda\\
\end{vmatrix}
\end{array}
\end{equation}
Multiplying the first row by $\frac{c}{a}$, we obtain
\begin{equation}
\begin{array}{rl}
\det(W-\lambda I_{n+1})= \frac{a}{c} \begin{vmatrix}
\frac{(1-\lambda)c}{a} & 0 &0 &\dots & 0& c\\
a&b-c-\lambda & 0&\dots& 0& c\\
a&0&b-c-\lambda&\dots&0&c\\
. & .&.&\dots&.&.\\
a&0&0&\dots&b-c-\lambda & c\\
a&c-b+\lambda&c-b+\lambda&\dots&c-b+\lambda & b-\lambda\\
\end{vmatrix}
\end{array}
\end{equation}
Subtracting the $j$-th row from the first row, $j=2,\dots,n$, we get
\begin{equation}
\begin{array}{rl}
\det(W-\lambda I_{n+1})= \frac{a}{c} \begin{vmatrix}
\frac{(1-\lambda)c}{a} & 0 &0 &\dots & 0& c\\
a-\frac{(1-\lambda)c}{a}&b-c-\lambda & 0&\dots& 0& 0\\
a-\frac{(1-\lambda)c}{a}&0&b-c-\lambda&\dots&0&0\\
. & .&.&\dots&.&.\\
a-\frac{(1-\lambda)c}{a}&0&0&\dots&b-c-\lambda & 0\\
a-\frac{(1-\lambda)c}{a}&c-b+\lambda&c-b+\lambda&\dots&c-b+\lambda & b-c-\lambda\\
\end{vmatrix}
\end{array}
\end{equation}
Expanding the determinant with respect to the first row, we get
\begin{equation}
\begin{array}{rl}
\det(W-\lambda I_{n+1})=& \frac{a}{c}\frac{(1-\lambda)c}{a}\begin{vmatrix}
b-c-\lambda & 0&\dots& 0& 0\\
0&b-c-\lambda&\dots&0&0\\
 .&.&\dots&.&.\\
0&0&\dots&b-c-\lambda & 0\\
c-b+\lambda&c-b+\lambda&\dots&c-b+\lambda & b-c-\lambda\\
\end{vmatrix}\\[30pt]
&+ \frac{a}{c}(-1)^nc\begin{vmatrix}
a-\frac{(1-\lambda)c}{a}&b-c-\lambda & 0&\dots& 0\\
a-\frac{(1-\lambda)c}{a}&0&b-c-\lambda&\dots&0\\
. & .&.&\dots&.\\
a-\frac{(1-\lambda)c}{a}&0&0&\dots&b-c-\lambda \\
a-\frac{(1-\lambda)c}{a}&c-b+\lambda&c-b+\lambda&\dots&c-b+\lambda\\
\end{vmatrix}\\[30pt]
=& (1-\lambda)(b-c-\lambda)^n- (-1)^n a[a-\frac{(1-\lambda)c}{a}](b-c-\lambda)^{n-1} T\\
=& (1-\lambda)(b-c-\lambda)^n- (-1)^n [a^2-(1-\lambda)c](b-c-\lambda)^{n-1} T\,,
\end{array}
\end{equation}
where 
\begin{equation}
T=\begin{vmatrix}
-1&1 & 0&\dots& 0\\
-1&0&1&\dots&0\\
. & .&.&\dots&.\\
-1&0&0&\dots&1 \\
-1&-1&-1&\dots&-1 \\
\end{vmatrix}\,.
\end{equation}
To compute $T$, adding to the $j$-th column the sum of the first $j-1$ columns, $j=2,\dots,n$, we get
\begin{equation}
T=\begin{vmatrix}
-1&0 & 0&\dots& 0& 0\\
-1&-1&0&\dots&0&0\\
-1&-1&-1&\dots&0&0\\
. & .&.&\dots&.&.\\
-1&-1&-1&\dots&-1 &0\\
-1&-2&-3&\dots&-(n-1) &-n \\
\end{vmatrix}=(-1)^n n\,.
\end{equation}
From this, the result follows.
\end{proof}

\begin{lemma}\label{lem:accuracy.interior-point}
Consider semidefinite program:
\begin{equation}\label{eq:primal}
\begin{array}{rl}
\lambda^\star=\inf\limits_{X}&\tr( CX)\\
\text{s.t.}&X\in\mathbb S^{N}_+\,,\,\tr(A_i X) =b_i\,,\,i=1,\dots,M\,.
\end{array}
\end{equation}
The dual of \eqref{eq:primal} is given by:
\begin{equation}\label{eq:dual}
\begin{array}{rl}
\tau^\star=\sup\limits_{\xi,Z}&b^\top \xi\\
\text{s.t.}&\xi\in\R^M\,,\,Z\in\mathbb S^N_+\,,\,\sum_{i=1}^M\xi_iA_i +Z =C\,.
\end{array}
\end{equation}
Assume that strong duality holds for problems \eqref{eq:primal}-\eqref{eq:dual}. Let $(X^{(\varepsilon)},\xi^{(\varepsilon)},Z^{(\varepsilon)})$ be an $\varepsilon$-optimal solution returned by the interior-point methods in \cite{jiang2020faster} and \cite{augustino2023quantum}.
Then $\lambda^\star=\tau^\star$ and
\begin{equation}
\begin{cases}
\lambda^\star\le \tr(CX^{(\varepsilon)})\le  \lambda^\star+\varepsilon N\,,\\
\tau^\star- \varepsilon N\le  b^\top \xi^{(\varepsilon)} \le\tau^\star \,.
\end{cases}
\end{equation}
\end{lemma}
\begin{proof}
By the strong duality theorem, $\lambda^\star=\tau^\star$.
From \cite[Section 7.3]{augustino2023quantum}, we know that:
\begin{equation}
\begin{cases}
X^{(\varepsilon)}\succ 0\,,\,\tr( A_iX^{(\varepsilon)}) =b_i\,,\,i=1,\dots,M\,,\\[5pt]
Z^{(\varepsilon)}\succ 0\,,\,\sum_{i=1}^M\xi_i^{(\varepsilon)}A_i +Z^{(\varepsilon)} =C\,,\\[5pt]
\frac{\tr(X^{(\varepsilon)} Z^{(\varepsilon)})}{N}\le \varepsilon\,.
\end{cases}
\end{equation}
Thus, $X_\varepsilon$ is a feasible solution to \eqref{eq:primal} and $(\xi_\varepsilon,Z_\varepsilon)$ is a feasible solution to \eqref{eq:dual}.
This implies:
\begin{equation}
\begin{array}{rl}
\tau^\star=\lambda^\star\le &\tr( CX^{(\varepsilon)})= \tr( (\sum_{i=1}^M\xi^{(\varepsilon)}_iA_i +Z^{(\varepsilon)})X^{(\varepsilon)})\\[5pt]
\le& \sum_{i=1}^M\xi^{(\varepsilon)}_i\tr(A_i X) +\tr( Z_\varepsilon X_\varepsilon)\\[5pt]
\le& \sum_{i=1}^M\xi^{(\varepsilon)}_i b_i +\varepsilon N\\[5pt]
=& b^\top \xi^{(\varepsilon)} +\varepsilon N\\
\le& \tau^\star+\varepsilon N=\lambda^\star+\varepsilon N\,.
\end{array}
\end{equation}
Hence, the desired inequalities hold.
\end{proof}

\section{Block encoding constructions for $B^{(\gamma)}$}
\label{sec:block.encoding}
Recall $B^{(\gamma)}\in\mathbb{S}^{\binom{n+k}{n}}$ from Definition~\ref{def:Bgamma}, whose entries satisfy
\(
(B^{(\gamma)})_{\alpha,\beta}=1
\) iff \(\alpha+\beta=\gamma\) (componentwise) for \(\alpha,\beta\in\N_k^n\), and \(0\) otherwise.
We now give a simple ``projected–unitary'' (normalization \(1\)) block–encoding that exploits the fact that
\(\alpha\mapsto \gamma-\alpha\) is an involution on its valid domain.

\begin{proposition}[Projected–unitary block–encoding]
Fix \(\gamma\in\N_{2k}^n\).
There exists a unitary \(U_\gamma\) acting on three registers
\(\mathcal{H}_A\otimes\mathcal{H}_W\otimes\mathcal{H}_b\),
where \(A\) stores an index \(\alpha\in\N_k^n\),
\(W\) is a work register of the same shape (initialised to \(|0\rangle\)),
and \(b\) is a one–qubit flag (initialised to \(|1\rangle\)),
such that with the projector
\(
\Pi:=I_A\otimes |0\rangle\!\langle 0|_W \otimes |1\rangle\!\langle 1|_b
\)
we have
\[
\Pi\,U_\gamma\,\Pi \;=\; B^{(\gamma)}\otimes |0\rangle\!\langle 0|_W \otimes |1\rangle\!\langle 1|_b\,.
\]
Consequently, \(U_\gamma\) is a \(1\)-block–encoding of \(B^{(\gamma)}\).
Moreover, \(U_\gamma\) uses \(O(n\log k)\) reversible arithmetic gates and one controlled–SWAP.
\end{proposition}

\begin{proof}
On computational basis states \(|\alpha\rangle_A|0\rangle_W|1\rangle_b\) with \(\alpha\in\N_k^n\),
define \(U_\gamma\) by the following reversible steps:
\begin{enumerate}[(1)]
\item \textbf{Compute \(\beta\) and validity.} Reversibly compute the componentwise difference
\(\beta:=\gamma-\alpha\) into \(W\).
Perform a controlled not on $b$ controlled on all components of \(\beta\) being nonnegative and \(|\beta|\le k\) (otherwise leave \(b=0\)).
\item \textbf{Controlled swap.} Controlled on \(b=1\), swap the contents of \(A\) and \(W\).
\item \textbf{Uncompute work.} Recompute \(\gamma-\alpha\) into \(W\) (where after the swap \(A=\beta\) on the valid branch), which resets \(W\) to \(|0\rangle\).
\end{enumerate}
All three steps are reversible; (3) restores the work register, so \(U_\gamma\) is unitary.

If \(\alpha\) is \emph{invalid}, then after (1) the flag is \(b=0\), the controlled swap does not fire, and uncomputation restores \(W=0\). Thus
\(
U_\gamma|\alpha,0,1\rangle = |\alpha,0,0\rangle
\),
which is orthogonal to the range of \(\Pi\), hence contributes zero to the compressed block.

If \(\alpha\) is \emph{valid}, then \(\beta=\gamma-\alpha\in\N_k^n\).
After (1) we have \(|\alpha,\beta,1\rangle\); the controlled swap yields \(|\beta,\alpha,1\rangle\);
the uncompute step maps \(W:\ \alpha\mapsto \alpha-(\gamma-\beta)=0\).
Therefore
\(
U_\gamma|\alpha,0,1\rangle = |\beta,0,1\rangle
\)
with \(\beta=\gamma-\alpha\).
Projecting in/out with \(\Pi\) produces exactly the transition amplitude
\(\langle \beta|\,B^{(\gamma)}\,|\alpha\rangle = 1\)
when \(\alpha+\beta=\gamma\), and \(0\) otherwise.
Hence \(\Pi U_\gamma \Pi = B^{(\gamma)}\otimes |0\rangle\!\langle 0|_W \otimes |1\rangle\!\langle 1|_b\).
The gate/space bounds follow from implementing componentwise add/subtract and comparisons on \(n\) registers of size \(O(\log k)\), and one controlled–SWAP on those registers.
\end{proof}
Thus, we see that we can efficiently construct the block encodings required by our algorithm in the unconstrained case. 
\paragraph{Extension to $B^{(\gamma)}_g$.}
For constrained POPs we use the block–diagonal matrices
\[
B^{(\gamma)}_g=\diag\!\big(\sum_{\zeta} g_{i,\zeta}\,B^{(\gamma-\zeta)}_i\big)_{i=0}^{m}
\]
(see~\eqref{eq.tilde.B.gamma}).
Selecting the block \(i\) and applying a linear combination of unitaries (LCU) over the at most \(s_g\) nonzero coefficients \(\{g_{i,\zeta}\}\),
with the above \(U_{\gamma-\zeta}\) as primitives,
gives a block–encoding of \(B^{(\gamma)}_g\).
Under the rescaling assumption \(\sum_{\zeta}|g_{i,\zeta}|\le 1\) for each \(i\) (Assumption~\ref{ass:rescale.g}),
the LCU normalization remains \(1\).
The per–application cost is \(O(s_g)\) invocations of the \(U_{\gamma-\zeta}\) gadgets plus \(O(\log s_g)\) overhead for state preparation/selection. Thus, all we need is access to the an auxilliary state with amplitudes \(\{g_{i,\zeta}\}\) to produce the relevant block encodings for constrained POP.

\if{

\begin{proof}
We have
\begin{equation}
\begin{array}{rl}
\binom{n+2k}n=\frac{(n+2k)!}{n!(2k)!}=&\frac{(n+k)!(n+k+1)\dots (n+2k)}{n!k!(k+1)\dots(k+k)}\\
=&\binom{n+k}n \frac{(n+k+1)\dots (n+2k)}{(k+1)\dots(k+k)}\\
\le &\binom{n+k}n  \frac{n(k+1)\dots n(2k)}{(k+1)\dots(k+k)}=\binom{n+k}n n^k \,.
\end{array}
\end{equation}
The last inequality is based on inequality $a+b\le ab$ for $a,b\in [2,\infty)$.
\end{proof}
}\fi

\end{document}